\documentclass[12pt]{article}
\usepackage[authoryear]{natbib}

\usepackage{amssymb}
\usepackage{amsfonts}
\usepackage{amsmath}
\usepackage[nohead]{geometry}
\usepackage[singlespacing]{setspace}
\usepackage[bottom]{footmisc}
\usepackage{indentfirst}
\usepackage{endnotes}
\usepackage{graphicx}
\usepackage{rotating}
\usepackage{verbatim}
\usepackage{setspace}
\usepackage{latexsym}
\usepackage{amsthm}
\usepackage{makeidx}
\usepackage{fancyhdr}
\usepackage{multirow}
\usepackage{enumitem}
\usepackage{todonotes}
\usepackage[colorlinks=true,citecolor=blue,linkcolor=red]{hyperref}
\usepackage{color}
\usepackage{xr}

\newcommand{\eq}[1]{\begin{align}#1\end{align}}

\newcommand{\dt}{\delta}

\newcommand{\ap}{\alpha}
\newcommand{\bt}{\beta}

\newcommand{\gm}{\gamma}

\newcommand{\cov}{\mathrm{cov}}

\newcommand{\diag}{\mathrm{diag}}

\newcommand{\beq}{\begin{eqnarray*}}
\newcommand{\eeq}{\end{eqnarray*}}
\newtheorem{thm}{Theorem}[section]

\newtheorem{lem}{Lemma}[section]

\newtheorem{assum}{Assumption}

\numberwithin{equation}{section}
\theoremstyle{definition}
\newtheorem{exm}{Example}[section]
\newtheorem{remark}{Remark}[section]
\makeatletter
\def\@biblabel#1{\hspace*{-\labelsep}}
\makeatother
\DeclareMathOperator*{\argmin}{argmin}

\input{tcilatex}

\begin{document}

\title{Oracle Estimation of a Change Point in High Dimensional Quantile Regression\thanks{We would like to thank Bernd Fitzenberger, an editor, an associate editor, and three anonymous referees for helpful comments. This work was supported in part
by Promising-Pioneering Researcher Program through Seoul National University,  
by the European Research Council (ERC-2014-CoG-646917-ROMIA),  and by the Research and Scholarship Award grant of University of Maryland. }}

\author{Sokbae Lee\thanks{Department of Economics, Columbia University, 1022 International Affairs Building
420 West 118th Street, New York, NY 10027, USA; Institute for Fiscal Studies, 7 Ridgmount Street, London, WC1E 7AE, UK. Email: \texttt{sl3841@columbia.edu}.}, Yuan Liao\thanks{Department of Economics, Rutgers University, 75 Hamilton Street, New Brunswick, NJ 08901, USA. Email: \texttt{yuan.liao@rutgers.edu}.}, 
Myung Hwan Seo\thanks{Department of Economics, Seoul National University, 1 Gwanak-ro,
Gwanak-gu, Seoul, 151-742, Republic of Korea. Email: \texttt{myunghseo@snu.ac.kr}.}, and Youngki Shin\thanks{Economics Discipline Group, University of Technology Sydney, PO Box 123, Broadway NSW 2007, Australia. Email: \texttt{yshin12@gmail.com}} }

\date{15 November 2016}

\maketitle

\begin{abstract}
In this paper, we consider a high-dimensional quantile regression model where the sparsity structure may differ  between two sub-populations. We develop $\ell_1$-penalized estimators of both regression coefficients 
and the threshold parameter.   Our penalized estimators not only select covariates but also discriminate between a model with homogeneous sparsity and a model with a change point.
As a result, it is not necessary to know or pretest whether the change point is present, or where it occurs.  
Our estimator of the change point achieves an oracle property in the sense 
that  its asymptotic distribution is the same as if the unknown active sets 
of regression coefficients  were known. 
Importantly, we establish this oracle property without  a perfect covariate selection, thereby avoiding the need for  the minimum level condition on the signals of active covariates. 
Dealing with high-dimensional quantile regression with an unknown change point calls for a  new proof    technique since
 the quantile loss function is non-smooth and furthermore the corresponding objective function is non-convex with respect to the change point.
The technique developed in this paper is applicable to a general M-estimation framework with a change point,  which may be of independent interest.
  The proposed methods are then  illustrated via Monte Carlo experiments and an application to tipping  in the dynamics of racial segregation. 


\end{abstract}


\section{Introduction}

\doublespacing

In this paper, 
we consider a high-dimensional quantile regression model where the sparsity structure (e.g., identities and effects of contributing regressors) may differ  between two sub-populations, thereby allowing for a possible change point in the model.   Let $Y \in \mathbb{R}$ be a response variable,
$Q \in \mathbb{R}$ be a scalar random variable that determines a possible 
change point,  and
$X \in \mathbb{R}^{p}$ be a $p$-dimensional vector of covariates. Here, $Q$ can 
be a component of $X$, and $p$ is potentially much larger than the sample size $n$.   
Specifically,  high-dimensional quantile regression with a change point is modelled as follows:
\begin{align}\label{qr-model}
Y &= X^T\beta_0+X^T\delta_0 1\{Q>\tau_0\}+U,
\end{align}
where $(\beta_0^T, \delta_0^T, \tau_0)$ is a vector of unknown parameters and the regression error $U$ satisfies 
$\mathbb{P}(U\leq 0|X,Q)=\gamma$  for some known $\gamma\in(0,1)$.
 Unlike mean regression, quantile regression analyzes the effects of active regressors on different parts of the conditional distribution of a response variable. 
Therefore, it  allows the sparsity patterns to differ at different quantiles and also handles heterogeneity due to either heteroskedastic variance
or other forms of non-location-scale covariate effects.  By taking into account a possible change point  in the model,  we provide a more realistic picture of the sparsity patterns. 
For instance,  when analyzing high-dimensional gene expression data,    the identities of contributing genes may depend on the environmental or demographical variables (e.g., exposed temperature, age or weights).

Our paper is closely related to the  literature on models with unknown change points (e.g., \cite{tong1990non}, \cite{chan1993consistency}, \cite{Hansen:1996, hansen2000sample}, \cite{Pons2003}, \cite{kosorok2007}, \cite{Seijo:Sen:11a,Seijo:Sen:11b} and \cite{Li:Ling:12} among many others).   
Recent papers on change points under high-dimensional setups include  \cite{enikeeva2013high, chan2013group}, \cite{Frick-et-al:14}, \cite{cho2012multiple},
\cite{Chan:et-al:2016},
\cite{Callot:et-al:2016},
 and \cite{lee2012lasso} among others; however, none of these papers  consider a change point in high-dimensional quantile regression.  
The literature on high-dimensional quantile regression  includes \cite{BC11}, \cite{Bradic2011}, \cite{Wang:Wu:Li:2012},
\cite{Wang:2013},  and \cite{FFB} among others. All the aforementioned papers on quantile regression are under the homogeneous sparsity framework (equivalently, assuming that $\delta_0=0$ in \eqref{qr-model}).  \cite{Ciuperca:13} considers penalized estimation of a quantile regression model with breaks, but  the corresponding analysis is restricted to the case when $p$ is small.


 In this paper, we consider estimating regression coefficients 
$\alpha_0 \equiv (\beta_0^T,\delta_0^T)^T$ as well as  the threshold parameter $\tau_0$ and selecting the contributing regressors based on $\ell_1$-penalized estimators.   One of the strengths of our proposed procedure  is that it  does not require to know or pretest whether $\delta_0=0$ or not, that is, whether the population's sparsity structure and covariate effects are invariant or not. In other words, we  do not need to know whether the  threshold $\tau_0$  is present in the model.

   For a sparse vector $v\in\mathbb{R}^p$, we denote the active set of $v$ as $J(v) \equiv \{j: v_j\neq0\}$.
One of the main contributions of this paper is that our proposed estimator of $\tau_0$ achieves an \emph{oracle property} in the sense 
that its asymptotic distribution 
is the same as if the unknown active sets $J(\beta_{0})$ and $J(\delta_0)$ were known. 
Importantly, we establish this oracle property without assuming  a perfect covariate selection, thereby avoiding the need for  the minimum level condition on the signals of active covariates.

The proposed estimation method in this paper consists of three main steps: in the first step, we obtain the initial estimators of $\alpha_0$
and $\tau_0$, whose rates of convergence may be suboptimal; in the second step, we re-estimate $\tau_0$ to obtain 
an improved estimator of $\tau_0$ that converges at the rate of $O_P(n^{-1})$ and  achieves the oracle property mentioned above;
in the third step, using the second step estimator of $\tau_0$, we update the estimator of $\alpha_0$. In particular, we propose two alternative estimators of $\alpha_0$, depending on the purpose of  estimation (prediction vs. variable selection). 

The most closely related work is \cite{lee2012lasso}. However, there are several important differences:
first, \cite{lee2012lasso} consider a high-dimensional mean regression model with a homoskedastic normal error and with deterministic  covariates; second, their method consists of one-step least squares estimation with an $ \ell _1 $ penalty; third, they derive  
non-asymptotic oracle inequalities similar to those in \cite{Bickeletal} but do not provide any distributional result on the estimator of the change point.
Compared to \cite{lee2012lasso}, dealing with high-dimensional quantile regression with an unknown change point calls for a  new  proof   technique since
 the quantile loss function is different from the least squares objective function and is non-smooth. 
In addition, we allow for heteroskesdastic and non-normal regression errors and  stochastic  covariates. These changes coupled with the fact that   
the quantile regression objective function is non-convex with respect to the threshold parameter $\tau_0$
raise new challenges. It requires careful derivation and multiple estimation steps to establish the oracle property for the estimator of $\tau_0$ and also to obtain desirable properties of the estimator of $\alpha_0$. 
The technique developed in this paper is applicable to a general M-estimation framework with a change point, 
which may be of independent interest.

One particular application of \eqref{qr-model} comes from tipping  in the racial segregation in social sciences \citep[see, e.g.][]{card2008tipping}. The empirical question addressed in \citet{card2008tipping} is whether and the extent to which 
the neighborhood's white population decreases substantially when the minority share
in the area exceeds a tipping point (or change point).
In Section  \ref{sec:data-example}, we use the US Census tract  dataset constructed by  \citet{card2008tipping} and confirm that the tipping exists in the neighborhoods of Chicago.

The remainder of the paper is organized as follows. 
  Section \ref{sec:model-estimator}  provides an informal description of  our  estimation methodology. 
  In Section \ref{sec:Lasso-theory-consistency}, we derive the consistency of the estimators in terms of the excess risk.  Further asymptotic properties of the proposed estimators are given in Sections 
 \ref{sec:Lasso-theory-case1} and \ref{sec:Lasso-theory-case2}.
In Section \ref{sec:MC}, we   present the results of extensive Monte Carlo experiments.
Section \ref{sec:data-example} illustrates the usefulness of our method by applying it to tipping  in the racial segregation.
Section \ref{sec:conclusion} concludes and Appendix \ref{sect:tau:ci-algorithm} describes in detail regarding how to 
construct the confidence interval for $\tau_0$.
In Appendix \ref{sec:Lasso-theory-assump}, we provide a set of regularity assumptions  to derive asymptotic properties of the proposed estimators in 
Sections   \ref{sec:Lasso-theory-case1} and \ref{sec:Lasso-theory-case2}.
Online supplements are comprised of 6 appendices for all the proofs as well as additional theoretical and numerical results that are left out for the brevity of the paper.

\textbf{Notation.}
Throughout the paper, we use $|v|_q$ for the $\ell_q$ norm for a vector $v$ with $q=0,1,2$.  We use $|v|_\infty$ to
denote the sup norm. For two sequences of positive real numbers $a_n$ and  $b_n$, we write $a_n\ll b_n$ and equivalently $b_n\gg a_n$ if $a_n=o(b_n)$.
If there exists a positive finite constant $c$ such that $a_n = c \cdot b_n$, then we write $a_n \propto b_n$.
Let $\lambda_{\min}(A)$ denote the minimum eigenvalue of a matrix $A$.
We use w.p.a.1 to mean ``with probability approaching one.''
We write $\theta_0 \equiv \beta_0+\delta_0$.
For a $2p$ dimensional vector $\alpha$, let $\alpha_J$ and $\alpha_{J^c}$ denote its subvectors formed by indices in $J(\alpha_0)$ and $\{1,...,2p\}\setminus J(\alpha_0)$, respectively. 
Likewise, let $X_{J}(\tau)$ denote the subvector of $X(\tau)\equiv (X^{T},X^{T}1\{Q>\tau \})^{T}$ whose indices are in $J(\alpha_0)$. 
The true parameter vectors $\beta_0$, $\delta_0$ and $\theta_0$ (except $\tau_0$) are
implicitly indexed by the sample size $n$, and we allow that the dimensions of 
$J(\beta_0)$,  $J(\delta_0)$ and $J(\theta_0)$
 can go to infinity as $n \rightarrow \infty$.
For simplicity, we omit their dependence on $n$ in our notation. 
We also use the terms  `change point' and `threshold' interchangeably throughout the paper.

\section{Estimators}\label{sec:model-estimator}

\subsection{Definitions}

In this section, we describe our estimation method.
We take the check function approach  of \cite{Koenker:Bassett:1978}.
Let  $\rho(t_1,t_2) \equiv (t_1-t_2)(\gamma-1\{t_1-t_2\leq 0\})$ denote the loss function for quantile regression.
Let $\mathcal{A}$ and $\mathcal{T}$ denote the parameter spaces for $\alpha_0 \equiv (\beta_0^T, \delta_0^T)^T$ and $\tau_0$, respectively. 
For each $\alpha \equiv (\beta,\delta)\in\mathcal{A}$ and $\tau\in\mathcal{T}$,  
we write  $X^T\beta+X^T%
\delta1\{Q>\tau\}=X(\tau)^T\alpha$ with the shorthand notation that 
$X(\tau) \equiv (X^{T},X^{T}1\{Q>\tau \})^{T}$.
We suppose that 
 the vector of true parameters is defined as the   minimizer of the expected loss:
 \begin{equation}\label{eq2.3add-QR}
(\alpha_0,\tau_0)=\argmin_{\alpha \in \mathcal{A},\tau \in \mathcal{T}} 
\mathbb{E} \left[ \rho(Y,X(\tau)^T\alpha) \right].
\end{equation} 
By construction, $\tau_0$ is not unique when $\delta_0=0$. However, if $\delta_0 = 0$, then the model reduces to the linear quantile regression model in which $\beta_0$ is identifiable under the standard assumptions.
In Appendix  \ref{app:iden}, we provide sufficient conditions under which $\alpha_0$ and $\tau_0$ are identified when $\delta_0 \neq 0$.

 Suppose we observe independent and identically distributed 
 samples $\{Y_i, X_i, Q_i\}_{i\leq n}$. Let $X_{i}(\tau )$ and $X_{ij}\left( \tau
\right) $ denote the $i$-th realization of $X(\tau )$ and $j$-th element of $%
X_{i}\left( \tau \right) ,$ respectively, $i=1,\ldots,n$ and $j=1,\ldots,2p$, so that 
$X_{ij}(\tau) \equiv X_{ij}$ if $j\leq p$ and $X_{ij}(\tau) \equiv X_{i,j-p}1\{Q_i>\tau\}$ otherwise.
 Define
$$
R_n(\alpha,\tau)\equiv
\frac{1}{n}%
\sum_{i=1}^{n}\rho (Y_{i},X_{i}(\tau )^{T}\alpha )= \frac{1}{n}\sum_{i=1}^{n}\rho (Y_{i},X_{i}^T\beta+ X_i^T\delta  1\{Q_i>\tau\} ).
$$
In addition, let $D_{j}(\tau ) \equiv \{ n^{-1} \sum_{i=1}^{n}X_{ij}(\tau )^{2} \}^{1/2}$, $j=1,\ldots,2p$. 

We describe the main steps of our $\ell_1$-penalized estimation method.  
For some tuning parameter $\kappa_{n}$, define:
\begin{eqnarray}\label{eq2.2add}
\textbf{Step 1: } (\breve{\alpha},\breve{\tau}) 
 &=& \text{argmin}_{\alpha \in \mathcal{A},\tau \in \mathcal{T}} R_n(\alpha,\tau)+\kappa_{n} \sum_{j=1}^{2p} D_{j}(\tau )| \alpha _{j}|.
\end{eqnarray}%
This step produces an initial estimator $(\breve\alpha,\breve\tau)$. 
The tuning parameter $\kappa_n$ is required to satisfy
\begin{equation}\label{kappa_n_rate}
\kappa_n \propto (\log p)(\log n)\sqrt{\frac{\log p}{n}}.
\end{equation}
Note that we take $\kappa_n$ that converges to zero at a rate slower than the standard 
$(\log p/n)^{1/2} $ rate in the literature. This modified rate of $\kappa_n$ is useful in our context
to deal with an unknown $\tau_0$. 
A data-dependent method of choosing $\kappa _{n}$ is discussed in 
Section \ref{tune-para}.

\begin{remark}
Define $d_j\equiv( \frac{1}{n%
}\sum_{i=1}^{n}X_{ij}^{2})^{1/2}$ and $d_j(\tau)\equiv( \frac{1}{n%
}\sum_{i=1}^{n}X_{ij}^{2}1\{Q_i>\tau\})^{1/2}$.
Note that 
$\sum_{j=1}^{2p} D_{j}(\tau )| \alpha _{j}| = \sum_{j=1}^pd_j|\beta_j|+\sum_{j=1}^pd_j(\tau)|\delta_j|$, so that 
the weight $D_j(\tau)$ adequately balances the regressors;  the weight $d_j$ regarding $|\beta_j|$ does not depend on $\tau$, while the weight $d_j(\tau)$ with respect to  $|\delta_j|$ does, which takes into account the effect of the threshold $\tau$ on the parameter change $\delta$. 
\end{remark}

\begin{remark}
The computational cost in Step 1 is the multiple of grid points to the computational time of estimating the  linear quantile model with an $\ell_1$ penalty, which is  solvable in polynomial time (see e.g. \cite{BC11} and \cite{Koenker:Mizera:14} among others). 
\end{remark}

The main purpose of the first step is to obtain an initial estimator of $\alpha_0$. 
The achieved convergence rates  of this step might be suboptimal due to the uniform control of the score functions over the space 
$\mathcal{T}$ of the unknown $\tau_0$.

In the second step, 
we introduce our improved estimator of the change point $\tau_0$. It does not use a penalty term, while using the first step estimator of $\alpha_0$. Define:
\begin{align}
\textbf{Step 2: } 
\widehat{\tau}&=
\operatorname*{argmin}_{\tau \in \mathcal{T}}R_n(\breve\alpha,\tau), \label{step2-tau-est}
 \end{align}%
where $\breve{\alpha}$ is the first step estimator of $\alpha_0$ in \eqref{eq2.2add}.
 In Section \ref{sec:Lasso-theory-case1}, 
we show that when $\tau_0$ is identifiable, $\widehat{\tau}$ is consistent for $\tau_0$ at a rate of $n^{-1}$. Furthermore, we
obtain the limiting distribution of $n(\widehat{\tau} - \tau_0)$, and establish conditions under which
its  asymptotic distribution is the same  as if the true $\alpha_0$ were known, without 
a perfect model selection on $\alpha_0$, nor assuming the minimum signal condition on the nonzero elements of $\alpha_0$.

In the third step, we update the Lasso estimator of $\alpha_0$ using a different value of the penalization tuning parameter and the second step estimator of $\tau_0$. 
In particular, we recommend two different estimators of $\alpha_0$\,: one for the prediction   and the other for the variable selection,
serving for different purposes of practitioners.
For two different tuning parameters $\omega_n$ and $\mu_n$ whose rates will be specified later by 
\eqref{omega-rate} and \eqref{e3.2},
define:
\begin{align}
\textbf{Step 3a} & \textbf{ (for prediction): } \notag \\ 
  \widehat{\alpha}
 &= \text{argmin}_{\alpha \in \mathcal{A}} R_n(\alpha,\widehat\tau) + \omega_{n} \sum_{j=1}^{2p} D_j(\widehat\tau)|\alpha_j|,  \label{step3-alpha-est}\\
\textbf{Step 3b}  & \textbf{ (for variable selection): }  \notag \\
\widetilde{\alpha}
&= \text{argmin}_{\alpha \in \mathcal{A}} R_n(\alpha,\widehat\tau) + \mu_n\sum_{j=1}^{2p}w_jD_j(\widehat\tau)|\alpha_j|, \label{step3-alpha-var-sel}
 \end{align}%
where $\widehat{\tau}$ is the second step estimator of $\tau_0$ in \eqref{step2-tau-est},
and  the ``signal-adaptive" weight $w_j$ in \eqref{step3-alpha-var-sel}, motivated by the local linear approximation  of the SCAD penalties \citep{Fan01, zouli}, is calculated based on the Step 3a estimator $\widehat\alpha$  from (\ref{step3-alpha-est}):
\begin{equation*}
w_j\equiv%
\begin{cases}
1, & |{\widehat{\alpha}}_j|<\mu_n \cr 
0, & |{\widehat{\alpha}}_j|>a\mu_n\cr 
\frac{a\mu_n-|{\widehat{\alpha}}_j|}{\mu_n(a-1)} & \mu_n\leq|{\widehat{\alpha}}_j|\leq a\mu_n.
\end{cases}%
\end{equation*}
Here $a>1$ is some prescribed constant, and $a=3.7$ is often used in the literature.
We take this as our choice of $a$.

\begin{remark}
For $\widehat{\alpha}$ in \eqref{step3-alpha-est}, we set $\omega_n$ to converge to zero at a rate of $(\log (p \vee n)/n)^{1/2}$:
\begin{align}\label{omega-rate}
\omega_{n} \propto  \sqrt{\frac{\log (p \vee n)}{n}} \;,
\end{align}
which is a more standard rate compared to $\kappa_n$ in  \eqref{kappa_n_rate}). 
Therefore, the  estimator $\widehat{\alpha}$ 
converges in probability to $\alpha_0$ faster than $\breve{\alpha}$.  
In addition, $\mu_n$ in \eqref{step3-alpha-var-sel} is chosen to be slightly larger than $\omega_n$
for the purpose of the variable selection. 
A data-dependent method of choosing $\omega _{n}$ as well as $\mu_n$ is  discussed in 
Section \ref{tune-para}.
In Sections \ref{sec:Lasso-theory-case1} and \ref{sec:Lasso-theory-case2}, we establish conditions under which $\widehat{\alpha}$  achieves the (minimax) optimal rate of convergence in probability for $\alpha_0$ regardless of the identifiability of $\tau_0$. 
\end{remark}

\begin{remark}
It is well known in linear models without the presence of an unknown $\tau_0$  (see, e.g. \cite{bulmann}) that the Lasso estimator may not perform well for the purpose of the variable selection. The estimator $\widetilde\alpha$ defined in Step 3b uses an entry-adaptive weight $w_j$ that corrects the shrinkage bias, and possesses similar merits of the asymptotic unbiasedness of the SCAD penalty. 
Therefore, we recommend $\widehat{\alpha}$ for the prediction; while suggesting $\widetilde\alpha$ for the variable selection.
\end{remark}

\begin{remark}
Note that the objective function is non-convex with respect to $\tau$ in the first 
and second steps. 
However,  the proposed estimators can be calculated efficiently using existing algorithms, and we 
describe the computation algorithms in Section \ref{tune-para}. 
\end{remark}

\begin{remark}
 Step 2  can be repeated using the updated estimator of $\alpha_0$ in Step 3. Analogously, Step 3 can be iterated after that. 
This would give asymptotically equivalent estimators but might improve the finite-sample performance especially when $p$ is very large. 
Repeating Step 2 might be useful especially when $\breve{\delta} = 0$ in the first step. In this case, there is no unique $\widehat{\tau}$ in Step 2. So, we skip  the second step 
by setting $\widehat{\tau} = \breve{\tau}$ and  move to the third step directly.
If a preferred estimator  of $\delta_0$ in the third step  (either $\widehat{\delta}$ or $\widetilde{\delta}$), depending on the estimation purpose, is different from zero,
we could go back to Step 2 and re-estimate $\tau_0$.   
If the third step estimator of $\delta_0$ is also zero, then we conclude that there is no change point and disregard the first-step estimator $\breve{\tau}$ since $\tau_0$ is not identifiable in this case. 
\end{remark}

\subsection{Comparison of Estimators  in Step 3 }\label{sec:comparison-step3}

Step 3 defines two estimators for $\alpha_0$. In this subsection we briefly explain their major differences and purposes.    Step 3b  is particularly useful when the variable selection consistency is the main objective, yet it often requires the minimum signal condition ($\min_{\alpha_{0j}\neq0}|\alpha_{0j}|$ is well separated from zero). In contrast,  Step 3a does not require the minimum signal condition,  and is recommended for prediction purposes.  More specifically:

\begin{enumerate}
\item If the minimum signal condition (\ref{e6.1})  indeed holds,   a perfect variable selection (variable selection consistency) is possible. Indeed, thanks to the signal-adaptive weights, the estimator of  Step 3b introduces little shrinkage biases. As a result, we  show in Theorem \ref{th3.1}  that under very mild conditions, this estimator 
 achieves  the variable selection consistency. In contrast,  Step 3a does not use  signal-adaptive weights. In order to achieve the variable selection consistency, it has to rely on much stronger conditions on the design matrix (i.e., the \textit{irrepresentable condition} of \cite{zhao2006model}) so as  to  ``balance out" the effects of shrinkage biases, and is less adaptive to correlated designs.

\item In the presence of  the minimum signal condition,   not only does Step 3b achieve the variable selection consistency,   it also has  a better rate of convergence than Step 3a (Theorem \ref{th3.1}). The faster rate of convergence is built on the variable selection consistency, and is still a consequence of  the signal-adaptive weights. Intuitively,  nonzero elements of $\alpha_0$ are easier to identify and estimate when the signal is strong. Such a phenomenon has been observed in the literature;    see, e.g., \cite{FL11} and many papers on variable selections using ``folded-concave" penalizations.

\item  In the absence of the minimum signal condition, neither method can achieve variable selection consistency. However, it  is not a requirement for the prediction purpose. In this case, we recommend the estimator of Step 3a, because   it achieves a fast (minimax) rate of convergence (Theorem \ref{th2.3-improved}), which is useful for predictions.

\item Finally, we show in Theorem \ref{th3.1-without} that  without  the minimum signal condition,  Step 3b, with the signal-adaptive weights,  does not perform badly, in the sense that it still results in    estimation and prediction consistency. However, the rate of convergence is slower than  that   of Step 3a.
\end{enumerate}

\subsection{Tuning parameter selection}\label{tune-para}

In this subsection, we provide details on how to choose tuning parameters in applications.
Recall that our procedure  involves three tuning parameters in the penalization: (1) $\kappa_n$   in Step 1 ought  to  dominate the score function uniformly over the range of $\tau$, and hence should be slightly larger than the others; (2)  $\omega_n$ is used in Step 3a for the prediction, and (3) $\mu_n$  in  Step 3b  for the variable selection  should be larger than $\omega_n$. Note that the tuning parameters in both Steps 3a and 3b are similar to those of the existing literature since the change point $\widehat\tau$ has been  estimated. 

We build on the data-dependent selection method in \cite{BC11}. Define
\begin{align}
\Lambda (\tau) := \max_{1\le j \le 2p} \left\vert \frac{1}{n}\sum_{i=1}^n \frac{X_{ij}(\tau) \left(\gamma - 1\{U_i \le \gamma \}\right)}{D_j(\tau)}\right\vert,
\end{align}
where $U_i$ is simulated from the \emph{i.i.d.\ }uniform distribution on the interval $[0,1]$; $\gamma$ is the quantile of interest (e.g. $\gamma = 0.5$ for median regression).  Note that $\Lambda(\tau)$ is a stochastic process indexed by $\tau$. Let $\overline{\Lambda}_{1-\epsilon^\ast}$ be the $(1-\epsilon^\ast)$-quantile of $\sup_{\tau \in \mathcal{T}} \Lambda(\tau)$, where 
$\epsilon^\ast$ is a small positive constant that will be selected by a user.
 Then, we select the tuning parameter in Step 1 by
$
\kappa_n = c_1 \cdot \overline{\Lambda}_{1-\epsilon^\ast}.
$
Similarly, let $\Lambda_{1-\epsilon^\ast}(\widehat{\tau})$ be the $(1-\epsilon^\ast)$-quantile of $\Lambda(\widehat{\tau})$, where $\widehat{\tau}$ is chosen in Step 2. We select $\omega_n$ and $\mu_n$ in Step 3 by
$
\omega_n =  c_1 \cdot \Lambda_{1-\epsilon^\ast}(\widehat{\tau})
$
and
$
\mu_n =      c_2 \cdot \omega_n.
$
It is also necessary to choose $\mathcal{T}$ in applications. In our Monte Carlo experiments in Section 
\ref{sec:MC}, 
we take $\mathcal{T}$ to be  the interval from the 15th percentile to the 85th percentile of the empirical distribution of the threshold variable $Q_i$. 
For example, \cite{Hansen:1996} employed the same range   in his application to U.S. GNP dynamics.

Based on the suggestions of \cite{BC11}
and 
some preliminary simulations, we choose to set  $c_1=1.1$, $c_2=\log\log n$, and $\epsilon^\ast=0.1$.
In addition, recall that we  set $a=3.7$ when calculating  the SCAD weights $w_j$ in Step 3b following the convention in the literature (e.g.\ \citet{Fan01} and \citet{loh2013regularized}).   In Step 1, we first solve the lasso problem for $\alpha$ given each grid point of $\tau\in \mathcal{T}$. Then, we choose $\breve{\tau}$ and the corresponding $\breve{\alpha}(\breve{\tau})$ that minimize the objective function. Step 2 can be solved simply by the grid search.  Step 3 is  a  standard lasso quantile regression estimation  given $\widehat{\tau}$,  whose numerical implementation is well established.
We use the \texttt{rq()} function of the R `quantreg' package with the \texttt{method = "lasso"} in each implementation of the standard lasso quantile regression estimation \citep{quantreg}.

\section{Risk Consistency}\label{sec:Lasso-theory-consistency}

Given the  loss function $\rho(t_1,t_2) \equiv (t_1-t_2)(\gamma-1\{t_1-t_2\leq 0\})$  for the quantile regression model,  define the \textit{excess risk} to be 
\begin{align}\label{def:excess-risk}
R(\alpha,\tau)\equiv\mathbb{E} \rho(Y, X(\tau)^T\alpha)-
\mathbb{E} \rho(Y,X(\tau_0)^T\alpha_0).
\end{align}
By the definition of $(\alpha_0,\tau_0)$ in \eqref{eq2.3add-QR}, we have that
$R(\alpha,\tau) \geq 0$ for any $\alpha \in \mathcal{A}$ and $\tau \in \mathcal{T}$.
What we mean by the ``risk consistency'' here is that the excess risk converges in probability to zero for the proposed estimators. The other asymptotic properties of the proposed estimators will be presented in Sections  \ref{sec:Lasso-theory-case1} and \ref{sec:Lasso-theory-case2}.

In this section, we begin by stating regularity conditions that are needed to develop our first theoretical result.
Recall that $X_{ij}$ denotes the $j$th element of $X_i$.



\begin{assum}[Setting]
\label{a:setting}
\begin{enumerate}[label=(\roman*)]
\item\label{a:setting:itm1}
 The data $\{(Y_{i},X_{i},Q_{i})\}_{i=1}^{n}$ are independent
and identically distributed. Furthermore,   for all $j$ and every integer $m\geq 1$, there is a constant  $K_{1}<\infty $ such that  $\mathbb{E}\left\vert X_{ij}\right\vert
^{m}\leq \frac{m!}{2}K_{1}^{m-2}$.
\item\label{a:dist-Q:itm1} 
$\mathbb{P}(\tau_1 < Q \leq \tau_2) \leq K_2 (\tau_2 - \tau_1)$
for any $\tau_1 < \tau_2$ and some constant $K_2 < \infty$.
\item\label{a:setting:itm2}  
$\alpha_0 \in \mathcal{A} \equiv \left\{ \alpha :\left\vert \alpha \right\vert
_{\infty }\leq M_1 \right\} $ for some constant $M_1 <\infty $, and $\tau_0 \in \mathcal{T} \equiv
\left[ \underline{\tau },\overline{\tau }\right] $. Furthermore,  the probability of 
$\left\{ Q < \underline{\tau }\right\} $ and that of $\left\{ Q > \overline{\tau 
}\right\} $ are strictly positive,  and 
$$
\sup_{j \leq p}\sup_{\tau \in \mathcal{T} } \mathbb{E}[ X_{ij}^2 | Q = \tau ] < \infty.
$$ 
\item\label{a:setting:itm3} 
 There exist universal constants $\underline{D}>0$ and $\overline{D}>0$ such
that w.p.a.1,
\begin{equation*}
0 < \underline{D}\leq \min_{j\leq 2p}\inf_{\tau \in \mathcal{T}}D_{j}(\tau )\leq
\max_{j\leq 2p}\sup_{\tau \in \mathcal{T}}D_{j}(\tau )\leq \overline{D} < \infty.
\end{equation*}%
\item\label{a:threshold:itm1}
$\mathbb{E}\left[ \left( X^{T}\delta _{0}\right) ^{2}|Q=\tau \right]
\leq  M_2 |\delta_0|_2^2$ for all $\tau \in \mathcal{T}$ and for some constant $M_2$ satisfying $0 < M_2 <\infty $.
\end{enumerate}
\end{assum}

In addition to the random sampling assumption, condition \ref{a:setting:itm1}    imposes mild moment restrictions on $X$.
Condition \ref{a:dist-Q:itm1}   imposes a weak restriction that the probability that $Q \in (\tau_1, \tau_2]$ is bounded by a constant times
$(\tau_2 - \tau_1)$. 
Condition \ref{a:setting:itm2} assumes that the parameter space is compact  and that the support of $Q$ is strictly larger than $\mathcal{T}$. These conditions are  standard   in the literature on change-point and threshold models (e.g., \cite{Seijo:Sen:11a,Seijo:Sen:11b}). 
 Condition \ref{a:setting:itm2} also assumes that the conditional expectation of $\mathbb{E}[ X_{ij}^2 | Q = \cdot ]$ is bounded on $\mathcal{T}$ uniformly in $j$.
Condition \ref{a:setting:itm3} requires that each regressor be of the same magnitude uniformly
over the threshold $\tau$. As the data-dependent weights $D_j(\tau)$ are the
sample second moments of the regressors, it is not stringent to assume them to
 be bounded away from both zero and infinity. 
Condition \ref{a:threshold:itm1} puts some weak upper bound on $\mathbb{E}[ \left( X^{T}\delta _{0}\right) ^{2}|Q=\tau ]$ for all $\tau \in \mathcal{T}$ when $\delta_0 \neq 0$. 
A simple sufficient condition for condition \ref{a:threshold:itm1} is that  the eigenvalues of $\mathbb{E}[ X_{J(\delta_0)}X_{J(\delta_0)}^T|Q=\tau ]$ are bounded uniformly in $\tau$, where $X_{J(\delta_0)}$ denotes the subvector of $X$ corresponding to the nonzero components of $\delta_0$.


Throughout the paper, we let $s \equiv |J(\alpha_0)|_0$, namely the cardinality of $J(\alpha_0)$. We allow that $s \rightarrow \infty$  as $n \rightarrow \infty$ and will give precise regularity conditions regarding its growth rates. 
The  following theorem is concerned about the convergence of $%
R(\breve{\alpha},\breve{\tau})$ with the first step estimator.

\begin{thm}[Risk Consistency]
\label{l2.1}
Let Assumption \ref{a:setting} hold. 
Suppose that 
the tuning parameter $\kappa_n$ satisfies \eqref{kappa_n_rate}. 
 Then, 
$
R(\breve{\alpha },\breve{\tau})=O_P\left( \kappa_{n}s\right) .
$
\end{thm}

Note that Theorem \ref{l2.1} holds regardless of the identifiability of  $\tau_0$ (that is, whether $\delta_0 = 0$ or not). In addition, the rate $O_P(\kappa_ns)$ is achieved regardless of whether $\kappa_ns$ converges, and we have  the risk consistency  if $\kappa_n s \rightarrow 0$ as $n \rightarrow \infty$. 
The restriction on $s$ is slightly stronger than that of the standard result 
$s = o ( \sqrt{n/\log p} )$
in the literature
for the M-estimation (see, e.g.  \cite{geer} and Chapter 6.6 of \cite{bulmann})
since the objective function $\rho(Y, X(\tau)^T\alpha)$ is non-convex in $\tau$,  due to the unknown change-point.

\begin{remark}
The extra logarithmic factor $(\log p)(\log n)$ in the definition of $\kappa_n$ (see \eqref{kappa_n_rate})
 is due to the existence of the unknown and possibly non-identifiable threshold parameter $\tau_0$. 
In fact, an inspection of the proof of 
Theorem \ref{l2.1} reveals that it suffices to assume that  $\kappa_n$ satisfies 
$\kappa_n \gg \log_2 (p/s) [\log (np)/n ]^{1/2}$. 
The term $\log_2 (p/s)$ and the additional $(\log n)^{1/2}$ term inside the brackets are needed to  establish the stochastic equicontinuity
of the empirical process
\begin{equation*}
\nu _{n}\left( \alpha ,\tau \right) \equiv\frac{1}{n}\sum_{i=1}^{n}\left[ \rho
\left( Y_{i},X_{i}\left( \tau \right) ^{T}\alpha \right) -
\mathbb{E} \rho \left(Y,X\left( \tau \right) ^{T}\alpha \right) \right]
\end{equation*}
uniformly over $(\alpha,\tau) \in \mathcal{A} \times \mathcal{T}$. 
\end{remark}

In Appendix \ref{app:add:theory2}, we show  that an improved rate of convergence, $O_P\left( \omega_{n}s\right)$, is possible for the excess risk by taking the second and third steps of estimation. 

\section{Asymptotic Properties: Case I. $\delta_0\neq0$}\label{sec:Lasso-theory-case1}

Sections   \ref{sec:Lasso-theory-case1} and \ref{sec:Lasso-theory-case2} provide asymptotic properties of the proposed estimators.  
In Appendix \ref{sec:Lasso-theory-assump}, we list a set of assumptions that are needed to derive these properties, in addition to Assumption \ref{a:setting}. 
We first establish the consistency of $\breve{\protect\tau}$ for $\tau_0$.

\begin{thm}[Consistency of $\breve{\protect\tau}$]
\label{th2.2} 
Let
Assumptions \ref{a:setting}, \ref{ass3.4-a}, \ref{a:dist-Q-add}, and \ref{a:moment} hold. Furthermore, assume that $\kappa_n s = o(1)$.
Then, 
$\breve{\tau}\overset{P}{\longrightarrow}\tau_{0}$.
\end{thm}

The following theorem presents the rates of convergence for the first step estimators of $\alpha_0$ and $\tau_0$. 
Recall that $\kappa_n$  is the first-step penalization  tuning parameter that satisfies \eqref{kappa_n_rate}.

\begin{thm} [Rates of Convergence When $\delta_0 \neq 0$]
\label{th2.3} Suppose that $\kappa_n s^2\log p=o(1)$.
 Then 
under Assumptions  \ref{a:setting}-\ref{a:moment},  
 we have:
\begin{equation*}
|\breve{\alpha}-\alpha _{0}|_{1}=O_P(\kappa _{n}s),
\;
R(\breve{\alpha},\breve{\tau})=O_P(\kappa _{n}^{2}s),
\ \ \text{ and } \ \ 
|\breve{\tau}-\tau _{0}|=O_P (\kappa _{n}^{2}s).
\end{equation*}
\end{thm}

In Theorem \ref{l2.1}, we have that $
R(\breve{\alpha },\breve{\tau})=O_P\left( \kappa_{n}s\right).
$
The improved rate of convergence for $
R(\breve{\alpha },\breve{\tau})$ in Theorem \ref{th2.3} is due to additional assumptions (in particular, compatibility conditions in Assumption \ref{ass2.7} among others). 
  It is worth noting that  $\breve\tau$ converges to $\tau_0$  faster than the standard parametric rate of $n^{-1/2}$,  as long as  $s^2(\log p)^6(\log n)^4  =o(n)$. 
The main reason  for  such  \textit{super-consistency} is that the objective function behaves locally linearly around $\tau_0$ with a kink at $\tau_0$, unlike in the regular estimation problem where the objective function behaves  locally quadratically around the true parameter value. 
Moreover,   the achieved convergence rate for $\breve\alpha$ is nearly minimax optimal, with an additional factor $(\log p)(\log n)$ compared to the rate of regular Lasso estimation (e.g., \cite{Bickeletal, Rasku09}). This factor arises due to the unknown change-point $\tau_0.$  
We will improve the rates of convergence for both $\tau_0$ and $\alpha_0$ further by taking the second and third steps of estimation. 
 
 
Recall that the second-step estimator of $\tau_0$ is defined as 
\begin{align*}
\widehat{\tau}&=
\operatorname*{argmin}_{\tau \in \mathcal{T}}R_n(\breve\alpha,\tau), 
 \end{align*}%
where $\breve{\alpha}$ is the first step estimator of $\alpha_0$ in \eqref{eq2.2add}.
Consider an oracle case for which $\alpha$ in $R_n(\alpha,\tau)$ is fixed at $\alpha_{0}$. Let
$R_{n}^{\ast}\left(  \tau\right)  =R_{n} \left(  \alpha_{0},\tau\right)
$ and
\[
\widetilde{\tau}=\operatorname*{argmin}_{\tau \in \mathcal{T}} R_{n}^{\ast}\left(  \tau\right)  .
\]


We now give one of the main results of this paper.

\begin{thm}[Oracle Estimation of $\tau_0$]\label{thm:tau-2nd}
Let Assumptions  \ref{a:setting}-\ref{a:moment} hold. Furthermore, 
suppose that $\kappa_n s^2\log p=o(1)$. Then,  we have that 
\[
\widehat{\tau} - \widetilde{\tau} = o_{P}\left(  n^{-1}\right).
\]
Furthermore, $n\left( \widehat{\tau}-\tau _{0}\right) $ converges in distribution to 
the smallest
minimizer of a compound Poisson process, which is given by%
\begin{equation*}
M\left( h\right) \equiv \sum_{i=1}^{N_{1}\left( -h\right) }\rho _{1i}1\left\{
h<0\right\} + \sum_{i=1}^{N_{2}\left( h\right) }\rho _{2i}1\left\{ h\geq
0\right\} ,
\end{equation*}%
where $N_{1}$ and $N_{2}$ are Poisson processes with the
same jump rate $f_{Q}\left( \tau _{0}\right) $, and $\left\{ \rho
_{1i}\right\} \ $and $\left\{ \rho _{2i}\right\} $ are two sequences of independent and identically distributed 
random variables. The   distributions of ${\rho}_{1i}$ and ${\rho}_{2i}$, respectively,  are identical to the
conditional distributions of $ \dot{\rho}\left( U _{i}-X_{i}^{T}\delta _{0}\right) 
 -\dot{\rho}\left( U _{i}\right) $ 
and $\dot{\rho}\left(U _{i}+X_{i}^{T}\delta _{0}\right) -\dot{\rho}\left( U_{i}\right)$ given $Q_{i}=\tau _{0}$, where $\dot{\rho}\left( t\right)  \equiv t\left( \gamma -1\left\{ t\leq 0\right\} \right) $ and 
$U_i \equiv Y_i - X_{i}^{T}\beta _{0}-X_{i}^{T}\delta _{0}1\left\{ Q_{i}>\tau
_{0}\right\}$ for each $i = 1,\dots,n$. Here, 
 $N_{1}$,  $N_{2}$,  $\left\{ \rho_{1i}\right\} \ $ and $\left\{ \rho _{2i}\right\} $ are mutually independent.
\end{thm}

The first conclusion of Theorem \ref{thm:tau-2nd} establishes that the second step estimator of $\tau_0$ 
is an oracle estimator in the sense that it is asymptotically equivalent to the infeasible, oracle estimator $\widetilde{\tau}$. 
As emphasized in the introduction, 
the oracle property is obtained  without relying on the perfect model selection in the first step
nor on the existence of the minimum signal condition on active covariates.
The second conclusion of Theorem \ref{thm:tau-2nd} follows from combining 
 well-known weak convergence results in  the literature (see e.g. \cite{Pons2003, kosorok2007, Lee:Seo:08})
with the argmax continuous mapping theorem by \cite{Seijo:Sen:11b}.

\begin{remark}\label{simulation:tauhat}
\cite{Li:Ling:12} propose  a numerical approach for constructing a confidence interval by simulating a compound Poisson process
in the context of least squares estimation. We adopt their approach to simulate the compound Poisson process for quantile regression.
See Appendix \ref{sect:tau:ci-algorithm}
 for a detailed description of how to construct a confidence interval for $\tau_0$. 
\end{remark} 

We now consider the Step 3a estimator of $\alpha_0$ defined in \eqref{step3-alpha-est}. 
Recall that $\omega_n$ is the Step 3a penalization tuning parameter that satisfies \eqref{omega-rate}.

\begin{thm} [Improved Rates of Convergence When $\delta_0 \neq 0$]
\label{th2.3-improved} 
Suppose that $\kappa_n s^2\log p=o(1)$. 
Then 
under Assumptions  \ref{a:setting}-\ref{a:moment},  
\begin{equation*}
|\widehat{\alpha}-\alpha _{0}|_{1}=O_P(\omega _{n}s)
\ \ \text{ and } \ \ 
R(\widehat{\alpha},\widehat{\tau})=O_P(\omega _{n}^{2}s).
\end{equation*}
\end{thm}

Theorem \ref{th2.3-improved}  shows that  the  estimator $\widehat\alpha$ defined in Step 3a achieves the optimal rate of convergence in terms of prediction and estimation. In other words, when $\omega_{n}$ is proportional 
to $\{ \log (p \vee n)/n \}^{1/2}$ in equation \eqref{omega-rate}  and $p$ is larger than $n$,  it obtains the  minimax rates as in  e.g., \cite{Rasku09}.

As we mentioned in Section \ref{sec:model-estimator}, the Step 3b estimator of $\alpha_0$ has the purpose of the variable selection. 
The nonzero components of  $\widetilde\alpha$
are expected to identify contributing regressors. Partition $\widetilde{\alpha }=(\widetilde{\alpha }_{J},\widetilde{%
\alpha }_{J^{c}})$ such that $\widetilde{\alpha }_{J}=(\widetilde{\alpha }%
_{j}:j\in J(\alpha_0))$ and $\widetilde{\alpha }_{J^{c}}= ( \widetilde{\alpha }%
_{j}:j\notin J(\alpha_0) )$. Note that $\widetilde\alpha_J$ consists of the estimators of $\beta_{0J}$ and $\delta_{0J}$, whereas  $\widetilde\alpha_{J^c}$ consists of the estimators of all the zero components of $\beta_0$ and $\delta_0.$
Let $\alpha_{0J}^{(j)}$ denote the $j$-th element of $\alpha_{0J}$.

We now establish conditions under which the estimator $\widetilde\alpha$   defined in Step 3b has the \textit{change-point-oracle properties}, meaning that it achieves the variable selection consistency and has the limiting distributions as though the identities of the important regressors and the location of the change point were known.   


\begin{thm}[Variable Selection When $\delta_0 \neq 0$]
\label{th3.1} 
Suppose that $\kappa_n s^2\log p=o(1)$,  $s^4\log s=o(n)$, 
 and 
 \begin{equation}\label{e3.2}
\omega_n+s\sqrt{\frac{\log s}{n}} \ll \mu _{n} \ll \min_{j\in J(\alpha_0)}|\alpha _{0J}^{(j)}|.
\end{equation}
Then 
under Assumptions  \ref{a:setting}-\ref{a:moment},  
 we have:
 (i)
\begin{equation*}
\left\vert \widetilde{\alpha }_{J}-\alpha _{0J}\right\vert _{2}=O_P \left( \sqrt{%
\frac{s\log s}{n}} \; \right),\quad \left\vert \widetilde{\alpha }_{J}-\alpha _{0J}\right\vert _{1}=O_P \left( s\sqrt{%
\frac{\log s}{n}} \; \right),
\end{equation*}%
(ii)
\begin{equation*}
P(\widetilde{\alpha }_{J^{c}}=0)\rightarrow 1,
\end{equation*}
and (iii)
$$
R(\widetilde\alpha,\widehat\tau) =O_P  \left( \; \mu_ns\sqrt{\frac{\log s}{n}} \; \right).
$$
\end{thm}


We see that  (\ref{e3.2})  provides a condition on the strength of the signal via $\min_{j\in J(\alpha_0)}|\alpha _{0J}^{(j)}|$, and the tuning parameter in Step 3b should satisfy $\omega_n \ll \mu_n$
and $s^2 \log s/n \ll \mu_n^2$.  
Hence the variable selection consistency   demands a larger tuning parameter than  in Step 3a. 


 To conduct statistical inference, we now discuss the asymptotic distribution of $\widetilde\alpha_J$. 
 Define $\widehat{\alpha}^{\ast}_J\equiv\operatorname*{argmin}_{\alpha_J
}R_{n}^{\ast}\left(  \alpha_J,\tau_{0}\right)$.  
Note that the asymptotic distribution for $\widehat{\alpha}^{\ast}_J$ corresponds to an oracle case that
we know $\tau _{0}$  as well as  the true active set $J(\alpha_0)$ \textit{a priori}. 
The limiting distribution of $\widetilde\alpha_J$ is  the same as that of 
$\widehat{\alpha}^{\ast}_J$. 
Hence,  we call this result the \emph{change-point-oracle property} of the Step 3b estimator and the following theorem establishes this property. 

\begin{thm}[Change-Point-Oracle Properties]\label{thm-AI}
Suppose that all the conditions imposed in Theorem \ref{th3.1} are satisfied. Furthermore, 
assume that $\frac{\partial}{\partial\alpha}E\left[
\rho\left(  Y,X^{T}\alpha\right)  |Q=t\right]  $ exists for all $t$ in a
neighborhood of $\tau_{0}$ and all its elements are continuous and bounded, and that 
$s^{3} (\log s) (\log n) =o\left(  n \right)$. 
Then, we have that  
$
\widetilde{\alpha}_J = \widehat{\alpha}^{\ast}_J + o_P( n^{-1/2}).
$
\end{thm}
 
Since the sparsity index ($s$) grows at a rate slower than  the sample size ($n$),  it is straightforward to establish  the  asymptotic normality of a linear transformation of
$\widetilde\alpha_J,$ i.e., $\mathbf{L} \widetilde\alpha_J,$ where $\mathbf{L}:\mathbb{R}^{s}\rightarrow\mathbb{R}$ with $|\mathbf{L}|_2=1$, by combing 
 the existing results on quantile regression with parameters of increasing dimension (see, e.g. \cite{He:Shao:00})
 with Theorem \ref{thm-AI}.
 
\begin{remark}
Without   the condition on the strength of minimal signals, it may not be possible to achieve the variable selection consistency or establish change-point-oracle properties. However,  the following theorem shows that  the SCAD-weighted  penalized estimation  can still achieve a satisfactory  rate of convergence in estimation of $\alpha_0$ without the condition that 
$\mu _{n} \ll \min_{j\in J(\alpha_0)}|\alpha _{0J}^{(j)}|$.
Yet, the rates of convergence are slower than those of Theorem \ref{th3.1}. 
 \end{remark}

\begin{thm}[Satisfactory Rates Without Minimum Signal Condition]\label{th3.1-without}
Assume that Assumptions  \ref{a:setting}-\ref{a:moment} hold.
Suppose that $\kappa_n s^2\log p=o(1)$ and $\omega_n \ll \mu _{n}$.
Then,  without the lower bound requirement on $\min_{j\in J(\alpha_0)}|\alpha _{0J}^{(j)}|$, we have 
that 
$
|\widetilde\alpha-\alpha_0|_1 =O_P \left( \mu_n s \right).
$ In addition,
$R(\widetilde\alpha,\widehat\tau)=O_P(\mu_n^2s).$
\end{thm}

\section{Asymptotic Properties: Case II. $\delta_0 = 0$}\label{sec:Lasso-theory-case2}
 
 In this section, we show that our estimators have desirable results even if there is no change point in the true model. 
The case of $\delta_0=0$ corresponds to the high-dimensional linear quantile regression  model.  
 Since   
$
 X^T\beta_0+X^T\delta_01\{Q>\tau_0\}=X^T\beta_0,
$
  $\tau_0$ is non-identifiable, and  there is no structural change on the coefficient.  But a new   analysis  different  from that of the  standard  high-dimensional model is still required  because in practice we do not know whether $\delta_0=0$ or not.
  Thus,  the proposed estimation method still estimates $\tau_0$ to account for possible structural changes.  
The following results show that in this case, the first step  estimator of $\alpha_0$ will asymptotically behave as if $\delta_0=0$  were \textit{a priori} known.

 \begin{thm}[Rates of Convergence When $\delta_0 = 0$]\label{th2.3-delta0}
Suppose that $\kappa_n s=o(1)$. 
Then 
under Assumptions  \ref{a:setting}-\ref{ass-rn},  
we have that 
\begin{equation*} 
|\breve{\alpha}-\alpha _{0}|_{1}=O_P(\kappa _{n}s)
\ \ \text{ and } \ \
R(\breve{\alpha},\breve{\tau})=O_P(\kappa _{n}^{2}s).
\end{equation*}%
\end{thm}

The results obtained in Theorem \ref{th2.3-delta0} combined with those obtained in Theorem \ref{th2.3} imply that 
the  first step  estimatior performs equally well in terms of rates of convergence for both the $\ell_1$ loss for $\breve\alpha$ and the excess risk regardless of the existence of the threshold effect. 
 It is straightforward to obtain an improved rate result for the Step 3a estimator, equivalent to 
 Theorem \ref{th2.3-improved} under  Assumptions  \ref{a:setting}-\ref{ass-rn}. We omit the details for brevity.

We now give a result that is  similar  to Theorem \ref{th3.1} and Theorem \ref{th3.1-without}.

\begin{thm}[Variable Selection When $\delta_0 = 0$]
\label{th4.2}  
Suppose that $\kappa_n s=o(1)$, $s^{4} \log s =o(n)$,    $\omega_n+s\sqrt{\frac{\log s}{n}} 
\ll \mu _{n}$, and Assumptions  \ref{a:setting}-\ref{ass-rn} hold. We have: \\
(i) If  the minimum signal condition holds: 
\begin{equation}\label{e6.1}
\mu _{n}=o\left( \min_{j\in J(\alpha_0)}|\alpha _{0J}^{(j)}|\right),
\end{equation}%
then 
\begin{equation*}
\left\vert \widetilde{\beta }_{J}-\beta _{0J}\right\vert _{2}=O_P \left( \sqrt{%
\frac{s\log s}{n}} \; \right),\quad \left\vert \widetilde{\beta }_{J}-\beta _{0J}\right\vert _{1}=O_P \left( s\sqrt{%
\frac{\log s}{n}} \; \right),
\end{equation*}%
\begin{equation*}
P(\widetilde{\beta }_{J^{c}}=0)\rightarrow 1,\quad P(\widetilde \delta=0)\rightarrow1,
\ \ \text{ and } \ \  R (\widetilde\alpha,\widehat\tau)=O_P \left( \; \mu_ns\sqrt{\frac{\log s}{n}} \; \right).
\end{equation*}
(ii) Without the minimum signal condition (\ref{e6.1}), we have:
  $$
 R(\widetilde\alpha,\widehat\tau)=O_P( \mu_n^2s),\quad\left|\widetilde\alpha-\alpha_0\right|_1=O_P(s\mu_n).
 $$  
\end{thm}

  Theorem \ref{th4.2}  demonstrates that when there is in fact no change point,  our estimator for $\delta_0$ is  exactly zero   with  a high probability. 
Therefore,  the estimator  can also be used as a diagnostic tool to check  whether there exists any change point.
Results similar to Theorems  \ref{thm-AI}  can be established straightforwardly as well; however, their
details are omitted for brevity.

\section{Monte Carlo Experiments}\label{sec:MC}

 In this section we provide the results of Monte Carlo experiments. The baseline model is based on the following data generating process: for $i=1,\ldots,n$,
\begin{align}
Y_i&=X_i'(\beta_0+\xi_{10} U_i) + 1\{Q_i>\tau_0\}X_i'(\delta_0 +\xi_{20}U_i) \label{eq:sim.qr},
\end{align} 
where $U_{i}$ follows $N(0,0.5^2)$, and $Q_i$ follows the uniform distribution on the interval $[0,1]$. The $p$-dimensional covariate $X_i$ is composed of a constant and $Z_i$, i.e.\ $X:=(1,Z_i^{T})^{T}$, where $Z_i$ follows the multivariate normal distribution $N(0,\Sigma)$  with a covariance matrix $\Sigma_{ij} = ( 1 / 2 )^{\vert i-j\vert}$.  Here, the variables $U_{i}, Q_i$ and $Z_i$ are  independent of each other.
Note that the conditional $\gamma$-quantile of $Y_i$ given $(X_i,Q_i)$ has the form:
\eq{
Quant_{\gamma}(Y_i|X_i,Q_i)=X_i'\bt_{\gamma} + 1\{Q_i<\tau_0\}X_i'\delta_{\gamma},
}
where $\bt_{\gamma}=\bt_0+\xi_{10}\cdot Quant_{\gamma}(U)$ and $\dt_{\gamma}=\dt_0+\xi_{20} \cdot Quant_{\gamma}(U)$. 

We consider three quantile regression models with $\gamma=0.25, 0.5$, and $0.75$. The $p$-dimensional parameters $\beta_0$, $\dt_0$, $\xi_{10}$, and $\xi_{20}$ are set to $\beta_0=(0,Quant_{0.75}(U)\approx0.34,0,\ldots,0)$, $\delta_0=(0,1,0,\ldots,0)$, $\xi_{10}=(0,1,0,\ldots,0)$, and $\xi_{20}=(0,0,0,\ldots,0)$, respectively. Because of the heteroskedasticity, the true parameter value $\bt_{\gamma}$ at each quantile is $\bt_{0.25}=(0,\ldots,0)$, $\bt_{0.5}=(0,0.34, \ldots,0)$, and $\bt_{0.75}=(0,0.68, \ldots,0)$. Note that nonzero coefficients are different between when $\gamma=0.25$ and when $\gamma=0.5$ or $\gamma=0.75$.

We set the change point parameter $\tau_0=0.5$ unless it is specified differently. The sample sizes are set to $n=200$ and $400$. The dimension of $X_i$ is set to $p=250$. Note that we have $500$ regressors in total. The change point $\tau$ is estimated over grid points of the sample observations $\{Q_i\}$, where the range is limited to those between the $0.15$-quantile and the $0.85$-quantile.  
We conduct $1{,}000$ replications of each design.

\begin{table}[htbp]
\centering
\caption{Baseline Model: $\gamma=0.25$}\label{tb:base-0.25}
\footnotesize
\begin{tabular}{lcclcccc}
  \hline
 & Excess Risk & E[$J(\widehat{\ap})$] & MSE of $\widehat{\ap} ~(\widehat{\ap}_{J_0} / \widehat{\ap}_{J_0^c})$ & Pred. Er. & RMSE of $\widehat{\tau}$ & C. Prob. of $\widehat{\tau}$ & Oracle Prop. \\ 
  \hline
  \underline{$n=200$}\\
  Oracle 1 & 0.004 & NA  & 0.005 ( NA  / NA )  & 0.203 & NA  & NA  &   NA\\ 
  Oracle 2 & 0.013 & NA  & 0.005 ( NA  / NA )  & 0.434 & 0.012 & 0.925 &  NA \\ 
  Step 1 & 0.032 & 4.266 & 0.433 (  0.389 / 0.044) & 0.727 & 0.013 & 0.943 & 0.032 \\ 
  Step 2 & 0.032 & NA  & NA  ( NA  / NA )  & 0.705 & 0.012 & 0.952 &  NA \\ 
  Step 3a & 0.031 & 4.249 & 0.413 (  0.370 / 0.043) & 0.691 & 0.012 & 0.951 & 0.032 \\ 
  Step 3b & 0.022 & 1.173 & 0.281 (  0.221 / 0.060) & 0.556 & 0.012 & 0.928 & 0.686 \\ 
  \underline{$n=400$}\\
  Oracle 1 & 0.002 & NA  & 0.003 ( NA  / NA )  & 0.145 & NA  & NA  &  NA \\ 
  Oracle 2 & 0.006 & NA  & 0.003 ( NA  / NA )  & 0.313 & 0.005 & 0.958 &   NA\\ 
  Step 1 & 0.017 & 4.352 & 0.214 (  0.193 / 0.021) & 0.502 & 0.006 & 0.959 & 0.035 \\ 
  Step 2 & 0.018 & NA  & NA  ( NA  / NA )  & 0.495 & 0.006 & 0.969 &  NA \\ 
  Step 3a & 0.015 & 4.361 & 0.207 ( 0.186 / 0.021) & 0.486 & 0.006 & 0.961 & 0.031 \\ 
  Step 3b & 0.009 & 1.176 & 0.062 (  0.048 / 0.014) & 0.315 & 0.005 & 0.955 & 0.816 \\       \hline
      \multicolumn{8}{p{\textwidth}}{\footnotesize \emph{Note: }Oracle 1 knows both $J(\alpha_{\gamma})$ and $\tau_0$ and Oracle 2 knows only $J(\alpha_{\gamma})$. Expectation (E) is calculated by the average of 1,000 iterations in each design. Note that $J(\alpha_{\gamma})=1$. `NA' denotes `Not Available' as the parameter is not estimated in the step. The estimation results for $\tau$ at the rows of Step 3a and Step 3b are based on the re-estimation of $\tau$ given estimates from Step 3a ($\widehat{\alpha}$) and Step 3b ($\widetilde{\alpha}$). 
}

\end{tabular}
\end{table}

\begin{table}[htbp]
\centering
\caption{Baseline Model: $\gamma=0.5$}\label{tb:base-0.5}
\footnotesize
\begin{tabular}{lcclcccc}
  \hline
 & Excess Risk & E[$J(\widehat{\ap})$] & MSE of $\widehat{\ap} ~(\widehat{\ap}_{J_0} / \widehat{\ap}_{J_0^c})$ & Pred. Er. & RMSE of $\widehat{\tau}$ & C. Prob. of $\widehat{\tau}$ & Oracle Prop. \\ 
  \hline
  \underline{$n=200$}\\
  Oracle 1 & 0.008 & NA  & 0.012 ( NA  / NA )  & 0.288 & NA  & NA  &  NA \\ 
  Oracle 2 & 0.018 & NA  & 0.012 ( NA  / NA )  & 0.465 & 0.011 & 0.948 &  NA \\ 
  Step 1 & 0.040 & 5.731 & 0.279 ( 0.245 / 0.034) & 0.723 & 0.011 & 0.950 & 0.015 \\ 
  Step 2 & 0.036 & NA    & NA  ( NA  / NA )  & 0.729 & 0.011 & 0.946 &   NA\\ 
  Step 3a & 0.039 & 5.776 & 0.272 ( 0.239 / 0.033) & 0.717 & 0.011 & 0.947 & 0.017 \\ 
  Step 3b & 0.040 & 2.364 & 0.182 ( 0.155 / 0.027) & 0.702 & 0.011 & 0.929 & 0.428 \\
  \underline{$n=400$}\\
  Oracle 1 & 0.004 & NA  & 0.006 ( NA  / NA ) & 0.201 & NA  & NA  &   NA\\ 
  Oracle 2 & 0.008 & NA  & 0.006 ( NA  / NA )& 0.337 & 0.005 & 0.956 &   NA\\ 
  Step 1 & 0.022 & 6.055 & 0.144 ( 0.128 / 0.017) & 0.512 & 0.005 & 0.953 & 0.020 \\ 
  Step 2 & 0.020 & NA    & NA  ( NA  / NA )  & 0.509 & 0.005 & 0.950 & NA  \\ 
  Step 3a & 0.019 & 6.056 & 0.142 ( 0.126 / 0.017) & 0.517 & 0.005 & 0.947 & 0.020 \\ 
  Step 3b & 0.018 & 2.250 & 0.061 ( 0.054 / 0.007) & 0.460 & 0.005 & 0.949 & 0.649 \\ 
   \hline
   \multicolumn{8}{p{\textwidth}}{\footnotesize \emph{Note: } $J(\alpha_{\gamma})=2$. See the note below Table \ref{tb:base-0.25} for other notation. 
}
\end{tabular}
\end{table}

\begin{table}[htbp]
\centering
\caption{Baseline Model: $\gamma=0.75$}\label{tb:base-0.75}
\footnotesize
\begin{tabular}{lcclcccc}
  \hline
 & Excess Risk & E[$J(\widehat{\ap})$] & MSE of $\widehat{\ap} ~(\widehat{\ap}_{J_0} / \widehat{\ap}_{J_0^c})$ & Pred. Er. & RMSE of $\widehat{\tau}$ & C. Prob. of $\widehat{\tau}$ & Oracle Prop. \\ 
  \hline
  \underline{$n=200$}\\
  Oracle 1 & 0.008 & NA  & 0.015 ( NA  / NA ) & 0.324 & NA  & NA  &  NA \\ 
  Oracle 2 & 0.016 & NA  & 0.015 ( NA  / NA ) & 0.508 & 0.011 & 0.941 &  NA \\ 
  Step 1 & 0.043 & 6.056 & 0.352 ( 0.310 / 0.042 )& 0.769 & 0.012 & 0.930 & 0.024 \\ 
  Step 2 & 0.036 & NA    & NA    ( NA  / NA ) & 0.787 & 0.013 & 0.911 &  NA \\ 
  Step 3a & 0.036 & 6.045 & 0.349 ( 0.308 / 0.042 )& 0.782 & 0.013 & 0.911 & 0.024 \\ 
  Step 3b & 0.029 & 2.160 & 0.232 ( 0.188 / 0.044 )& 0.629 & 0.012 & 0.925 & 0.688 \\ 

  \underline{$n=400$}\\
  Oracle 1 & 0.004 & NA  & 0.007 ( NA  / NA ) & 0.218 & NA  & NA  &  NA \\ 
  Oracle 2 & 0.008 & NA  & 0.007 ( NA  / NA ) & 0.354 & 0.005 & 0.952 &   NA\\ 
  Step 1 & 0.018 & 6.007 & 0.169 ( 0.150 / 0.020)  & 0.538 & 0.005 & 0.962 & 0.013 \\ 
  Step 2 & 0.019 & NA    & NA  ( NA  / NA )  & 0.571 & 0.005 & 0.944 &  NA \\ 
  Step 3a & 0.019 & 6.032 & 0.169 ( 0.149 / 0.020)  & 0.548 & 0.005 & 0.942 & 0.016 \\ 
  Step 3b & 0.010 & 2.128 & 0.052 (0.041 / 0.011)  & 0.367 & 0.005 & 0.953 & 0.860 \\ 
   \hline
   \multicolumn{8}{p{\textwidth}}{\footnotesize \emph{Note: } $J(\alpha_{\gamma})=2$. See the note below Table \ref{tb:base-0.25} for other notation. 
}
\end{tabular}
\end{table}

We compare estimation results of each step. To assess the performance of our estimators, we also compare the results with two ``oracle estimators". Specifically,  Oracle 1 knows the true active set $J(\alpha_{\gamma})$ and the change point parameter $\tau_0$, and Oracle 2 knows   only  $J(\alpha_{\gamma})$.  
The threshold parameter $\tau_0$ is re-estimated in 
Steps 3a and 3b using updated estimates of $\alpha_{\gamma}$. 

Tables \ref{tb:base-0.25}--\ref{tb:base-0.75} summarize the simulation results. 
We abuse notation slightly and denote all estimators by $(\widehat{\alpha},\widehat{\tau})$. They would be understood as $(\breve{\alpha},\breve{\tau})$ in Step 1, $\widehat{\tau}$ in Step 2, and so on.  
We report the excess risk, the average number of parameters selected, $\mathbb{E}[J(\widehat{\alpha})]$, and the sum of the mean squared error of $\widehat{\ap}$ ($\widehat{\ap}_{J_0}$ / $\widehat{\ap}_{J^c_0}$). For each sample, the excess risk is calculated by the simulation, $S^{-1}\sum_{s=1}^S \left[\rho(Y_s,X_s^T(\widehat{\tau})\widehat{\alpha})-\rho(Y_s,X_s^T({\tau}_0){\alpha}_0)\right]$, where $S=10{,}000$ is the number of simulations; then we report the average value of 1,000 replications. Similarly,  we also calculate prediction errors by the simulation, $\left(S^{-1}\sum_{s=1}^S \left(X^T_s(\widehat{\tau})\widehat{\alpha} - X^T_s(\tau_0)\ap_{\gamma}\right)^2\right)^{1/2}$, and report the average value.

We also report the root-mean-squared error (RMSE) and the coverage probability of the 95\% confidence interval of $\widehat{\tau}$ (C.\ Prob.\ of $\widehat{\tau}$). 
The confidence intervals for $\tau_0$ are calculated by simulating the two-sided compound Poisson process in Theorem \ref{thm:tau-2nd}  by adopting the approach proposed by \citet{Li:Ling:12}.
The details are provided in Section \ref{sect:tau:ci-algorithm}. 
\citet{Li:Ling:12} showed that it is valid to simulate the compound poisson process by simulating the poisson process and the  compounding factors from empirical distributions separately in the context of least squares estimation. We build on their suggestion and modify their procedure to quantile regression. 
We did not prove a formal justification for our procedure in this paper; however, it seems working well in simulations. It is an interesting topic for future research.

Note that the root-mean-squared error of $\widehat{\tau}$ and the coverage probability of the confidence interval at the rows of Step 3a and Step 3b in the tables are estimation results of updated $\widehat{\tau}$: we re-estimate $\tau$ as in Step 2 using $(\widehat{U}_i,\widehat{\alpha})$ and $(\widetilde{U}_i,\widetilde{\alpha})$ from Step 3a and Step 3b instead of $(\breve{U}_i,\breve{\alpha})$. Finally, we also report the oracle proportion (Oracle Prop.), namely the ratio of the correct model selection out of 1,000 replications.

Overall, the simulation results confirm the asymptotic theory developed in the previous sections. First, these results show the advantage of quantile regression models over the existing mean regression models with a change point, e.g.\ \citet{lee2012lasso}. The proposed estimator (Step 3b) selects different nonzero coefficients at different quantile levels. The estimator in \citet{lee2012lasso} cannot detect these heterogeneous models. In general the proposed estimators show better performance for heteroskedastic designs and for the fat-tail error distributions as will be discussed in detail  below. 
Second, when we look at the finite sample performance  of the proposed estimators in Step 3, their prediction errors  are within a reasonable bound from those of Oracles 1 and 2. Recall that we estimate models with 250 times or 500 times more regressors in each design.  
Third, the root-mean-squared error of $\widehat{\tau}$ decreases quickly and confirms the super-consistency result of $\widehat{\tau}$. 
Fourth, the coverage probabilities of the confidence interval are close to 95\%, especially when $n=400$. Thus,  we recommend practitioners to use $\widehat{\tau}$ in Step 2 or the re-estimated version of it based on the estimates from Step 3a or Step 3b. 
Finally, the oracle proportion of Step 3b is quite satisfactory and confirms our results in model selection consistency. 

\subsection{Comparison with Mean Regression with a Change Point}

Table \ref{tb:mean-all} compares the performance of the proposed estimator with that of the mean regression method in \citet{lee2012lasso}. 
For the purpose of direct comparison between mean and median regression models, the tuning parameter $\lambda$ is fixed to be the same as that in Step 1 from median regression. 
 We consider three different simulation designs at $\gamma=0.5$ with $n=200$. The first model is a homoskedastic model by setting $\xi_{01}=(1,0,\ldots,0)$ in the baseline design. The second model is the same as the heteroskedastic median regression in Table \ref{tb:base-0.5}. The third model is a fat-tail model, where $U_i$ follows a Cauchy distribution with a scale parameter 0.25 while keeping the heteroskedastic design as the second model. The mean regression method shows slight over-selection but its performance looks reasonable in the homoskedastic model. However, the method in  \citet{lee2012lasso} is not robust to the heteroskedastic errors, which we can observe in Panel B of Table  \ref{tb:mean-all}. Furthermore, it cannot detect different nonzero coefficients at different quantile levels while the quantile method shows such a result in Table \ref{tb:base-0.25}. Finally, the quantile method works well when the error distribution follows a Cauchy distribution in Panel C of Table  \ref{tb:mean-all}. However, the mean regression method performs poorly with a Cauchy error distribution as the conditional mean function is not well-defined in this case.

\begin{table}[htbp]
\centering
\caption{Comparison between mean and median regression models with a change point}\label{tb:mean-all}

\footnotesize

\bigskip

\medskip

\begin{tabular}{lclccc}
\multicolumn{6}{l}{\emph{Panel A---Homoskedastic Model: $\gamma=0.5$ and $n=200$}} \\ 
  \hline
  &  E[$J(\widehat{\ap})$] & MSE of $\widehat{\ap} ~(\widehat{\ap}_{J_0} / \widehat{\ap}_{J_0^c})$ & Pred. Er. & RMSE of $\widehat{\tau}$ &   Oracle Prop. \\ 
  \hline
  Oracle 1  & NA  & 0.000 (NA / NA) & 0.056 & NA  &    NA \\ 
  Oracle 2 & NA  & 0.000 (NA / NA)& 0.199 & 0.003 &   NA \\ 
  Step 1 & 5.919 & 0.011 ( 0.010 / 0.001) & 0.259 & 0.003 &   0.026 \\ 
  Step 2 & NA  & NA (NA / NA) & 0.248 & 0.003 &    NA \\ 
  Step 3a  & 5.900 & 0.011 ( 0.010 / 0.001) & 0.257 & 0.003 &   0.024 \\ 
  Step 3b & 2.001 & 0.001 ( 0.001 / 0.000) & 0.213 & 0.003 &   0.999 \\
  Mean Reg & 8.162 & 0.010 (0.008 / 0.001) & 0.256 & 0.003 & 0.000 \\ 
\hline
\\
\multicolumn{6}{l}{\emph{Panel B---Heteroskedastic Model: $\gamma=0.5$ and $n=200$}} \\
  \hline
 &  E[$J(\widehat{\ap})$] & MSE of $\widehat{\ap} ~(\widehat{\ap}_{J_0} / \widehat{\ap}_{J_0^c})$ & Pred. Er. & RMSE of $\widehat{\tau}$ &   Oracle Prop. \\ 
  \hline
  Oracle 1 & NA  & 0.012 ( NA  / NA )  & 0.288 & NA          &  NA \\ 
  Oracle 2 & NA  & 0.012 ( NA  / NA )  & 0.465 & 0.011       &  NA \\ 
  Step 1   & 5.731 & 0.279 ( 0.245 / 0.034) & 0.723 & 0.011  & 0.015 \\ 
  Step 2   & NA    & NA  ( NA  / NA )  & 0.729 & 0.011       &   NA\\ 
  Step 3a  & 5.776 & 0.272 ( 0.239 / 0.033) & 0.717 & 0.011  & 0.017 \\ 
  Step 3b  & 2.364 & 0.182 ( 0.155 / 0.027) & 0.702 & 0.011  & 0.428 \\
Mean Reg & 93.550 & 2.537 ( 0.326 / 2.211 ) & 1.572 & 0.011 & 0.000 \\ 
   \hline
\\
\multicolumn{6}{l}{\emph{Panel C---Fat-tail Model: $\gamma=0.5$, $U_i \sim Cauchy(0.25)$  and $n=200$}} \\
  \hline
 &  E[$J(\widehat{\ap})$] & MSE of $\widehat{\ap} ~(\widehat{\ap}_{J_0} / \widehat{\ap}_{J_0^c})$ & Pred. Er. & RMSE of $\widehat{\tau}$ &   Oracle Prop. \\ 
  \hline
  Oracle 1 &  NA & 0.005( NA  / NA ) & 0.185 &  NA          &  NA \\ 
  Oracle 2 &  NA & 0.005( NA  / NA ) & 0.392 & 0.011        &  NA \\ 
  Step 1   & 5.843 & 0.148 ( 0.131 / 0.017) & 0.566 & 0.011 & 0.022 \\ 
  Step 2   &  NA &  NA ( NA  / NA ) & 0.576 & 0.011         &  NA \\ 
  Step 3a  & 5.806 & 0.143 ( 0.126 / 0.017) & 0.575 & 0.011 & 0.019 \\ 
  Step 3b  & 2.582 & 0.074 ( 0.066 / 0.008) & 0.575 & 0.011 & 0.483 \\ 
  Mean Reg & 218.991 & $5.55\times 10^6$ ( $5.45\times 10^3$ / $5.50\times 10^6$ ) & 137.985 & 0.221 & 0.000 \\ 

 \hline

\end{tabular}
\end{table}

\subsection{When There Is No Change Point}

Table \ref{tb:no-change} shows the performance of the estimator when there does not exist any change point. We use the baseline design with $\gamma=0.75$ and set $\dt=(0,\ldots,0)$. As we are interested in the performance of $\widehat{\dt}$, we report the average number of parameters selected in $\widehat{\dt}$, the MSE of $\widehat{\dt}$, and the proportion of detecting no-change point (No-change Prop.). As predicted by the theory, all measures on $\widehat{\dt}$ indicate that the estimator (Step 3b) detects no-change point models quite well. Both $\mathbb{E}[J(\widehat{\dt})]$ and MSE of $\widehat{\dt}$ are quite low and no-change proportion is high. We can also observe much improvement in these measure when the sample size increases from $n=200$ to $n=400$. 

\begin{table}[htbp]
\centering
\caption{No Change Point: $\gamma=0.75$, $\dt_{\gamma}=0$}\label{tb:no-change}
\footnotesize
\begin{tabular}{lcccccccc}
  \hline
 & Excess Risk & E[$J(\widehat{\ap})$] &E[$J(\widehat{\dt})$]  & MSE of $\widehat{\ap}$  & MSE of $\widehat{\dt}$   & Pred. Er. & No-change Prop. & Oracle Prop. \\  
  \hline
  \underline{$n=200$}\\
  Oracle 1 & 0.004 &NA  &NA  & 0.002 &NA  & 0.196 &NA  & NA  \\
  Oracle 2 & 0.004 &NA  &NA  & 0.002 &NA  & 0.196 &NA  & NA  \\
  Step 1 & 0.030 & 4.796 & 1.149 & 0.228 & 0.006 & 0.618 & 0.221 & 0.008 \\ 
  Step 2 & 0.024 &NA  &NA  &NA  &NA  & 0.617 &NA  & NA  \\
  Step 3a & 0.026 & 4.915 & 1.309 & 0.226 & 0.008 & 0.602 & 0.142 & 0.008 \\ 
  Step 3b & 0.017 & 1.520 & 0.334 & 0.178 & 0.007 & 0.436 & 0.722 & 0.541 \\ 

  \underline{$n=400$}\\
  Oracle 1 & 0.002 &NA  &NA  & 0.001 &NA  & 0.143 &NA  & NA  \\
  Oracle 2 & 0.002 &NA  &NA  & 0.001 &NA  & 0.143 &NA  & NA  \\
  Step 1 & 0.015 & 4.933 & 1.137 & 0.126 & 0.003 & 0.451 & 0.223 & 0.013 \\ 
  Step 2 & 0.014 &NA  &NA  &NA  &NA  & 0.449 &NA  & NA  \\
  Step 3a & 0.015 & 5.042 & 1.301 & 0.124 & 0.004 & 0.440 & 0.123 & 0.010 \\ 
  Step 3b & 0.005 & 1.208 & 0.141 & 0.037 & 0.002 & 0.197 & 0.867 & 0.805 \\ 
   \hline
   \multicolumn{9}{p{\textwidth}}{\footnotesize \emph{Note: } $J(\alpha_{\gamma})=1$ and $J(\dt_{\gamma})=0$.
}
\end{tabular}
\end{table}

\subsection{When the Minimal Signal in $\dt$ is Low}

In this subsection, we consider the case when the model  contains low \emph{minimal} signals in $\delta$.  Specifically, we consider the median regression model and set  $\bt_{0.5} = (0,0.34,0,\ldots,0)$	and	$\dt_{0.5} = (0,1,1/2,1/4,1/8,1/16,0\ldots,0)$. Table \ref{table-diminishing} reports  simulation results
in this design. 
Note that the simulation design in Table \ref{table-diminishing} is the same as that reported in Table \ref{tb:base-0.5}
except that $\dt_{0.5} = (0,1,0\ldots,0)$ in Table \ref{tb:base-0.5}.
Therefore, we may view that  the simulation design in Table \ref{tb:base-0.5} satisfies 
the minimum signal condition, whereas that of this subsection does not. The simulation results in Table \ref{table-diminishing} are consistent with asymptotic theory in Section \ref{sec:Lasso-theory-case1} and remarks in Section \ref{sec:comparison-step3} comparing estimators in step 3.
The step 3b estimator performs better than the step 3a estimator in Table \ref{tb:base-0.5}, but it performs worse in Table \ref{table-diminishing}. Also note that the oracle proportion is zero for the step 3b estimator, which is expected given low signals in coefficients. Finally, it is important to note that the performance of the estimators of $\tau_0$ is good in terms of the MSE in the presence of low signals in $\delta$. 
The coverage probability of the confidence interval is much higher than the nominal level, which was not observed in previous simulations. Since the MSE and coverage probability between the infeasible oracle 2 estimator and other estimators are
very similar, we interpret that the over-coverage result  is not driven by high-dimensionality of regressors and variable selection. Perhaps this is due to a larger number of coefficients to estimate for the oracle 2 estimator, compared to Table \ref{tb:base-0.5}.

\begin{table}[htbp]
\centering
\caption{When the minimal signal in $\delta$ is low}\label{table-diminishing}
\footnotesize
\begin{tabular}{lccccccc}
  \hline
 & Excess Risk & $\mathbb{E}[J_0(\widehat{\ap})$] & MSE of $\widehat{\ap}$ ($\widehat{\ap}_{J_0}$ / $\widehat{\ap}_{J_0^c}$) & Pred. Er. & MSE of $\widehat{\tau}$ & C. Prob. of $\widehat{\tau}$ & Oracle Prop \\ 
  \hline
\underline{$n=200$}\\
 Oracle 1 & 0.024 &   NA        & 0.729 (NA / NA )         & 0.522 &   NA       &   NA        &  NA \\ 
 Oracle 2 & 0.037 &   NA        & 0.729 (NA / NA )         & 0.816 & 0.004 & 0.995 &  NA \\ 
    Step 1 & 0.054 & 9.949 & 0.517 (0.414 / 0.104)  & 0.946 & 0.004 & 0.995 & 0.000 \\ 
    Step 2 & 0.054 &   NA       &   NA (NA / NA )             & 0.949 & 0.004 & 0.992 &  NA \\ 
  Step 3a & 0.056 & 9.923 & 0.517 ( 0.414 / 0.104) & 0.903 & 0.004 & 0.991 & 0.000 \\ 
  Step 3b & 0.062 & 3.327 & 1.293 ( 1.174 / 0.119) & 1.002 & 0.004 & 0.990 & 0.000 \\ 
\underline{$n=400$}\\
Oracle 1 & 0.012 & NA & 0.339 (NA / NA )   & 0.365 & NA  & NA &  NA \\ 
  Oracle 2 & 0.017 & NA & 0.339 (NA / NA )   & 0.522 & 0.002 & 0.999 &  NA \\ 
  Step 1 & 0.029 & 11.058 & 0.333 (0.275 / 0.058) & 0.647 & 0.002 & 1.000 & 0.000 \\ 
  Step 2 & 0.029 & NA &  NA (NA / NA )     & 0.694 & 0.002 & 0.997 &  NA \\ 
  Step 3a & 0.027 & 11.067 & 0.332 (0.274 / 0.058) & 0.678 & 0.002 & 0.998 & 0.000 \\ 
  Step 3b & 0.032 & 3.574 & 0.648 ( 0.585 / 0.063) & 0.691 & 0.002 & 0.999 & 0.000 \\ 
   \hline
\end{tabular}
\end{table}

\subsection{Additional Simulation Results}

We have carried out additional Monte Carlo experiments. For the sake of brevity, we only report main findings here and show full results in the appendices.
In Appendix \ref{sec:add-sim-results-1}, we report  simulation results when the change point $\tau_0$ and the distribution of $Q_i$ vary. In particular,  we consider three different distributions of $Q_i$: Uniform$[0,1]$, $N(0,1)$, and $\chi^2(1)$. The change point parameter $\tau_0$ varies over $0.3,0.4,\ldots,0.7$ quantiles of each $Q_i$ distribution. We find that the performance of $\widehat{\tau}$ measured by the root-mean-squared error depends on the density of $Q_i$ distribution, as is expected from asymptotic theory. For instance, it is quite uniform over different $\tau_0$ when $Q_i$ follows Uniform$[0,1]$. However, when $Q_i$ follows $N(0,1)$ or $\chi^2(1)$, it performs better when $\tau_0$ is located at a point with a higher density of $Q_i$ distribution. 
Sensitivity analyses provided 
in Appendix \ref{sec:add-sim-results}  show that the main simulation results are robust when we 
make changes over the range between $-15\%$ and $+15\%$ of the suggested tuning parameter values.

In summary, the proposed estimation procedure works well in finite samples and confirms the theoretical results developed earlier.  
The simulation studies show some advantages of the proposed estimator over the existing mean regression method, e.g.\ \citet{lee2012lasso}. It also detects no-change-point models well without any pre-test. The main qualitative results are not sensitive to different simulation designs on $\tau_0$ and $Q_i$ as well as to some variation on tuning parameter values.

\section{Estimating a Change Point in Racial Segregation}\label{sec:data-example}

As an empirical illustration, we investigate the existence of tipping  in the dynamics of racial segregation 
using the dataset constructed by  \citet{card2008tipping}. They show that 
the neighborhood's white population decreases substantially when the minority share
in the area exceeds a tipping point (or threshold point),  using U.S. Census tract-level data.
\citet{lee2012testing}  develop a test for the existence of threshold effects and apply their test to this dataset.
Different from these existing studies, we consider a high-dimensional setup by allowing both possibly highly nonlinear effects of  the main covariate (minority share in the neighborhood) and possibly higher-order interactions between additional covariates.

We build on the specifications used in \citet{card2008tipping} and \citet{lee2012testing} to choose 
the following median regression with a constant shift due to the tipping effect:
\begin{align}
Y_i = g_0(Q_i) + \delta_0 1\{Q_i>\tau_0\} +   X_i'\beta_0 + U_i,
\end{align}
where for census tract $i$, the dependent variable $Y_i$  is the ten-year change in the neighborhood's white population, $Q_i$ is the base-year minority share in the neighborhood, and $X_i$ is a vector of six tract-level control variables and their various interactions depending on the model specification. Both $Y_i$ and $Q_i$ are in percentage terms. The basic six variables in $X_i$ include the unemployment rate, the log of mean family income, the fractions of single-unit, vacant, and renter-occupied housing units, and the fraction of workers who use public transport to travel to work. The function $g(\cdot)$ is approximated by the cubic b-splines with 15 knots over equi-quantile locations, so the degrees of freedom are $19$ including an intercept term. 
In our empirical illustration, we use the census-tract-level sample of Chicago whose base year is 1980.

In the first set of models, we consider possible interactions among the six tract-level control variables up to six-way interactions. Specifically, 
the vector $X$ in the six-way interactions will be composed of the following 63 regressors, 
$$
\{X^{(1)}, \ldots,X^{(6)}, X^{(1)}X^{(2)},\ldots, X^{(5)}X^{(6)},\ldots, X^{(1)} X^{(2)}X^{(3)}X^{(4)}X^{(5)}X^{(6)}\},
$$
 where $X^{(j)}$ is the $j$-th element among those tract-level control variables. Note that the lower order interaction vector (e.g.\ two-way or three-way) is nested by the higher order interaction vector (e.g.\ three-way or four-way).
 The total number of regressors varies from 26 (19 from b-splines, 6 from $X_i$ and $1\{Q_i>\tau\}$) when there is no interaction to 83 when there are full six-way interactions. In the next set of models, we add the square of each tract-level control variable and generate similar interactions up to six.
 In this case the total number of regressors varies from 32 to 2,529. For example, the number of regressors in the largest model consists of $\#(\text{b-spline basis})+\#(\text{indicator function}) + \#(\text{interactions up to six-way out of 12})=19 + 1 + \sum_{k=1}^6 \binom {12}k = 2,529$. This number is much larger than the sample size ($n=1,813$).

\begin{table}[htbp]
\footnotesize
\centering
\caption{Estimation Results from Quantile Regression}\label{tb:all-quantiles}
\begin{tabular}{llccccc}
  \\
\hline
 & & \multirow{2}{*}{No.\  of Reg.\ } & No.\  of Selected    & \multirow{2}{*}{$\widehat{\tau}$} &  \multirow{2}{*}{CI  for $\tau_0$} &\multirow{2}{*}{$\widehat{\delta}$} \\ 
 & & & Reg. in Step 3b &    &  \\ 
  \hline
$\gamma = 0.25$ \\
  &  \underline{6 control variables}\\

&  No Interaction & 26 & 17 & 5.65 & [4.75, 6.17] & -4.07 \\ 
&  Two-way Interaction & 41 & 20 & 2.35 & [1.00, 4.44] & -1.82 \\ 
&  Three-way Interaction & 61 & 24 & 2.35 & [1.00, 4.15] & -2.19 \\ 
&  Four-way Interaction & 76 & 21 & 5.65 & [4.69, 6.08] & -5.50 \\ 
&  Five-way Interaction & 82 & 22 & 2.45 & [1.00, 4.93] & -1.55 \\ 
&  Six-way Interaction & 83 & 22 & 2.45 & [1.00, 4.75] & -1.55 \\ 
 &   \underline{12 control variables}\\

&  No Interaction & 32 & 17 & 5.65 & [4.75, 6.17] & -4.07 \\ 
&  Two-way Interaction & 98 & 18 & 5.25 & [3.55, 6.09] & -3.40 \\ 
&  Three-way Interaction & 318 & 22 & 5.25 & [3.63, 5.94] & -3.61 \\ 
&  Four-way Interaction & 813 & 26 & 5.25 & [3.79, 5.97] & -3.53 \\ 
&  Five-way Interaction & 1605 & 27 & 5.25 & [4.57, 5.65] & -5.37 \\ 
&  Six-way Interaction & 2529 & 28 & 5.65 & [4.96, 6.06] & -5.50 \\ 
     \hline
$\gamma = 0.50$ \\
&  \underline{6 control variables}\\
&  No Interaction & 26 & 15 & 5.65 & [1.67, 11.85] & -2.24 \\ 
&  Two-way Interaction & 41 & 17 & 5.05 & [2.25, 7.46] & -2.63 \\ 
&  Three-way Interaction & 61 & 20 & 5.25 & [4.22, 6.38] & -4.15 \\ 
&  Four-way Interaction & 76 & 19 & 5.05 & [3.60, 7.00] & -3.14 \\ 
&  Five-way Interaction & 82 & 20 & 5.05 & [1.23, 9.16] & -1.90 \\ 
&  Six-way Interaction & 83 & 20 & 5.05 & [1.33, 9.39] & -1.90 \\ 
&   \underline{12 control variables}\\
&  No Interaction & 32 & 16 & 1.95 & [0.77, 4.61] & -3.69 \\ 
&  Two-way Interaction & 98 & 21 & 6.75 & [1.00, 45.57] & 0.48 \\ 
&  Three-way Interaction & 318 & 25 & 4.05 & [1.00, 13.15] & -0.97 \\ 
&  Four-way Interaction & 813 & 27 & 3.65 & [1.00, 15.91] & -0.56 \\ 
&  Five-way Interaction & 1605 & 29 & 3.25 & [1.00, 13.16] & -0.68 \\ 
&  Six-way Interaction & 2529 & 28 & 3.25 & [1.00, 11.67] & -0.74 \\ 
\hline
$\gamma = 0.75$ \\
&    \underline{6 control variables}\\
& No Interaction & 26 & 15 & 10.05 & [9.37, 11.29] & -10.62 \\ 
&  Two-way Interaction & 41 & 14 & NA & NA & 0.00 \\ 
&  Three-way Interaction & 61 & 21 & NA & NA & 0.00 \\ 
&  Four-way Interaction & 76 & 18 & NA & NA & 0.00 \\ 
&  Five-way Interaction & 82 & 18 & NA & NA & 0.00 \\ 
&  Six-way Interaction & 83 & 18 & NA & NA & 0.00 \\

&     \underline{12 control variables}\\
& No Interaction & 32 & 14 & 10.05 & [8.44, 11.94] & -7.14 \\ 
&  Two-way Interaction & 98 & 20 & NA & NA & 0.00 \\ 
&  Three-way Interaction & 318 & 21 & NA & NA & 0.00 \\ 
&  Four-way Interaction & 813 & 25 & NA & NA & 0.00 \\ 
&  Five-way Interaction & 1605 & 28 & NA & NA & 0.00 \\  
&   Six-way Interaction & 2529 & 24 & NA & NA & 0.00 \\ 

     \hline
     \multicolumn{7}{p{\textwidth}}{\footnotesize{\emph{Note}: The sample size is $n=1,813$.  The parameter ${\tau}_0$ is estimated by the grid search on the 591 equi-spaced points over $[1,60]$. Both $\widehat{\tau}$ and the 95\% confidence interval are based on re-estimation after Step 3b: that is, $\tau$ is estimated again using $(\widetilde{U}_i, \widetilde{\alpha})$ from Step 3b.
  }}
\end{tabular}
\end{table}

Table \ref{tb:all-quantiles} summarizes the estimation results at the 0.25, 0.5, and 0.75 quantiles, respectively. We report the total number of regressors in each model and the number of selected regressors in Step 3b. The change point $\tau$ is estimated by the grid search over 591 equi-spaced points in $[1,60]$. The lower bound value $1 \%$ corresponds to the 1.6 sample percentile of $Q_i$ and the upper bound value $60 \%$, which is about the upper sample quartile of $Q_i$, is the same as one used in  \citet{card2008tipping}. In this empirical example, we report the estimates of $\tau_0$ and the confidence intervals updated after Step 3b (that is, $\tau$ is re-estimated using the estimates of $\alpha_0$ in Step 3b). If this estimate is different from the previous one in Step 2, then we repeat Step 3b and Step 2 until it converges.

The estimation results suggest several interesting points. First, at each quantile, the proposed method
selects sparse representations in all model specifications even when the number of regressors is relatively large. Furthermore, the number of selected regressors does not grow rapidly when we increase the number of possible covariates. It seems that the set of selected covariates overlaps across different dictionaries at each quantile. See Appendix \ref{app:emp-tables-figures} for details on selected regressors.   Second, the estimation results are different across different quantiles, indicating that there may exist heterogeneity in this application. The confidence intervals for $
\tau_0$ at the 0.25 quantile are quite tight in all cases and they provide convincing evidence of the tipping effect. If we look at the case of six-way interactions with 12 control variables, the estimated tipping point is 5.65\% and the estimated jump size is $-5.50\%$. However, this strong tipping effect becomes weaker at the 0.50 and 0.75 quantiles as shown either by wider confidence intervals or by the zero jump size, i.e.\ $\widehat{\delta}=0$.

\begin{table}[htbp]
\footnotesize
\centering
\caption{Estimation Results from Mean Regression}\label{tb:mean}
\begin{tabular}{lcccc}
  \\
  \hline
 & \multirow{2}{*}{No.\  of Reg.\ } & No.\  of    & \multirow{2}{*}{$\widehat{\tau}$} &  \multirow{2}{*}{$\widehat{\delta}$} \\ 
  & & Selected Reg.     &  \\ 
  \hline
{\underline{6 Control Variables, Six-way Interaction}}\\

Untrimmed     &    83        &        50  &  3.25  &     -16.14 \\
Trimmed       &    83        &        41  &  3.25 &        -6.53 \\
\\
{\underline{12 Control Variables, Six-way Interaction}}\\
Untrimmed      &  2529       &        142  &  3.25  &     -15.55 \\
Trimmed      &   2529    &            107  &   3.25  &     -5.19 \\

     \hline
     \multicolumn{5}{p{\textwidth}}{\footnotesize{\emph{Note}: The sample size of the  untrimmed original data is $n=1,813$. The trimmed data drop top and bottom $5\%$ observations based on $\{Y_i\}$ and the sample sizes decreases to $n=1,626$.  The parameter ${\tau}_0$ is estimated by the grid search on the 591 equi-spaced points over $[1,60]$. As in the simulation studies, the tuning parameters are set from Step 1 in median regression.
  }}
\end{tabular}
\end{table}

We now compare the estimation results from quantile regression with those from mean regression, which are reported in  Table \ref{tb:mean} (full estimation results are in Appendix \ref{app:emp-tables-figures}). We show two kinds of mean regression estimates: one with the   untrimmed original data and the other with the trimmed data for which we  drop top and bottom $5\%$ observations based on $\{Y_i\}$. The estimated tipping points are the same between the two datasets but the estimated jump size is much larger with the original data. Figure \ref{fg:15knots} shows the fitted values over $Q_i$ at the sample mean of the six basic covariates. They are from the model of six-way interactions with 12 control variables and the vertical line indicates the location of a tipping point.  The left panel of Figure \ref{fg:15knots} compares the results between the mean and median regression results (without trimming the data) and the right panel shows the the interquartile range of the conditional distribution of $Y_i$ as a function of $Q_i$ given other regressors.  It can be seen that the mean regression estimates are much more volatile around the tipping point than the median regression estimates, although the estimated tipping point is the same. In Figure \ref{fg:trim}, we compare the mean regression estimates with and without trimming. Removing observations with top and bottom 5\% $Y_i$'s stablize the estimates, thus demonstrating that the median regression estimates have the built-in feature that they are more stable with outliers of $Y_i$ than the mean estimates. Finally, looking at the right panel of Figure \ref{fg:15knots}, we can see that the 25 percentile of the conditional distribution drops at the tipping point of 5.65\% but no such change at the 75\% quantile. This shows that the quantile regression estimates can provide insights into \emph{distributional} threshold effects in racial segregation.

\begin{figure}[htbp]
\centering
\caption{Estimation Results: 12 Control Variables and Six-way Interaction}\label{fg:15knots}
\centering
\includegraphics[scale=0.5]{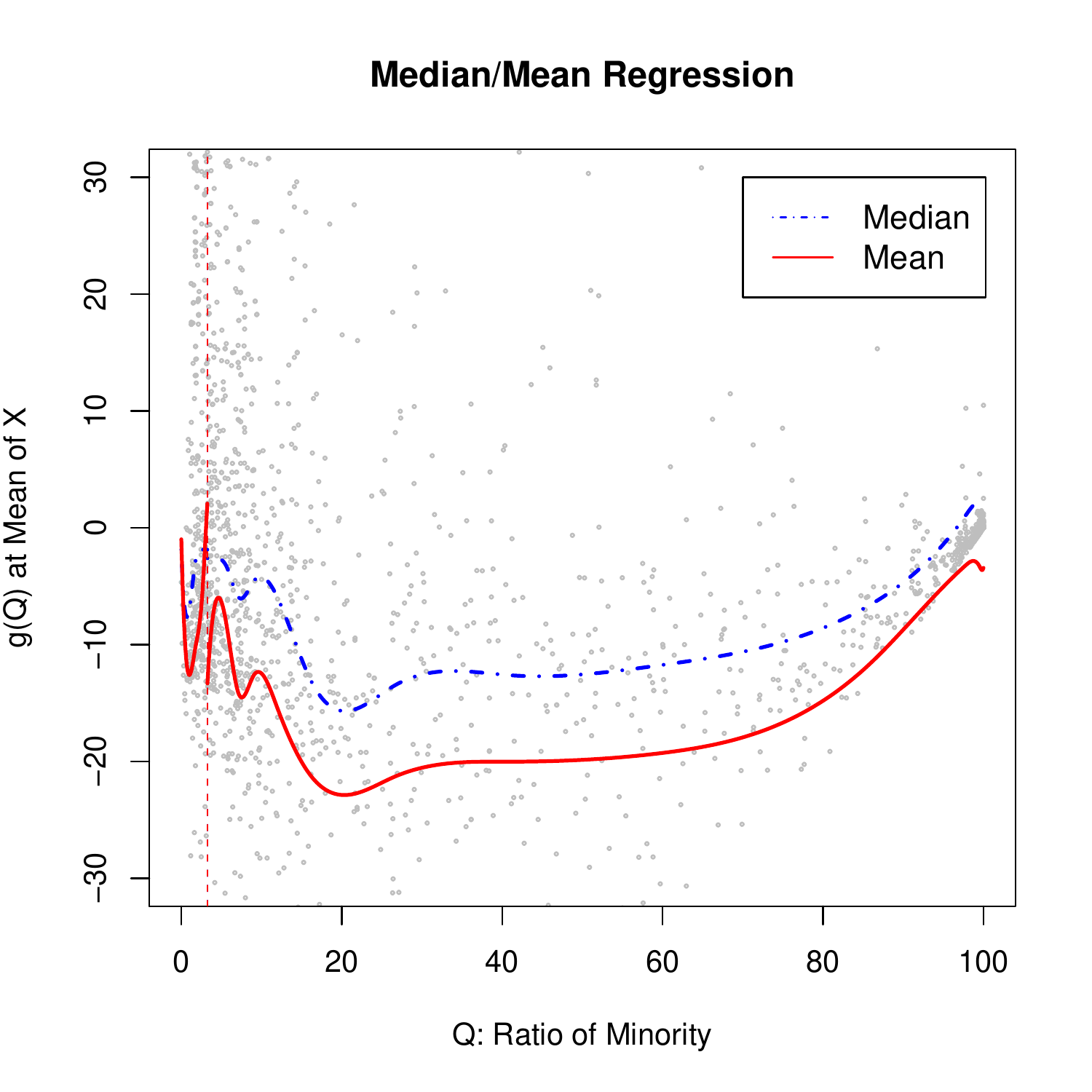}
\includegraphics[scale=0.5]{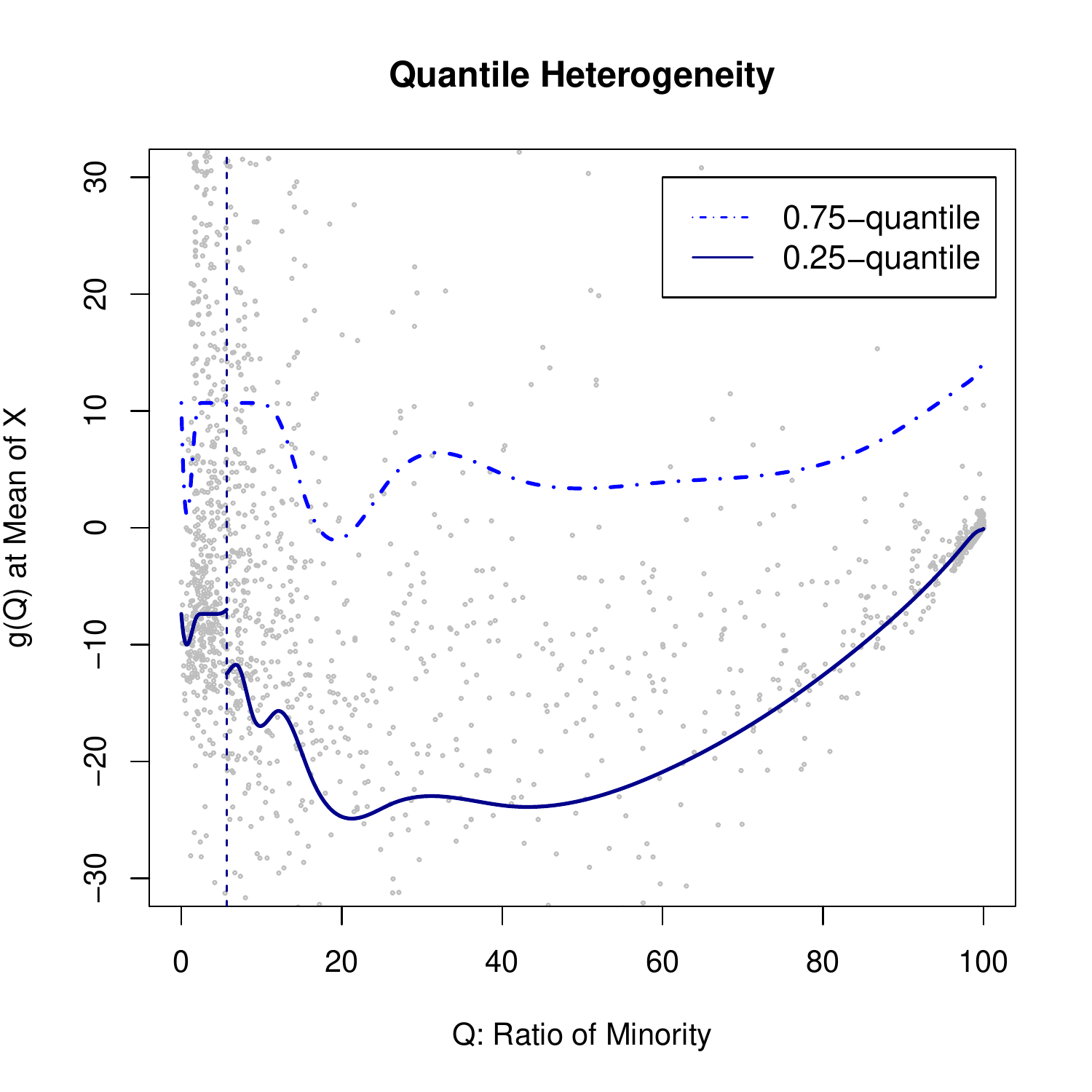}
\end{figure}

\begin{figure}[htbp]
\centering
\caption{Estimation Results: Mean Regression with Untrimmed/Trimmed Data}\label{fg:trim}
\centering
\includegraphics[scale=0.50]{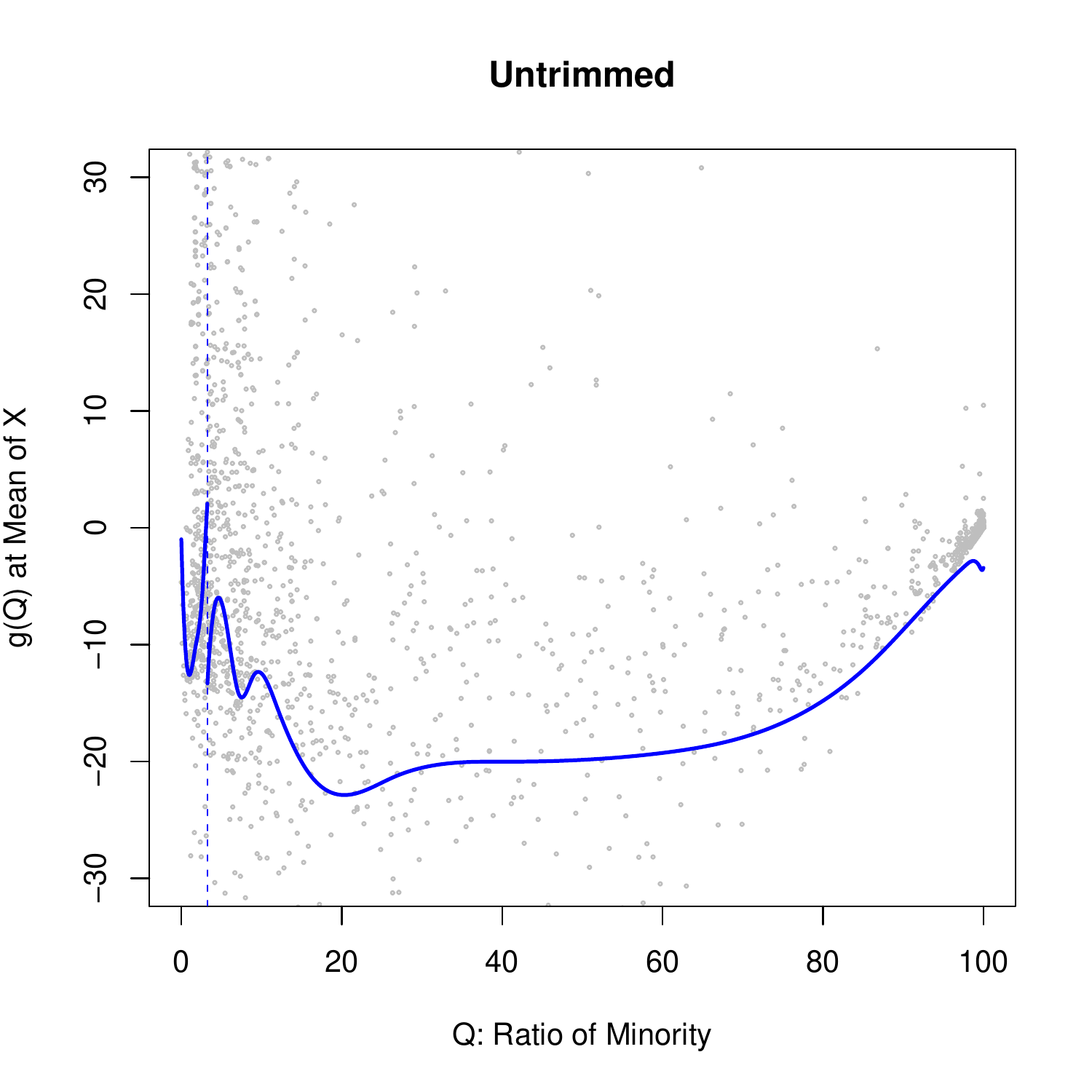}
\includegraphics[scale=0.50]{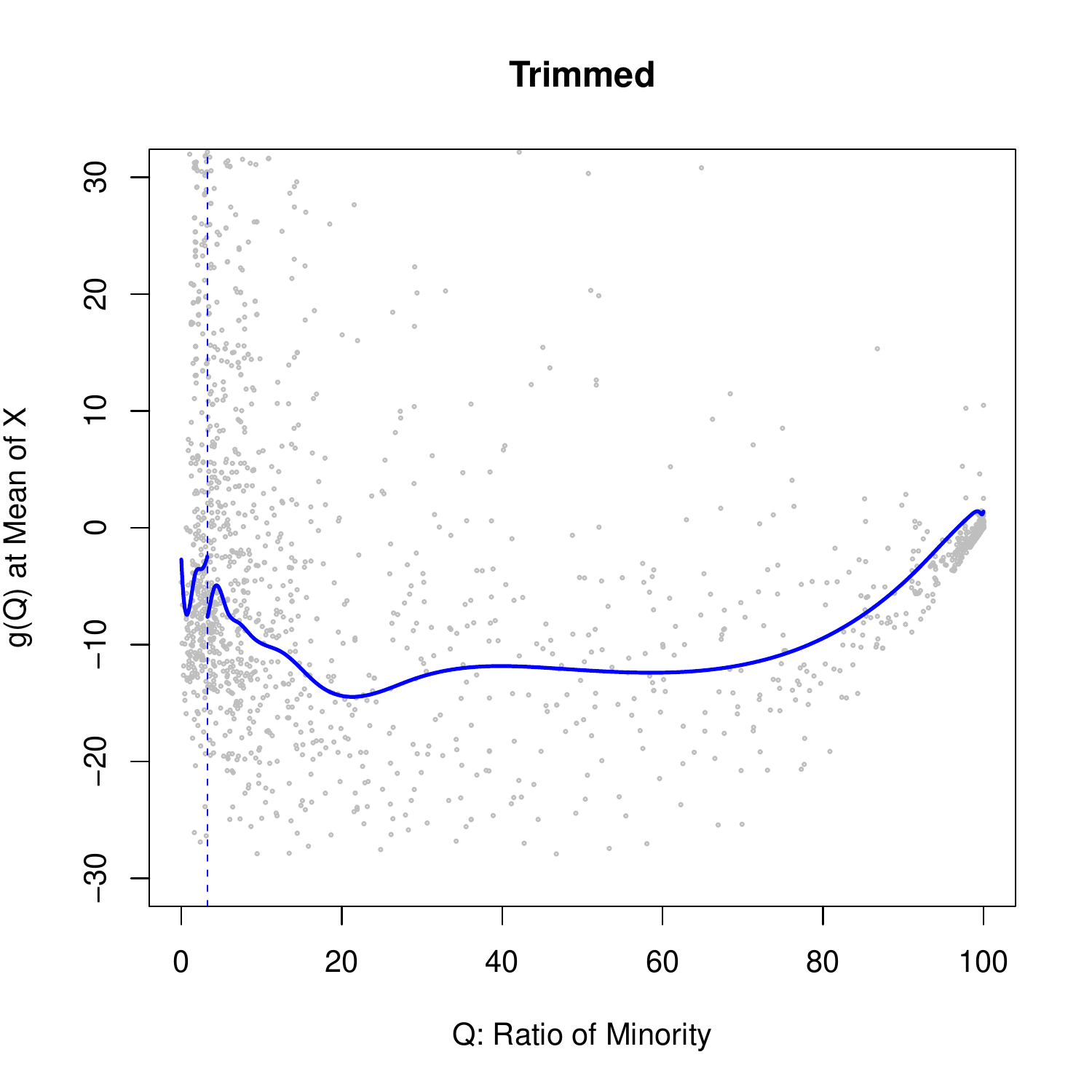}
\end{figure}

In summary, this empirical example shows that the proposed method works well in the
real empirical setup and is robust to outliers compared to the mean regression approach. The estimation results also confirm that there exists a tipping point in the racial segregation at the 0.25 quantile and that the tipping effect is heterogeneous over different quantiles.

\section{Conclusions}\label{sec:conclusion}

In this paper, we have developed $\ell_1$-penalized estimators of a high-dimensional quantile regression model 
with an unknown change point due to a covariate threshold. We have shown among other things that our estimator of the change point achieves an oracle property  without relying on a perfect covariate selection, thereby avoiding the need for  the minimum level condition on the signals of active covariates. 
We have illustrated the usefulness of our estimation methods via  Monte Carlo experiments and an application to tipping  in the racial segregation.

 In a recent working paper,  \cite{leonardi2016} consider a high-dimensional \emph{mean} regression model with 
 multiple change points whose number may grow as  the sample size increases.
  They have proposed a binary search algorithm to choose the number of change points. It is an important future research topic to develop a computationally efficient algorithm to detect multiple changes for high-dimensional quantile regression models.

\newpage

\appendix
\section*{Appendices}

In Appendix \ref{sect:tau:ci-algorithm}, we  provide the algorithm of constructing the confidence interval for $\tau_0$. 
In Appendix \ref{sec:Lasso-theory-assump}, we provide a set of regularity assumptions  to derive asymptotic properties of the proposed estimators in 
Sections   \ref{sec:Lasso-theory-case1} and \ref{sec:Lasso-theory-case2}.
In Appendix \ref{app:add-theory}, we provide sufficient conditions for 
the identification of $(\alpha_0, \tau_0)$ in \eqref{eq2.3add-QR} and show that an improved rate of convergence is possible for the excess risk by taking the second and third steps of estimation. 
To prove the theoretical results in the main text, we consider a general M-estimation framework that includes quantile regression as a special case.
We provide high-level regularity conditions on the loss function in Appendix \ref{high-level-appendix}. Under these conditions, we derive asymptotic properties
and then we verify all the high level assumptions for the quantile regression model in Appendix \ref{sec:proof-lasso}. 
Hence, our general results are of independent interest and can be applicable to other models, for example logistic regression models.
Appendices \ref{sec:add-sim-results-1} and \ref{sec:add-sim-results} provide additional simulation results, and Appendix \ref{app:emp-tables-figures}
gives additional results for the empirical example. 

\section{The Algorithm of Constructing the Confidence Interval for $\tau_0$}\label{sect:tau:ci-algorithm}

The detailed algorithm for constructing the confidence interval based on the Step 2 estimator is as follows:
\begin{enumerate}
\item Simulate two independent Poisson processes $N_1(-h)$ for $h<0$ and $N_2(h)$ for $f>0$ with the same jump rate $\widehat{f}_Q(\widehat{\tau})$ over $h \in [-\overline{H}n,\overline{H}n]$, where $f_Q(\cdot)$ is the pdf of $Q$, $n$ is the sample size, and $\overline{H}>0$ is a large constant. For estimating ${f}_Q(\cdot)$, we use the kernel density estimator with a normal density kernel and the rule-of-thumb bandwidth, $1.06\cdot \min\{s, (Q_{0.75}-Q_{0.25})/1.34 \}\cdot n^{-1/5}$, where $s$ is the standard deviation of $Q$ and $Q_{0.75}-Q_{0.25}$ is the interquartile range of $Q$. A Poisson process $N(h)$ is generated by the following algorithm:
	\begin{enumerate}
{\singlespacing
	\item Set $h=0$ and $k=0$.
	\item Generate $\epsilon$ from the uniform distribution on $[0,1]$.
	\item $h=h+[-(1/\widehat{f}_Q(\widehat{\tau}))\log(\epsilon)]$.
	\item If $h>n\overline{H}$, then stop and goto Step (f). Otherwise, set $k=k+1$ and $h_k=h$.
	\item Repeat Steps (b)--(d).
	\item The algorithm generates $\{h_k\}$ for $k=1,\ldots,\overline{K}$. Transform it into the Poisson process $N(h)\equiv \sum_{k=1}^{\overline{K}} 1\{h_k \le h\}$ for $h\in [0,n\overline{H}]$.
}
	\end{enumerate}
\item  Using the residuals $\{\breve{U}_i\}$ and the estimate $\breve{\delta}$ from Step 1, simulate ${\rho}_{1j}$ for $j=1,\ldots ,N_1(-h)$  
from the empirical distribution  of $ \{\dot{\rho}(\breve{U}_i-X_i^T\breve{\delta})-\dot{\rho}(\breve{U}_i)\}_{i\leq n}$; 
 simulate  ${\rho}_{2j}$ for $j=1,\ldots , N_2(h)$
from the empirical distribution  of    $  \{\dot{\rho}(\breve{U}_i+ X_i^T\breve{\delta})-\dot{\rho}(\breve{U}_i)\}_{i\leq n}$. Here  $\dot{\rho}\left( t\right)  \equiv t\left( \gamma -1\left\{ t\leq 0\right\} \right) $ is the check function as defined in Section \ref{sec:Lasso-theory-case1}.

\item Recall that 
\begin{equation*}
M\left( h\right) \equiv \sum_{i=1}^{N_{1}\left( -h\right) }\rho _{1i}1\left\{
h<0\right\} + \sum_{i=1}^{N_{2}\left( h\right) }\rho _{2i}1\left\{ h\geq
0\right\}
\end{equation*}from Section \ref{sec:Lasso-theory-case1}.
Construct the function $M(\cdot)$ for $h \in [-\overline{H}n,\overline{H}n]$ using values from Steps 1--3 above.  Find the smallest minimizer $h$ of $M(\cdot)$.
\item  Repeat Steps 1--4 above and generate $\{h_1,\ldots,h_B\}$.

\item Construct the $95\%$ confidence interval  of $\widehat{\tau}$ from the empirical distribution of $\{h_b\}$ by $[\widehat{\tau}+h_{0.025}/n,\widehat{\tau}+h_{0.975}/n]$, 
where $h_{0.025}$ and $h_{0.975}$ are $2.5$ and $97.5$ percentiles of $\{h_b\}$, respectively.
\end{enumerate}

It is straightforward to modify the algorithm above for the confidence intervals with Step 3a and Step 3b estimators. 
We set $\overline{H}=0.5$, and $B=1,000$ in this simulation studies.

\section{Assumptions for Oracle Properties}\label{sec:Lasso-theory-assump}

In this section, we list a set of assumptions that will be useful to derive asymptotic properties of the proposed estimators in 
Sections   \ref{sec:Lasso-theory-case1} and \ref{sec:Lasso-theory-case2}.
In the following, we divide our discussions into two important cases: (i) $\delta_0\neq0$ and $\tau_0$ is identified, and (ii) $\delta_0=0$ and thus $\tau_0$ is not identified. The asymptotic properties are derived under both cases.  Note that such a distinction is only needed for presenting  our theoretical results. In practice, we  do not need to know whether $\delta_0=0$  or not.


\begin{assum}[Underlying Distribution]
\label{ass3.4-a}
\begin{enumerate}[label=(\roman*)]
\item\label{ass3.4-a:itm1} 
The conditional distribution $Y|X,Q$ has a continuously differentiable density
function $f_{Y|X,Q}(y|x,q)$  with respect to $y$, whose derivative is denoted
by $\tilde{f}_{Y|X,Q}(y|x,q)$. 
\item\label{ass3.4-a:itm2}  
There are constants $C_1, C_2 >0$ such that for all $(y,x,q)$ in the support of $(Y,X,Q)$, 
\begin{align*}
|\tilde{f}_{Y|X,Q}(y|x,q)| \leq C_1, \quad  f_{Y|X,Q}(x(\tau_0)^T\alpha_0|x,q) \geq C_2.
\end{align*}
\item\label{ass3.4-a:itm5}  
When $\delta_0 \neq 0$, 
$\Gamma(\tau,\alpha_0)$ is positive definite uniformly in a neighborhood of $\tau_0$, where
\begin{equation*}
\Gamma(\tau,\alpha_0)\equiv
\frac{\partial^2 \mathbb{E}[\rho(Y, X_J(\tau)^T\alpha_{0J})]}{\partial\alpha_J\partial%
\alpha_J^T}= \mathbb{E}[X_J(\tau)X_J(\tau)^Tf_{Y|X,Q}(X(\tau)^T\alpha_0|X, Q)].
\end{equation*}
When $\delta_0 = 0$, the matrix
$
\mathbb{E}[X_{J(\beta_0)}X_{J(\beta_0)}^Tf_{Y|X,Q}(X_{J(\beta_{0})}^{T}\beta _{0J(\beta_{0})}|X, Q)]
$
is positive definite.
\end{enumerate}
\end{assum}

Conditions \ref{ass3.4-a:itm1}  and \ref{ass3.4-a:itm2} are standard assumptions for quantile regression models. 
To follow the notation in condition \ref{ass3.4-a:itm5}, recall that $\alpha_J$ denotes the subvector of $\alpha$ whose indices are in $J(\alpha_0)$. Expressions $X_J(\tau)$, $X_{J(\beta_{0})}$, $\alpha_{0J}$ and $\beta _{0J(\beta_{0})}$ can be understood similarly. 
Condition \ref{ass3.4-a:itm5} is a weak condition that imposes non-singularity of the Hessian matrix of the population objective function uniformly in a neighborhood of $\tau_0$ in case of $\delta_0 \neq 0$. This condition reduces to the usual 
non-singularity condition when $\delta_0 = 0$.

\subsection{Compatibility   Conditions}

We now make  an assumption that is an extension of  the well-known \textit{compatibility condition} (see \cite{bulmann}, Chapter 6).  
In particular, the following condition is  a uniform-in-$\tau$ version of the compatibility condition. 
Recall that for a $2p$ dimensional vector $\alpha$, we use $\alpha_J$ and $\alpha_{J^c}$ to denote its subvectors formed by indices in $J(\alpha_0)$ and $\{1,...,2p\}\setminus J(\alpha_0)$, respectively.

\begin{assum} [Compatibility Condition]
\label{ass2.7} 
\begin{enumerate}[label=(\roman*)]
\item\label{ass2.7:itm1}  
When $\delta_0 \neq 0$, there is a neighborhood $\mathcal{T}_0\subset\mathcal{T}$ of $\tau_0$, and  a constant $\phi>0 $ such that for all $\tau\in \mathcal{T}_0$ and all $\alpha  \in\mathbb{R}^{2p}$ satisfying 
$|\alpha_{J^c}|_1\leq 5|\alpha_J|_1$, 
\begin{equation}  \label{eq2.4}
\phi |\alpha_J|_1^2\leq s\alpha^T \mathbb{E}[X(\tau)X(\tau)^T] \alpha.
\end{equation}
\item\label{ass2.7:itm2}  
When $\delta_0 = 0$, there is   a constant $\phi>0 $ such that for all $\tau\in \mathcal{T}$ and all $\alpha  \in\mathbb{R}^{2p}$ satisfying 
$|\alpha_{J^c}|_1\leq 4|\alpha_J|_1$, 
\begin{equation}  \label{eq2.4-delta0}
\phi |\alpha_J|_1^2\leq s\alpha^T \mathbb{E}[X(\tau)X(\tau)^T] \alpha.
\end{equation}
\end{enumerate}
\end{assum}
 
Assumption \ref{ass2.7} requires that the compatibility condition hold uniformly in $\tau$ over a neighbourhood of $\tau_0$ when $\delta_0 \neq 0$ and
 over the entire parameter space $\mathcal{T}$ when $\delta_0 = 0$. 
 Note that this assumption  is
imposed on the population covariance matrix $\mathbb{E}[X(\tau)X(\tau)^T]$; thus, a simple
sufficient condition of Assumption \ref{ass2.7} is that the   smallest
eigenvalue of $\mathbb{E}[X(\tau)X(\tau)^T]$ is bounded away from zero uniformly in $%
\tau$. Even if $p>n$, the population covariance can still be
strictly positive definite while the sample covariance  is not.

\subsection{Restricted Nonlinearity Conditions}

In this subsection, we make an assumption called a \emph{restricted nonlinear condition} to deal with the quantile loss function. 
We extend condition D.4 in \cite{BC11}  to  accommodate  the possible existence of the unknown threshold in our model (specifically, a uniform-in-$\tau$ version of the restricted nonlinear condition as in the compatibility condition). 

Note that when $Q\leq \tau_0$, $X(\tau_0)^T\alpha_0=X^T\beta_0$, while
when $Q>\tau_0$, $X(\tau_0)^T\alpha_0=X^T\theta_0$,
where $\theta_0 \equiv \beta_0+\delta_0$.
 Hence we define the
``prediction balls" with radius $r$ and corresponding centers as follows: 
\begin{align}  \label{eq2.3}
\begin{split}
\mathcal{B}(\beta_0, r) &=\{\beta\in \mathcal{B} \subset \mathbb{R}^p:
\mathbb{E}[(X^T(\beta-\beta_0))^2 1\{Q\leq\tau_0\}]\leq r^2\}, \\
\mathcal{G}(\theta_0, r)&=\{\theta\in  \mathcal{G} \subset \mathbb{R}^p:
\mathbb{E}[(X^T(\theta-\theta_0))^2 1\{Q>\tau_0\}]\leq r^2\},
\end{split}
\end{align}
where $\mathcal{B}$ and $\mathcal{G}$ are parameter spaces for $\beta_0$
and $\theta_0$, respectively. 
To deal with the case that $\delta_0 = 0$, we also
define 
\begin{align}  \label{eq2.3-delta0}
\begin{split}
\mathcal{\tilde{B}}(\beta_0, r, \tau) 
&=\{\beta\in \mathcal{B} \subset \mathbb{R}^p: \mathbb{E}[(X^T(\beta-\beta_0))^2 1\left\{ Q\leq \tau \right\}]\leq r^2\}, \\
\mathcal{\tilde{G}}(\beta_0, r, \tau) 
&=\{\theta \in \mathcal{G} \subset \mathbb{R}^p: \mathbb{E}[(X^T(\theta-\beta_0))^2 1\left\{ Q > \tau \right\}]\leq r^2\}.
\end{split}
\end{align}

\begin{assum}[Restricted Nonlinearity]
\label{ass-rn}
The following holds for the constants $C_1$ and $C_2$ defined in Assumption \ref{ass3.4-a} \ref{ass3.4-a:itm2}.
\begin{enumerate}[label=(\roman*)]
\item\label{ass3.4-a:itm3}  
When $\delta_0 \neq 0$, there exists a constant $r^\ast_{QR} > 0$ such that
\begin{align}\label{nl-condition-qr}
  \inf_{\beta \in \mathcal{B}(\beta_0,r^\ast_{QR}), \beta  \neq \beta_0} \frac{\mathbb{E}[| X^{T}(\beta-\beta_0)|^{2} 1\{Q \leq \tau_0 \}]^{3/2}}{\mathbb{E}[| X^{T}(\beta-\beta_0)|^{3} 1\{Q \leq \tau_0 \}]} \geq r^\ast_{QR} \frac{2C_1}{3C_2} > 0
\end{align}
and that 
\begin{align}\label{nl-condition-qr-2}
  \inf_{\theta \in \mathcal{G}(\theta_0,r^\ast_{QR}), \theta  \neq \theta_0} \frac{\mathbb{E}[| X^{T}(\theta-\theta_0)|^{2} 1\{Q > \tau_0 \}]^{3/2}}{\mathbb{E}[| X^{T}(\theta-\theta_0)|^{3} 1\{Q > \tau_0 \}]} \geq r^\ast_{QR} \frac{2C_1}{3C_2} > 0.
\end{align}
\item\label{ass3.4-a:itm3-delta0}  
When $\delta_0 = 0$, there exists a constant $r^\ast_{QR} > 0$ such that
\begin{align}\label{nl-condition-qr-delta0}
\inf_{\tau \in \mathcal{T}}  \inf_{\beta \in \mathcal{\tilde{B}}(\beta_0, r^\ast_{QR}, \tau), \beta  \neq \beta_0} \frac{\mathbb{E}[| X^{T}(\beta-\beta_0)|^{2} 1\{Q \leq \tau \}]^{3/2}}{\mathbb{E}[| X^{T}(\beta-\beta_0)|^{3} 1\{Q \leq \tau \}]} \geq r^\ast_{QR} \frac{2C_1}{3C_2} > 0
\end{align}
and that 
\begin{align}\label{nl-condition-qr-2-delta0}
\inf_{\tau \in \mathcal{T}}  \inf_{\theta \in \mathcal{\tilde{G}}(\beta_0, r^\ast_{QR}, \tau), \beta  \neq \beta_0} \frac{\mathbb{E}[| X^{T}(\theta-\theta_0)|^{2} 1\{Q > \tau \}]^{3/2}}{\mathbb{E}[| X^{T}(\theta-\theta_0)|^{3} 1\{Q > \tau \}]} \geq r^\ast_{QR} \frac{2C_1}{3C_2} > 0.
\end{align}
\end{enumerate}
\end{assum}


\begin{remark}
As pointed out by \cite{BC11}, 
if $X^{T}c$ follows  a logconcave distribution conditional on $Q$
for any nonzero $c$ (e.g. if the distribution of $X$ is
multivariate normal), then  Theorem 5.22 of \cite{logcave:07} and the H\"{o}lder inequality 
imply that for all $\alpha \in \mathcal{A}$, 
\begin{align*}
\mathbb{E}[|X(\tau_0)^{T}(\alpha-\alpha_0)|^{3}|Q] \leq 
6 \left\{ \mathbb{E}[ \{X(\tau_0)^{T}(\alpha-\alpha_0)\}^{2}|Q] \right\}^{3/2},
\end{align*} 
which provides a sufficient condition for Assumption \ref{ass-rn}.  On the other hand, this assumption can hold more generally since 
equations \eqref{nl-condition-qr}-\eqref{nl-condition-qr-2-delta0} in Assumption \ref{ass-rn} need to hold only locally around true parameters $\alpha_0$.
\end{remark}

\subsection{Additional Assumptions When $\delta_0\neq0$}

We first describe the additional conditions on the distribution of $(X, Q)$. 

\begin{assum}[Additional Conditions on the Distribution of $(X, Q)$]
\label{a:dist-Q-add}
Assume $\delta_0 \neq 0$. In addition, there exists a neighborhood $\mathcal{T}_0 \subset \mathcal{T}$ of $\tau_0$ that satisfies the following.
\begin{enumerate}[label=(\roman*)]
\item\label{a:dist-Q:itm2} 
$Q$ has a density function $f_Q(\cdot)$ that  is continuous and bounded away from
zero on $\mathcal{T}_0$.
\item\label{a:dist-Q:itm3}  
Let 
$\tilde{X}$ denote all the components of $X$ excluding $Q$ in case that $Q$ is an element of $X$. 
The conditional distribution of $Q$ given $\tilde{X}$ has a density function $f_{Q|\tilde{X}}(q|\tilde{x})$ that is   bounded uniformly in both $q\in\mathcal{T}_0$ and $\tilde{x}$.
\item\label{a:threshold}
There exists $M_3 > 0$ such that $M_3^{-1} \leq \mathbb{E}[(X^T\delta_0)^2|Q=\tau] \leq    M_3$ for all $\tau \in \mathcal{T}_{0}$. 
\end{enumerate}
\end{assum}

Condition   \ref{a:dist-Q:itm2} implies that $\mathbb{P} \left\{ \left\vert Q-\tau
_{0}\right\vert <\varepsilon \right\} >0$ for any $\varepsilon >0$,
and condition  \ref{a:dist-Q:itm3} requires that the conditional density of $Q$ given $\tilde{X}$ be uniformly bounded. 
When $\tau_0$ is identified, we require $\delta_0$ to be considerably different from zero. This requirement is 
given in condition  \ref{a:threshold}. Note that this condition is concerned with $\mathbb{E}[ \left( X^{T}\delta _{0}\right) ^{2}|Q=\tau ]$, which is an important quantity to develop asymptotic results
when $\delta_0 \neq 0$.
Note that condition \ref{a:threshold}  is a local condition with respect to $\tau$ in the sense that it has to hold only locally in a neighborhood of $\tau_0$. 

The following additional moment conditions are useful to derive our theoretical results.

\begin{assum}[Moment Bounds]
\label{a:moment}
\begin{enumerate}[label=(\roman*)]
\item\label{a:moment:itm1}
There exist finite positive constants $\widetilde{C}$ and $r$ such that for
all $ \beta \in \mathcal{B}(\beta_0, r)$
and for any $ \theta \in \mathcal{G}(\theta_0, r) $, 
\begin{align*}
 \mathbb{E} [|X^T (\beta -\beta_0 )| 1\{Q > \tau_0 \} ] 	
  &\leq \widetilde{C} \; \mathbb{E}[|X^T (\beta -\beta_0 )| 1\{Q \leq \tau_0 \}], \\ 
 \mathbb{E}[|X^T (\theta -\theta_0 )| 1\{Q \leq \tau_0 \}] 	
  &\leq \widetilde{C}  \; \mathbb{E}[|X^T (\theta -\theta_0 )| 1\{Q  > \tau_0 \}].	
\end{align*}
\item\label{a:moment:itm2}
There exist finite positive constants $M, r$ and the neighborhood $\mathcal{T}_0$ of $\tau_0$  such that 
\begin{align*}  
\mathbb{E} \left[ (X^T[(\theta-\beta)-(\theta_0-\beta_0)])^2 \big|Q  = \tau \right] &\leq M, \\
\mathbb{E}[|X^T(\beta-\beta_0)|\big{|}Q=\tau] &\leq M,   \\
\mathbb{E}[|X^T(\theta-\theta_0)|\big{|}Q=\tau] &\leq M, \\
\sup_{\tau\in\mathcal{T}_0: \tau > \tau_0} \mathbb{E} \left[ |X^T(\beta-\beta_0)| \frac{ 1\{\tau _0 < Q \leq \tau \}}{ (\tau-\tau_0) } \right]
& \leq M  \mathbb{E}[|X^T(\beta-\beta_0)|1\{Q \leq \tau_0\}], \\ 
\sup_{\tau\in\mathcal{T}_0: \tau < \tau_0}  \mathbb{E}\left[ |X^T(\theta-\theta_0)| \frac{1\{\tau < Q \leq \tau_0 \} }{ (\tau_0-\tau) } \right]
& \leq M \mathbb{E}[|X^T(\theta-\theta_0)|1\{Q>\tau_0\}],
\end{align*}
uniformly in  $\beta\in\mathcal{B}(\beta_0, r)$, $\theta\in\mathcal{G}(\theta_0, r)$ and $\tau\in\mathcal{T}_0$.
\end{enumerate}
\end{assum}

\begin{remark}
Condition \ref{a:moment:itm1}
requires that $Q$ have non-negligible support on both sides of $%
\tau_0$. This condition can be viewed as a rank condition for identification of $\alpha_0$. 
In the standard threshold model with a fixed dimension, our condition is trivially satisfied by the rank condition such that 
both $ \mathbb{E} [XX^T 1\{Q \leq \tau_0 \}]$ and $\mathbb{E} [XX^T 1\{Q> \tau_0 \}]$ are positive definite
(see e.g. \cite{chan1993consistency} or \cite{hansen2000sample}). 
	If the rank condition fails, the regression coefficient may not be identified and thus affecting the identification of the change point. In the high-dimensional setup, it is undesirable to impose the same rank condition due to the high-dimensionality. Instead, we replace it with  condition \ref{a:moment:itm1}. 
Condition \ref{a:moment:itm2}  requires the boundedness and certain smoothness
of the conditional expectation functions $\mathbb{E} [ (X^T[(\theta-\beta)-(\theta_0-\beta_0)])^2 \big|Q  = \tau ]$,
$\mathbb{E}[|X^T(\beta-\beta_0)|\big{|}Q=\tau]$,
and $\mathbb{E}[|X^T(\theta-\theta_0)|\big{|}Q=\tau]$, 
 and  prohibits degeneracy in one regime.
The last two inequalities in condition \ref{a:moment:itm2} are satisfied if 
\[
\frac{\mathbb{E}\left[\left|X^{T}\beta\right||Q=\tau\right]}{\mathbb{E}\left[\left|X^{T}\beta\right|\right]}\leq M
\]
 for all $\tau\in\mathcal{T}_{0}$ and for all $\beta$ satisfying  $0<\mathbb{E}\left|X^{T}\beta\right|\leq c$
for some small $c>0$. 
\end{remark}
 
\section{Additional Theoretical Results}\label{app:add-theory} 

In this part of the appendix, we consider the identification of $(\alpha_0, \tau_0)$ in \eqref{eq2.3add-QR} and show that an improved rate of convergence is possible for the excess risk by taking the second and third steps of estimation.

\subsection{Identification}\label{app:iden} 
 
The following theorem establishes the identification of $(\alpha_0, \tau_0)$ in \eqref{eq2.3add-QR}.

\begin{thm}[Identification]\label{iden-thm}
\begin{enumerate}[label=(\roman*)]
\item
Assume that $\delta_0 \neq 0$ and that the $\gamma$-th conditional quantile of $Y$ given $X$ and $Q$ is uniquely given as 
\begin{align}\label{conditional-qr-form}
\text{\emph{Quantile}}_{Y|X,Q}(\tau|X=x,Q=q) = x^T\beta_0+x^T\delta_0 1\{q>\tau_0\}.
\end{align}
\item The distribution of $Q$ is absolutely continuous with respect to Lebesgue measure. 
\item $\tau_0 \in \mathcal{T} \equiv \left[ \underline{\tau },\overline{\tau }\right] $, which is contained in 
a strict interior of the support of $Q$. 
\item For any $\tau_1 \in \mathcal{T}$ satisfying $\tau_1 <\tau_0$, we have that $P(\tau_1<Q\leq\tau_0)>0$; for any $\tau_2 \in \mathcal{T}$ satisfying $\tau_2  >\tau_0$, $P(\tau_0 < Q \leq \tau_2)>0$.
\item For every $\tau  \in \mathcal{T}$, we have that
$\inf_{q\in[\underline{\tau},\bar{\tau}]}\lambda_{\min} \{\mathbb{E}(X(\tau)X(\tau)^T|Q=q)\}>0$,
where $X(\tau) \equiv (X^{T},X^{T}1\{Q>\tau \})^{T}$.
\item $\inf_{q\in[\underline{\tau},\bar{\tau}]}\mathbb{E}((X^T\delta_0)^2|Q=q)>0$.
\end{enumerate}
Then $(\alpha_0, \tau_0)$ is identified. 
\end{thm}

Theorem  \ref{iden-thm} establishes sufficient  conditions under which  $\alpha_{0}$ and $\tau_{0}$ are identified. Conditions (i)-(v) in Theorem  \ref{iden-thm} are standard.  The non-singularity condition (v)  is uniform in $\tau \in \mathcal{T}$  and can be viewed as a natural extension of the usual rank condition in the linear model. Condition (vi) is a condition that imposes that the model is well separated from the case that there is no change point in the model.

\begin{proof}[Proof of Theorem \ref{iden-thm}]
Since the conditional quantile function is uniquely  given as \eqref{conditional-qr-form}, it suffices to show that
\[
X\left(  \tau\right)  ^T\alpha=X\left(  \tau_{0}\right)  ^{\prime
}\alpha_{0}\ \text{a.s.\ } \ \  \Longleftrightarrow \ \ \alpha=\alpha
_{0}\ \text{and }\tau=\tau_{0}.
\]
To begin with, write, assuming $\tau \leq \tau_{0}$, 
\begin{align}\label{iden:eq:1}
\begin{split}
\mathcal{D}\left(  \alpha,\tau\right)  &\equiv X\left(  \tau\right)  ^T%
\alpha-X\left(  \tau_{0}\right)  ^T\alpha_{0} \\
&=X^T\left(
\beta-\beta_{0}\right)  +X^T\left(  \delta-\delta_{0}\right)  1\left\{
Q>\tau\right\}  +X^T\delta_{0}1\left\{  \tau<Q\leq\tau_{0}\right\}.
\end{split}
\end{align}
Now suppose $\mathcal{D}\left(  \alpha,\tau\right)  $ in 
(\ref{iden:eq:1}) is zero a.s. Then, it is also zero on the following event $E$:
\begin{align}\label{event-E-def}
E \equiv \{1\{\tau<Q\leq\tau_0\}=0\}=\{Q\notin(\tau,\tau_0]\}.
\end{align}
on the other hand, $P(E)>0$ because $P(E)=P(Q\notin(\tau,\tau_0])\geq P(Q>\tau_0)>0.$
However, on event $E$,
\[
\mathcal{D}\left(  \alpha,\tau\right)  =X^T\left(  \beta-\beta
_{0}\right)  +X^T\left(  \delta-\delta_{0}\right)  1\left\{
Q>\tau\right\}  = 0 \ \text{a.s.\ } \ \
\]
Thus, we have that 
  $$
  X(\tau)^T(\alpha-\alpha_0)1_{E}=0 \ \text{a.s.\ } \ \
  $$
 This is equivalent to 
 $$
 \mathbb{E}\{[X(\tau)^T(\alpha-\alpha_0)]^21_E\}=  \mathbb{E}\{ \mathbb{E}([X(\tau)^T(\alpha-\alpha_0)]^2|Q)1_E\}=0.
 $$
However, we have that
\begin{align*}
0 &=    \mathbb{E}\{ \mathbb{E}([X(\tau)^T(\alpha-\alpha_0)]^2|Q)1_E\} \\
&\geq \inf_{q\in[\underline{\tau},\bar{\tau}]} \mathbb{E}\{[X(\tau)^T(\alpha-\alpha_0)]^2|Q=q\}P(E) \\
&   \geq \inf_{q\in[\underline{\tau},\bar{\tau}]}\lambda_{\min} \{\mathbb{E}(X(\tau)X(\tau)^T|Q=q)\}P(E)  \|\alpha-\alpha_0\|^2_2.
\end{align*}
This result combined with \eqref{iden:eq:1} implies that
\[
X^T\delta_{0}1\left\{  \tau<Q\leq\tau_{0}\right\}  =0 \ \text{a.s.\ } \ \
\]
This also implies that
\begin{align*}
0 & =\mathbb{E}[ (X^T\delta_0)^21\{\tau<Q\leq\tau_0\}] \\
&=\mathbb{E}\{1\{\tau<Q\leq\tau_0\}\mathbb{E}((X^T\delta_0)^2|Q)\} \\
& \geq\inf_{q\in[\underline{\tau},\bar{\tau}]}\mathbb{E}((X^T\delta_0)^2|Q=q)P(\tau<Q\leq\tau_0).
\end{align*}
Since it is  assumed that $\inf_{q\in[\underline{\tau},\bar{\tau}]}\mathbb{E}((X^T\delta_0)^2|Q=q)>0$, thus $P(\tau<Q\leq\tau_0)=0$. However, we also assume that $P(\tau<Q\leq\tau_0)>0$ if $\tau<\tau_0$. Hence we must have $\tau=\tau_0$.

Now consider the other case, that is $\tau < \tau_0$. In this case, we have that 
\begin{align}\label{iden:eq:2}
\mathcal{D}\left(  \alpha,\tau\right)  
&=X^T\left(
\beta-\beta_{0}\right)  +X^T\left(  \delta-\delta_{0}\right)  1\left\{
Q>\tau\right\}  +X^T\delta_{0}1\left\{  \tau_0 <Q\leq\tau \right\}.
\end{align}
Hence, in this case, modifying the definition of $E$ in \eqref{event-E-def} to be 
\begin{align*}
E \equiv \left\{  1\left\{  \tau_0 <Q\leq\tau
\right\}  =0\right\} =\{Q\notin(\tau_0,\tau]\}.
\end{align*}
and proceeding the arguments identical to those above gives the desired result.
\end{proof} 

\subsection{Improved Risk Consistency}\label{app:add:theory2}
  
The following theorem shows that an improved rate of convergence is possible for the excess risk by taking the second and third steps of estimation. 
Recall that  
\begin{align*}
\omega_{n} \propto  \sqrt{\frac{\log (p \vee n)}{n}} \;.
\end{align*}

\begin{thm}[Improved Risk Consistency]
\label{l2.1-improved}
Let Assumption \ref{a:setting} hold. 
In addition, assume that $|\widehat{\tau}-\tau _{0}|=O_P (n^{-1})$  when $\delta_0 \neq 0$.
 Then, whether $\delta_0 = 0$ or not, 
\[
R\left(\widehat{\alpha},\widehat{\tau}\right)=O_{P}\left(\omega_{n}s\right).
\]
\end{thm}

The proof of this theorem is given in Appendix \ref{proof-theorem-c-2}.
For the sake of not introducing additional assumptions in this section, we have  assumed in Theorem \ref{l2.1-improved} that $|\widehat{\tau}-\tau _{0}|=O_P (n^{-1})$ when $\tau_0$ is identifiable.
Its formal statement  is given by  Theorem \ref{thm:tau-2nd} in Section \ref{sec:Lasso-theory-case1}.

\begin{remark}
As in  Theorem \ref{l2.1}, 
the risk consistency part of Theorem \ref{l2.1-improved} holds whether or not $\delta_0 = 0$.
We obtain the improved rate of convergence in probability for the excess risk by combining the fact that
 our objective function is convex with respect to $\alpha$ given each $\tau$ with  the second-step estimation results that 
(i)  if $\delta \neq 0$, then $\widehat{\tau}$ is within a shrinking local neighborhood of $\tau_0$, and
(ii) when $\delta_{0}=0$, $\widehat{\tau}$ does not affect the excess risk in the sense that $R\left(\alpha_{0},\tau\right) = 0$ 
 for all $\tau \in \mathcal{T}$ .  
\end{remark}

\section{Regularity conditions on the general loss function}\label{high-level-appendix}

Let $Y$ be a scalar variable of outcome and   $X$ be a  vector of $p$-dimensional observed characteristics. Suppose there is an observable scalar variable $Q$ such that  the conditional distribution of $Y$  or some feature of that (given $X$) depends on:
$$
X^T\beta_01\{Q\leq \tau_0\}+ X^T\theta_01\{Q>\tau_0\}=X^T\beta_0+X^T\delta_01\{Q>\tau_0\},
$$
where $\delta_0=\theta_0-\beta_0$.  Let   $\rho:  \mathbb{R}\times\mathbb{R}\rightarrow\mathbb{R}^+$ be a  loss function under consideration, whose analytical form is clear in  specific models.  Suppose  
 the true parameters are defined as the   minimizer of the expected loss:
\begin{equation}\label{eq2.1add-appendix}
(\beta_0,\delta_0,\tau_0)\equiv\argmin_{(\beta,\delta) \in \mathcal{A},\tau \in \mathcal{T}} 
\mathbb{E} \left[ \rho(Y,X^T\beta+X^T%
\delta1\{Q>\tau\}) \right],
\end{equation}
where $\mathcal{A}$ and $\mathcal{T}$ denote the parameter spaces for $(\beta_0, \delta_0)$ and $\tau_0$.   Here $\beta$ represents the components of ``baseline parameters", while $\delta$ represents the structural changes; $\tau$ is the change point value where the structural changes occur, if any.  By construction, $\tau_0$ is not unique when $\delta_0=0$. For each $(\beta,\delta)\in\mathcal{A}$ and $\tau\in\mathcal{T}$,  define $2p\times 1$ vectors:
$$
\alpha \equiv (\beta ^{T},\delta ^{T})^{T}, \quad X(\tau
) \equiv (X^{T},X^{T}1\{Q>\tau \})^{T}.
$$
Then  $X^T\beta+X^T%
\delta1\{Q>\tau\}=X(\tau)^T\alpha$, and by letting 
$\alpha_0 \equiv (\beta_0^T,\delta_0^T)^T$, 
we can write (\ref{eq2.1add-appendix}) more compactly as:
 \begin{equation}\label{eq2.3add}
(\alpha_0,\tau_0)=\argmin_{\alpha \in \mathcal{A},\tau \in \mathcal{T}} 
\mathbb{E} \left[ \rho(Y,X(\tau)^T\alpha) \right].
\end{equation} 
In quantile regression models, for a given quantile $\gamma\in (0,1)$, recall that 
$$\rho(t_1,t_2) = (t_1-t_2)(\gamma-1\{t_1-t_2\leq 0\}).$$

\subsection{When $\delta_0\neq0$ and $\tau_0$ is identified}

For a constant $\eta > 0$, define
\begin{align*}
r_1(\eta) \equiv &\sup_r \Big\{ r: \mathbb{E} \left( \left[  \rho \left( Y,X^{T}\beta \right) -\rho
\left( Y,X^{T}\beta_0 \right) \right] 1\left\{ Q\leq \tau
_{0}\right\} \right) \\
& \;\;\;\;\;\;\;\;\;\;\;
\geq \eta \mathbb{E}[(X^T(\beta-\beta_0))^2 1\{Q\leq\tau_0\}]  
\textrm{ for all $\beta \in \mathcal{B}(\beta_0, r)$}
\Big \}
\end{align*}
and
\begin{align*}
r_2(\eta) \equiv &\sup_r \Big\{ r: \mathbb{E} \left( \left[  \rho \left( Y,X^{T}\theta \right) -\rho
\left( Y,X^{T}\theta_0 \right) \right] 1\left\{ Q >  \tau
_{0}\right\} \right) \\
& \;\;\;\;\;\;\;\;\;\;\;
\geq \eta \mathbb{E}[(X^T(\theta-\theta_0))^2 1\{Q > \tau_0\}]  
\textrm{ for all $\theta \in \mathcal{G}(\theta_0, r)$}
\Big \},
\end{align*}
where $\mathcal{B}(\beta_0, r)$ and $\mathcal{G}(\theta_0, r)$ are defined in  \eqref{eq2.3}.
Note that
$r_1(\eta)$ and $r_2(\eta)$ 
are the maximal radii  over which  the excess risk 
can be bounded below by the quadratic loss on $\{Q \leq \tau_0\}$ and  
$\{Q > \tau_0\}$, respectively.

\begin{assum}
\label{a:obj-ftn}
\begin{enumerate}[label=(\roman*)]
\item\label{a:setting:L-cont}  
Let $\mathcal{Y}$ denote the support of $Y$.  There is a Liptschitz constant $L>0$ such that for all $y \in \mathcal{Y}$, $\rho (y,\cdot)$ is convex,  and  \begin{equation*}
|\rho (y,t_{1})-\rho (y,t_{2})|\leq L|t_{1}-t_{2}|, \forall t_1, t_2\in\mathbb{R}.
\end{equation*}
 \item\label{a:obj-ftn:itm2} 
For all $\alpha \in \mathcal{A}$, almost surely,
\begin{equation*}
\mathbb{E} \left[ \rho (Y,X(\tau _{0})^{T}\alpha )-\rho (Y,X(\tau _{0})^{T}\alpha
_{0})|Q \right]\geq 0.
\end{equation*}%
 
  \item\label{a:obj-ftn:itm3} 
 There exist constants $\eta^\ast > 0$ and $r^\ast > 0$ such that 
$r_1(\eta^\ast) \geq r^\ast$ and
$r_2(\eta^\ast) \geq r^\ast$.
 \item\label{a:obj-ftn:itm4} 
There is a constant $c_0 > 0$ such that for all $\tau \in \mathcal{T}_{0}$,
\begin{align*}
\mathbb{E}\left[ \left( \rho \left( Y,X^{T}\theta _{0}\right) -\rho \left( Y,X^{T}\beta
_{0}\right) \right) 1\left\{ \tau <Q\leq \tau _{0}\right\} \right] & \geq
c_0 \mathbb{E}\left[  (X^{T}\left( \beta _{0}-\theta _{0})\right) ^{2}1\left\{ \tau <Q\leq \tau
_{0}\right\} \right], \\
\mathbb{E}\left[ \left( \rho \left( Y,X^{T}\beta _{0}\right) -\rho \left( Y,X^{T}\theta
_{0}\right) \right) 1\left\{ \tau _{0}<Q\leq \tau \right\} \right] & \geq
c_0 \mathbb{E}\left[ (X^{T}\left( \beta _{0}-\theta _{0})\right) ^{2}1\left\{ \tau
_{0}<Q\leq \tau \right\} \right].
\end{align*}%

\end{enumerate}
\end{assum}

We  focus on a convex Lipchitz  loss function, which is assumed in condition \ref{a:setting:L-cont}.   It might be possible to 
weaken the convexity to a ``restricted strong convexity  condition" as in \cite{loh2013regularized}. For simplicity, we focus on the case of a convex loss, which is satisfied for
quantile regression. However,  unlike the framework of M-estimation in \cite{negahban2012} and \cite{loh2013regularized}, we do allow $\rho(t_1, t_2)$ to be non-differentiable, which admits the quantile regression model as a special case.

Condition \ref{a:obj-ftn:itm3} requires that the excess risk can be bounded below by 
a quadratic function locally when $\tau$ is fixed at $\tau_0$, while 
condition \ref{a:obj-ftn:itm4} is an analogous condition when $\alpha$ is fixed at $\alpha_0$.
conditions \ref{a:obj-ftn:itm3} and \ref{a:obj-ftn:itm4}, combined with the convexity of $\rho(Y,\cdot)$,
helps us derive the rates  of convergence (in the $\ell_1$ norm) of the Lasso estimators of $(\alpha_0,\tau_0)$.
Furthermore, these two conditions separate the conditions for $\alpha $ and $\tau $, making them easier to interpret and verify.

\begin{remark}
Condition \ref{a:obj-ftn:itm3}  of Assumption \ref{a:obj-ftn} is  similar to  \textit{the restricted nonlinear impact (RNI)} condition of  \cite{BC11}.
One may consider an alternative formulation as in  \cite{geer} and \cite{bulmann} (Chapter 6), which is known as the \textit{margin condition}. But the margin condition needs to be adjusted to account for structural changes as in condition \ref{a:obj-ftn:itm4}. It would be an interesting future
research topic to develop a general theory of high-dimensional M-estimation with an unknown sparsity-structural-change with  general margin conditions.
\end{remark}

\begin{remark}\label{remark-tau-whole-set}
Assumptions \ref{a:obj-ftn} \ref{a:obj-ftn:itm4} and \ref{a:dist-Q-add} \ref{a:threshold}  together  imply that 
for all $\tau \in \mathcal{T}_{0}$, there exists a constant $c_0 > 0$ such that 
\begin{align}\label{iden-tau0}
\begin{split}
\Delta_1(\tau) \equiv \mathbb{E}\left[ \left( \rho \left( Y,X^{T}\theta _{0}\right) -\rho \left( Y,X^{T}\beta
_{0}\right) \right) 1\left\{ \tau <Q\leq \tau _{0}\right\} \right] & \geq
c_0^2 \mathbb{P}\left[ \tau <Q\leq \tau
_{0} \right], \\
\Delta_2(\tau) \equiv \mathbb{E}\left[ \left( \rho \left( Y,X^{T}\beta _{0}\right) -\rho \left( Y,X^{T}\theta
_{0}\right) \right) 1\left\{ \tau _{0}<Q\leq \tau \right\} \right] & \geq
c_0^2 \mathbb{P}\left[  \tau
_{0}<Q\leq \tau  \right].
\end{split}
\end{align}%
Note that Assumption \ref{a:obj-ftn} \ref{a:obj-ftn:itm2} implies that
$\Delta_1(\tau)$
is monotonely non-increasing when $\tau<\tau_{0}$, and 
$\Delta_2(\tau)$ is monotonely non-decreasing
when $\tau>\tau_{0}$, respectively. Therefore, Assumptions \ref{a:obj-ftn} \ref{a:obj-ftn:itm2}, \ref{a:obj-ftn} \ref{a:obj-ftn:itm4} and \ref{a:dist-Q-add} \ref{a:threshold}   all together imply that \eqref{iden-tau0} holds  for all $\tau$ in the $\mathcal{T}$, not just
in the $\mathcal{T}_{0}$ since $\mathcal{T}$ is compact.  Equation \eqref{iden-tau0} 
plays an important role in  achieving a super-efficient convergence rate for $\tau_0$, 
since it states the presence of a kink in the expected loss and that of a jump in the loss function  at $\tau_0$.
\end{remark}

We now move to the set of assumptions that are useful to deal with the Step 3b estimator. 
Define
\begin{equation*}
m_j(\tau,\alpha) \equiv \frac{\partial \mathbb{E}[  \rho(Y, X(\tau)^T\alpha)]}{\partial\alpha_j}%
,\quad m(\tau,\alpha) \equiv (m_1(\tau,\alpha),...,m_{2p}(\tau,\alpha))^T.
\end{equation*}
Also, let $m_J(\tau,\alpha) \equiv (m_j(\tau,\alpha): j\in J(\alpha_0))$.

 \begin{assum}
\label{ass3.2} $\mathbb{E}[ \rho(Y, X(\tau)^T\alpha)]$ is three times continuously differentiable with respect
to $\alpha$, and there are constants $c_1, c_2, L>0$  and a neighborhood $\mathcal{T}_0$ of $\tau_0$ such that the following conditions hold:
  for all large $n$ and all $%
\tau\in\mathcal{T}_0$,
\begin{enumerate}[label=(\roman*)]
\item\label{ass3.2:itm1}
 there is $M_n>0$, which  may depend on the sample size $n$, such that 
\begin{equation*}
\max_{j\leq 2p}\left|m_j(\tau,\alpha_0)-m_j(\tau_0,\alpha_0)\right|<M_n{%
|\tau-\tau_0|};
\end{equation*}
\item\label{ass3.2:itm2}
there is $r>0$ such that for all   $\beta\in\mathcal{B}(\beta_0, r)$, $\theta\in\mathcal{G}(\theta_0, r)$,   $\alpha=(\beta^T,\theta^T-\beta^T)^T$ satisfies:
\begin{equation*}
\max_{j\leq 2p}
\sup_{\tau\in\mathcal{T}_0}\left|m_j(\tau,\alpha)-m_j(\tau,\alpha_0)\right|<L
\left|\alpha-\alpha_0\right|_1;
\end{equation*}
\item\label{ass3.2:itm3}   
$\alpha_0$ is in the interior of the parameter space $\mathcal{A}$, and 
 \begin{equation*}
\inf_{\tau\in\mathcal{T}_0} \lambda_{\min} \left( \frac{\partial^2 \mathbb{E} [\rho(Y,
X_J(\tau)^T\alpha_{0J})]}{\partial\alpha_J\partial\alpha_J^T} \right) >c_1,
\end{equation*}
\begin{equation*}
\sup_{\left|\alpha_J-\alpha_{0J}\right|_1<c_2,}\sup_{\tau\in\mathcal{T}_0}%
\max_{i,j,k\in J} \left| \frac{\partial^3\mathbb{E} [\rho(Y, X_J(\tau)^T\alpha_J)]}{%
\partial\alpha_i\partial\alpha_j\partial\alpha_k} \right| < L.
\end{equation*}
\end{enumerate} 
\end{assum}

The score-condition in the population level is expressed by $m(\tau_0,\alpha_0)=0$ 
since $\alpha_0$ is in the interior of  $\mathcal{A}$ by condition \ref{ass3.2:itm3}. 
Conditions \ref{ass3.2:itm1} and \ref{ass3.2:itm2} regulate the continuity of the score $m(\tau,\alpha)$,
and 
condition \ref{ass3.2:itm3} assumes the higher-order differentiability  of 
the expectation of the loss function.
Condition \ref{ass3.2:itm1} requires the Lipschitz continuity of the score function with
respect to the threshold. The Lipschitz constant may grow with $n$, since
it is assumed uniformly over $j\leq 2p$. In many  examples, $M_n$ in fact grows  slowly; as a result, it does not affect the
asymptotic behavior of $\widetilde\alpha$. For quantile regression models, we will show that $M_n=C s^{1/2}$ for
some constant $C>0.$
 Condition \ref{ass3.2:itm2} requires the local equicontinuity at $\alpha_0$ in the 
$\ell_1$ norm of the class 
\begin{equation*}
\{m_j(\tau,\alpha): \tau\in\mathcal{T}_0, j\leq 2p\}.
\end{equation*}

We now establish that Assumptions \ref{a:obj-ftn} and  \ref{ass3.2}  are satisfied
for quantile regression models. 

\begin{lem} \label{l4.1}
Suppose that  Assumptions  \ref{a:setting} and  \ref{ass3.4-a} hold. Then  Assumptions \ref{a:obj-ftn} and  \ref{ass3.2}  are satisfied
by the loss function for the quantile regression model, 
with $M_n=C s^{1/2}$ for some constant $C>0$.
\end{lem}

\subsubsection{Proof of Lemma \protect\ref{l4.1}}

\begin{proof}[Verification of Assumption \ref{a:obj-ftn} \ref{a:setting:L-cont}]
It is straightforward to show that the loss function for quantile regression is convex and satisfies the Liptschitz condition. \end{proof}

\begin{proof}[Verification of  Assumption \ref{a:obj-ftn} \ref{a:obj-ftn:itm2}]
Note that $\rho (Y,t)=h_{\gamma
}(Y-t)$, where $h_{\gamma }(t)=t(\gamma -1\{t\leq 0\})$. By (B.3) of \cite{BC11}, 
\begin{align}\label{Knight-ineq}
h_{\gamma }(w-v)-h_{\gamma }(w)=-v(\gamma-1\{w\leq 0\})+\int_{0}^{v}(1\{w\leq
z\}-1\{w\leq 0\})dz
\end{align}
where $w=Y-X(\tau_0)^{T}\alpha_0$ and $v=X(\tau_0)^{T}(\alpha-\alpha_0)$. Note that 
\begin{equation*}
\mathbb{E}[ v(\gamma -1\{w\leq 0\})|Q]  =-\mathbb{E}[X(\tau_0)^{T}(\alpha-\alpha_0)(\gamma
-1\{U\leq 0\})|Q]=0,
\end{equation*}%
since $\mathbb{P}(U\leq 0|X,Q)=\gamma$. 
Let $F_{Y|X,Q}$ denote the CDF of the conditional distribution $Y|X,Q$.
Then
\begin{eqnarray*}
&&\mathbb{E}\left[ \rho (Y,X(\tau _{0})^{T}\alpha )-\rho (Y,X(\tau _{0})^{T}\alpha_{0})|Q \right] %
\\
&&=\mathbb{E}\left[ \int_{0}^{X(\tau_0)^{T}(\alpha-\alpha_0)}(1\{U \leq z)-1\{U \leq
0 \})dz \Big| Q  \right] \\
&&= \mathbb{E}\left[ \int_{0}^{X(\tau_0)^{T}(\alpha-\alpha_0)}[F_{Y|X,Q}(X(\tau_0)^{T}\alpha_{0}+z|X,Q)-F_{Y|X,Q}(X(\tau_0)^{T}\alpha_{0}|X,Q)]dz\bigg{|}Q\right] \\
&&\geq 0,
\end{eqnarray*}%
where the last inequality follows immediately from the fact that $F_{Y|X,Q}(\cdot|X,Q)$ is the CDF. 
Hence, we have verified Assumption \ref{a:obj-ftn} \ref{a:obj-ftn:itm2}.
\end{proof}

\begin{proof}[Verification of  Assumption \ref{a:obj-ftn} \ref{a:obj-ftn:itm3}]
Following the arguments analogous those used in (B.4) of \cite{BC11},
the mean value expansion implies: 
\begin{eqnarray*}
&&\mathbb{E}\left[ \rho (Y,X(\tau _{0})^{T}\alpha )-\rho (Y,X(\tau _{0})^{T}\alpha_{0})|Q \right] %
\\
&&=
\mathbb{E}\left\{ \int_{0}^{X(\tau_0)^{T}(\alpha-\alpha_0) } \left[z f_{Y|X,Q}(X(\tau _{0})^{T}\alpha_{0}|X,Q)+%
\frac{z^{2}}{2}\tilde{f}_{Y|X,Q}(X(\tau _{0})^{T}\alpha_{0}+t|X,Q) \right]dz\bigg{|}Q\right\} %
\\
&&=\frac{1}{2}(\alpha-\alpha_0) ^{T}\mathbb{E} \left[ X(\tau_0)X(\tau_0)^{T}f_{Y|X,Q}(X(\tau _{0})^{T}\alpha_{0}|X,Q)|Q \right](\alpha-\alpha_0) \\
&&+\mathbb{E}\left\{ \int_{0}^{X(\tau_0)^{T}(\alpha-\alpha_0)}\frac{z^{2}}{2}%
\tilde{f}_{Y|X,Q}(X(\tau _{0})^{T}\alpha_{0}+t|X,Q)dz\bigg{|}Q\right\} 
\end{eqnarray*}%
for some intermediate value $t$ between $0$ and $z$.
By condition \ref{ass3.4-a:itm2} of Assumption \ref{ass3.4-a}, 
\begin{align*}
|\tilde{f}_{Y|X,Q}(X(\tau _{0})^{T}\alpha_{0}+t|X,Q)| \leq C_{1} 
\ \ \text{ and } \ \ f_{Y|X,Q}(X(\tau _{0})^{T}\alpha_{0}|X,Q) \geq C_2.
\end{align*} 
Hence,  taking the expectation on $\{ Q \leq \tau_0 \}$ gives
\begin{align*}
&\mathbb{E}\left[ \rho (Y,X^{T}\beta )-\rho (Y,X^{T}\beta_{0}) 1\{Q \leq \tau_0 \} \right] %
\\
&\geq \frac{C_2}{2} \mathbb{E}[( X^{T}(\beta-\beta_0))^{2} 1\{Q \leq \tau_0 \}]- \frac{C_{1}}{6}\mathbb{E}[|X^{T}(\beta-\beta_0)|^{3}1\{Q \leq \tau_0 \} ] 
\\
&\geq \frac{C_2}{4} \mathbb{E}[| X^{T}(\beta-\beta_0)|^{2} 1\{Q \leq \tau_0 \}],
\end{align*}%
where the last inequality follows from 
\begin{align}\label{nl-term}
\frac{C_2}{4} \mathbb{E}[| X^{T}(\beta-\beta_0)|^{2} 1\{Q \leq \tau_0 \}] \geq \frac{C_{1}}{6}\mathbb{E}[|X^{T}(\beta-\beta_0)|^{3}1\{Q \leq \tau_0 \} ].
\end{align}
To see why \eqref{nl-term} holds,
note that
by \eqref{nl-condition-qr}, for any nonzero $\beta \in \mathcal{B}(\beta_0,r^\ast_{QR})$, 
\begin{align*}
  \frac{\mathbb{E}[| X^{T}(\beta-\beta_0)|^{2} 1\{Q \leq \tau_0 \}]^{3/2}}{\mathbb{E}[| X^{T}(\beta-\beta_0)|^{3} 1\{Q \leq \tau_0 \}]} \geq  r^\ast_{QR} \frac{2C_1}{3C_2} \geq  \frac{2C_1}{3C_2}\mathbb{E}[| X^{T}(\beta-\beta_0)|^{2} 1\{Q \leq \tau_0 \}]^{1/2},
\end{align*}
which proves \eqref{nl-term} immediately. Thus, we have shown that 
Assumption \ref{a:obj-ftn} \ref{a:obj-ftn:itm3} holds for $r_1(\eta)$ with $\eta^\ast = C_2/4$ and $r^\ast = r^\ast_{QR}$ defined 
in \eqref{nl-condition-qr} in Assumption \ref{ass3.4-a}.
The case for $r_2(\eta)$ is similar and hence is omitted. 
\end{proof}

\begin{proof}[Verification of  Assumption \ref{a:obj-ftn} \ref{a:obj-ftn:itm4}]
We again start from \eqref{Knight-ineq} but  with different choices of $(w,v)$ such that  $w=Y-X(\tau_0)^{T}\alpha_0$ and $v = X^{T}\delta
_{0}[1\{Q \leq \tau _{0}\} - 1\{Q>\tau _{0}\}]$.
Then arguments similar to those used in verifying Assumptions \ref{a:obj-ftn} \ref{a:obj-ftn:itm2}-\ref{a:obj-ftn:itm3} yield that 
for $\tau < \tau_0$,
\begin{align}\label{qr-important-lower-bound}
& \mathbb{E}\left[  \rho \left( Y,X^{T}\theta _{0}\right) -\rho \left( Y,X^{T}\beta
_{0}\right) | Q= \tau  \right] \\
&= \mathbb{E}\left\{ \int_{0}^{X^{T}\delta_{0}} z f_{Y|X,Q}(X^{T}\beta
_{0}+t|X,Q) dz\bigg{|}Q = \tau \right\} \\
&\geq \mathbb{E}\left\{ \int_{0}^{\widetilde{\varepsilon}  (X^{T}\delta_{0})} z f_{Y|X,Q}(X^{T}\beta
_{0}+t|X,Q) dz\bigg{|}Q = \tau \right\} \\
&\geq 
\frac{\widetilde{\varepsilon}^2  C_3}{2} \mathbb{E}\left[ (X^{T}\delta_{0})^2 |Q = \tau \right],
\end{align}%
where $t$ is an intermediate value $t$ between $0$ and $z$.
 Thus, we have that 
\begin{align*}
\mathbb{E}\left[ \left( \rho \left( Y,X^{T}\theta _{0}\right) -\rho \left( Y,X^{T}\beta
_{0}\right) \right) 1\left\{ \tau <Q\leq \tau _{0}\right\} \right] & \geq
 \frac{\widetilde{\varepsilon}^2  C_3}{2} \mathbb{E}\left[  (X^{T}\left( \beta _{0}-\theta _{0})\right) ^{2}1\left\{ \tau <Q\leq \tau
_{0}\right\} \right].
\end{align*}
The case that $\tau > \tau_0$ is similar.  
\end{proof}

\begin{proof}[Verification of  Assumption \ref{ass3.2}]
Note that  $$m_j(\tau,\alpha)=\mathbb{E} [X_j(%
\tau)(1\{Y-X(\tau)^T\alpha\leq0\}-\gamma)] .$$
Hence,  $m_j(\tau_0,%
\alpha_0)=0,$ for all $j\leq 2p$.
For condition \ref{ass3.2:itm1} of
Assumption \ref{ass3.2}, for all $j\leq 2p$, 
\begin{align*}
&| m_j(\tau,\alpha_0)- m_j(\tau_0,\alpha_0)| \\
&= |\mathbb{E} X_j(\tau)[1\{Y\leq
X(\tau)^T\alpha_0\}-1\{Y\leq X(\tau_0)^T\alpha_0\}]|\cr 
&=|\mathbb{E} X_j(\tau)[\mathbb{P} (Y%
\leq X(\tau)^T\alpha_0|X,Q)-\mathbb{P} (Y\leq X(\tau_0)^T\alpha_0|X,Q)]|\cr 
&\leq C
\mathbb{E} |X_j(\tau)||(X(\tau)-X(\tau_0))^T\alpha_0| \\
&=C
\mathbb{E} |X_j(\tau)||X^T\delta_0(1\{Q>\tau\}-1\{Q>\tau_0\})|\cr 
&\leq
C\mathbb{E} |X_j(\tau)||X^T\delta_0| (1\{\tau<Q<\tau_0\}+1\{\tau_0<Q<\tau\} )\cr %
&\leq C
(\mathbb{P} (\tau_0<Q<\tau)+\mathbb{P} (\tau<Q<\tau_0))\sup_{\tau, \tau' \in \mathcal{T}_0}\mathbb{E} (|X_j(\tau) X^T\delta_0||Q=\tau') \\
&\leq C
(\mathbb{P} (\tau_0<Q<\tau)+\mathbb{P} (\tau<Q<\tau_0))
\sup_{\tau, \tau' \in \mathcal{T}_0} [\mathbb{E} (|X_j(\tau)|^2||Q=\tau')]^{1/2}  [\mathbb{E} (|X^T\delta_0|^2|Q=\tau') ]^{1/2} \\
&\leq C M_2 K_2 |\delta_0|_2 |\tau_0-\tau| 
\end{align*}
for some constant $C$, where the last inequality follows from conditions 
 \ref{a:dist-Q:itm1},
 \ref{a:setting:itm2} and
 \ref{a:threshold:itm1} of Assumption  \ref{a:setting}. Therefore, we have verified  condition \ref{ass3.2:itm1} of Assumption \ref{ass3.2}
with $M_n = C M_2 K_2 |\delta_0|_2$.

We now verify condition \ref{ass3.2:itm2} of Assumption \ref{ass3.2}. For all $j$ and $%
\tau$ in a neighborhood of $\tau_0$, 
\begin{eqnarray*}
&&| m_j(\tau,\alpha)- m_j(\tau,\alpha_0)|=|\mathbb{E} X_j(\tau)(1\{Y\leq
X(\tau)^T\alpha\}-1\{Y\leq X(\tau)^T\alpha_0\})|\cr &&=|\mathbb{E} X_j(\tau)(\mathbb{P} (Y\leq
X(\tau)^T\alpha|X,Q)-\mathbb{P} (Y\leq X(\tau)^T\alpha_0|X,Q))|\cr &&\leq
C\mathbb{E} |X_j(\tau)||X(\tau)^T(\alpha-\alpha_0)|\leq
C|\alpha-\alpha_0|_1\max_{j\leq 2p, i\leq 2p}\mathbb{E} |X_j(\tau)X_i(\tau)|,
\end{eqnarray*}
which implies the result immediately in view of Assumption \ref{a:setting}. Finally, it is straightforward to verify  condition \ref{ass3.2:itm3} using Assumption \ref{ass3.4-a} \ref{ass3.4-a:itm5}.
\end{proof}

\subsection{When $\delta_0=0$}

We now consider the case when  $\delta_0=0$. In this case, $\tau_0$ is not identifiable, and there is actually  no structural change in the sparsity.  
If $\alpha_0$ is in the interior of $\mathcal{A}$, then
 $m(\tau,\alpha _{0})=0$ for all $\tau\in\mathcal{T}$.

For a constant $\eta > 0$, define
\begin{align*}
\tilde{r}_1(\eta) \equiv &\sup_r \Big\{ r: \mathbb{E} \left( \left[  \rho \left( Y,X^{T}\beta \right) -\rho
\left( Y,X^{T}\beta_0 \right) \right] 1\left\{ Q\leq \tau \right\} \right) \\
& \;\;\;\;\;\;\;\;\;\;\;
\geq \eta \mathbb{E}[(X^T(\beta-\beta_0))^2 1\{Q\leq\tau \}]  
\textrm{ for all $\beta \in \mathcal{\tilde{B}}(\beta_0, r,\tau)$ and for all  $\tau \in \mathcal{T}$}
\Big \}
\end{align*}
and
\begin{align*}
\tilde{r}_2(\eta) \equiv &\sup_r \Big\{ r: \mathbb{E} \left( \left[  \rho \left( Y,X^{T}\theta \right) -\rho
\left( Y,X^{T}\beta_0 \right) \right] 1\left\{ Q >  \tau \right\} \right) \\
& \;\;\;\;\;\;\;\;\;\;\;
\geq \eta \mathbb{E}[(X^T(\theta-\beta_0))^2 1\{Q > \tau \}]  
\textrm{ for all $\theta \in \mathcal{\tilde{G}}(\beta_0, r, \tau)$ and for all  $\tau \in \mathcal{T}$}
\Big \},
\end{align*}
where $\mathcal{\tilde{B}}(\beta_0, r,\tau)$ and $\mathcal{\tilde{G}}(\beta_0, r, \tau)$ are defined in
\eqref{eq2.3-delta0}.

\begin{assum}
\label{ass2.2-delta0} 
\begin{enumerate}[label=(\roman*)]
\item\label{ass2.2-delta0:L-cont}  
Let $\mathcal{Y}$ denote the support of $Y$.  There is a Liptschitz constant $L>0$ such that for all $y \in \mathcal{Y}$, $\rho (y,\cdot)$ is convex,  and  \begin{equation*}
|\rho (y,t_{1})-\rho (y,t_{2})|\leq L|t_{1}-t_{2}|, \forall t_1, t_2\in\mathbb{R}.
\end{equation*}
\item\label{ass2.2-delta0:itm1}
For all $\alpha \in \mathcal{A}$ and for all $\tau \in \mathcal{T}$, almost surely, 
$$
\mathbb{E}[\rho(Y, X(\tau)^T\alpha)-\rho(Y,X^T\beta_0)|Q] \geq 0,  
$$
\item\label{ass2.2-delta0:itm2} 
There exist constants $\eta^\ast > 0$ and $r^\ast > 0$ such that 
$\tilde{r}_1(\eta^\ast) \geq r^\ast$ and
$\tilde{r}_2(\eta^\ast) \geq r^\ast$.

\item\label{ass2.2-delta0:itm2:more} $\mathbb{E}[ \rho(Y, X(\tau)^T\alpha)]$ is three times differentiable with respect
to $\alpha$, and there are universal constants $r>0$ and $L>0$    such that 
 for all   $\beta\in\mathcal{\tilde{B}}(\beta_0, r, \tau)$, $\theta\in \mathcal{\tilde{G}}(\beta_0, r, \tau)$,   $\alpha=(\beta^T,\theta^T-\beta^T)^T$ satisfies:
\begin{equation*}
\max_{j\leq 2p}
 \left|m_j(\tau,\alpha)-m_j(\tau,\alpha_0)\right|<L
\left|\alpha-\alpha_0\right|_1.
\end{equation*}
 for all large $n$ and for all $\tau \in \mathcal{T}$.
\item\label{ass2.2-delta0:itm3}
 $\alpha_0$ is in the interior of the parameter space $\mathcal{A}$, and there are 
constants $c_1$ an $c_2>0$ such that 
 \begin{equation*}
 \lambda_{\min} \left( \frac{\partial^2 \mathbb{E} [\rho(Y,
X_{J(\beta_0)}^T\beta_{0J})]}{\partial\beta_J\partial\beta_J^T} \right) >c_1,
\end{equation*}
\begin{equation*}
\sup_{\left|\alpha_J-\alpha_{0J}\right|_1<c_2,} 
\max_{i,j,k\in J(\beta_0)} \left| \frac{\partial^3\mathbb{E} [\rho(Y, X_{J(\beta_0)}^T\beta_J)]}{%
\partial\beta_i\partial\beta_j\partial\beta_k} \right| < L.
\end{equation*}
\end{enumerate}
\end{assum}

 As in Lemma \ref{l4.1}, we now establish that Assumption \ref{ass2.2-delta0} is satisfied
for quantile regression models when $\delta_0 = 0$. 

\begin{lem} \label{l4.2}
Suppose that  Assumptions  \ref{a:setting} and  \ref{ass3.4-a} hold.  Then   Assumption \ref{ass2.2-delta0} is satisfied.
\end{lem}

\subsubsection{Proof of Lemma \protect\ref{l4.2}}

\begin{proof}[Verification of Assumption \ref{ass2.2-delta0} \ref{ass2.2-delta0:L-cont}]
This is the same as the verification of 
Assumption \ref{a:obj-ftn} \ref{a:setting:L-cont}. 
\end{proof}

\begin{proof}[Verification of Assumption \ref{ass2.2-delta0} \ref{ass2.2-delta0:itm1}]
This can be verified exactly as in verification of 
 Assumption \ref{a:obj-ftn} \ref{a:obj-ftn:itm2} with $\alpha_0 = \beta_0$ now. 
\end{proof}

\begin{proof}[Verification of  Assumption \ref{ass2.2-delta0} \ref{ass2.2-delta0:itm2}]
By the arguments identical to those used to verify Assumption \ref{a:obj-ftn} \ref{a:obj-ftn:itm3}, we have that 
\begin{align*}
&\mathbb{E}\left[ \rho (Y,X^{T}\beta )-\rho (Y,X^{T}\beta_{0}) 1\{Q \leq \tau \} \right] %
\\
&\geq \frac{C_2}{2} \mathbb{E}[( X^{T}(\beta-\beta_0))^{2} 1\{Q \leq \tau \}]- \frac{C_{1}}{6}\mathbb{E}[|X^{T}(\beta-\beta_0)|^{3}1\{Q \leq \tau \} ] 
\\
&\geq \frac{C_2}{4} \mathbb{E}[| X^{T}(\beta-\beta_0)|^{2} 1\{Q \leq \tau \}],
\end{align*}%
where the last inequality follows from \eqref{nl-condition-qr-delta0}. This proves the case for $\tilde{r}_1(\eta)$. 
The case for $\tilde{r}_2(\eta)$ is similar and hence is omitted. 
\end{proof}

\begin{proof}[Verification of Assumptions \ref{ass2.2-delta0} \ref{ass2.2-delta0:itm2:more}
and \ref{ass2.2-delta0:itm3}]
They can be verified similarly as in verification of 
Assumption \ref{ass3.2} in the proof of Lemma Lemma \ref{l4.1}. 
For all $j$ and $\tau \in \mathcal{T}$, 
\begin{eqnarray*}
&&| m_j(\tau,\alpha)- m_j(\tau,\alpha_0)|=|\mathbb{E} X_j(\tau)(1\{Y\leq
X(\tau)^T\alpha\}-1\{Y\leq X(\tau)^T\alpha_0\})|\cr &&=|\mathbb{E} X_j(\tau)(\mathbb{P} (Y\leq
X(\tau)^T\alpha|X,Q)-\mathbb{P} (Y\leq X(\tau)^T\alpha_0|X,Q))|\cr &&\leq
C\mathbb{E} |X_j(\tau)||X(\tau)^T(\alpha-\alpha_0)|\leq
C|\alpha-\alpha_0|_1\max_{j\leq 2p, i\leq 2p}\mathbb{E} |X_j(\tau)X_i(\tau)|,
\end{eqnarray*}
which implies condition \ref{ass2.2-delta0} \ref{ass2.2-delta0:itm2:more} in view of Assumption \ref{a:setting}. It is also straightforward to verify  condition \ref{ass2.2-delta0} \ref{ass2.2-delta0:itm3} using Assumption \ref{ass3.4-a} \ref{ass3.4-a:itm5}.
\end{proof}

\section{Proofs of Theorems}\label{sec:proof-lasso}

Throughout the proofs, we define 
\begin{equation*}
\nu _{n}\left( \alpha ,\tau \right) \equiv\frac{1}{n}\sum_{i=1}^{n}\left[ \rho
\left( Y_{i},X_{i}\left( \tau \right) ^{T}\alpha \right) -
\mathbb{E} \rho \left(Y,X\left( \tau \right) ^{T}\alpha \right) \right].
\end{equation*}%
Without loss of generality let $\nu_{n}\left( \alpha _{J},\tau \right)=n^{-1}\sum_{i=1}^{n}\left[ \rho \left( Y_{i},X_{iJ}\left( \tau
\right) ^{T}\alpha _{J}\right) - \mathbb{E} \rho \left( Y,X_{J}\left( \tau \right)
^{T}\alpha _{J}\right) \right]$. 

In this section, we suppose that Assumptions \ref{a:obj-ftn} and  \ref{ass3.2} 
hold when $\delta_0 \neq 0$ and 
that Assumption \ref{ass2.2-delta0} holds when $\delta_0 = 0$, respectively.

\subsection{Useful Lemmas}

For the positive constant $K_1$ in Assumption  \ref{a:setting} \ref{a:setting:itm1},
define $$c_{np}\equiv\sqrt{\frac{2\log
\left( 4np\right) }{n}}+\frac{K_{1}\log \left( 4np\right) }{n}.$$
Let $\lceil x \rceil$ denote  the smallest integer greater than or equal to
a real number $x$.
The following lemma bounds 
$\nu _{n}\left( \alpha ,\tau \right)$. 


\begin{lem}
\label{lem-emp}  For any positive  sequences $ m_{1n}$ and $ m_{2n}$, and any $\widetilde{\delta}\in(0,1)$, there are constants $L_1, L_2$ and $L_3>0$ such that for   $a_{n}=L_1c_{np}\widetilde{\delta} ^{-1}$, $b_{n}=L_2c_{np} \lceil \log
_{2}\left( m_{2n}/m_{1n}\right) \rceil \widetilde{\delta} ^{-1},$ and $%
c_{n}=L_3n^{-1/2}\widetilde{\delta} ^{-1},$  
\begin{equation}
\mathbb{P} \left\{ \sup_{\tau \in \mathcal{T}}\sup_{\left\vert \alpha -\alpha
_{0}\right\vert _{1}\leq m_{1n}}\left\vert \nu _{n}\left( \alpha ,\tau
\right) -\nu _{n}\left( \alpha _{0},\tau \right) \right\vert \geq
a_{n}m_{1n}\right\} \leq \widetilde{\delta} ,  \label{emp1}
\end{equation}%
\begin{equation}
\mathbb{P} \left\{ \sup_{\tau \in \mathcal{T}}\sup_{m_{1n}\leq \left\vert \alpha
-\alpha _{0}\right\vert _{1}\leq m_{2n}}\frac{\left\vert \nu _{n}\left(
\alpha ,\tau \right) -\nu _{n}\left( \alpha _{0},\tau \right) \right\vert }{%
\left\vert \alpha -\alpha _{0}\right\vert _{1}}\geq b_{n}\right\} \leq
\widetilde{\delta} ,  \label{emp2}
\end{equation}%
and for any $\eta > 0$ and $\mathcal{T}_{\eta }=\left\{ \tau \in \mathcal{T}%
:\left\vert \tau -\tau _{0}\right\vert \leq \eta \right\} $,
\begin{equation}
\mathbb{P} \left\{ \sup_{\tau \in \mathcal{T}_{\eta }}\left\vert \nu _{n}\left(
\alpha _{0},\tau \right) -\nu _{n}\left( \alpha _{0},\tau _{0}\right)
\right\vert \geq c_{n} |\delta_0|_2 \sqrt{\eta} \right\} \leq \widetilde{\delta} .  \label{emp3}
\end{equation}
\end{lem}

\noindent \textbf{Proof of \eqref{emp1}}: 
Let $\epsilon _{1},...,\epsilon _{n}$ denote a Rademacher sequence,
independent of $\{Y_{i},X_{i},Q_{i}\}_{i\leq n}$.
By the symmetrization theorem (see, for example, 
Theorem 14.3 of \cite{bulmann}) and then by the contraction theorem (see, for example, Theorem 14.4 of \cite{bulmann}),
\begin{align*}
& \mathbb{E} \left( \sup_{\tau \in \mathcal{T} }\sup_{\left\vert \alpha -\alpha _{0}\right\vert
_{1}\leq m_{1n}}\left\vert \nu _{n}\left( \alpha ,\tau \right) -\nu
_{n}\left( \alpha _{0},\tau \right) \right\vert \right) \\
& \leq 2\mathbb{E} \left( \sup_{\tau \in \mathcal{T} }\sup_{\left\vert \alpha -\alpha
_{0}\right\vert _{1}\leq m_{1n}}\left\vert \frac{1}{n}\sum_{i=1}^{n}\epsilon
_{i} \left[ \rho \left( Y_{i},X_{i}\left( \tau \right) ^{T}\alpha \right) -\rho
\left( Y_{i},X_{i}\left( \tau \right) ^{T}\alpha _{0}\right) \right] \right\vert
\right) \\
& \leq 4L\mathbb{E}\left( \sup_{\tau \in \mathcal{T} }\sup_{\left\vert \alpha -\alpha
_{0}\right\vert _{1}\leq m_{1n}}\left\vert \frac{1}{n}\sum_{i=1}^{n}\epsilon
_{i}X_{i}\left( \tau \right) ^{T}\left( \alpha -\alpha _{0}\right)
\right\vert \right).
\end{align*}%
Note that 
\begin{align}\label{holder-ineq}
\begin{split}
\lefteqn{
\sup_{\tau \in \mathcal{T} } \sup_{\left\vert \alpha -\alpha _{0}\right\vert _{1}  \leq m_{1n} } \left\vert 
\frac{1}{n}\sum_{i=1}^{n}\epsilon _{i}X_{i}\left( \tau \right) ^{T}\left( \alpha
-\alpha _{0}\right) \right\vert } \\
&= \sup_{\tau \in \mathcal{T} } \sup_{\left\vert \alpha -\alpha _{0}\right\vert _{1}  \leq m_{1n} } \left\vert 
\sum_{j=1}^{2p} \left( \alpha_j
-\alpha _{0j}\right) \frac{1}{n}\sum_{i=1}^{n} \epsilon _{i}X_{ij}\left( \tau \right)  \right\vert \\
&\leq  \sup_{\left\vert \alpha -\alpha _{0}\right\vert _{1}  \leq m_{1n} } 
 \sum_{j=1}^{2p} \left| \alpha_j
-\alpha _{0j}\right| 
\sup_{\tau \in \mathcal{T} }  \max_{j\leq 2p}\left\vert \frac{1}{
n}\sum_{i=1}^{n}\epsilon _{i}X_{ij}\left( \tau \right) \right\vert  
\\
&\leq m_{1n} \sup_{\tau \in \mathcal{T} } 
\max_{j\leq 2p}\left\vert \frac{1}{%
n}\sum_{i=1}^{n}\epsilon _{i}X_{ij}\left( \tau \right) \right\vert .
\end{split}
\end{align}%
For all $\tilde{L} > K_1$, 
\begin{align*}
\mathbb{E}\left( \sup_{\tau \in \mathcal{T} } \max_{j\leq 2p}\left\vert  \sum_{i=1}^{n}\epsilon _{i}X_{ij}\left( \tau \right) \right\vert \right) 
&\leq_{(1)}
\tilde{L} \log \mathbb{E} \left[ \exp \left( \tilde{L}^{-1} \sup_{\tau \in \mathcal{T} } \max_{j\leq 2p}\left\vert  \sum_{i=1}^{n}\epsilon _{i}X_{ij}\left( \tau \right) \right\vert \right) \right] \\
&\leq_{(2)}
\tilde{L} \log \mathbb{E} \left[ \exp \left( \tilde{L}^{-1} \max_{\tau \in 
\{Q_1,\ldots, Q_n \} } \max_{j\leq 2p}\left\vert  \sum_{i=1}^{n}\epsilon _{i}X_{ij}\left( \tau \right) \right\vert \right) \right] \\
&\leq_{(3)} \tilde{L} \log  \left[  4np 
\exp \left( \frac{n}{2(\tilde{L}^2 - \tilde{L}K_1)} \right) \right], 
\end{align*}
where inequality $(1)$ follows from Jensen's inequality,
inequality $(2)$ comes from the fact that $X_{ij}\left( \tau
\right)$ is a step function with jump points on $\mathcal{T} \cap
\{Q_1,\ldots, Q_n \}$, and
inequality $(3)$ is by Bernstein's inequality for the exponential moment of an average (see, for example, Lemma 14.8 of \cite{bulmann}),
combined with the simple inequalities that 
$\exp(|x|) \leq \exp(x) + \exp(-x)$
and  that $\exp(\max_{1 \leq j \leq J} x_j) \leq \sum_{j=1}^J \exp( x_j)$.
Then it follows that
\begin{align}\label{eqa.5}
\mathbb{E}\left( \sup_{\tau \in \mathcal{T} } \max_{j\leq 2p}\left\vert  \frac{1}{n}\sum_{i=1}^{n}\epsilon _{i}X_{ij}\left( \tau \right) \right\vert \right) 
&\leq
  \frac{\tilde{L} \log (4np)}{n}  
+  \frac{1}{2(\tilde{L} - K_1)} = c_{np},
\end{align}
where the last equality follows by  taking $\tilde{L} = K_1 + \sqrt{n/[2\log(4np)]}$.
Thus, by Markov's inequality,
\begin{equation*}
\mathbb{P} \left\{ \sup_{\tau \in \mathcal{T} }\sup_{\left\vert \alpha -\alpha _{0}\right\vert
_{1}\leq m_{1n}}\left\vert \nu _{n}\left( \alpha ,\tau \right) -\nu
_{n}\left( \alpha _{0},\tau \right) \right\vert >a_{n}m_{1n}\right\} \leq
\left( a_{n}m_{1n}\right) ^{-1}4Lm_{1n}c_{np}=\widetilde{\delta},
\end{equation*}
where the last equality follows by setting $L_1 = 4L$.

\noindent \textbf{Proof of \eqref{emp2}}: 
Recall that $\epsilon _{1},...,\epsilon _{n}$ is a Rademacher sequence,
independent of $\{Y_{i},X_{i},Q_{i}\}_{i\leq n}$.
Note that 
\begin{align*}
& \mathbb{E}
\left( \sup_{\tau \in \mathcal{T} }\sup_{m_{1n}\leq \left\vert \alpha -\alpha
_{0}\right\vert _{1}\leq m_{2n}}\frac{\left\vert \nu _{n}\left( \alpha ,\tau
\right) -\nu _{n}\left( \alpha _{0},\tau \right) \right\vert }{\left\vert
\alpha -\alpha _{0}\right\vert _{1}} \right) \\
& \leq_{(1)} 2\mathbb{E}\left( \sup_{\tau \in \mathcal{T} }\sup_{m_{1n}\leq \left\vert \alpha
-\alpha _{0}\right\vert _{1}\leq m_{2n}}\left\vert \frac{1}{n}%
\sum_{i=1}^{n}\epsilon _{i}\frac{\rho \left( Y_{i},X_{i}\left( \tau \right)
^{T}\alpha \right) -\rho \left( Y_{i},X_{i}\left( \tau \right) ^{T}\alpha
_{0}\right) }{\left\vert \alpha -\alpha _{0}\right\vert _{1}}\right\vert
\right) \\
& \leq_{(2)} 2\sum_{j=1}^{k}\mathbb{E} \left( \sup_{\tau \in \mathcal{T} }\sup_{2^{j-1}m_{1n}\leq
\left\vert \alpha -\alpha _{0}\right\vert _{1}\leq 2^{j}m_{1n}}\left\vert 
\frac{1}{n}\sum_{i=1}^{n}\epsilon _{i}\frac{\rho \left( Y_{i},X_{i}\left( \tau
\right) ^{T}\alpha \right) -\rho \left( Y_{i},X_{i}\left( \tau \right)
^{T}\alpha _{0}\right) }{2^{j-1}m_{1n}}\right\vert \right) \\
& \leq_{(3)} 4L\sum_{j=1}^{k}\mathbb{E}\left( \sup_{\tau \in \mathcal{T} }\sup_{2^{j-1}m_{1n}\leq
\left\vert \alpha -\alpha _{0}\right\vert _{1}\leq 2^{j}m_{1n}}\left\vert 
\frac{1}{n}\sum_{i=1}^{n}\epsilon _{i}\frac{X_{i}\left( \tau \right) ^{T}\left(
\alpha -\alpha _{0}\right) }{2^{j-1}m_{1n}}\right\vert \right),
\end{align*}
where inequality (1) is by  the symmetrization theorem, 
inequality (2) holds for  some $k \equiv \lceil \log_{2}\left( m_{2n}/m_{1n}\right) \rceil$, and inequality (3) follows from  the contraction theorem.

Next, the identical arguments showing \eqref{holder-ineq} yield
\begin{align*}
\sup_{2^{j-1}m_{1n}\leq \left\vert \alpha -\alpha _{0}\right\vert _{1}\leq
2^{j}m_{1n}}\left\vert \frac{1}{n}\sum_{i=1}^{n}\epsilon _{i}\frac{X_{i}\left(
\tau \right) ^{T}\left( \alpha -\alpha _{0}\right) }{2^{j-1}m_{1n}}%
\right\vert \leq 2\max_{j\leq 2p}\left\vert \frac{1}{n}\sum_{i=1}^{n}\epsilon
_{i}X_{ij}\left( \tau \right) \right\vert
\end{align*}
uniformly in $\tau \in \mathcal{T}$.
Then, as in the proof of \eqref{emp1},  
Bernstein's and Markov's inequalities imply that 
\begin{equation*}
\mathbb{P} \left\{ \sup_{\tau \in \mathcal{T}}\sup_{m_{1n}\leq \left\vert \alpha -\alpha
_{0}\right\vert _{1}\leq m_{2n}}\frac{\left\vert \nu _{n}\left( \alpha ,\tau
\right) -\nu _{n}\left( \alpha _{0},\tau \right) \right\vert }{\left\vert
\alpha -\alpha _{0}\right\vert _{1}}>b_{n}\right\} \leq
b_{n}^{-1}8Lkc_{np}=\widetilde{\delta},
\end{equation*}
where the last equality follows by setting $L_2 = 8L$.

\noindent \textbf{Proof of \eqref{emp3}}: As above, by the
symmetrization and contraction theorems, we have that 
\begin{eqnarray*}
&&\mathbb{E} \left( \sup_{\tau \in \mathcal{T}_{\eta }}\left\vert \nu _{n}\left(
\alpha _{0},\tau \right) -\nu _{n}\left( \alpha _{0},\tau _{0}\right)
\right\vert \right) \\
&\leq &2\mathbb{E} \left( \sup_{\tau \in \mathcal{T}_{\eta }}\left\vert \frac{1}{n}%
\sum_{i=1}^{n}\epsilon _{i} \left[ \rho \left( Y_{i},X_{i}\left( \tau \right) ^{T}\alpha
_{0}\right) -\rho \left( Y_{i},X_{i}\left( \tau _{0}\right) ^{T}\alpha
_{0}\right) \right\vert \right] \right) \\
&\leq &4L\mathbb{E} \left( \sup_{\tau \in \mathcal{T}_{\eta }}\left\vert \frac{1}{n}%
\sum_{i=1}^{n}\epsilon _{i}X_{i}^{T}\delta _{0}\left( 1\left\{ Q_{i}>\tau
\right\} -1\left\{ Q_{i}>\tau _{0}\right\} \right) \right\vert \right) \\
&\leq &\frac{4 L C_{1} ( M_2 |\delta_0|_2^2 K_2 \eta)^{1/2} }{\sqrt{n}}
\end{eqnarray*}%
for some constant $C_1 < \infty$, where the last inequality is due to Theorem 2.14.1 of 
\cite{VW} with $M_2$  in Assumption \ref{a:setting} \ref{a:threshold:itm1}
and $K_2$ in Assumption \ref{a:setting} \ref{a:dist-Q:itm1}. Specifically, we apply the second inequality of this theorem to the class $\mathcal{F}=\{f(\epsilon,  X, Q, \tau)=\epsilon X^T\delta_0(1\{Q>\tau\}-1\{Q>\tau_0\}), \tau\in\mathcal{T}_{\eta}\}$. Note that   $\mathcal{F}$ is a Vapnik-Cervonenkis class, which has a uniformly bounded entropy integral and thus $ J ( 1, \mathcal{F} ) $ in their theorem is bounded, and that the $ L_2 $ norm of the envelope $ |\epsilon_i X^T_i \delta_0|1\{|Q_i - \tau_0 |< \eta \} $ is proportional to the square root of the length of $ \mathcal{T}_{\eta} $: $$
(E|\epsilon_i X^T_i \delta_0|^21\{|Q_i - \tau_0 |< \eta \})^{1/2}\leq (2M_2 |\delta_0|_2^2 K_2\eta)^{1/2}.
$$ This implies the last inequality with $C_1$ being $\sqrt{2}$ times  the entropy integral of the class $\mathcal{F}$. 
Then, by Markov's inequality, we obtain \eqref{emp3} with $L_3 = 4 L C_{1} ( M_2 K_2)^{1/2}$.

\subsection{Proof of Theorem \protect\ref{l2.1}}

Define $D(\tau )=\mathrm{diag}%
(D_{j}(\tau ):j\leq 2p)$; and also let $D_{0}=D\left( \tau _{0}\right) $   and
$\breve{D}=D\left( \breve{\tau}\right)$.  
It follows from the definition of $(\breve{\alpha},\breve{\tau})$ in \eqref{eq2.2add}
that 
\begin{align}\label{basic-ineq-use}
\frac{1}{n}\sum_{i=1}^{n}\rho (Y_{i},X_{i}(\breve{\tau})^{T}\breve{%
	\alpha })+\kappa _{n}|\breve{D}\breve{\alpha }|_{1}\leq \frac{1}{n}%
\sum_{i=1}^{n}\rho (Y_{i},X_{i}(\tau _{0})^{T}\alpha _{0})+\kappa
_{n}|D_{0}\alpha _{0}|_{1}.
\end{align}
From \eqref{basic-ineq-use}  we obtain the following inequality 
\begin{eqnarray}
R(\breve{\alpha },\breve{\tau}) &\leq &
\left[ \nu _{n}(\alpha_{0},\tau_{0})-\nu _{n}(\breve{\alpha},\breve{\tau })\right] +\kappa
_{n}|D_{0}\alpha _{0}|_{1}-\kappa _{n}|\breve{D}\breve{\alpha }|_{1} 
\notag \\
&=&\left[ \nu _{n} (\alpha _{0},\breve{\tau})-\nu _{n}(\breve{\alpha },\breve{\tau})%
\right] +\left[ \nu _{n} (\alpha _{0},\tau _{0})-\nu _{n}(\alpha _{0},\breve{\tau})%
\right]  \label{basic ineq} \\
&&+\kappa _{n}\left( |D_{0}\alpha _{0}|_{1}-|\breve{D}\breve{\alpha }%
|_{1}\right) .  \notag
\end{eqnarray}%
Note that the second component $\left[ \nu _{n}(\alpha _{0},\tau _{0})-\nu _{n}(%
\alpha _{0},\breve{\tau})\right] =o_P\left[ (s/n)^{1/2}\log n\right] $ due to
\eqref{emp3} of Lemma \ref{lem-emp}
with taking $\mathcal{T}_{\eta } = \mathcal{T}$
by choosing some sufficiently large $\eta > 0$. Thus, we focus on the other two terms in the following
discussion. We consider two cases respectively:  $\left\vert \breve{\alpha}-\alpha _{0}\right\vert _{1}\leq
\left\vert \alpha _{0}\right\vert _{1}$ and $\left\vert \breve{\alpha}-\alpha _{0}\right\vert _{1}>
\left\vert \alpha _{0}\right\vert _{1}$.  

Suppose that $\left\vert \breve{\alpha}-\alpha _{0}\right\vert _{1}\leq
\left\vert \alpha _{0}\right\vert _{1}.$ Then, $\left\vert \breve{D}\breve{%
	\alpha}\right\vert _{1}\leq \left\vert \breve{D}\left( \breve{\alpha}-\alpha
_{0}\right) \right\vert _{1}+\left\vert \breve{D}\alpha _{0}\right\vert
_{1}\leq 2\bar D \left\vert \alpha _{0}\right\vert _{1},$ and%
\begin{equation*}
\left\vert \kappa _{n}\left( |D_{0}\alpha _{0}|_{1}-|\breve{D}\breve{%
	\alpha }|_{1}\right) \right\vert \leq 3\kappa _{n}\bar D \left\vert
\alpha _{0}\right\vert _{1}.
\end{equation*}%
Applying 
\eqref{emp1} in Lemma \ref{lem-emp} with $m_{1n}=\left\vert \alpha _{0}\right\vert _{1}$, we  obtain
\begin{equation*}
\left\vert \nu_n (\alpha _{0},\breve{\tau})-\nu _{n}(\breve{\alpha }, \breve{\tau})\right\vert \leq a_{n}\left\vert \alpha _{0}\right\vert _{1}\leq \kappa
_{n}\left\vert \alpha _{0}\right\vert _{1}  \; \text{ w.p.a.1},
\end{equation*}%
where the last inequality
follows from the fact that 
$a_n \ll \kappa_n$  with $\kappa_n$ satisfying \eqref{kappa_n_rate}.
Thus, the theorem follows in this case.

Now assume that $\left\vert \breve{\alpha}-\alpha _{0}\right\vert
_{1}>\left\vert \alpha _{0}\right\vert _{1}$. 
In this case, apply \eqref{emp2} of Lemma \ref{lem-emp} with $m_{1n}=\left\vert \alpha _{0}\right\vert _{1}$
and $m_{2n}=2M_1 p$, where $M_1$ is defined in Assumption \ref{a:setting}\ref{a:setting:itm2},  to obtain
\begin{align*}
\frac{ \left\vert \nu_n (\alpha _{0},\breve{\tau})-\nu _{n}(\breve{\alpha },\breve{\tau})\right\vert}{\left\vert \breve{\alpha}-\alpha _{0}\right\vert
	_{1}}  \leq b_{n}  
\end{align*}%
with probability arbitrarily close to one for small enough $ \widetilde{\delta} $. Since $b_n \ll \underline{D} \kappa_n$, we have
\begin{align*}
\left\vert \nu_n (\alpha _{0},\breve{\tau})-\nu _{n}(\breve{\alpha },\breve{\tau})\right\vert
\leq \kappa_{n} \underline{D} \left\vert \breve{\alpha}-\alpha _{0}\right\vert
_{1} 
\leq \kappa_{n} \left\vert 
\breve{D}\left( \breve{\alpha}-\alpha _{0}\right) \right\vert _{1}  \; \text{ w.p.a.1}. 
\end{align*}%
Therefore, 
\begin{align*}
R(\breve{\alpha },\breve{\tau}) +o_P\left( n^{-1/2}\log n\right) 
&\leq
\kappa _{n}\left( |D_{0}\alpha _{0}|_{1}-|\breve{D}\breve{\alpha }%
|_{1}\right) +\kappa _{n}\left\vert \breve{D}\left( \breve{\alpha}-\alpha
_{0}\right) \right\vert _{1} \\
&\leq
\kappa _{n}\left( |D_{0}\alpha _{0}|_{1}-|\breve{D}\breve{\alpha}_{J}%
|_{1}\right) +\kappa _{n}\left\vert \breve{D}\left( \breve{\alpha}-\alpha
_{0}\right)_{J} \right\vert _{1},
\end{align*}%
where the last inequality follows from the fact that $\breve{\alpha}-\alpha _{0}=\breve{\alpha}_{J^{C}}+\left( \breve{\alpha}%
-\alpha _{0}\right) _{J}.$ 
Thus, the theorem follows in this case as well.

\subsection{Proof of Theorem \ref{l2.1-improved}}\label{proof-theorem-c-2} 

Define
\begin{align}\label{def-M-star}
M^{*} \equiv 4 \max_{\tau\in T_{n}}\left(R\left(\alpha_{0},\tau\right)+2\omega_{n}\bar{D}\left|\alpha_{0}\right|_{1}\right)/(\omega_{n}\underline{D}),
\end{align}
where $T_{n} \subset \mathcal{T}$ will be specified below.
For each $\tau$, define 
\begin{align}\label{alpha-hat-tau}
 \widehat{\alpha}(\tau)
 &= \text{argmin}_{\alpha \in \mathcal{A}} R_n(\alpha,\tau) + \omega_{n} \sum_{j=1}^{2p} D_j(\tau)|\alpha_j|.
\end{align}
It follows from the definition of $\widehat{\alpha}(\tau)$ in \eqref{alpha-hat-tau} that
\begin{align}\label{basic-ineq-convex-case}
\frac{1}{n}\sum_{i=1}^{n}\rho (Y_{i},X_{i}(\tau)^{T}\widehat{\alpha}(\tau))
+\omega _{n}| D(\tau) \widehat{\alpha}(\tau)|_{1}\leq 
\frac{1}{n}\sum_{i=1}^{n}\rho (Y_{i},X_{i}(\tau)^{T}\alpha _{0})+\omega
_{n}|D(\tau)\alpha _{0}|_{1}.
\end{align}
Next, let 
\[
t\left(\tau\right)=\frac{M^{*}}{M^{*}+\left|\widehat{\alpha}\left(\tau\right)-\alpha_{0}\right|_{1}}
\]
and $\bar{\alpha}\left(\tau\right)=t\left(\tau\right)\widehat{\alpha}\left(\tau\right)+\left(1-t\left(\tau\right)\right)\alpha_{0}$.
By construction, it follows that $\left|\bar{\alpha}\left(\tau\right)-\alpha_{0}\right|_{1}\leq M^{*}$.
And also note that
\begin{equation}
 \left|\bar{\alpha}\left(\tau\right)-\alpha_{0}\right|_{1}\leq M^{*}/2
 \text{ implies } \left|\widehat{\alpha}\left(\tau\right)-\alpha_{0}\right|_{1}\leq M^{*} \label{apbar-aphat}
\end{equation}
since $\bar{\alpha}\left(\tau\right)-\alpha_{0}=t\left(\tau\right)\left(\widehat{\alpha}\left(\tau\right)-\alpha_{0}\right)$. 

For each $\tau$, \eqref{basic-ineq-convex-case} and the convexity of the following map
\[
\alpha \mapsto 
\frac{1}{n}\sum_{i=1}^{n}\rho (Y_{i},X_{i}(\tau)^{T}\alpha)
+\omega _{n}| D(\tau) \alpha|_{1}
\]
implies that
\begin{align*}
&\frac{1}{n}\sum_{i=1}^{n}\rho (Y_{i},X_{i}(\tau)^{T}\bar{\alpha}\left(\tau\right))
+\omega _{n}| D(\tau) \bar{\alpha}\left(\tau\right)|_{1} \\
&\leq
t\left(\tau\right) \left[ \frac{1}{n}\sum_{i=1}^{n}\rho (Y_{i},X_{i}(\tau)^{T}\widehat{\alpha}(\tau))
+\omega _{n}| D(\tau) \widehat{\alpha}(\tau)|_{1} \right] \\
&+ [1-t\left(\tau\right)] \left[ \frac{1}{n}\sum_{i=1}^{n}\rho (Y_{i},X_{i}(\tau)^{T}\alpha _{0})+\omega
_{n}|D(\tau)\alpha _{0}|_{1} \right] \\
&\leq
 \left[ \frac{1}{n}\sum_{i=1}^{n}\rho (Y_{i},X_{i}(\tau)^{T}\alpha _{0})+\omega
_{n}|D(\tau)\alpha _{0}|_{1} \right],
\end{align*}
which in turn yields  the following inequality 
\begin{align} \label{basic ineq new1}
R(\bar{\alpha}(\tau),\tau) + \omega_{n}|D(\tau)   \bar{\alpha}(\tau)|_{1} 
&\leq 
\left[ \nu _{n}(\alpha_{0},\tau)-\nu _{n}(\bar{\alpha}(\tau),\tau )\right] +
R(\alpha_{0},\tau) 
+\omega_{n}  |D(\tau)\alpha _{0}|_{1} .
\end{align}%
Furthermore, by the triangle inequality, \eqref{basic ineq new1} can be written as 
\begin{align}\label{basic ineq new2before}
R(\bar{\alpha}(\tau),\tau) +\omega_{n}\underline{D}\left|\bar{\alpha}(\tau)-\alpha_{0}\right|_{1}
&\leq 
\left[ \nu _{n}(\alpha_{0},\tau)-\nu _{n}(\bar{\alpha}(\tau),\tau )\right] +
R(\alpha_{0},\tau) +2\omega_{n} \overline{D} |\alpha _{0}|_{1}. 
\end{align}
Now let $Z_{M}=\sup_{\tau\in T_{n}}\sup_{\left|\alpha-\alpha_{0}\right|\leq M}\left|\nu_{n}\left(\alpha,\tau\right)-\nu_{n}\left(\alpha_{0},\tau\right)\right|$ for each $M > 0$.
Then, by Lemma \ref{lem-emp}, $Z_{M^\ast}=o_{P}\left(\omega_{n}M^\ast\right)$ by the simple fact that
$\log(np) \leq 2\log(n \vee p)$.
Thus, in view of the definition of $M^\ast$ in \eqref{def-M-star},
the following inequality holds w.p.a.1,
\begin{align}\label{basic ineq new2}
R(\bar{\alpha}(\tau),\tau) +\omega_{n}\underline{D}\left|\bar{\alpha}(\tau)-\alpha_{0}\right|_{1}
&\leq 
 \omega_n \underline{D} M^\ast /2
\end{align}
uniformly in $\tau \in T_n$. 

We can repeat the same arguments for $\widehat{\alpha}(\tau)$ instead of $\bar{\alpha}(\tau)$ due to (\ref{apbar-aphat}) and (\ref{basic ineq new2}),  
to obtain
\begin{align}\label{basic ineq new3}
R(\widehat{\alpha}(\tau),\tau) +\omega_{n}\underline{D}\left|\widehat{\alpha}(\tau)-\alpha_{0}\right|_{1}
&\leq 
\omega_n \underline{D} M^\ast = O (\omega_n s), \text{  w.p.a.1},
\end{align}
uniformly in $\tau \in T_n$.
It remains to show that there exists a set $ T_n $
such that $ \widehat{\tau} \in T_n $ w.p.a.1 and the corresponding $ M^\ast = O(s) $. We split the remaining part of the proof into two cases: $ \delta_0 \neq 0 $ and $ \delta_0 = 0 $.
\medskip

\noindent\textbf{(Case 1: $\delta_{0}\neq0$)}

Let 
\[
T_{n}=\left\{ \tau:\left|\tau-\tau_{0}\right|\leq C n^{-1} \log \log n \right\} 
\]
for some constant $C > 0$.
Note that we assume that if $\delta_0 \neq 0$, then 
\begin{equation*}
|\widehat{\tau}-\tau _{0}|=O_P (n^{-1}),
\end{equation*}
which implies that $ \widehat{\tau} \in T_n $ w.p.a.1.
Furthermore, note that 
\begin{align}  \label{R-alpha0-tau}
\begin{split}
R\left( \alpha_0, \tau\right) & =\mathbb{E} \left(   \left[ \rho\left( Y,X^T\theta_{0}\right) -\rho\left( Y,X^T\beta_{0}\right)
\right] 1\left\{ \tau <Q\leq\tau_{0}\right\} \right) \\
&\;\;\;+\mathbb{E} \left(   \left[ \rho\left( Y,X^T\beta_{0}\right) -\rho\left( Y,X^T\theta_{0}\right)
\right] 1\left\{ \tau_0<Q\leq\tau\right\} \right).
\end{split} 
\end{align}
Combining the fact that the objective function is Liptschitz continuous by Assumptions \ref{a:obj-ftn}  \ref{a:setting:L-cont}
with Assumption \ref{a:setting}, we have that 
\begin{align*}
\sup_{\tau \in T_n} |R\left( \alpha_0, \tau\right)|
&\leq L \sup_{\tau \in T_n} \Big[  \mathbb{E} \left(   |X^T \delta_0| |  1\left\{ \tau <Q\leq\tau_{0}\right\} \right)
+ \mathbb{E} \left(   |X^T \delta_0| |  1\left\{ \tau_0<Q\leq\tau\right\} \right) \Big] \\
&= O\left( \left|\delta_{0}\right|_{1} n^{-1} \log \log n \right) \\
&= o \left(\left|\delta_{0}\right|_{1} \omega_n^2 \right).
\end{align*}
Thus, $M^{*}=O\left(|\alpha_0|_1\right)=O\left(s\right)$.

\medskip

\noindent \textbf{(Case 2: $\delta_{0}=0$)} Redefine $M^{*}$ with $T_n = \mathcal{T}$ as the maximum over the
whole parameter space for $\tau$. Note that when $\delta_{0}=0$, we have that $R\left(\alpha_{0},\tau\right) = 0$  and $M^{*}=O\left(|\alpha_0|_1\right)=O\left(s\right)$.
Therefore, the desired result follows immediately.

\subsection{Proof of Theorem \protect\ref{th2.2}}

\begin{remark}
We first briefly provide  the logic behind the proof of Theorem \ref{th2.2} here.
Note that for all $
\alpha \equiv (\beta^T, \delta^T)^T\in\mathbb{R}^{2p}$ and $\theta \equiv \beta+\delta$,
the excess risk has the following decomposition: when $\tau_1<\tau_0$,
\begin{align}  \label{eq2.1}
\begin{split}
R\left( \alpha, \tau_1\right) & =\mathbb{E} \left( 
\left[\rho\left( Y,X^T\beta\right)
-\rho\left( Y,X^T\beta _{0}\right) \right] 1\left\{ Q\leq\tau_1\right\} \right) \\
&\;\;\;+\mathbb{E} \left(  \left[ \rho\left( Y,X^T\theta\right) -\rho\left( Y,X^T\theta_{0}\right)
\right] 1\left\{ Q>\tau_{0}\right\} \right) \\
&\;\;\;+\mathbb{E} \left(   \left[ \rho\left( Y,X^T\theta\right) -\rho\left( Y,X^T\beta_{0}\right)
\right] 1\left\{ \tau_1<Q\leq\tau_{0}\right\} \right),
\end{split} 
\end{align}
and when $\tau_2>\tau_0$, 
\begin{align}  \label{eq2.2}
\begin{split}
R\left( \alpha, \tau_2\right) & =\mathbb{E} \left( 
\left[\rho\left( Y,X^T\beta\right)
-\rho\left( Y,X^T\beta _{0}\right) \right] 1\left\{ Q\leq\tau_0\right\} \right) \\
&\;\;\;+\mathbb{E} \left(  \left[ \rho\left( Y,X^T\theta\right) -\rho\left( Y,X^T\theta_{0}\right)
\right] 1\left\{ Q>\tau_{2}\right\} \right) \\
&\;\;\;+\mathbb{E} \left(   \left[ \rho\left( Y,X^T\beta\right) -\rho\left( Y,X^T\theta_{0}\right)
\right] 1\left\{ \tau_0<Q\leq\tau_{2}\right\} \right).
\end{split} 
\end{align}
The key observations are that  all the six terms
in the above decompositions are non-negative, and are   stochastically negligible when taking $%
\alpha =\breve{\alpha}$, and $\tau _{1}=\breve{\tau}$ if $\breve{\tau}<\tau _{0}$
or $\tau _{2}=\breve{\tau}$ if $\breve{\tau}>\tau _{0}.$ This follows from the
risk consistency of $R(\breve{\alpha},\breve{\tau})$. 
Then, the identification conditions for $\alpha_0$ and $\tau_0$ (Assumptions \ref{a:obj-ftn} \ref{a:obj-ftn:itm2}-\ref{a:obj-ftn:itm4}), along with 
Assumption \ref{a:moment} \ref{a:moment:itm1},
 are useful to show that the risk consistency implies  the consistency of 
$\breve \tau$. 
\end{remark}

\begin{proof}[Proof of Theorem \protect\ref{th2.2}]
Recall from \eqref{eq2.2} that for all $\alpha=(\beta^T, \delta^T)^T\in\mathbb{R}^{2p}$ and $%
\theta=\beta+\delta$, the excess risk has the following decomposition: when $%
\tau>\tau_0$, 

\begin{align}  \label{eqaa2.2}
\begin{split}
R\left( \alpha, \tau\right) & =\mathbb{E} \left( 
\left[\rho\left( Y,X^T\beta\right)
-\rho\left( Y,X^T\beta _{0}\right) \right] 1\left\{ Q\leq\tau_0\right\} \right) \\
&\;\;\;+\mathbb{E} \left(  \left[ \rho\left( Y,X^T\theta\right) -\rho\left( Y,X^T\theta_{0}\right)
\right] 1\left\{ Q>\tau \right\} \right) \\
&\;\;\;+\mathbb{E} \left(   \left[ \rho\left( Y,X^T\beta\right) -\rho\left( Y,X^T\theta_{0}\right)
\right] 1\left\{ \tau_0<Q\leq\tau \right\} \right).
\end{split} 
\end{align}
We split the proof into five steps.

\noindent
\textbf{Step 1}: 
All the three terms on the right
hand side (RHS) of \eqref{eqaa2.2} are nonnegative. As a consequence,  all the three
terms on the RHS of (\ref{eqaa2.2}) are bounded by $R(\alpha,\tau)$.

\begin{proof}[Proof of Step 1]
Step 1 is implied by the  condition that $\mathbb{E}[\rho(Y, X(\tau_0)^T\alpha)-\rho(Y,
X(\tau_0)^T\alpha_0)|Q]\geq0$ a.s. for all $\alpha 
\in \mathcal{A}$.  To see this, the first two terms are
nonnegative by simply multiplying $\mathbb{E}[\rho(Y, X(\tau_0)^T\alpha)-\rho(Y,
X(\tau_0)^T\alpha_0)|Q]\geq0$ with $1\{Q\leq\tau_0\}$ and $1\{Q>\tau\}$
respectively. To show that the third term is nonnegative for all $\beta\in%
\mathbb{R}^p$ and $\tau>\tau_0,$ set $\alpha=(\beta/2,\beta/2)$ in the
inequality $1\{\tau_0<Q\leq\tau\}\mathbb{E}[\rho(Y, X(\tau_0)^T\alpha)-\rho(Y,
X(\tau_0)^T\alpha_0)|Q]\geq0$. Then we have that  
\begin{equation*}
1\{\tau_0<Q\leq\tau\} \mathbb{E}[\rho(Y, X^T(\beta/2+\beta/2))-\rho(Y, X
^T\theta_0)|Q]\geq0,
\end{equation*}
which yields the nonnegativeness of the third term. 
\end{proof}

\noindent
\textbf{Step 2}:  Let $a\vee b=\max(a, b)$ and $a \wedge b=\min(a, b)$. Prove:
\begin{align*}
\mathbb{E} \left[ |X^{T}(\beta -\beta _{0})|1\{Q\leq \tau _{0}\} \right] 
\leq \frac{1}{\eta^\ast r^\ast} R(\alpha ,\tau) \vee
\left[ \frac{1}{\eta^\ast} R(\alpha ,\tau) \right]^{1/2}.
\end{align*}

\begin{proof}[Proof of Step 2]

Recall that
\begin{align*}
r_1(\eta) \equiv &\sup_r \Big\{ r: \mathbb{E} \left( \left[  \rho \left( Y,X^{T}\beta \right) -\rho
\left( Y,X^{T}\beta_0 \right) \right] 1\left\{ Q\leq \tau
_{0}\right\} \right) \\
& \;\;\;\;\;\;\;\;\;\;\;
\geq \eta \mathbb{E}[(X^T(\beta-\beta_0))^2 1\{Q\leq\tau_0\}]  
\textrm{ for all $\beta \in \mathcal{B}(\beta_0, r)$}
\Big \}.
\end{align*}
For notational simplicity, write $$\mathbb{E}[(X^T(\beta-\beta_0))^2 1\{Q\leq\tau_0\}]  \equiv \|\beta-\beta_0\|_{q}^2,$$
and 
$$
F(\delta) \equiv \mathbb{E} \left( \left[  \rho \left( Y,X^{T}(\beta_0+\delta) \right) -\rho
\left( Y,X^{T}\beta_0 \right) \right] 1\left\{ Q\leq \tau
_{0}\right\} \right). 
$$ 
Note that $ F(\beta-\beta_0)=\mathbb{E} \left( \left[  \rho \left( Y,X^{T}\beta \right) -\rho
\left( Y,X^{T}\beta_0 \right) \right] 1\left\{ Q\leq \tau
_{0}\right\} \right)$, and $\beta\in\mathcal{B}(\beta_0,r)$ if and only if $\|\beta-\beta_0\|_q\leq r.$

For any $\beta$, if $\|\beta-\beta_0\|_q\leq r_1(\eta^*)$, then by the definition of $r_1(\eta^*)$, we have:
$$
F(\beta-\beta_0)\geq \eta^*\mathbb{E}[(X^T(\beta-\beta_0))^2 1\{Q\leq\tau_0\}]. 
$$
If $\|\beta-\beta_0\|_q> r_1(\eta^*)$,  let $t=r_1(\eta^*)\|\beta-\beta_0\|_q^{-1}\in(0,1)$. Since $F(\cdot)$ 
is convex, and $F(0)=0$, we have $F(\beta-\beta_0)\geq t^{-1}F(t(\beta-\beta_0))$.  Moreover, 
define 
$$
\check\beta=\beta_0+r_1(\eta^*)\frac{\beta-\beta_0}{\|\beta-\beta_0\|_q},
$$
 then $\|\check\beta-\beta_0\|_q=r_1(\eta^*)$ and $t(\beta-\beta_0)=\check\beta-\beta_0.$ 
  Hence still by the definition of $r_1(\eta^*)$, 
 $$
 F(\beta-\beta_0)\geq \frac{1}{t}F(\check\beta-\beta_0)\geq \frac{\eta^*}{t}\mathbb{E}[(X^T(\check\beta-\beta_0))^2 1\{Q\leq\tau_0\}]=\eta^*r_1(\eta^*)\|\beta-\beta_0\|_q.
 $$
Therefore, by Assumption \ref{a:obj-ftn} \ref{a:obj-ftn:itm3}, and Step 1, 
 \begin{align*}
R(\alpha ,\tau) & \geq \mathbb{E} \left( \left[  \rho \left( Y,X^{T}\beta \right) -\rho
\left( Y,X^{T}\beta_0 \right) \right] 1\left\{ Q\leq \tau
_{0}\right\} \right) \\
 &\geq \eta^\ast \mathbb{E}[(X^T(\beta-\beta_0))^2 1\{Q\leq\tau_0\}]  
\wedge \eta^\ast r^\ast \{\mathbb{E}[(X^T(\beta-\beta_0))^2 1\{Q\leq\tau_0\}]\}^{1/2} \\
&\geq
\eta^\ast \left( \mathbb{E} \left[ |X^{T}(\beta -\beta _{0})|1\{Q\leq \tau _{0}\} \right] \right)^{2} 
\wedge \eta^\ast r^\ast  \mathbb{E} \left[ |X^{T}(\beta -\beta _{0})|1\{Q\leq \tau _{0}\} \right],
\end{align*}
where the last inequality follows from 
 Jensen's inequality. 
\end{proof}

\noindent
\textbf{Step 3}: For any $r>0$, w.p.a.1, $\breve\beta\in\mathcal{B}(\beta_0, r)$ and  $\breve\theta\in\mathcal{G}(\theta_0, r)$.

\begin{proof}[Proof of Step 3]
Suppose that $\breve \tau >\tau_0$.
The proof of Step 2
implies that when $\tau > \tau_0$,
\begin{align*}
\mathbb{E} \left[ (X^{T}(\beta -\beta _{0}))^2 1\{Q\leq \tau _{0}\} \right] 
\leq \frac{R(\alpha ,\tau)^2}{(\eta^\ast r^\ast)^2}  \vee
\frac{R(\alpha ,\tau)}{\eta^\ast}.
\end{align*}
For any $r>0$, note that $R(\breve\alpha,\breve\tau)=o_P(1)$ implies that the event $R(\breve\alpha,\breve\tau)<r^2$ holds w.p.a.1. 
Therefore, we have shown that $\breve\beta\in\mathcal{B}(\beta_0,r)$. 

We now show that $\breve\theta\in\mathcal{G}(\theta_0,r)$. 
 When $\tau > \tau_0$, we have that
\begin{align*}
R(\alpha ,\tau) 
& \geq_{(1)} \mathbb{E} 
\left( \left[  \rho \left( Y,X^{T}\theta \right) -\rho \left( Y,X^{T}\theta_0 \right) \right] 1\left\{ Q > \tau \right\} \right) \\
& = \mathbb{E} 
\left( \left[  \rho \left( Y,X^{T}\theta \right) -\rho \left( Y,X^{T}\theta_0 \right) \right] 1\left\{ Q > \tau_0 \right\} \right) \\
&- \mathbb{E} 
\left( \left[  \rho \left( Y,X^{T}\theta \right) -\rho \left( Y,X^{T}\theta_0 \right) \right] 1\left\{ \tau_0 < Q \leq \tau \right\}
 \right)  \\
& 
\geq_{(2)}  \eta^\ast  \mathbb{E} \left[ |X^{T}(\theta -\theta_{0})|^2 1\{Q > \tau_0\} \right]  
\wedge \eta^\ast r^\ast  \left( \mathbb{E} \left[ |X^{T}(\theta -\theta _{0})|^2 1\{Q > \tau_0\} \right] \right)^{1/2} \\
&\;\;\;\;\;\ -  \mathbb{E} 
\left( \left[  \rho \left( Y,X^{T}\theta \right) -\rho \left( Y,X^{T}\theta_0 \right) \right] 1\left\{ \tau_0 < Q \leq \tau \right\} \right), 
\end{align*}
where (1) is from \eqref{eq2.2}  
and (2) can be proved using arguments similar to those used in the proof  of Step 2.
This implies that 
\begin{align*}
\mathbb{E} \left[ (X^{T}(\theta -\theta _{0}))^2 1\{Q >  \tau_0 \} \right] 
\leq \frac{\tilde{R}(\alpha ,\tau)^2}{(\eta^\ast r^\ast)^2}  \vee
\frac{\tilde{R}(\alpha ,\tau)}{\eta^\ast},
\end{align*}
where $\tilde{R}(\alpha ,\tau) \equiv 
R(\alpha ,\tau) +\mathbb{E} 
\left( \left[  \rho \left( Y,X^{T}\theta \right) -\rho \left( Y,X^{T}\theta_0 \right) \right] 1\left\{ \tau_0 < Q \leq \tau \right\} \right)$.
Thus,  it suffices to show that 
$
\tilde{R}(\breve \alpha ,\breve \tau) = o_P(1)
$
in order to establish that $\breve\theta\in\mathcal{G}(\theta_0,r)$. 
Note that
for some constant $C > 0$,
\begin{align*}
& \mathbb{E}\left[ (\rho(Y,X^T\theta)-\rho(Y,X^T\theta_0))1\{\tau_0<Q\leq\tau\} \right] \\
&\leq_{(1)}  L \mathbb{E}\left[ |X^T(\theta-\theta_0)|1\{\tau_0<Q\leq\tau\} \right] \\
& \leq_{(2)}  L|\theta-\theta_0|_1 \mathbb{E} \left[ \max_{j\leq p}|\tilde{X}_j|1\{\tau_0<Q\leq\tau\} \right]
+  L|\theta-\theta_0|_1 \mathbb{E} \left[ |Q|1\{\tau_0<Q\leq\tau\} \right]\\
&\leq_{(3)}   L|\theta-\theta_0|_1 \mathbb{E} \left[ \max_{j\leq p}|\tilde{X}_j| \sup_{\tilde{x}} \mathbb{P}(\tau_0<Q\leq\tau|\tilde{X}=\tilde{x})\right] 
+  L|\theta-\theta_0|_1 \mathbb{E} \left[ |Q|1\{\tau_0<Q\leq\tau\} \right]\\
& \leq_{(4)} C (\tau-\tau_0)|\theta-\theta_0|_1 \mathbb{E} \left\{ \left[ \max_{j\leq p}|\tilde{X}_j| \right] + 1 \right\},
\end{align*}
where (1) is by the Lipschitz continuity of $\rho(Y,\cdot)$,
(2) is from the fact that $|X^T(\theta-\theta_0)| \leq |\theta-\theta_0|_1 (\max_{j\leq p}|\tilde{X}_j| + |Q|)$,
(3) is by taking the conditional probability, 
and
(4)  is from Assumption \ref{a:dist-Q-add} \ref{a:dist-Q:itm3}. 

By the expectation-form of the Bernstein inequality (Lemma 14.12 of \cite{bulmann}), 
$\mathbb{E}[ \max_{j\leq p}|X_j| ] \leq K_1\log (p+1)+\sqrt{2\log (p+1)}$.  By (\ref{aphat}), which will be shown below, $|\breve\theta-\theta_0|_1=O_P(s)$.  Hence by (\ref{eqa.10add}) which will  also be shown below, when $\breve\tau>\tau_0$,  $$|\breve\tau-\tau_0| |\breve\theta-\theta_0|_1 \mathbb{E}[ \max_{j\leq p}|X_j|]=O_P(\kappa_ns^2\log p)=o_P(1).$$ Note that when $\breve\tau>\tau_0$, the proofs of  (\ref{eqa.10add}) and (\ref{aphat}) do not require  $\breve\theta\in\mathcal{G}(\theta_0, r)$,  so there is no problem of applying them here. This implies that
$\tilde R(\breve \alpha ,\breve \tau) = o_P(1)$.

The same argument yields that w.p.a.1, $\breve\theta\in\mathcal{G}(\theta_0,r)$ and $\breve\beta\in\mathcal{B}(\beta_0,r)$ when   $\breve\tau\leq\tau_0$; 
hence it is omitted to avoid repetition.
\end{proof}

\noindent
\textbf{Step 4}: 
For any $\epsilon'>0$ and any $r > 0$, there is an 
$\varepsilon >0$ such that for all $\tau$,  $\beta\in\mathcal{B}(\beta_0, r)$ and  $\theta\in\mathcal{G}(\theta_0, r)$, 
$R(\alpha,\tau)< \varepsilon$ implies $|\tau-\tau_0|<\epsilon'.$

\begin{proof}[Proof of Step 4]
We first prove that, for any $\epsilon'>0$, there is 
$\varepsilon >0$ such that for all $\tau >\tau _{0}$, $\beta\in\mathcal{B}(\beta_0, r)$ and  $\theta\in\mathcal{G}(\theta_0, r)$,  
 $R(\alpha ,\tau )< \varepsilon $ implies that $\tau <\tau _{0}+\epsilon'.$

Suppose that $R(\alpha ,\tau )< \varepsilon$.
Applying the triangle inequality, for all $\beta $ and $\tau >\tau _{0},$ 
\begin{align}\label{ER3}
\begin{split}
\lefteqn{ \mathbb{E} \left[ \left( \rho \left( Y,X^{T}\beta _{0}\right) -\rho \left( Y,X^{T}\theta
_{0}\right) \right) 1\left\{ \tau _{0}<Q\leq \tau \right\} \right] } \\
& \leq \left\vert \mathbb{E} \left[ \left( \rho \left( Y,X^{T}\beta \right) -\rho \left(
Y,X^{T}\theta _{0}\right) \right) 1\left\{ \tau _{0}<Q\leq \tau \right\} \right]
\right\vert \\
&+\left\vert \mathbb{E} \left[ \left( \rho \left( Y,X^{T}\beta \right) -\rho
\left( Y,X^{T}\beta _{0}\right) \right) 1\left\{ \tau _{0}<Q\leq \tau
\right\} \right] \right\vert .  
\end{split}
\end{align}%
First, note that the first term on the RHS of (\ref{ER3}) is the third
term on the RHS of (\ref{eqaa2.2}), hence is bounded by $R(\alpha ,\tau
)<\varepsilon $.

We now consider the second term on the RHS of (\ref{ER3}).
Assumption \ref{a:moment} \ref{a:moment:itm1} implies that for all $\beta\in\mathcal{B}(\beta_0, r)$ and  $\theta\in\mathcal{G}(\theta_0, r)$,
\begin{equation}\label{eqa.10}
C_{2}^{* } \mathbb{E} \left[ |X^{T}\beta |1\left\{ Q>\tau _{0}\right\} \right] \leq \mathbb{E} \left[ |X^{T}\beta
|1\left\{ Q \leq \tau _{0}\right\} \right] \leq C_{1}^{* } \mathbb{E} \left[|X^{T}\beta |1\left\{
Q>\tau _{0}\right\} \right].
\end{equation}%
It follows from the Lipschitz condition,
Step 2, and Assumption \ref{a:moment} \ref{a:moment:itm1} that for all $\beta\in\mathcal{B}(\beta_0, r)$,
\begin{align*}
\left| \mathbb{E} \left[ \left( \rho \left( Y,X^{T}\beta \right) -\rho \left( Y,X^{T}\beta
_{0}\right) \right) 1\left\{ \tau _{0}<Q\leq \tau \right\} \right] \right|
&\leq L \mathbb{E} \left[ \left\vert
X^{T}\left( \beta -\beta _{0}\right) \right\vert 1\left\{ \tau _{0}<Q\leq
\tau \right\} \right] \\
& \leq L \mathbb{E} \left[ \left\vert X^{T}\left( \beta -\beta _{0}\right) \right\vert
1\left\{ \tau _{0}<Q\right\} \right] \\
&\leq L \widetilde{C}
\, \mathbb{E}\left[ \left\vert
X^{T}\left( \beta -\beta _{0}\right) \right\vert 1\left\{ Q\leq \tau
_{0}\right\} \right] 
\\
&\leq L\widetilde{C}
\left\{  \varepsilon/(\eta^\ast r^\ast)  \vee
\sqrt{\varepsilon/\eta^\ast} \, \right\}
 \\
&\equiv  C(\varepsilon).
\end{align*}%
Thus, we have shown that  (\ref{ER3}) is bounded by $C(\varepsilon) +\varepsilon$.

For any $\epsilon'>0$, it   follows from Assumptions \ref{a:obj-ftn} \ref{a:obj-ftn:itm2}, \ref{a:obj-ftn} \ref{a:obj-ftn:itm4} and \ref{a:dist-Q-add} \ref{a:threshold}   (see also Remark \ref{remark-tau-whole-set})
that there is a $c>0$ such that if $\tau >\tau _{0}+\epsilon'$,
\begin{align*}
c \mathbb{P} \left( \tau _{0} <Q\leq \tau _{0}+\epsilon' \right) 
&\leq c \mathbb{P} \left( \tau _{0}<Q\leq
\tau \right) \\
&\leq \mathbb{E} \left[ \left( \rho \left( Y,X^{T}\beta _{0}\right) -\rho \left( Y,X^{T}\theta
_{0}\right) \right) 1\left\{ \tau _{0}<Q\leq \tau \right\} \right] \\
& \leq C(\varepsilon) +\varepsilon.
\end{align*}%
Since $\varepsilon \mapsto C(\varepsilon) +\varepsilon$ converges to zero as 
$\varepsilon$ converges to zero,
for a given $\epsilon' > 0$ choose a sufficient small $\varepsilon > 0$
such that $C(\varepsilon) +\varepsilon <  c \mathbb{P}(\tau _{0}<Q\leq \tau _{0}+\epsilon')$, so that the above inequality cannot hold. Hence we  infer that for this $\varepsilon $%
, when $R(\alpha ,\tau )<\varepsilon$, we must
have $\tau <\tau _{0}+\epsilon'$.

By the same argument, if $\tau<\tau_0,$ then we must have $%
\tau>\tau_0-\epsilon'.$ Hence, $R(\alpha,\tau)<\varepsilon$ implies $|\tau-\tau_0|<\epsilon'.$
\end{proof}

\noindent
\textbf{Step 5}:  $\breve{\tau}\overset{p}{\longrightarrow}\tau_{0}$.

\begin{proof}[Proof of Step 5]
For the $\varepsilon$ chosen in Step 4,  consider the event $\{R(\breve{\alpha },\breve{\tau})<\varepsilon \}$, which occurs w.p.a.1, due to Theorem \ref{l2.1}. On this event, $|\breve{\tau}-\tau
_{0}|<\epsilon'$ by Step 4. Because $\epsilon'$ is taken
arbitrarily, we have proved the consistency of $\breve{\tau}.$
\end{proof}

\end{proof}

\subsection{Proof of Theorem \protect\ref{th2.3}}

The proof consists of multiple  steps. First, 
we obtain an intermediate
convergence rate for $\breve{\tau}$ based on the consistency of the risk and that of $%
\breve{\tau}$.  
Second, we use   the compatibility condition to obtain a tighter bound. 
\medskip

\noindent
\textbf{Step 1}: 
Let $\bar{c}_0(\delta_0) \equiv c_0  \inf_{\tau \in \mathcal{T}_0}\mathbb{E}[(X^T\delta_0)^2|Q=\tau]$,
which is bounded away from zero and bounded above due to Assumption \ref{a:dist-Q-add} \ref{a:threshold}. 
Then
$\bar{c}_0(\delta_0)  \left\vert \breve{\tau}-\tau
_{0}\right\vert \leq  4 R\left( 
\breve{\alpha},\breve{\tau}\right) $  w.p.a.1. As a result, 
$|\breve\tau-\tau_0|=O_P\left[ \kappa _{n}s/\bar{c}_0(\delta_0) \right]$.

\begin{proof}[Proof of Step 1]
For any  $\tau_0<\tau$ and $\tau\in\mathcal{T}_0$,  and any $\beta\in\mathcal{B}(\beta_0,r)$, $\alpha=(\beta,\delta)$ with arbitrary $\delta$,
   for some $L, M>0$
which do not depend on $\beta $ and $\tau ,$ 
\begin{align*}
& \left\vert \mathbb{E}\left( \rho \left( Y,X^{T}\beta \right) -\rho \left(
Y,X^{T}\beta _{0}\right) \right) 1\left\{ \tau _{0}<Q\leq \tau \right\}
\right\vert   \\
& \leq_{(1)} L\mathbb{E} \left[ \left\vert X^{T}\left( \beta -\beta
_{0}\right) \right\vert 1\left\{ \tau _{0}<Q\leq \tau \right\}  \right]
\\
& \leq_{(2)} M L  ( \tau- \tau _{0}) 
\mathbb{E} \left[\left\vert X^{T}\left( \beta -\beta _{0}\right) \right\vert
1\left\{ Q  \leq \tau _{0}\right\}\right]   \\
 & \leq_{(3)} M L  ( \tau- \tau _{0}) 
  \left\{ \mathbb{E} \left[ \left( X^{T}\left( \beta -\beta _{0}\right) \right)^2 1\left\{ Q\leq \tau _{0}\right\}  \right] \right\}^{1/2}
 \\
&\leq_{(4)} \left(M L  ( \tau- \tau _{0})  \right) ^{2}/\left( 4 \eta^\ast \right)
+\eta^\ast \mathbb{E} \left[ \left( X^{T}\left( \beta -\beta _{0}\right) \right)^2 1\left\{ Q\leq \tau _{0}\right\}  \right]  \\
&\leq_{(5)} \left( M L  ( \tau- \tau _{0})  \right) ^{2}/\left( 4\eta^\ast \right)
+\mathbb{E} \left[\left( \rho \left( Y,X^{T}\beta \right) -\rho \left( Y,X^{T}\beta
_{0}\right) \right) 1\left\{ Q\leq \tau _{0}\right\}  \right] \\
&\leq_{(6)} \left( M L  ( \tau- \tau _{0})  \right) ^{2}/\left( 4\eta^\ast \right)
+ R(\alpha,\tau),
\end{align*}%
where  (1) follows from the 
Lipschitz condition on the objective function,
(2) is by Assumption \ref{a:moment} \ref{a:moment:itm2},
(3) is by Jensen's inequality, 
(4) follows from  the fact that $uv\leq v^{2}/\left( 4c\right)+cu^2$ for any $c>0$,
(5) is from Assumption \ref{a:obj-ftn} \ref{a:obj-ftn:itm3},
and (6) is from Step 1 in the proof of Theorem \protect\ref{th2.2}.

In addition, 
\begin{align*}
&\left\vert \mathbb{E} \left[ \left( \rho \left( Y,X^{T}\beta \right) -\rho
\left( Y,X^{T}\beta _{0}\right) \right) 1\left\{ \tau _{0}<Q\leq \tau
\right\} \right] \right\vert \\
&\geq_{(1)} 
 \mathbb{E} \left[ \left( \rho \left( Y,X^{T}\beta _{0}\right) -\rho \left( Y,X^{T}\theta
_{0}\right) \right) 1\left\{ \tau _{0}<Q\leq \tau \right\} \right] \\
& 
- \left\vert \mathbb{E} \left[ \left( \rho \left( Y,X^{T}\beta \right) -\rho \left(
Y,X^{T}\theta _{0}\right) \right) 1\left\{ \tau _{0}<Q\leq \tau \right\} \right]
\right\vert \\
&\geq_{(2)} 
 \mathbb{E} \left[ \left( \rho \left( Y,X^{T}\beta _{0}\right) -\rho \left( Y,X^{T}\theta
_{0}\right) \right) 1\left\{ \tau _{0}<Q\leq \tau \right\} \right]
- R(\alpha,\tau) \\
&\geq_{(3)} c_0 \left\{ \inf_{\tau \in \mathcal{T}_0}\mathbb{E}[(X^T\delta_0)^2|Q=\tau] \right\}(\tau-\tau_0)-R(\alpha,\tau),
\end{align*}%
where (1) is by the triangular inequality, (2) is from \eqref{eq2.2},
and (3) is by Assumption \ref{a:obj-ftn} \ref{a:obj-ftn:itm4}.
Therefore, we have established that there exists a constant $\tilde{C} > 0$, independent of $(\alpha,\tau)$, such that
\begin{align}\label{der-c0}
\bar{c}_0(\delta_0)(\tau-\tau_0) &\leq \tilde{C} (\tau - \tau_0)^2   + 2R(\alpha,\tau).
\end{align} 
Note that when $0 < (\tau-\tau_0) < \bar{c}_0(\delta_0)(2\tilde{C})^{-1}$, \eqref{der-c0} implies that 
\begin{align*}
\bar{c}_0(\delta_0)(\tau-\tau_0)  &\leq     \frac{\bar{c}_0(\delta_0)}{2}(\tau - \tau_0)  +  2R(\alpha,\tau),
\end{align*} 
which in turn implies that  
 $\tau-\tau_0\leq \frac{4}{\bar{c}_0(\delta_0)}R(\alpha,\tau)$.
By the same argument, when 
$- \bar{c}_0(\delta_0)(2\tilde{C})^{-1} <(\tau - \tau_0) \leq 0$,  we have  $\tau_0-\tau\leq \frac{4}{\bar{c}_0(\delta_0)}R(\alpha,\tau)$   for $\alpha=(\beta,\delta)$, with any $\theta\in\mathcal{G}(\theta_0, r)$ and arbitrary $\beta.$

 Hence when $\breve\tau>\tau_0$, on the event $\breve\beta\in\mathcal{B}(\beta_0,  r)$,   and  $\breve\tau-\tau_0< \bar{c}_0(\delta_0)(2\tilde{C})^{-1}$, we have
 \begin{equation}\label{eqa.10add}
 \breve\tau-\tau_0\leq \frac{4}{\bar{c}_0(\delta_0)}R(\breve\alpha,\breve\tau).
 \end{equation}
 When $\breve\tau\leq \tau_0$, on the event $\breve\theta\in\mathcal{G}(\theta_0,  r)$,   and  $\tau_0-\breve\tau< \bar{c}_0(\delta_0)(2\tilde{C})^{-1}$, we have
 $
\tau_0- \breve\tau\leq \frac{4}{\bar{c}_0(\delta_0)}R(\breve\alpha,\breve\tau).
 $ Hence due to Step 3 in the proof of Theorem \ref{th2.2}  and the consistency of $\breve\tau$, we have
 \begin{equation}\label{G1tau}
  \left\vert \breve{\tau}-\tau
_{0}\right\vert \leq   \frac{4}{\bar{c}_0(\delta_0)}R\left( 
\breve{\alpha},\breve{\tau}\right) \; \text{ w.p.a.1}.
\end{equation}%
 This also implies $|\breve\tau-\tau_0|=O_P\left[ \kappa _{n}s/\bar{c}_0(\delta_0) \right]$ 
 in view of the proof of Theorem \ref{l2.1}.
\end{proof} 
 
 \noindent
\textbf{Step 2}: 
Define $\nu _{1n}\left( \tau \right) \equiv \nu _{n}\left( \alpha _{0},\tau \right) -\nu
_{n}\left( \alpha _{0},\tau _{0}\right)$  and $c_{\alpha } \equiv \kappa _{n}\left( \left\vert D_{0}\alpha
_{0}\right\vert _{1}-\left\vert \breve{D}\alpha _{0}\right\vert _{1}\right)
+\left\vert \nu _{1n}\left( \breve{\tau}\right) \right\vert $.
Then, 
\begin{align}\label{BI1}
R\left( \breve{\alpha},\breve{\tau}\right) +\frac{1}{2}\kappa
_{n}\left\vert \breve{D}\left( \breve{\alpha}-\alpha _{0}\right) \right\vert
_{1} \leq  c_{\alpha } + 2\kappa _{n}\left\vert \breve{D}\left( \breve{\alpha}%
-\alpha _{0}\right) _{J}\right\vert _{1}\; \text{ w.p.a.1}.
\end{align}%

\begin{proof}[Proof of Step 2]
 Recall the following basic inequality in \eqref{basic ineq}: 
\begin{align}\label{basic-ineq2}
R(\breve{\alpha },\breve{\tau}) 
&\leq 
\left[ \nu _{n} (\alpha _{0},\breve{\tau})-\nu _{n}(\breve{\alpha },\breve{\tau})%
\right] - \nu _{1n}\left( \breve \tau \right) 
+\kappa _{n}\left( |D_{0}\alpha _{0}|_{1}-|\breve{D}\breve{\alpha }%
|_{1}\right).  
\end{align}
Now applying Lemma \ref{lem-emp} 
to 
$[\nu _{n} (\alpha _{0},\breve{\tau})-\nu _{n}(\breve{\alpha },\breve{\tau})]$
 with $a_{n}\ $and $b_{n}$ replaced by $a_{n}/2\ 
$and $b_{n}/2$,
we can rewrite the basic inequality in \eqref{basic-ineq2} by%
\begin{equation*}
\kappa _{n}\left\vert D_0\alpha _{0}\right\vert _{1}\geq R\left( \breve{\alpha},%
\breve{\tau}\right) +\kappa _{n}\left\vert \breve{D}\breve{\alpha}\right\vert
_{1}-\frac{1}{2}\kappa _{n}\left\vert \breve{D}\left( \breve{\alpha}-\alpha
_{0}\right) \right\vert _{1}-\left\vert \nu _{1n}\left( \breve{\tau}\right)
\right\vert \; \text{ w.p.a.1}.
\end{equation*}
Now adding $\kappa _{n}\left\vert \breve{D}\left( \breve{\alpha}-\alpha
_{0}\right) \right\vert _{1}$ on both sides of the inequality above
and using the fact that $ \left\vert  \alpha_{0j} \right\vert _{1} - \left\vert \breve{\alpha}_{j} \right\vert
_{1} + \left\vert \left( \breve{\alpha}_{j} -\alpha
_{0j}\right) \right\vert _{1} = 0$ for $j \notin J$,
we have that
\begin{align*}
& \kappa _{n}\left( \left\vert D_{0}\alpha _{0}\right\vert _{1}-\left\vert 
\breve{D}\alpha _{0}\right\vert _{1}\right) +\left\vert \nu _{1n}\left( \breve{%
\tau}\right) \right\vert +2\kappa _{n}\left\vert \breve{D}\left( \breve{\alpha}%
-\alpha _{0}\right) _{J}\right\vert _{1}  \\
& 
\geq R\left( \breve{\alpha},\breve{\tau}\right) +\frac{1}{2}\kappa
_{n}\left\vert \breve{D}\left( \breve{\alpha}-\alpha _{0}\right) \right\vert
_{1}\; \text{ w.p.a.1}.
\end{align*}%
Therefore, we have proved Step 2.
\end{proof}

We prove the remaining part of the steps by considering 
 two cases: (i) $\kappa _{n}\left\vert \breve{D}\left( \breve{\alpha}-\alpha
_{0}\right) _{J}\right\vert_1 \leq c_{\alpha };$ (ii) $\kappa _{n}\left\vert 
\breve{D}\left( \breve{\alpha}-\alpha _{0}\right) _{J}\right\vert_1 >c_{\alpha }.$
We first consider  Case (ii). 
\medskip

\noindent
\textbf{Step 3}:   
Suppose that $\kappa _{n}\left\vert \breve{D}\left( \breve{\alpha}-\alpha
_{0}\right) _{J}\right\vert_1 > c_{\alpha }$. 
Then
\begin{align*}
\left\vert \breve{\tau}-\tau _{0}\right\vert &
= O_P\left[\kappa_n^2s/\bar{c}_0(\delta_0)\right] 
\; \text { and }  \;
\left\vert  \breve{\alpha}-\alpha _{0} \right\vert = O_P\left( \kappa _{n}s\right).
\end{align*}

\begin{proof}[Proof of Step 3]
By $\kappa _{n}\left\vert 
\breve{D}\left( \breve{\alpha}-\alpha _{0}\right) _{J}\right\vert_1 >c_{\alpha }$ and the basic inequality (\ref{BI1}) in Step 2, 
\begin{equation}
6\left\vert \breve{D}\left( \breve{\alpha}-\alpha _{0}\right) _{J}\right\vert
_{1}\geq \left\vert \breve{D}\left( \breve{\alpha}-\alpha _{0}\right)
\right\vert _{1}=\left\vert \breve{D}\left( \breve{\alpha}-\alpha _{0}\right)
_{J}\right\vert _{1} +\left\vert \breve{D}\left( \breve{\alpha}-\alpha _{0}\right)
_{J^{c}}\right\vert _{1},  \label{aphat}
\end{equation}%
which enables us to apply the compatibility condition in Assumption \ref%
{ass2.7}.

Recall that $\|Z\|_2=(EZ^2)^{1/2}$ for a random variable $Z.$ 
Note that for $s=|J(\alpha_0)|_0$,
\begin{align}\label{a.11add}
\begin{split}
& R\left( \breve{\alpha},\breve{\tau}\right) +\frac{1}{2}\kappa
_{n}\left\vert \breve{D}\left( \breve{\alpha}-\alpha _{0}\right) \right\vert
_{1} \\
&\leq_{(1)} 3\kappa _{n}\left\vert \breve{D}\left( \breve{\alpha}-\alpha _{0}\right)
_{J}\right\vert _{1} \\
& \leq_{(2)} 3\kappa _{n}\bar D\left\Vert X(\breve{\tau})^T(\breve\alpha-\alpha_0)\right\Vert _{2}\sqrt{s}/\phi  \\
&\leq_{(3)} \frac{9\kappa _{n}^{2}\bar D^{2}s}{2 \tilde{c} \phi ^{2}}+\frac{\tilde{c}}{2}%
\left\Vert X(\breve{\tau})^T(\breve\alpha-\alpha_0)\right\Vert _{2}^{2},
\end{split}
\end{align}
where (1) is from the basic inequality \eqref{BI1} in Step 2,
(2) is by the compatibility condition (Assumption \ref{ass2.7}), and
(3) is from the
inequality  that $uv\leq v^{2}/(2\tilde{c})+\tilde{c} u^2/2$ for any $\tilde{c} >0$.

We will show below in Step 4 that there is a  constant $C_0 > 0$ such that 
\begin{align}\label{step-4b-statement}
\left\Vert X(\breve\tau)^T(\breve\alpha-\alpha_0)\right\Vert _{2}^{2}\leq C_0 R(\breve\alpha,\breve\tau)+C_0\bar{c}_0(\delta_0)| \breve\tau-\tau_0|,
\text{ w.p.a.1}.
\end{align}
Recall that by \eqref{G1tau}, 
$\bar{c}_0(\delta_0)  \left\vert \breve{\tau}-\tau
_{0}\right\vert \leq   4 R\left( 
\breve{\alpha},\breve{\tau}\right)$.
Hence,  (\ref{a.11add})  with $\tilde{c}=(5C_0)^{-1}$
implies that 
\begin{align}\label{oracle-ineq}
 R\left( \breve{\alpha},\breve{\tau}\right) +\kappa
_{n}\left\vert \breve{D}\left( \breve{\alpha}-\alpha _{0}\right) \right\vert
_{1}\leq \frac{9\kappa _{n}^{2}\bar D^{2}s}{\tilde{c} \phi ^{2}}.
\end{align}
By \eqref{oracle-ineq} and \eqref{G1tau}, $\left\vert \breve{\tau}-\tau _{0}\right\vert=O_P\left[\kappa_n^2s/\bar{c}_0(\delta_0)\right]$.
Also, by \eqref{oracle-ineq}, 
$\left\vert  \breve{\alpha}-\alpha _{0} \right\vert = O_P\left( \kappa _{n}s\right)$
since $D(\breve\tau)\geq \underline{D}$ w.p.a.1
by Assumption \ref{a:setting} \ref{a:setting:itm3}.
\end{proof}


\noindent
\textbf{Step 4}: There is a  constant $C_0 > 0$ such that 
$
\left\Vert X(\breve\tau)^T(\breve\alpha-\alpha_0)\right\Vert _{2}^{2}\leq C_0 R(\breve\alpha,\breve\tau)+C_0 \bar{c}_0(\delta_0) |\breve\tau-\tau_0|,
$
w.p.a.1.
   
\begin{proof}[Proof Step 4] 
Note that
\begin{align}\label{decom-step4}
\begin{split}
 \left\Vert X(\tau)^T(\alpha-\alpha_0)\right\Vert _{2}^{2} 
&\leq 2\left\Vert X(\tau)^T\alpha-X(\tau_0)^T\alpha\right\Vert _{2}^{2} \\
&+4\left\Vert X(\tau_0)^T\alpha-X(\tau_0)^T\alpha_0\right\Vert _{2}^{2}
+4\left\Vert X(\tau_0)^T\alpha_0-X(\tau)^T\alpha_0\right\Vert _{2}^{2}.
\end{split}
\end{align}
We bound the three terms on the right hand side of \eqref{decom-step4}. 
When $\tau >\tau _{0}$,  there is a constant $C_1 > 0$ such that 
\begin{align*}
& \left\Vert X(\tau)^T\alpha-X(\tau_0)^T\alpha\right\Vert _{2}^{2} \\
&=\mathbb{E} \left[ (X^T\delta)^2 1\{\tau_0 \leq Q <\tau\} \right] \\
&= \int_{\tau_0}^\tau \mathbb{E} \left[ (X^T\delta)^2 \big|Q  = t \right]  dF_Q(t) \\
&\leq 2\int_{\tau_0}^\tau \mathbb{E} \left[ (X^T\delta_0)^2 \big|Q  = t \right]  dF_Q(t) 
+ 2\int_{\tau_0}^\tau \mathbb{E} \left[ (X^T(\delta-\delta_0))^2 \big|Q  = t \right]  dF_Q(t) \\
&\leq C_1  \bar{c}_0(\delta_0) (\tau-\tau_0),
\end{align*}
where the last inequality is by Assumptions \ref{a:setting},  
\ref{a:dist-Q-add} \ref{a:dist-Q:itm3}, 
\ref{a:dist-Q-add} \ref{a:threshold},
and
\ref{a:moment} \ref{a:moment:itm2}.

Similarly,
$\left\Vert X(\tau_0)^T\alpha_0-X(\tau)^T\alpha_0\right\Vert _{2}^{2}
=\mathbb{E} \left[ (X^T\delta_0)^2 1\{\tau_0\leq Q<\tau\} \right] \leq C_1 \bar{c}_0(\delta_0) (\tau-\tau_0).
$
Hence, the first and third terms of the right hand side of of \eqref{decom-step4}
are bounded by $6 C_1 \bar{c}_0(\delta_0)(\tau-\tau_0)$.

To bound the second term, note that there exists a constant $C_2 > 0$ such that 
\begin{align*}
\lefteqn{\left\Vert X(\tau_0)^T\alpha-X(\tau_0)^T\alpha_0\right\Vert _{2}^{2} } \\
&=_{(1)} \mathbb{E}\left[ (X^T(\theta-\theta_0))^21\{Q > \tau_0\} \right]
+ \mathbb{E}\left[ (X^T(\beta-\beta_0))^21\{Q \leq \tau_0\} \right] \\
&\leq_{(2)} (\eta^\ast)^{-1} \mathbb{E}\left[ \left( \rho \left( Y,X^{T}\theta \right) -\rho \left( Y,X^{T}\theta_{0}\right) \right) 1\left\{ Q>\tau _{0}\right\}  \right] \\
&+  (\eta^\ast)^{-1} \mathbb{E}\left[ \left( \rho \left( Y,X^{T}\beta \right) -\rho \left( Y,X^{T}\beta_{0}\right) \right) 1\left\{ Q\leq \tau _{0}\right\} \right] \\
&\leq_{(3)} (\eta^\ast)^{-1} R(\alpha, \tau) + (\eta^\ast)^{-1} \mathbb{E}\left[ \left( \rho \left( Y,X^{T}\theta \right) -\rho \left( Y,X^{T}\theta
_{0}\right) \right) 1\left\{ \tau_0< Q \leq \tau\right\} \right] \\
&\leq_{(4)}  (\eta^\ast)^{-1}  R(\alpha, \tau) + (\eta^\ast)^{-1} L \mathbb{E}\left[ |X^T(\theta-\theta_0)|1\left\{ \tau_0< Q \leq \tau\right\} \right] \\
&=_{(5)} (\eta^\ast)^{-1}  R(\alpha, \tau) 
+ (\eta^\ast)^{-1} L \int_{\tau_0}^\tau \mathbb{E} \left[ |X^T(\theta-\theta_0)| \big|Q  = t \right]  dF_Q(t) \\
&\leq_{(6)} (\eta^\ast)^{-1}  R(\alpha, \tau) +C_3 (\tau-\tau_0),
\end{align*}
where (1) is simply an identity,  (2) from  Assumption \ref{a:obj-ftn} \ref{a:obj-ftn:itm3}, (3) is  
due to \eqref{eqaa2.2}: namely,
\begin{align*}
\mathbb{E}\left[ \left( \rho \left( Y,X^{T}\theta \right) -\rho \left( Y,X^{T}\theta
_{0}\right) \right) 1\left\{ Q>\tau \right\} \right]
 +\mathbb{E}\left[ \left( \rho \left( Y,X^{T}\beta \right) -\rho \left( Y,X^{T}\beta
_{0}\right) \right) 1\left\{ Q\leq \tau _{0}\right\} \right] \leq R(\alpha,\tau),
\end{align*}
(4) is by  the Lipschitz continuity of $\rho(Y,\cdot)$,
(5) is by rewriting the expectation term, and (6) is by Assumptions \ref{a:setting} \ref{a:dist-Q:itm1} and
\ref{a:moment} \ref{a:moment:itm2}.
Therefore, we have shown that $\left\Vert X(\tau)^T(\alpha-\alpha_0)\right\Vert _{2}^{2}\leq C_0 R(\alpha,\tau)+C_0 \bar{c}_0(\delta_0)(\tau-\tau_0)$ for some constant $C_0 > 0$. 
The case of $\tau\leq\tau_0$ can be proved using the same argument. Hence, setting $\tau=\breve\tau$, and $\alpha=\breve\alpha$, we obtain the desired result. 
\end{proof}

\noindent
\textbf{Step 5}:   
 We now consider Case (i). 
Suppose that $\kappa _{n}\left\vert \breve{D}\left( \breve{\alpha}-\alpha
_{0}\right) _{J}\right\vert _{1} \leq c_{\alpha }$. 
Then
\begin{align*}
\left\vert \breve{\tau}-\tau _{0}\right\vert &
= O_P\left[ \kappa _{n}^{2}s/\bar{c}_0(\delta_0) \right]
\; \text { and }  \;
\left\vert  \breve{\alpha}-\alpha _{0} \right\vert = O_P\left( \kappa _{n}s\right).
\end{align*}

\begin{proof}[Proof of Step 5]
Recall that  $X_{ij}$ is the $j$th element of $X_i$, where $i\leq n, j\leq p.$
By Assumption \ref{a:setting}  and Step 1,
\begin{align*}
\sup_{1 \leq j \leq p} \frac{1%
}{n}\sum_{i=1}^{n}\left\vert X_{ij} \right\vert ^{2}\left\vert
1\left( Q_{i}<\breve \tau \right) -1\left( Q_{i}<\tau_0 \right)
\right\vert = O_P\left[ \kappa _{n} s/\bar{c}_0(\delta_0)\right].
\end{align*}
By the
mean value theorem, 
\begin{align}
&\kappa _{n}\left\vert \left\vert D_{0}\alpha _{0}\right\vert
_{1}-\left\vert \breve{D}\alpha _{0}\right\vert _{1}\right\vert  \notag \\
&\leq \kappa _{n}\sum_{j=1}^{p}\left( \frac{4}{n}\sum_{i=1}^{n}\left\vert
X_{ij}1\left\{ Q_{i}>\overline{\tau }\right\} \right\vert ^{2}\right)
^{-1/2}\left\vert \delta_0 ^{\left( j\right) }\right\vert \frac{1}{n}%
\sum_{i=1}^{n}\left\vert X_{ij}\right\vert ^{2}\left\vert
1\left\{ Q_{i}>\breve{\tau}\right\} -1\left\{ Q_{i}>\tau _{0}\right\}
\right\vert  \notag \\
&=O_P\left[ \kappa _{n}^{2}s |J(\delta_0)|_0/\bar{c}_0(\delta_0)\right].  \label{D-Dhat}
\end{align}
Here, recall that $\overline{\tau }$ is the right-end point of $\mathcal{T}$ and $|J(\delta_0)|_0$ is the dimension of 
nonzero elements of $\delta_0$. 

Due to Step 1  and (\ref{emp3})  in Lemma \ref{lem-emp}, 
\begin{align}\label{v1n-hat-rate}
\left\vert \nu _{1n}\left( \breve{%
\tau}\right) \right\vert 
=O_P \left[ \frac{|\delta_0|_2}{\sqrt{\bar{c}_0(\delta_0)}}  \left( \kappa _{n} s/n \right)^{1/2}\right].
\end{align}
Thus, under Case (i), we have that, by \eqref{G1tau}, \eqref{BI1}, \eqref{D-Dhat},
and \eqref{v1n-hat-rate}, 
\begin{align}\label{tauhat}
\begin{split}
\frac{\bar{c}_0(\delta_0)}{4}\left\vert \breve{\tau}-\tau _{0}\right\vert &\leq \frac{\kappa _{n}}{2}%
\left\vert \breve{D}\left( \breve{\alpha}-\alpha _{0}\right) \right\vert
_{1}+R\left( \breve{\alpha},\breve{\tau}\right)   \\
&\leq 3\kappa _{n}\left( \left\vert D_{0}\alpha _{0}\right\vert
_{1}-\left\vert \breve{D}\alpha _{0}\right\vert _{1}\right) +3\left\vert \nu
_{1n}\left( \breve{\tau}\right) \right\vert   \\
&= O_P\left( \kappa _{n}^{2}s^{2}\right) +O_P \left[ s^{1/2} \left(  \kappa _{n} s/n \right)^{1/2}\right],  
\end{split}
\end{align}
where the last equality uses the fact that $|J(\delta_0)|_0/\bar{c}_0(\delta_0) = O(s)$ 
and ${|\delta_0|_2}/{\sqrt{\bar{c}_0(\delta_0)}} = O(s^{1/2})$ at most (both could be bounded in some cases).

Therefore, we now have an improved rate of convergence in probability for $\breve \tau$
from $r_{n0, \tau} \equiv \kappa _{n} s$ to $r_{n1, \tau} \equiv [\kappa _{n}^{2}s^{2} + s^{1/2} ( \kappa _{n} s/n )^{1/2}]$. 
Repeating the arguments identical to those to prove \eqref{D-Dhat} and \eqref{v1n-hat-rate} yields that
\begin{align*}
\kappa _{n}\left\vert \left\vert D_{0}\alpha _{0}\right\vert
_{1}-\left\vert \breve{D}\alpha _{0}\right\vert _{1}\right\vert  
= O_P \left[ r_{n1, \tau} \kappa_n s  \right]
\ \ \text{and} \ \
\left\vert \nu _{1n}\left( \breve{%
\tau}\right) \right\vert 
= O_P \left[ s^{1/2} \left( r_{n1, \tau} /n \right)^{1/2}\right].
\end{align*}
Plugging these improved rates into \eqref{tauhat} gives 
\begin{align*}
\bar{c}_0(\delta_0) \left\vert \breve{\tau}-\tau _{0}\right\vert 
&
= O_P\left( \kappa _{n}^{3}s^{3}\right) 
+ O_P \left[ s^{1/2} (\kappa _{n} s)^{3/2} / n^{1/2}\right]
+ O_P\left( \kappa _{n}s^{3/2} / n^{1/2} \right)
+ O_P \left[ s^{3/4} (\kappa _{n} s)^{1/4} /n^{3/4} \right] \\
&= O_P\left( \kappa _{n}^{2} s^{3/2} \right) 
+ O_P \left[  s^{3/4} (\kappa _{n} s)^{1/4} /n^{3/4} \right] \\
&\equiv O_P( r_{n2, \tau} ),
\end{align*}
where the second equality comes from the fact that 
the first three terms are $O_P\left( \kappa _{n}^{2}s^{3/2} \right)$ since $\kappa _{n} s^{3/2} =o(1)$, $\kappa_n n/s \rightarrow \infty$, and $\kappa_n \sqrt{n} \rightarrow \infty$
in view of the assumption   that $\kappa_n s^2 \log p=o(1)$.
Repeating the same arguments again with the further improved rate $r_{n2, \tau}$, we have that
\begin{align*}
\left\vert \breve{\tau}-\tau _{0}\right\vert 
&
= O_P\left( \kappa _{n}^{2}s^{5/4} \right) 
+ O_P \left[ s^{7/8} (\kappa _{n} s)^{1/8} /n^{7/8} \right] 
\equiv O_P( r_{n3, \tau} ).
\end{align*}
Thus, repeating the same arguments $k$ times yields 
\begin{align*}
\bar{c}_0(\delta_0) \left\vert \breve{\tau}-\tau _{0}\right\vert 
&
= O_P\left( \kappa _{n}^{2}s^{1 + 2^{-k}} \right) 
+ O_P \left[ s^{(2^k-1)/2^k} (\kappa _{n} s)^{1/2^{k}} /n^{(2^k-1)/2^k} \right] 
\equiv O_P( r_{nk, \tau} ).
\end{align*}
Then letting $k \rightarrow \infty$ gives the desired result that  
$\bar{c}_0(\delta_0) \left\vert \breve{\tau}-\tau _{0}\right\vert = O_P\left( \kappa _{n}^{2}s \right)$.  
Finally, the same iteration based on \eqref{tauhat}  gives 
$\left\vert \breve{D}\left( \breve{\alpha}-\alpha _{0}\right)\right\vert =o_P\left( \kappa _{n}s\right) $,
which proves the desired result 
since $D(\breve\tau)\geq \underline{D}$ w.p.a.1
by Assumption \ref{a:setting} \ref{a:setting:itm3}.
\end{proof}

\subsection{Proof of Theorem \ref{thm:tau-2nd}}

\begin{proof}[Proof of Theorem \ref{thm:tau-2nd}]
The asymptotic property of $\widetilde
{\tau}$ is well-known in the literature (see Lemma \ref{lem:tau-dist-oracle}
below for its asymptotic distribution). Specifically, we can apply
Theorem 3.4.1 of \citet{VW} (by defining the criterion $\mathbb{M}_{n}\left(
\cdot\right)  \equiv R_{n}^{\ast}\left(  \cdot\right)  $, 
$M_{n}\left(\cdot\right)  \equiv \mathbb{E} R_{n}^{\ast}\left(  \cdot\right) = R(\alpha_0,\tau) $,
the distance function
$d\left(  \tau,\tau_{0}\right)  \equiv\left\vert \tau-\tau_{0}\right\vert
^{1/2}$, and $\phi_{n}\left(  \delta\right)  \equiv\delta$) to characterize
the convergence rate of $\widetilde{\tau}$, which  results in the super-consistency
in the sense that $\widetilde{\tau}-\tau_{0}=O_{P}(n^{-1})$. 
See e.g. Section 14.5 of \cite{Kosorok}.

Furthermore, it is worth noting that the same theorem also implies that if
\begin{align}\label{equality3-below}
\left[  R_{n}^{\ast}\left(  \widehat{\tau}\right)  -R_{n}^{\ast}\left(
\tau_{0}\right)  \right]  -\left[  R_{n}\left(  \breve{\alpha},\widehat{\tau
}\right)  -R_{n}\left(  \breve{\alpha},\tau_{0}\right)  \right]  =O_{P}%
(r_{n}^{-2})
\end{align}
$\ $for some sequence $r_{n}$ satisfying $r_{n}^{2}\phi_{n}\left(  r_{n}%
^{-1}\right)  \leq\sqrt{n},$ then
\[
r_{n}d\left(  \widehat{\tau},\tau_{0}\right)  =O_{P}\left(  1\right)  .
\]
This is because
\begin{align*}
R_{n}^{\ast}\left(  \widehat{\tau}\right)   &  =R_{n}^{\ast}\left(
\widehat{\tau}\right)  -\left[  R_{n}\left(  \breve{\alpha},\widehat{\tau
}\right)  -R_{n}\left(  \breve{\alpha},\tau_{0}\right)  +R_{n}^{\ast}\left(
\tau_{0}\right)  \right]  +\left[  R_{n}\left(  \breve{\alpha},\widehat{\tau
}\right)  -R_{n}\left(  \breve{\alpha},\tau_{0}\right)  +R_{n}^{\ast}\left(
\tau_{0}\right)  \right]  \\
&  \leq_{(1)} R_{n}^{\ast}\left(  \widehat{\tau}\right)  -\left[  R_{n}\left(
\breve{\alpha},\widehat{\tau}\right)  -R_{n}\left(  \breve{\alpha},\tau
_{0}\right)  +R_{n}^{\ast}\left(  \tau_{0}\right)  \right]  +\left[
R_{n}\left(  \breve{\alpha},\tau_{0}\right)  -R_{n}\left(  \breve{\alpha}%
,\tau_{0}\right)  +R_{n}^{\ast}\left(  \tau_{0}\right)  \right]  \\
&  =_{(2)}\left\{  \left[  R_{n}^{\ast}\left(  \widehat{\tau}\right)  -R_{n}^{\ast
}\left(  \tau_{0}\right)  \right]  -\left[  R_{n}\left(  \breve{\alpha
},\widehat{\tau}\right)  -R_{n}\left(  \breve{\alpha},\tau_{0}\right)
\right]  \right\}  +R_{n}^{\ast}\left(  \tau_{0}\right)  \\
&  =_{(3)} O_{P}\left(  r_{n}^{-2}\right)  +R_{n}^{\ast}\left(  \tau_{0}\right),
\end{align*}
where  inequality (1) uses the fact that $\widehat{\tau}$ is a minimizer
of $R_{n}\left(  \breve{\alpha},\tau\right)$, equality (2)
follows since $R_{n}\left(  \breve{\alpha},\tau_{0}\right)  -R_{n}\left(
\breve{\alpha},\tau_{0}\right)  +R_{n}^{\ast}\left(  \tau_{0}\right)
=R_{n}^{\ast}\left(  \tau_{0}\right)  $, and equality (3) comes from \eqref{equality3-below}.

Then, note that we can set $r_{n}^{-2}=a_{n}s_{n}\log(np)$ with $s_{n}=1$ and
$a_{n}=\kappa_{n}s\log n$ due to Lemma \ref{lem:ortho} and the rate of
convergence $\breve{\alpha}-\alpha_{0}=O_{P}\left(  \kappa_{n}s\right)
\ $given by Theorem \ref{th2.3}. Next, we will apply a chaining argument to
obtain the convergence rate of $\widehat{\tau}$
by repeatedly verifying the condition $R_n^*(\widehat\tau) \leq R_n^*(\tau_0)+O_P(r_n^{-2})$, 
with an iteratively improved rate $r_n$.
Applying Theorem 3.4.1 of
\citet{VW} with $r_{n}=\left(  a_{n}\log(np)\right)  ^{-1/2},$ we have
\[
\widehat{\tau}-\tau_{0}=O_{P}\left(  a_{n}\log(np)\right)  =O_{P}\left(
\kappa_{n}s\log n\log(np)\right)  .
\]
Next, we reset $s_{n}=\kappa_{n}s\left(  \log n\right)  ^{2}\log(np)$ and
$a_{n}=\kappa_{n}s\log n$ to apply Lemma \ref{lem:ortho} again and then
 Theorem 3.4.1 of \citet{VW} with $r_{n}=\left(  s_{n}a_{n}\log(np)\right)
^{-1/2}$. It follows that 
\[
\widehat{\tau}-\tau_{0}=O_{P}\left(  [\kappa_{n}s]^{2}\left(  \log n\right)
^{3}\left(  \log(np)\right)  ^{2}\right)  .
\]
In the next step, we set $r_{n}=\sqrt{n}$ since it should satisfy the
constraint that $r_{n}^{2}\phi_{n}\left(  r_{n}^{-1}\right)  \leq\sqrt{n}$ as
well. Then, we conclude that $\widehat{\tau}=\tau_{0}+O_{P}\left(
n^{-1}\right)  $. Furthermore, in view of Lemma \ref{lem:ortho},
$\widehat{\tau}=\tau_{0}+O_{P}\left(  n^{-1}\right)  $ implies that the
asymptotic distribution of $n\left(  \widehat{\tau}-\tau_{0}\right)  $ is
identical to $n\left(  \widetilde{\tau}-\tau_{0}\right)  $ since each of them is
characterized by the minimizer of the weak limit of $n\left(  R_{n}\left(
\alpha,\tau_{0}+tn^{-1}\right)  -R_{n}\left(  \alpha,\tau_{0}\right)  \right)
$ with $\alpha=\breve{\alpha}$ and $\alpha=\alpha_{0}$, respectively. That is,
the weak limits of the processes are identical due to Lemma \ref{lem:ortho}.
Therefore, we have proved the first conclusion of the theorem. Lemma
\ref{lem:tau-dist-oracle} establishes the second conclusion.
\end{proof}

\begin{lem}\label{lem:ortho}
Suppose that $\alpha\in \mathcal{A}_{n} \equiv \{
\alpha=\left(  \beta^{T},\delta^{T}\right)  ^{T}:\left\vert \alpha-\alpha
_{0}\right\vert _{1}\leq Ka_{n} \}  $ and $\tau\in\mathcal{T}_{n} \equiv \left\{  \left\vert \tau-\tau_{0}\right\vert \leq Ks_{n}\right\}  $ for
some $K<\infty$ and for some sequences $a_{n}$ and $s_{n}$ as
$n\rightarrow\infty$. Then,
\[
\sup_{\alpha\in\mathcal{A}_{n},\tau\in\mathcal{T}_{n}}\Big\vert 
\left\{ R_{n}\left(  \alpha,\tau\right)  -R_{n}\left(  \alpha,\tau_{0}\right) \right\}
-\left\{ R_{n}\left(  \alpha_{0},\tau\right)  -R_{n}\left(
\alpha_{0},\tau_{0}\right)  \right\}  \Big\vert =O_{P}\left[  a_{n}s_{n}\log
(np) \right].
\]
\end{lem}

\begin{proof}[Proof of Lemma \ref{lem:ortho}]
Noting that
\[
\rho\left(  Y_{i},X_{i}^{T}\beta+X_{i}^{T}\delta1\left\{  Q_{i}>\tau\right\}
\right)  =\rho\left(  Y_{i},X_{i}^{T}\beta\right)  1\left\{  Q_{i}\leq
\tau\right\}  +\rho\left(  Y_{i},X_{i}^{T}\beta+X_{i}^{T}\delta\right)
1\left\{  Q_{i}>\tau\right\}  ,
\]
we have, for $\tau>\tau_{0},$%
\begin{align*}
 D_{n}\left(  \alpha,\tau\right)  
&  := \left\{ R_{n}\left(  \alpha,\tau\right)  -R_{n}\left(
\alpha,\tau_{0}\right) \right\}  -\left\{  R_{n}\left(  \alpha_{0},\tau\right)
-R_{n}\left(  \alpha_{0},\tau_{0}\right)  \right\}  \\
&  =\frac{1}{n}\sum_{i=1}^{n}\left[  \rho\left(  Y_{i},X_{i}^{T}\beta\right)
-\rho\left(  Y_{i},X_{i}^{T}\beta_{0}\right)  \right]  1\left\{  \tau
_{0}<Q_{i}\leq\tau\right\}  \\
&  -\frac{1}{n}\sum_{i=1}^{n}\left[  \rho\left(  Y_{i},X_{i}^{T}\theta\right)
-\rho\left(  Y_{i},X_{i}^{T}\theta_{0}\right)  \right]  1\left\{  \tau
_{0}<Q_{i}\leq\tau\right\}  \\
&  =:D_{n1}\left(  \alpha,\tau\right)  -D_{n2}\left(  \alpha,\tau\right)  .
\end{align*}
However, the Lipschitz property of $\rho$ yields that%
\begin{align*}
\left| D_{n1}\left(  \alpha,\tau\right) \right| 
& = \left\vert \frac{1}{n}\sum_{i=1}^{n}\left[  \rho\left(  Y_{i},X_{i}%
^{T}\beta\right)  -\rho\left(  Y_{i},X_{i}^{T}\beta_{0}\right)  \right]
1\left\{  \tau_{0}<Q_{i}\leq\tau\right\}  \right\vert \\
&  \leq L\max_{i,j}\left\vert X_{ij}\right\vert \left\vert \beta-\beta
_{0}\right\vert _{1}\frac{1}{n}\sum_{i=1}^{n}1\left\{  \tau_{0}<Q_{i}\leq
\tau\right\}  \\
&  =O_{P}\left[  \log (np) \cdot a_{n}\cdot s_{n}\right] \;
  \text{ uniformly in $(\alpha, \tau) \in \mathcal{A}_{n} \times \mathcal{T}_{n}$},
\end{align*}
where $\log (np)$ term comes from the Bernstein inequality and the $s_{n}$ term
follows from the fact that $\mathbb{E} \left\vert \frac{1}{n}\sum_{i=1}^{n}1\left\{
\tau_{0}<Q_{i}\leq\tau\right\}  \right\vert =\mathbb{E} 1\left\{  \tau_{0}<Q_{i}\leq
\tau\right\}  \leq C\cdot Ks_{n}$ due to the boundedness of the density of
$Q_{i}$ around $\tau_{0}$.
The other term $D_{n2}\left(  \alpha,\tau\right)$ can be bounded similarly.
The case of  $\tau <\tau_{0}$ can be treated analogously and hence details are omitted.
\end{proof}

\begin{lem}\label{lem:tau-dist-oracle}
We have that 
$n\left( \widetilde{\tau}-\tau _{0}\right) $ converges in distribution to 
the smallest
minimizer of a compound Poisson process  defined in 
Theorem \ref{thm:tau-2nd}.
\end{lem}

\begin{proof}[Proof of Lemma \ref{lem:tau-dist-oracle}]
The convergence rate of $ \widetilde{\tau} $ is standard as commented in the beginning 
of the proof of Theorem \ref{thm:tau-2nd} and thus details are omitted here. We present the characterization 
of the asymptotic distribution for the given convergence rate $n$. 

Recall that $\rho \left( t,s\right) =\dot{\rho}\left( t-s\right) ,$ where $%
\dot{\rho}\left( t\right) =t\left( \gamma -1\left\{ t\leq 0\right\} \right) $.
Note that 
\begin{eqnarray*}
&&nR_{n}^{\ast }\left( \tau \right)  \\
&=&\sum_{i=1}^{n}\dot{\rho}\left( Y_{i}-X_{i}^{T}\beta _{0}-X_{i}^{T}\delta
_{0}1\left\{ Q_{i}>\tau \right\} \right) -\dot{\rho}\left(
Y_{i}-X_{i}^{T}\beta _{0}-X_{i}^{T}\delta _{0}1\left\{ Q_{i}>\tau
_{0}\right\} \right)  \\
&=&\sum_{i=1}^{n}\left[ \dot{\rho}\left( U _{i}-X_{i}^{T}\delta
_{0}\left( 1\left\{ Q_{i}>\tau \right\} -1\left\{ Q_{i}>\tau _{0}\right\}
\right) \right) -\dot{\rho}\left( U _{i}\right) \right] \left(
1\left\{ \tau <Q_{i}\leq \tau _{0}\right\} +1\left\{ \tau _{0}<Q_{i}\leq
\tau \right\} \right)  \\
&=&\sum_{i=1}^{n}\left[ \dot{\rho}\left( U _{i}-X_{i}^{T}\delta
_{0}\right) -\dot{\rho}\left( U _{i}\right) \right] 1\left\{ \tau
<Q_{i}\leq \tau _{0}\right\}  \\
&&+\sum_{i=1}^{n}\left[ \dot{\rho}\left( U _{i}+X_{i}^{T}\delta
_{0}\right) -\dot{\rho}\left( U _{i}\right) \right] 1\left\{ \tau
_{0}<Q_{i}\leq \tau \right\}.
\end{eqnarray*}%
Thus, the asymptotic distribution of $n\left( \widetilde{\tau}-\tau _{0}\right) $
is characterized by the smallest minimizer of the weak limit of%
\begin{equation*}
M_{n}\left( h\right) =\sum_{i=1}^{n}\dot{\rho}_{1i}1\left\{ \tau _{0}+\frac{h%
}{n}<Q_{i}\leq \tau _{0}\right\} +\sum_{i=1}^{n}\dot{\rho}_{2i}1\left\{ \tau
_{0}<Q_{i}\leq \tau _{0}+\frac{h}{n}\right\} 
\end{equation*}%
for $\left\vert h\right\vert \leq K$ for some large $K,$ where $\dot{\rho}%
_{1i} \equiv \dot{\rho}\left( U _{i}-X_{i}^{T}\delta _{0}\right) -\dot{%
\rho}\left( U _{i}\right) $ and $\dot{\rho}_{2i} \equiv \dot{\rho}\left(
U _{i}+X_{i}^{T}\delta _{0}\right) -\dot{\rho}\left( U
_{i}\right) $. The weak limit of the empirical process $M_{n}\left( \cdot
\right) $ is well developed in the literature, (see e.g. \cite{Pons2003,
kosorok2007, Lee:Seo:08}) and the argmax continuous mapping theorem by \cite%
{Seijo:Sen:11b} yields the asymptotic distribution, namely the smallest
minimizer of a compound Poisson process, which is  defined in 
Theorem \ref{thm:tau-2nd}.
\end{proof}

\subsection{Proof of Theorem \ref{th2.3-improved}}


Let  $\widehat{D} \equiv D\left( \widehat{\tau}\right)$.  
It follows from the definition of $\widehat{\alpha}$ in \eqref{step3-alpha-est}
that 
\begin{align*}
\frac{1}{n}\sum_{i=1}^{n}\rho (Y_{i},X_{i}(\widehat{\tau})^{T}\widehat{\alpha })
+\omega _{n}|\widehat{D}\widehat{\alpha }|_{1}\leq \frac{1}{n}%
\sum_{i=1}^{n}\rho (Y_{i},X_{i}(\widehat{\tau})^{T}\alpha _{0})+\omega
_{n}|\widehat{D}\alpha _{0}|_{1}.
\end{align*}
From this, we obtain the following inequality 
\begin{align}\label{basic_ineq_2nd_step}
R(\widehat{\alpha },\widehat{\tau}) &\leq 
\left[ \nu _{n}(\alpha_{0},\widehat{\tau})-\nu _{n}(\widehat{\alpha},\widehat{\tau })\right] 
+ R(\alpha_0,\widehat{\tau})
+\omega
_{n}|\widehat{D}\alpha _{0}|_{1}-\omega _{n}|\widehat{D}\widehat{\alpha }|_{1}. 
\end{align}
Now applying Lemma \ref{lem-emp} 
to 
$[\nu _{n} (\alpha _{0},\widehat{\tau})-\nu _{n}(\widehat{\alpha },\widehat{\tau})]$, 
we rewrite the basic inequality in \eqref{basic_ineq_2nd_step} by%
\begin{equation*}
\omega _{n}\left\vert \widehat{D} \alpha _{0}\right\vert _{1}\geq R(\widehat{\alpha },\widehat{\tau}) 
 +\omega  _{n}\left\vert \widehat{D}\widehat{\alpha}\right\vert
_{1}-\frac{1}{2}\omega _{n}\left\vert \widehat{D}\left( \widehat{\alpha}-\alpha
_{0}\right) \right\vert _{1}-\left\vert R(\alpha_0,\widehat{\tau})
\right\vert \; \text{ w.p.a.1}.
\end{equation*}
As before, adding $\omega _{n}\left\vert \widehat{D}\left( \widehat{\alpha}-\alpha
_{0}\right) \right\vert _{1}$ on both sides of the inequality above
and using the fact that $ \left\vert  \alpha_{0j} \right\vert _{1} - \left\vert \widehat{\alpha}_{j} \right\vert
_{1} + \left\vert \left( \widehat{\alpha}_{j} -\alpha
_{0j}\right) \right\vert _{1} = 0$ for $j \notin J$,
we have that 
\begin{align}\label{dev-eq1}
R\left( \widehat{\alpha},\widehat{\tau}\right)
 +\frac{1}{2}\omega
_{n}\left\vert \widehat{D}\left( \widehat{\alpha}-\alpha _{0}\right) \right\vert
_{1} & \leq 
\left\vert R(\alpha_0,\widehat{\tau}) \right\vert
+ 2\omega _{n}\left\vert \widehat{D}\left( \widehat{\alpha}-\alpha _{0}\right) _{J}\right\vert _{1}\; \text{ w.p.a.1}.
\end{align}%

As in the proof of Theorem \protect\ref{th2.3}, we   consider 
 two cases: (i) $\omega _{n}\left\vert \widehat{D}\left( \widehat{\alpha}-\alpha _{0}\right) _{J}\right\vert _{1} \leq \left\vert R(\alpha_0,\widehat{\tau}) \right\vert$; (ii) $\omega _{n}\left\vert \widehat{D}\left( \widehat{\alpha}-\alpha _{0}\right) _{J}\right\vert _{1} > \left\vert R(\alpha_0,\widehat{\tau}) \right\vert$.
 We first consider case (ii). 
Recall that $\|Z\|_2=(EZ^2)^{1/2}$ for a random variable $Z.$  
It follows from 
the compatibility condition (Assumption \ref{ass2.7}) and 
the same arguments as in \eqref{a.11add} that 
\begin{align}\label{dev-eq2}
\begin{split}
\omega _{n}\left\vert \widehat{D}\left( \widehat{\alpha}-\alpha _{0}\right) _{J}\right\vert _{1}
& \leq
\omega _{n}\bar D\left\Vert X(\widehat{\tau})^T(\widehat\alpha-\alpha_0)\right\Vert _{2}\sqrt{s}/\phi  \\
&\leq \frac{\omega _{n}^{2}\bar D^{2}s}{2 \tilde{c} \phi ^{2}}+\frac{\tilde{c}}{2}%
\left\Vert X(\widehat{\tau})^T(\widehat\alpha-\alpha_0)\right\Vert _{2}^{2}
\end{split}
\end{align}
 for any $\tilde{c} >0$.
Recall that $\bar{c}_0(\delta_0) \equiv c_0  \inf_{\tau \in \mathcal{T}_0}\mathbb{E}[(X^T\delta_0)^2|Q=\tau]$.
 As in Step 5 of the proof of Theorem  \protect\ref{th2.3}, there is a  constant $C_0 > 0$ such that 
\begin{align}\label{dev-eq3}
\left\Vert X(\widehat\tau)^T(\widehat\alpha-\alpha_0)\right\Vert _{2}^{2}\leq C_0 R(\widehat\alpha,\widehat\tau)+C_0 \bar{c}_0(\delta_0) |\widehat\tau-\tau_0|,
\end{align}
w.p.a.1.
Combining \eqref{dev-eq1}-\eqref{dev-eq3} with a sufficiently small $\tilde{c}$ yields
\begin{align}\label{dev-eq4}
R\left( \widehat{\alpha},\widehat{\tau}\right)
 +\omega
_{n}\left\vert \widehat{D}\left( \widehat{\alpha}-\alpha _{0}\right) \right\vert
_{1} & \leq 
  C \left( \omega _{n}^{2} s +
 |\widehat\tau-\tau_0| \right)
\end{align}
for some finite constant $C > 0$. 
Since $|\widehat{\tau}-\tau _{0}|=O_P (n^{-1})$ by Theorem \ref{thm:tau-2nd},
the desired results follow \eqref{dev-eq4} immediately.

Now we consider case (i). In this case, 
\begin{align}\label{dev-eq11}
R\left( \widehat{\alpha},\widehat{\tau}\right)
 +\frac{1}{2}\omega
_{n}\left\vert \widehat{D}\left( \widehat{\alpha}-\alpha _{0}\right) \right\vert
_{1} & \leq 
3\left\vert R(\alpha_0,\widehat{\tau}) \right\vert.
\end{align}%
As shown in the proof of Theorem \ref{l2.1-improved}, we have that  
\begin{align}\label{dev-eq22}
 |R\left( \alpha_0, \widehat{\tau} \right)|
&= O_P \left( \left|\delta_{0}\right|_{1} n^{-1} \log n \right)
= O_P \left( \omega _{n}^{2} s \right).
\end{align}
Therefore, we obtain the desired results in case (i) as well by combining \eqref{dev-eq22} with \eqref{dev-eq11}.

\subsection{Proof of Theorems \ref{th3.1}}


We write $\alpha_J$ be a subvector of $\alpha$ whose components' indices are in $J(\alpha_0)$. 
Define  $\bar{Q}_{n}(\alpha _{J}) \equiv \widetilde{S%
}_{n}((\alpha _{J},0))$,  so that%
\begin{equation*}
\bar{Q}_{n}(\alpha _{J})=\frac{1}{n}\sum_{i=1}^{n}\rho (Y_{i},X_{iJ}(%
\widehat{\tau})^{T}\alpha _{J})+\mu _{n}\sum_{j\in J(\alpha_0)}w_{j}\widehat{D}_{j}|\alpha
_{j}|.
\end{equation*}%
For notational simplicity, here we write $\widehat D_j\equiv D_j(\widehat\tau)$. When $\tau_0$ is identifiable, our argument is conditional on 
\begin{equation}
\widehat{\tau}\in \mathcal{T}_{n}=\left\{ \left\vert \tau -\tau _{0}\right\vert
\leq n^{-1} \log n\right\} ,  \label{Tn}
\end{equation}%
whose probability goes to $1$ due to Theorem \ref{thm:tau-2nd}.  


We first prove the following two lemmas.
Define \begin{align}\label{bar-alpha-J-def}
\bar{\alpha}_{J}\equiv\limfunc{argmin}_{\alpha _{J}}\bar{Q}_{n}(\alpha _{J}).
\end{align}

\begin{lem}
\label{la.1}    Suppose that
$M_n^2 (\log n)^2/(s \log s) = o(n)$, $s^4\log s=o(n)$,
$s^2 \log n/\log s =o(n)$
 and $\widehat\tau\in\mathcal{T}_n$ if $\delta_0\neq0$; suppose that $s^4\log s=o(n)$ and $\widehat\tau$  is any value in $\mathcal{T}$ if $\delta_0=0$.   Then 
\begin{equation*}
|\bar{\alpha }_{J}-\alpha _{0J}|_{2}=O_P\left( \sqrt{\frac{s\log s}{n}} \; \right).
\end{equation*}
\end{lem}

\begin{proof}[Proof of Lemma \ref{la.1}.]
Let $k_{n}=\sqrt{\frac{s\log s}{n}}$. We first prove that for any $\epsilon >0$%
, there is $C_{\epsilon}>0$, with probability at least $1-\epsilon$, 
\begin{equation}
\inf_{|\alpha _{J}-\alpha _{0J}|_{2}=C_{\epsilon}k_{n}}\bar{Q}%
_{n}(\alpha _{J})>\bar{Q}_{n}(\alpha _{0J})  \label{1.1}
\end{equation}%
Once this is proved, then by the continuity of $\bar{Q}_{n}$, there is
a local minimizer of $\bar{Q}_{n}(\alpha _{J})$ inside $B(\alpha
_{0J},C_{\epsilon}k_{n}) \equiv \{\alpha _{J}\in \mathbb{R}^{s}:|\alpha _{0J}-\alpha
_{J}|_{2}\leq C_{\epsilon}k_{n}\}$. Due to the convexity of $\bar Q_n$, such a  local minimizer is also global. We now prove (\ref{1.1}). 

Write 
\begin{equation*}
l_{J}(\alpha _{J})=\frac{1}{n}\sum_{i=1}^{n}\rho (Y_{i},X_{iJ}(\widehat{\tau}%
)^{T}\alpha _{J}),\quad 
L_{J}(\alpha _{J}, \tau)= \mathbb{E}[ \rho (Y,X_{J}(\tau)^{T}\alpha _{J})].
\end{equation*}%
Then  for all $|\alpha _{J}-\alpha _{0J}|_{2}=C_{\epsilon}k_{n}$,
\begin{eqnarray*}
&&\bar{Q}_{n}(\alpha _{J})-\bar{Q}_{n}(\alpha _{0J}) \\
&=&l_{J}(\alpha _{J})-l_{J}(\alpha _{0J})+\sum_{j\in J(\alpha_0)}w_{j}\mu _{n}\widehat{D}%
_{j}(|\alpha _{j}|-|\alpha _{0j}|) \\
&\geq &\underbrace{L_{J}(\alpha _{J},\widehat{\tau})-L_{J}(\alpha
_{0J}, \widehat{\tau})}_{(1)}-\underbrace{\sup_{|\alpha _{J}-\alpha _{0J}|_{2}\leq C_{\delta
}k_{n}}|\nu _{n}(\alpha _{J},\widehat{\tau})-\nu _{n}(\alpha _{0J},\widehat{\tau})|}%
_{(2)}+\underbrace{\sum_{j\in J(\alpha_0)}\mu _{n}\widehat{D}_{j}w_{j}(|\alpha
_{j}|-|\alpha _{0j}|)}_{(3)}.
\end{eqnarray*}%

To analyze (1), note that $|\alpha _{J}-\alpha _{0J}|_{2}=C_{\epsilon}k_{n}$
and $m_{J}(\tau _{0},\alpha _{0})=0$ and when $\delta_0=0$, $m_J(\tau,\alpha_{0J})$ is free of $\tau$. Then there is $c_{3}>0$, 
\begin{eqnarray*}
&&L_{J}(\alpha _{J},\widehat{\tau})-L_{J}(\alpha _{0J},\widehat{\tau})\cr &\geq
&m_{J}(\tau _{0},\alpha _{0J})^{T}(\alpha _{J}-\alpha _{0J})+(\alpha
_{J}-\alpha _{0J})^{T}\frac{\partial ^{2} \mathbb{E}[ \rho (Y,X_{J}(\widehat{\tau}%
)^{T}\alpha _{0J})]}{\partial \alpha _{J}\partial \alpha _{J}^{T}}(\alpha
_{J}-\alpha _{0J})\cr 
&&-|m_{J}(\tau _{0},\alpha _{0J})-m_{J}(\widehat{\tau},\alpha _{0J})|_{2}|\alpha
_{J}-\alpha _{0J}|_{2}-c_{3}|\alpha _{0J}-\alpha _{J}|_{1}^{3}\cr &\geq
&\lambda _{\min } \left(\frac{\partial ^{2} \mathbb{E}[ \rho (Y,X_{J}(\widehat{\tau})^{T}\alpha_{0J})]}{\partial \alpha _{J}\partial \alpha _{J}^{T}} \right)|\alpha _{J}-\alpha
_{0J}|_{2}^{2}\cr 
&&-(|m_{J}(\tau _{0},\alpha _{0J})-m_{J}(\widehat{\tau},\alpha _{0J})|_{2})|\alpha
_{J}-\alpha _{0J}|_{2}-c_{3}s^{3/2}|\alpha _{0J}-\alpha _{J}|_{2}^{3}\cr %
&\geq &c_{1}C_{\epsilon}^{2}k_{n}^{2}-(|m_{J}(\tau _{0},\alpha _{0J})-m_{J}(%
\widehat{\tau},\alpha _{0J})|_{2})C_{\epsilon}k_{n}-c_{3}s^{3/2}C_{\delta
}^{3}k_{n}^{3} \\
&\geq &C_{\epsilon}k_{n}(c_{1}C_{\epsilon}k_{n}-M_n n^{-1} 
\log n-c_{3}s^{3/2}C_{\epsilon}^{2}k_{n}^{2})\geq c_{1}C_{\delta
}^{2}k_{n}^{2}/3,
\end{eqnarray*}
where the last inequality follows from $M_n n^{-1}
\log n<1/3c_{1}C_{\epsilon}k_{n}$ and  $c_{3}s^{3/2}C_{\epsilon}^{2}k_{n}^{2}<1/3c_{1}C_{\epsilon}k_{n}$. These follow from the 
  conditions $M_n^2 (\log n)^2/(s \log s) = o(n)$ and $s^4\log s=o(n)$.

To analyze (2), by the symmetrization theorem   and the contraction theorem (see, for example, Theorems 14.3 and 14.4 of \cite{bulmann}), there is a Rademacher sequence $\epsilon _{1},...,\epsilon _{n}$
independent of $\{Y_{i},X_{i},Q_{i}\}_{i\leq n}$ such that (note that when $\delta_0=0$, $\alpha_J=\beta_J$,  \begin{equation*}
\nu _{n}\left( \alpha_J ,\tau \right) \equiv\frac{1}{n}\sum_{i=1}^{n}\left[ \rho
\left( Y_{i},X_{J(\beta_0)i}  ^{T}\beta_J \right) -
\mathbb{E} \rho \left(Y,X_{J(\beta_0)} ^{T}\beta_J \right) \right],
\end{equation*}
which is free of $\tau$)
\begin{eqnarray*}
V_{n} &=&\mathbb{E} \left( \sup_{\tau \in \mathcal{T}_{n}}\sup_{|\alpha _{J}-\alpha
_{0J}|_{2}\leq C_{\epsilon}k_{n}}|\nu _{n}(\alpha _{J},\tau)-\nu _{n}(\alpha _{0J},\tau)| \right) \\
&\leq &2\mathbb{E} \left( \sup_{\tau \in \mathcal{T}_{n}}\sup_{|\alpha _{J}-\alpha
_{0J}|_{2}\leq C_{\epsilon}k_{n}}
\left|\frac{1}{n}\sum_{i=1}^{n}\epsilon _{i}[\rho
(Y_{i},X_{iJ}(\tau )^{T}\alpha _{J})-\rho (Y_{i},X_{iJ}(\tau )^{T}\alpha
_{0J})]\right| \right) \\
&\leq &4L \mathbb{E} \left( \sup_{\tau \in \mathcal{T}_{n}}\sup_{|\alpha _{J}-\alpha
_{0J}|_{2}\leq C_{\epsilon}k_{n}}\left| \frac{1}{n}\sum_{i=1}^{n}\epsilon
_{i}(X_{iJ}(\tau )^{T}\left( \alpha _{J}-\alpha _{0J}\right) )\right| \right),
\end{eqnarray*}%
which is bounded by the sum of the following two terms, $V_{1n}+V_{2n}$,  due to the triangle
inequality and the fact that $|\alpha _{J}-\alpha _{0J}|_{1}\leq |\alpha
_{J}-\alpha _{0J}|_{2}\sqrt{s}$:  first, 
when $\delta_0=0$, $V_{1n}\equiv0$; second, when $\delta_0\neq0$ and $\tau_0$ is identifiable,  we have that
\begin{eqnarray*}
V_{1n} &=&4L\mathbb{E} \left( \sup_{\tau \in \mathcal{T}_{n}}\sup_{|\alpha _{J}-\alpha
_{0J}|_{1}\leq C_{\epsilon}k_{n}\sqrt{s}} \left| \frac{1}{n}\sum_{i=1}^{n}\epsilon
_{i}(X_{iJ}(\tau )-X_{iJ}(\tau _{0}))^{T}(\alpha _{J}-\alpha _{0J}) \right| \right) \\
&\leq &4L\mathbb{E} \left( \sup_{\tau \in \mathcal{T}_{n}}\sup_{|\delta _{J}-\delta
_{0J}|_{1}\leq C_{\epsilon}k_{n}\sqrt{s}} \left|\frac{1}{n}\sum_{i=1}^{n}%
\epsilon _{i}X_{iJ(\delta_0)}^{T}(1\{Q_{i}>\tau \}-1\{Q_{i}>\tau _{0}\})(\delta
_{J}-\delta _{0J}) \right| \right) \\
&\leq &4LC_{\epsilon}k_{n}\sqrt{s} \mathbb{E} \left( \sup_{\tau \in \mathcal{T}_{n}}\max_{j\in
J(\delta_0)}\left\vert \frac{1}{n}\sum_{i=1}^{n}\epsilon _{i}X_{ij}(1\{Q_{i}>\tau
\}-1\{Q_{i}>\tau _{0}\})\right\vert \right) \\
&\leq &4LC_{\epsilon}k_{n}\sqrt{s}C_{1}\left\vert J(\delta_0)\right\vert_0 \sqrt{%
\frac{\log n}{n^2}}
\end{eqnarray*}%
by bounding the maximum over $j$ with  summation and using the maximal inequality in Theorem 2.14.1 in \cite{VW} since the class of transformations $ \epsilon_i X_{ij} (1\{Q_i > \tau \} - 1\{Q_i > \tau_0 \}) $ constitutes a VC class of functions. Here the bound is uniform and determined by the $ L_2 $-norm of the envelope, which is proportional to $$ \sqrt{\mathbb{E} \left( 1\{|Q_i - \tau_0| \leq n^{-1} \log n \}  \right)}.$$
Note that 	
\begin{eqnarray*}
V_{2n} &=&4L \mathbb{E}\left( \sup_{|\alpha _{J}-\alpha _{0J}|_{1}\leq C_{\epsilon}k_{n}\sqrt{s%
}} \left|\frac{1}{n}\sum_{i=1}^{n}\epsilon _{i}X_{iJ}(\tau _{0})^{T}(\alpha
_{J}-\alpha _{0J})\right| \right) \\
&\leq &4LC_{\epsilon}k_{n}\sqrt{s} \mathbb{E} \left( \max_{j\in J(\alpha_0)} \left|\frac{1}{n}\sum_{i=1}^{n}\epsilon _{i}X_{ij}(\tau
_{0}) \right| \right)\leq 4LC_{\epsilon} C_{2}k_{n}^2,
\end{eqnarray*}%
due to the Bernstein's moment inequality (Lemma 14.12 of \cite{bulmann} for some $C_{2}>0.$ Therefore, 
$$
V_n\leq 4LC_{\epsilon}k_{n}\sqrt{s}C_{1}\left\vert J(\delta_0)\right\vert_0 \sqrt{%
\frac{\log n}{n^2}}+4LC_{\epsilon} C_{2}k_{n}^2<5LC_{\epsilon} C_{2}k_{n}^2, 
$$
where the last inequality is due to  
 $s^2 \log n/\log s =o(n)$.
Therefore, conditioning on the event $\widehat{\tau}\in \mathcal{T}_{n}$ when $\delta_0\neq0$, or for $\widehat\tau\in\mathcal{T}$  when $\delta_0=0$, with
probability at least $1-\epsilon$, $(2)\leq \frac{1}{\epsilon }5LC_{2}C_{\epsilon}k_{n}^{2}$.

In addition, note that $P(\max_{j\in J(\alpha_0)}|w_{j}|=0)=1$, so $(3)=0$ with
probability approaching one. Hence 
\begin{equation*}
\inf_{|\alpha _{J}-\alpha _{0J}|_{2}=C_{\epsilon}k_{n}}\bar{Q}%
_{n}(\alpha _{J})-\bar{Q}_{n}(\alpha _{0J})\geq \frac{c_{1}C_{\epsilon}^{2}k_{n}^{2}}{3}-\frac{1}{\epsilon}5LC_{2}C_{\epsilon}k_{n}^{2}>0.
\end{equation*}%
The last inequality holds for $C_{\epsilon}>\frac{15LC_{2}}{c_{1}\epsilon }$.
By the continuity of $\bar{Q}_{n}$, there is a local minimizer of $%
\bar{Q}_{n}(\alpha _{J})$ inside $\{\alpha _{J}\in \mathbb{R}%
^{s}:|\alpha _{0J}-\alpha _{J}|_{2}\leq C_{\epsilon}k_{n}\}$, which is also a global minimizer due to the convexity. $\hfill $
\end{proof}

On $\mathbb{R}^{2p}$, recall that 
\begin{equation*}
R_n(\tau, \alpha)=\frac{1}{n}\sum_{i=1}^n\rho(Y_i, X_i(\tau)^T\alpha).
\end{equation*}

For $\bar\alpha_J=(\bar\beta_{J(\beta_0)},\bar\delta_{J(\delta_0)})\equiv (\bar \beta_J,\bar\delta_J)$ in
the previous lemma, define 
\begin{equation*}
\bar\alpha=(\bar\beta_{J}^T,0^T, \bar\delta_{J}^T,0^T)^T.
\end{equation*}
Without introducing confusions, we also write $\bar\alpha=(\bar%
\alpha_J,0)$ for notational simplicity. This notation indicates that $%
\bar\alpha$ has zero entries on the indices outside the oracle index set $J(\alpha_0)$. We prove the following lemma.

\begin{lem}
\label{la.2} With probability approaching one, there is a random
neighborhood of $\bar\alpha$ in $\mathbb{R}^{2p}$, denoted by $%
\mathcal{H}$, so that $\forall\alpha=(\alpha_J,\alpha_{J^c})\in \mathcal{H}$%
, if $\alpha_{J^c}\neq0$, we have $\widetilde S_n(\alpha_J,0)<\widetilde
Q_n(\alpha).$
\end{lem}

\begin{proof}[Proof of Lemma \ref{la.2}]
Define an $l_{2}$-ball, for $r_{n} \equiv \mu _{n}/\log n,$ 
\begin{equation*}
\mathcal{H}=\{\alpha \in \mathbb{R}^{2p}:|\alpha -\bar{\alpha }%
|_{2}<r_{n}/(2p)\}.
\end{equation*}%
Then $\sup_{\alpha \in \mathcal{H}}|\alpha -\bar{\alpha }%
|_{1}=\sup_{\alpha \in \mathcal{H}}\sum_{l\leq 2p}|\alpha _{l}-\bar{%
\alpha }_{l}|<r_{n}.$ Consider any $\tau \in \mathcal{T}_{n}$. For any $%
\alpha =(\alpha _{J},\alpha _{J^{c}})\in \mathcal{H}$, write 
\begin{eqnarray*}
&&R_n(\tau ,\alpha _{J},0)-R_n(\tau ,\alpha ) \\
&&= R_n(\tau ,\alpha
_{J},0)-\mathbb{E}R_n(\tau ,\alpha _{J},0)+\mathbb{E}R_n(\tau ,\alpha _{J},0)-R_n(\tau
,\alpha )+\mathbb{E}R_n(\tau ,\alpha )-\mathbb{E}R_n(\tau ,\alpha )\cr
&&\leq \mathbb{E}R_n(\tau
,\alpha _{J},0)-\mathbb{E}R_n(\tau ,\alpha )+|R_n(\tau ,\alpha
_{J},0)-\mathbb{E}R_n(\tau ,\alpha _{J},0)+\mathbb{E}R_n(\tau ,\alpha )-R_n(\tau ,\alpha
)|\cr
&&\leq \mathbb{E}R_n(\tau ,\alpha _{J},0)-\mathbb{E}R_n(\tau ,\alpha )+|\nu _{n}(\alpha _{J},0,\tau)-\nu _{n}(\alpha,\tau )|.
\end{eqnarray*}

Note that $|(\alpha_J,0)-\bar\alpha|_2^2=|\alpha_J-\bar\alpha_J|_2^2\leq
|\alpha_J-\bar\alpha_J|_2^2+|\alpha_{J^c}-0|_2^2=|\alpha-\bar\alpha|_2^2$. Hence $\alpha\in\mathcal{H}$ implies $(\alpha_J,0)\in\mathcal{H%
}$. In addition, by definition of $\bar\alpha=(\bar\alpha_J,0)$
and $|\bar\alpha_J-\alpha_{0J}|_2=O_P(\sqrt{\frac{s\log s}{n}})$ (Lemma \ref{la.1}), we
have $|\bar\alpha-\alpha_0|_1=O_P(s\sqrt{\frac{\log s}{n}})$, which
also implies 
\begin{equation*}
\sup_{\alpha\in\mathcal{H}}|\alpha-\alpha_0|_1=O_P
\left( s\sqrt{\frac{\log s}{n}} \right)+r_n,
\end{equation*}
where the randomness in $\sup_{\alpha\in\mathcal{H}}|\alpha-\alpha_0|_1$
comes from that of $\mathcal{H}$.

By the mean value theorem, there is $h$ in the segment between $\alpha$ and $%
(\alpha_J,0)$, 
\begin{eqnarray*}
\mathbb{E}R_n(\tau, \alpha_J,0)-\mathbb{E}R_n(\tau,\alpha)&=&\mathbb{E}\rho(Y,
X_J(\tau)^T\alpha_J)-\mathbb{E}\rho(Y,X_J(\tau)^T\alpha_J+X_{J^c}(\tau)^T\alpha_{J^c})%
\cr & =&-\sum_{j\notin J(\alpha_0)}\frac{\partial \mathbb{E}\rho(Y, X(\tau)^Th)}{\partial
\alpha_j}\alpha_j\equiv \sum_{j\notin J(\alpha_0)}m_j(\tau, h)\alpha_j
\end{eqnarray*}
where $m_j(\tau, h)=-\frac{\partial \mathbb{E}\rho(Y, X(\tau)^Th)}{\partial \alpha_j}%
. $ Hence, $\mathbb{E}R_n(\tau, \alpha_J,0)-\mathbb{E}R_n(\tau, \alpha)\leq \sum_{j\notin
J}|m_j(\tau, h)||\alpha_j|. $

Because $h$ is on the segment between $\alpha $ and $(\alpha _{J},0)$, so $%
h\in \mathcal{H}$. So for all $j\notin J(\alpha_0)$, 
\begin{equation*}
|m_{j}(\tau ,h)|\leq \sup_{\alpha \in \mathcal{H}}|m_{j}(\tau ,\alpha )|\leq
\sup_{\alpha \in \mathcal{H}}|m_{j}(\tau ,\alpha )-m_{j}(\tau ,\alpha
_{0})|+|m_{j}(\tau ,\alpha _{0})-m_{j}(\tau _{0},\alpha _{0})|.
\end{equation*}
We now argue that  we can apply Assumption \ref{ass3.2}  \ref{ass3.2:itm2}. Let 
$$c_{n} \equiv s\sqrt{\left( \log s\right) /n}+r_{n}.$$ 
For any $\epsilon >0$, there is $C_{\epsilon}>0$, with probability at last $%
1-\epsilon$,  $\sup_{\alpha \in 
\mathcal{H}}|\alpha -\alpha _{0}|_{1}\leq C_{\epsilon}c_n$. $\forall \alpha\in\mathcal{H}$,  write $\alpha=(\beta, \delta)$ and $\theta=\beta+\delta$. On the event $|\alpha-\alpha_0|_1\leq C_{\epsilon} c_n$,   we have $|\beta-\beta_0|_1\leq C_{\epsilon} c_n$ and $|\theta-\theta_0|_1\leq C_{\epsilon} c_n$. Hence $\mathbb{E}[(X^T(\beta-\beta_0))^2 1\{Q\leq\tau_0\}]\leq |\beta-\beta_0|_1^2\max_{i,j\leq p}E|X_iX_j|<r^2$, yielding $\beta\in\mathcal{B}(\beta_0, r)$. Similarly, $\theta\in\mathcal{G}(\theta_0, r)$. Therefore, by Assumption \ref{ass3.2} \ref{ass3.2:itm2}, with probability at least $1-\epsilon$,  (note that neither $C_{\epsilon}, L$ nor $c_n$ depend on $\alpha$)
\begin{align*}
\max_{j\notin J(\alpha_0)}\sup_{\tau \in \mathcal{T}_{n}}\sup_{\alpha \in \mathcal{H}}|m_{j}(\tau ,\alpha )-m_{j}(\tau ,\alpha _{0})| & \leq L\sup_{\alpha \in 
\mathcal{H}}|\alpha -\alpha _{0}|_{1}\leq L(C_{\epsilon}c_n),\\
\max_{j\leq 2p}\sup_{\tau \in \mathcal{T}_{n}}|m_{j}(\tau
,\alpha _{0})-m_{j}(\tau _{0},\alpha _{0})| &\leq 
M_{n}n^{-1} \log n.
\end{align*} 
In particular, when $\delta_0=0$, $m_j(\tau,\alpha_0)=0$ for all $\tau$.
Therefore, when $\delta_0\neq 0$, $$\sup_{j\notin J(\alpha_0)}\sup_{\tau \in \mathcal{T}_{n}}|m_{j}(\tau
,h)|=O_P(c_n+M_{n}n^{-1} \log
n)=o_P(\mu _{n});$$
when $\delta_0=0$,  $\sup_{j\notin J(\alpha_0)}\sup_{\tau \in \mathcal{T}}|m_{j}(\tau
,h)|=O_P(c_n)=o_P(\mu _{n})$.

Let  $\epsilon_1,...,\epsilon_n$  be  a  Rademacher  sequence  independent of $\{Y_i, X_i, Q_i\}_{i\leq n}$. Then by the  symmetrization and contraction theorems, 
\begin{eqnarray*}
 && \mathbb{E} \left( \sup_{\tau \in \mathcal{T}} |\nu _{n}(\alpha _{J},0, \tau)-\nu _{n}(\alpha,\tau)| \right) \cr
&&\leq 2\mathbb{E} \left( \sup_{\tau \in \mathcal{T}} 
\left|\frac{1}{n}\sum_{i=1}^{n}\epsilon _{i}[\rho
(Y_{i},X_{iJ}(\tau )^{T}\alpha _{J})-\rho (Y_{i},X_{i}(\tau )^{T}\alpha)]\right| \right) \\
&&\leq 4L \mathbb{E} \left( \sup_{\tau \in \mathcal{T}} \left| \frac{1}{n}\sum_{i=1}^{n}\epsilon_{i}[X_{iJ}(\tau )^{T} \alpha _{J}-X_i(\tau)^T\alpha]\right| \right)\cr
&&\leq 4L \mathbb{E} \left( \sup_{\tau \in \mathcal{T}} \left\| \frac{1}{n}\sum_{i=1}^{n}\epsilon_{i}X_{i}(\tau )\right\|_{\max} \right)\sum_{j\notin J(\alpha_0)}|\alpha_j|\leq 2\omega_n\sum_{j\notin J(\alpha_0)}|\alpha_j|,
\end{eqnarray*}%
where the last equality follows from (\ref{eqa.5}).

  Thus uniformly over $\alpha\in\mathcal{H}$, 
  $
R_n(\tau ,\alpha _{J},0)-R_n(\tau ,\alpha )=o_P(\mu _{n})\sum_{j\notin
J(\alpha_0)}|\alpha _{j}|.
$ 
On the other hand, $$\sum_{j\in J(\alpha_0)}w_{j}\mu _{n}\widehat{D}_{j}|\alpha
_{j}|-\sum_{j}w_{j}\mu _{n}\widehat{D}_{j}|\alpha _{j}|=\sum_{j\notin J(\alpha_0)}\mu
_{n}w_{j}\widehat{D}_{j}|\alpha _{j}|.$$
 Also, w.p.a.1, $%
w_{j}=1$ and $\widehat{D}_{j}\geq \overline{D}$ for all $j\notin J(\alpha_0)$. Hence with
probability approaching one, $\widetilde{Q}_{n}(\alpha _{J},0)-\widetilde{Q}%
_{n}(\alpha )$ equals 
\begin{equation*}
R_n(\widehat{\tau},\alpha _{J},0)+\sum_{j\in J(\alpha_0)}\widehat{D}_{j}w_{j}\lambda
_{n}|\alpha _{j}|-R_n(\widehat{\tau},\alpha )-\sum_{j\leq 2p}\widehat{D}%
_{j}w_{j}\omega_n|\alpha _{j}|\leq -\underline{D}\frac{\mu _{n}}{2}%
\sum_{j\notin J(\alpha_0)}|\alpha _{j}|<0.  \text{\qed}
\end{equation*}

\noindent
\textbf{Proof of Theorem \ref{th3.1}.}
Conditions in Lemmas \ref{la.1} and \ref{la.2} are expressed in terms of $M_n$. 
By Lemma \ref{l4.1}, we verify that 
in quantile  regression models, $M_n=Cs^{1/2}$ for some $C>0$. 
Then all the required conditions in Lemmas \ref{la.1} and \ref{la.2}
are satisfied by the conditions imposed in Theorem \ref{th3.1}.

By Lemmas \ref{la.1} and \ref{la.2},  w.p.a.1,  for any $\alpha =(\alpha _{J},\alpha
_{J^{c}})\in \mathcal{H}$, 
\begin{equation*}
\widetilde{S}_{n}(\bar \alpha_J, 0)=\bar Q_n(\bar\alpha_J)\leq \bar Q_n(\alpha_J)=\widetilde S_n(\alpha_J, 0)\leq \widetilde S_n(\alpha).
\end{equation*}%
Hence $(\bar\alpha_J,0)$ is a local minimizer of $\widetilde S_n$, which is also a global minimizer due to the convexity. This implies that w.p.a.1,  $\widetilde\alpha=(\widetilde\alpha_J,\widetilde\alpha_{J^c})$ satisfies: $\widetilde\alpha_{J^c}=0$, and  $\widetilde\alpha_J=\bar\alpha_J$, so 
\begin{equation*}
\left\vert \widetilde{\alpha }_{J}-\alpha _{0J}\right\vert _{2}=O_P \left( \sqrt{%
\frac{s\log s}{n}} \; \right),\quad \left\vert \widetilde{\alpha }_{J}-\alpha _{0J}\right\vert _{1}=O_P \left( s\sqrt{%
\frac{\log s}{n}} \; \right).
\end{equation*}
Finally, by (\ref{ec.47}), and that $R(\alpha_0,\widehat\tau)\leq Cs|\widehat\tau-\tau_0|=O_P(sn^{-1})$,
$$
R(\widetilde\alpha,\widehat\tau) \leq 2R(\alpha_0,\widehat\tau)+ 3\mu_n\bar D|\tilde\alpha-\alpha_0|_1=O_P(sn^{-1}+\mu_ns\sqrt{\frac{\log s}{n}})=O_P(\mu_ns\sqrt{\frac{\log s}{n}}).
$$

\end{proof}

\subsection{Proof of Theorem \ref{thm-AI}}

Recall that by Theorems \ref{thm:tau-2nd} and  \ref{th3.1}, we have 
\begin{equation}\label{local-nbd}
\left\vert \widetilde{\alpha }_{J}-\alpha _{0J}\right\vert _{2}=O_P \left( \sqrt{%
\frac{s\log s}{n}} \; \right) \ \ \text{ and } \ \
|\widehat{\tau}-\tau _{0}|=O_P (n^{-1} ),
\end{equation}
and the set of regressors with nonzero coefficients  is   recovered w.p.a.1. Hence we can restrict ourselves on the oracle space $J(\alpha_0)$. 
In view of \eqref{local-nbd}, 
define
$r_{n}\equiv\sqrt{n^{-1}s\log s}$ and $s_{n}$.
Let%
\[
R_{n}^{\ast}\left(  \alpha_J,\tau\right)  \equiv\frac{1}{n}\sum_{i=1}^{n}\rho\left(
Y_{i},X_{iJ}(\tau)^{T}\alpha _{J}  \right)  ,
\]
where $\alpha_J \in\mathcal{A}_n \equiv \left\{  \alpha_J :\left\vert \alpha_J -\alpha_{0J}\right\vert _{2}\leq
Kr_{n}\right\}  \subset\mathbb{R}^{s}$ and $\tau\in\mathcal{T}_n \equiv\left\{ \tau: 
\left\vert \tau-\tau_{0}\right\vert \leq Ks_{n}\right\}  $ for some $K<\infty$, where  $K$ is a generic finite constant.

The following lemma is useful to establish that  $\alpha_{0}$ can be estimated  as if $\tau_0$ were known.

\begin{lem}[Asymptotic Equivalence]\label{lem-AI}
Assume that $\frac{\partial}{\partial\alpha}E\left[
\rho\left(  Y,X^{T}\alpha\right)  |Q=t\right]  $ exists for all $t$ in a
neighborhood of $\tau_{0}$ and all its elements are continuous and bounded.  
Suppose that  $s^{3} (\log s) (\log n) =o\left(  n \right)$. 
Then 
\[
\sup_{\alpha_J \in \mathcal{A}_n,\tau \in \mathcal{T}_n}\left\vert \left\{ R_{n}^{\ast}\left(  \alpha_J,\tau\right)
-R_{n}^{\ast}\left(  \alpha_J,\tau_{0}\right) \right\} - \left\{  R_{n}^{\ast}\left(
\alpha_{0J},\tau\right)  -R_{n}^{\ast}\left(  \alpha_{0J},\tau_{0}\right)
\right\}  \right\vert =o_{P}\left(  n^{-1} \right).
\]
\end{lem}

This lemma implies that the asymptotic distribution of $ \operatorname*{argmin}_{\alpha_J}R_{n}^{\ast}\left(  \alpha_J,\widehat{\tau}\right)
$  can be characterized by $\widehat{\alpha}^{\ast}_J\equiv\operatorname*{argmin}_{\alpha_J
}R_{n}^{\ast}\left(  \alpha_J,\tau_{0}\right)$.  It then follows immediately from the variable selection
consistency  that the asymptotic distribution of $\widetilde{\alpha}_J$ is
equivalent to that of $\widehat{\alpha}^{\ast}_J$.
Therefore, we have proved the theorem.

\begin{proof}[Proof of Lemma \ref{lem-AI}]
Noting that
\[
\rho\left(  Y_{i},X_{i}^{T}\beta+X_{i}^{T}\delta1\left\{  Q_{i}>\tau\right\}
\right)  =\rho\left(  Y_{i},X_{i}^{T}\beta\right)  1\left\{  Q_{i}\leq
\tau\right\}  +\rho\left(  Y_{i},X_{i}^{T}\beta+X_{i}^{T}\delta\right)
1\left\{  Q_{i}>\tau\right\}  ,
\]
we have, for $\tau>\tau_{0},$%
\begin{align*}
&  D_{n}\left(  \alpha,\tau\right) \\
&  \equiv \left\{ R_{n}\left(  \alpha,\tau\right)  -R_{n}\left(  \alpha
,\tau_{0}\right) \right\}  -\left\{  R_{n}\left(  \alpha_{0},\tau\right)
-R_{n}\left(  \alpha_{0},\tau_{0}\right)  \right\} \\
&  =\frac{1}{n}\sum_{i=1}^{n}\left[  \rho\left(  Y_{i},X_{i}^{T}\beta\right)
-\rho\left(  Y_{i},X_{i}^{T}\beta_{0}\right)  \right]  1\left\{  \tau
_{0}<Q_{i}\leq\tau\right\} \\
&  -\frac{1}{n}\sum_{i=1}^{n}\left[  \rho\left(  Y_{i},X_{i}^{T}\beta
+X_{i}^{T}\delta\right)  -\rho\left(  Y_{i},X_{i}^{T}\beta_{0}+X_{i}^{T}%
\delta_{0}\right)  \right]  1\left\{  \tau_{0}<Q_{i}\leq\tau\right\} \\
&=: D_{n1}\left(  \alpha,\tau\right) - D_{n2}\left(  \alpha,\tau\right).
\end{align*}
To prove this lemma, we consider empirical processes
\[
\mathbb{G}_{nj}\left(  \alpha_J,\tau\right)  \equiv\sqrt{n}\left(  D_{nj}\left(  \alpha_J
,\tau\right)  - \mathbb{E} D_{nj}\left(  \alpha_J,\tau\right)  \right), \ \ (j = 1,2),
\]
and apply the maximal inequality in Theorem 2.14.2 of \cite{VW}. 

First, for $\mathbb{G}_{n1}\left(  \alpha_J,\tau\right)$, we consider  the following class of functions indexed by $(\beta_J,\tau)$:
\[
\mathcal{F}_n \equiv \{ \left(  \rho\left(
Y_{i},X_{iJ}^{T}\beta_J \right)  -\rho\left(  Y_{i},X_{iJ}^{T}\beta_{0J}\right)
\right)  1\left(  \tau_{0}<Q_{i}\leq\tau\right): 
|\beta_J - \beta_{0J}|_2 \leq K r_n  \text{ and } \left\vert \tau-\tau_{0}\right\vert \leq Ks_{n} \}.
\]
Note that the
Lipschitz property of $\rho$ yields that%
\[
\left\vert \rho\left(  Y_{i},X_{iJ}^{T}\beta_J \right)  -\rho\left(  Y_{i}%
,X_{iJ}^{T}\beta_{0J}\right)  \right\vert 1\left\{  \tau_{0}<Q_{i}\leq
\tau\right\}  \leq \left\vert X_{iJ}^{T}\right\vert _{2} |\beta_J - \beta_{0J}|_2 
1\left\{
\left\vert Q_{i}-\tau_{0}\right\vert \leq Ks_{n}\right\}  .
\]
Thus, we let the envelope function be $F_{n}(X_{iJ},Q_i)\equiv\left\vert X_{iJ}\right\vert _{2}Kr_{n}1\left\{  \left\vert
Q_{i}-\tau_{0}\right\vert \leq Ks_{n}\right\}  \ $ and note that its $L_{2}$
norm is $O\left(  \sqrt{s}r_{n}\sqrt{s_{n}}\right).$ \

To compute the bracketing integral 
$$
J_{[]}\left(1,  \mathcal{F}_{n}, L_2 \right) \equiv \int_0^1 
\sqrt{1 + \log N_{[]} \left( \varepsilon \| F_{n} \|_{L_2}, \mathcal{F}_{n},  L_2 \right)} d\varepsilon,
$$
note that 
its $2\varepsilon$ bracketing
number is bounded by the product of the $\varepsilon$ bracketing numbers of two classes
$\mathcal{F}_{n1} \equiv\left\{  \rho\left(  Y_{i},X_{iJ}^{T}\beta_J \right)  -\rho\left(  Y_{i},X_{iJ}%
^{T}\beta_{0}\right) : |\beta_J - \beta_{0J}|_2 \leq K r_n  \right\}  $ and 
$\mathcal{F}_{n2} \equiv\{ 1\left(  \tau_{0}<Q_{i}\leq
\tau\right) : \left\vert \tau-\tau_{0}\right\vert \leq Ks_{n} \}$
by Lemma 9.25 of \cite{Kosorok} since both classes are bounded w.p.a.1 (note that w.p.a.1,
$\left\vert X_{iJ}\right\vert _{2}Kr_{n} < C < \infty$ for some constant $C$). 
That is, 
\begin{align*}
N_{[]} \left( 2\varepsilon \| F_{n} \|_{L_2}, \mathcal{F}_{n},  L_2 \right)
\leq N_{[]} \left( \varepsilon \| F_{n} \|_{L_2}, \mathcal{F}_{n1},  L_2 \right)
 N_{[]} \left( \varepsilon \| F_{n} \|_{L_2}, \mathcal{F}_{n2},  L_2 \right).
\end{align*}

Let $F_{n1}(X_{iJ})\equiv\left\vert X_{iJ}\right\vert _{2}Kr_{n}$ and  $ l_n(X_{iJ}) \equiv \left\vert X_{iJ}\right\vert _{2}  $. Note that by Theorem 2.7.11 of \cite{VW}, the  Lipschitz property of $\rho$ implies that 
\begin{align*}
N_{[]} \left( 2\varepsilon \| l_n \|_{L_2}, \mathcal{F}_{n1},  L_2 \right)
\leq N( \varepsilon, \{\beta_J :  |\beta_J - \beta_{0J}|_2 \leq K r_n \} , |\cdot|_2),
\end{align*}
which in turn implies that, for some constant $C$,
\begin{align*}
N_{[]} \left( \varepsilon \| F_{n} \|_{L_2}, \mathcal{F}_{n1},  L_2 \right)
&\leq N \left( \frac{\varepsilon \| F_{n} \|_{L_2}}{2\| l_{n} \|_{L_2}} , \{\beta_J :  |\beta_J - \beta_{0J}|_2 \leq K r_n \} , |\cdot|_2 \right) \\
&\leq C \left(  \frac{\sqrt{s}}{\varepsilon \sqrt{s_n}}  \right)^{s} 
= C \left(\frac{\sqrt{n s}}{\varepsilon}\right)^s,
\end{align*}
where the last inequality holds since a $\varepsilon $-ball contains a hypercube with side length $\varepsilon / \sqrt{s}$ in the $s$-dimensional Euclidean space. On the other hand, for the second class of functions $\mathcal{F}_{n2}$
with the envelope function $F_{n2}(Q_i)\equiv1\left\{  \left\vert
Q_{i}-\tau_{0}\right\vert \leq Ks_{n}\right\}$, we have that 
\begin{align*}
N_{[]} \left( \varepsilon \| F_{n} \|_{L_2}, \mathcal{F}_{n2},  L_2 \right) 
&\leq C \frac{\sqrt{s_n}}{\varepsilon \| F_{n} \|_{L_2}}
= \frac{C}{\varepsilon \sqrt{s} r_n }
=  \frac{C \sqrt{n}}{\varepsilon s \sqrt{\log s} },
\end{align*}
for some constant $C$. Combining these results together yields that
\begin{align*}
N_{[]} \left( \varepsilon \| F_{n} \|_{L_2}, \mathcal{F}_{n},  L_2 \right)
\leq
\frac{C^2 \sqrt{n}}{\varepsilon s \sqrt{\log s} } \left(\frac{\sqrt{ns}}{\varepsilon }\right)^s
\leq C^2 \varepsilon^{-s-1} n^{(s+1)/2} 
\end{align*}
for all sufficiently large $n$.
Then we have that 
\begin{align*}
J_{[]}\left(1,  \mathcal{F}_{n}, L_2 \right) \leq C^2 (\sqrt{s \log n} + \sqrt{s})
\end{align*}
for all sufficiently large $n$.
Thus, 
by the maximal inequality in Theorem 2.14.2 of \cite{VW},
\begin{align*}
n^{-1/2} \; \mathbb{E} \sup_{\mathcal{A}_n \times \mathcal{T}_n}\left\vert \mathbb{G}_{n1}\left(
\alpha_J,\tau\right)  \right\vert
&\leq O\left[  n^{-1/2}\sqrt{s}r_{n}\sqrt{s_{n}} (\sqrt{s\log n} + \sqrt{s}) \right]  \\
&= O \left[ \frac{s}{n^{3/2}} \sqrt{\log s}  (\sqrt{s \log n} + \sqrt{s}) \right] 
\\
&=o\left(  n^{-1}\right),
\end{align*}
where the last equality follows from the restriction that 
$s^3 (\log s)  (\log n)  =o\left(  n \right)$.
Identical arguments also apply to $\mathbb{G}_{n2}\left(  \alpha_J,\tau\right)$.

Turning to $\mathbb{E} D_{n}\left(  \alpha,\tau\right)  ,$ note that by the condition that  $\frac{\partial}{\partial\alpha}E\left[
\rho\left(  Y,X^{T}\alpha\right)  |Q=t\right]  $ exists for all $t$ in a
neighborhood of $\tau_{0}$ and all its elements are continuous and bounded, we have that for some mean value $\tilde{\beta}_J$ between $\beta_J$ and $\beta_{0J}$, 
\begin{align*}
&  \left\vert \mathbb{E}\left(  \rho\left(  Y,X_J^{T}\beta_J\right)  -\rho\left(
Y,X_J^{T}\beta_{0J}\right)  \right)  1\left\{  \tau_{0}<Q\leq\tau\right\}
\right\vert \\
&  =\left\vert \mathbb{E}\left[  \frac{\partial}{\partial\beta}\mathbb{E}\left[  \rho\left(
Y,X^{T}\tilde{\beta}_J\right)  |Q\right]  1\left\{  \tau_{0}<Q\leq\tau\right\}
\right]  \left(  \beta-\beta_{0}\right)  \right\vert \\
&  
=O\left(  s r_{n} s_{n}\right)  \\
&=
O \left[ \frac{s^{3/2}}{n^{3/2}} \sqrt{\log s}   \right] \\
&=o\left(  n^{-1}\right),
\end{align*}
where the last equality follows  from the restriction that $s^{3} (\log s)  =o\left(  n \right)$.
Since the same holds for the other term in $\mathbb{E}D_{n},$ $\sup\left\vert
\mathbb{E} D_{n}\left(  \alpha,\tau\right)  \right\vert =o\left(  n^{-1}\right)  $ as desired.
\end{proof}

 
\subsection{Proof of Theorem \ref{th3.1-without}}
  By definition, \begin{align*}
\frac{1}{n}\sum_{i=1}^{n}\rho (Y_{i},X_{i}(\widehat{\tau})^{T}\widetilde{%
\alpha })+\mu _{n}|W\widehat{D}\widetilde{\alpha }|_{1}\leq \frac{1}{n}%
\sum_{i=1}^{n}\rho (Y_{i},X_{i}(\widehat \tau)^{T}\alpha _{0})+\mu
_{n}|W\widehat D\alpha _{0}|_{1}.
\end{align*}
where $W=\diag\{w_1,...,w_{2p}\}$.  From this, we obtain the following inequality 
$$
R(\widetilde\alpha,\widehat\tau) +\mu_n|W\widehat D\widetilde\alpha|_1\leq |\nu_n(\alpha_0, \widehat\tau)-\nu_n(\widetilde\alpha, \widehat\tau)|+R(\alpha_0,\widehat\tau)+\mu_n|W\widehat D\alpha_0|_1.
$$
Now applying Lemma \ref{lem-emp} yields, when $\sqrt{\log(np)/n}=o(\mu_n)$ (which is true under the assumption that $\omega_n \ll \mu_n$), we have that w.p.a.1,  
$
|\nu_n(\alpha_0, \widehat\tau)-\nu_n(\widetilde\alpha, \widehat\tau)|\leq \frac{1}{2}\mu_n|\widehat D(\alpha_0-\widetilde\alpha)|_1.
$
Hence on this event, 
 $$
R(\widetilde\alpha,\widehat\tau) +\mu_n|W\widehat D\widetilde\alpha|_1\leq \frac{1}{2}\mu_n|\widehat D(\alpha_0-\widetilde\alpha)|_1+R(\alpha_0,\widehat\tau)+\mu_n|W\widehat D\alpha_0|_1.
$$
Note that $\max_j w_j\leq 1$, so for $\Delta:=\widetilde\alpha-\alpha_0$,
 $$
R(\widetilde\alpha,\widehat\tau) +\mu_n|(W\widehat D\Delta)_{J^c}|_1\leq \frac{3}{2}\mu_n|\widehat D\Delta_J|_1+\frac{1}{2}\mu_n|\widehat D\Delta_{J^c}|_1+R(\alpha_0,\widehat\tau).
$$
By Theorem \ref{th2.3},   $\max_{j\notin J}|\widehat\alpha_j|=O_P(\omega_ns)$.  Hence for any $\epsilon>0$, there is $C>0$,  $\max_{j\notin J}|\widehat\alpha_j|\leq C\omega_ns<\mu_n$ with probability at least $1-\epsilon$. On the event $\max_{j\notin J}|\widehat\alpha_j|\leq C\omega_ns<\mu_n$,  by definition, $w_j=1$ $\forall j\notin J$. Hence on this event, 
 \begin{equation}\label{ec.47add}
R(\widetilde\alpha,\widehat\tau) +\frac{1}{2}\mu_n|(\widehat D\Delta)_{J^c}|_1\leq \frac{3}{2}\mu_n|\widehat D\Delta_J|_1+R(\alpha_0,\widehat\tau).
\end{equation}
 We now consider two cases: (i) $\frac{3}{2}\mu_n|\widehat D\Delta_J|_1\leq R(\alpha_0,\widehat\tau)$; (ii) $\frac{3}{2}\mu_n|\widehat D\Delta_J|_1>R(\alpha_0,\widehat\tau)$.

\textbf{case 1: $\frac{3}{2}\mu_n|\widehat D\Delta_J|_1\leq R(\alpha_0,\widehat\tau)$}

We have: for $C=14\underline{D}^{-1}/3$,
 $
 \mu_n|\Delta|_1\leq CR(\alpha_0,\widehat\tau).
$
If $\widehat\tau>\tau_0$, for $\tau=\widehat\tau$ in the inequalities below,  
\begin{eqnarray*}
&&R(\alpha_0,\widehat\tau) =\mathbb{E}(\rho(Y,X^T\beta_0)-\rho(X^T\theta_0))1\{\tau_0<Q<\tau\}\leq L\mathbb{E}|X^T\delta_0|1\{\tau_0<Q<\tau\}\cr
&&\leq  L|\delta_0|_1 \max_{j\leq p}E|X_j|1\{\tau_0<Q<\tau\}\leq L|\delta_0|_1 \max_{j\leq p}\sup_qE(|X_j||Q=q)P(\tau_0<Q<\tau)\cr
&&\leq Cs(\tau-\tau_0).
\end{eqnarray*}
The case for $\tau\leq\tau_0$ follows from the same argument. Hence 
$
\mu_n|\Delta|_1\leq C|\widehat\tau-\tau_0|s.
$

\textbf{case 2: $\frac{3}{2}\mu_n|\widehat D\Delta_J|_1> R(\alpha_0,\widehat\tau)$}  

Then by the compatibility property, 
 $$
R(\widetilde\alpha,\widehat\tau) +\frac{1}{2}\mu_n|(\widehat D\Delta)_{J^c}|_1\leq 3\mu_n|\widehat D\Delta_J|_1\leq  3\mu_n\bar D\sqrt{s}\|X(\tau_0)^T\Delta\|_2/\sqrt{\phi}.
$$
The same argument as that of Step 5 in the proof of Theorem \protect\ref{th2.3} yields
 $$
 \|X(\tau_0)^T\Delta\|_2^2\leq CR(\widetilde\alpha,\widehat\tau)+C|\widehat\tau-\tau_0|
 $$
 for some generic constant  $C>0$. This implies, for some generic constant $C>0$, 
 $$
 R(\widetilde\alpha,\widehat\tau)^2\leq \mu_n^2sC( R(\widetilde\alpha,\widehat\tau)+|\widehat\tau-\tau_0|).
 $$
  It follows that 
 $
 R(\widetilde\alpha,\widehat\tau)\leq C(\mu_n^2s +|\widehat\tau-\tau_0|), 
 $ and 
 $
 \|X(\tau_0)\Delta\|_2^2\leq C(\mu_n^2s +|\widehat\tau-\tau_0|).
 $
 Hence
 $$
| \Delta|_1^2\leq Cs\|X(\tau_0)\Delta\|_2^2\leq   C(\mu_n^2s^2 +|\widehat\tau-\tau_0|s).
$$
Combining both cases, we reach:
$$
|\widetilde\alpha-\alpha_0|_1^2\leq C(\mu_n^2s^2 +|\widehat\tau-\tau_0|s+\frac{1}{\mu_n^2}|\widehat\tau-\tau_0|^2s^2),
$$
which gives the desired result since the first term $\mu_n^2s^2$ dominates the other two terms.

\textbf{Rate of convergence for $R(\tilde\alpha,\widehat\tau)$}

In the proofs above, we have in fact shown that 
 \begin{equation}\label{ec.47}
R(\widetilde\alpha,\widehat\tau) \leq 2R(\alpha_0,\widehat\tau)+ 3\mu_n\bar D|\tilde\alpha-\alpha_0|_1,
\end{equation}
and when $\delta_0\neq0$,  $R(\alpha_0,\widehat\tau)\leq Cs|\widehat\tau-\tau_0|$. Note that $\widehat\tau-\tau_0=O_P(n^{-1})$. Hence $R(\tilde\alpha,\widehat\tau)=O_P(sn^{-1}+\mu_n^2s)=O_P(\mu_n^2s).$

\subsection{Proof of Theorem \protect\ref{th2.3-delta0}}

If $\delta_0 = 0$, $\tau_0$ is non-identifiable. In this case, we decompose the excess risk  in the following way:
\begin{align}  \label{eqaa2.2-delta0}
\begin{split}
R\left( \alpha, \tau\right) & =\mathbb{E} \left( 
\left[\rho\left( Y,X^T\beta\right)
-\rho\left( Y,X^T\beta _{0}\right) \right] 1\left\{ Q\leq\tau \right\} \right) \\
&\;\;\;+\mathbb{E} \left(  \left[ \rho\left( Y,X^T\theta\right) -\rho\left( Y,X^T\beta_{0}\right)
\right] 1\left\{ Q>\tau \right\} \right).
\end{split} 
\end{align}
We split the proof into three steps.

\noindent
\textbf{Step 1}: 
For any $r>0$, we have that w.p.a.1, 
$\breve\beta \in \mathcal{\tilde{B}}(\beta_0, r, \breve\tau)$
and  $\breve\theta \in \mathcal{\tilde{G}}(\beta_0, r, \breve\tau)$.

\begin{proof}[Proof of Step 1]
As in the proof of Step 1 in the proof of Theorem \protect\ref{th2.3}, Assumption \ref{ass2.2-delta0} \ref{ass2.2-delta0:itm2}  implies that 
\begin{align*}
\mathbb{E} \left[ (X^{T}(\beta -\beta _{0}))^2 1\{Q\leq \tau \} \right] 
\leq \frac{R(\alpha ,\tau)^2}{(\eta^\ast r^\ast)^2}  \vee
\frac{R(\alpha ,\tau)}{\eta^\ast}.
\end{align*}
For any $r>0$, note that $R(\breve\alpha,\breve\tau)=o_P(1)$ implies that the event $R(\breve\alpha,\breve\tau)<r^2$ holds w.p.a.1. 
Therefore, we have shown that $\breve\beta\in\mathcal{\tilde{B}}(\beta_0, r, \breve\tau)$. The other case can be proved similarly.
\end{proof}

\noindent
\textbf{Step 2 }: Suppose that $\delta_0 = 0$. Then
\begin{align}\label{BI1-delta0}
R\left( \breve{\alpha},\breve{\tau}\right) +\frac{1}{2}\kappa
_{n}\left\vert \breve{D}\left( \breve{\alpha}-\alpha _{0}\right) \right\vert
_{1}  \leq 2\kappa _{n}\left\vert \breve{D}\left( \breve{\alpha}%
-\alpha _{0}\right) _{J}\right\vert _{1}  \; \text{w.p.a.1}.
\end{align}%

\begin{proof}
The proof of this step is similar to that of Step 3 in the proof of Theorem \ref{th2.3}.
Since  $(\breve{\alpha},\breve{\tau})$ minimizes the $\ell_1$-penalized objective function in \eqref{eq2.2add}, 
we have that 
\begin{align}\label{delta0-eq1}
\frac{1}{n}\sum_{i=1}^{n}\rho (Y_{i},X_{i}(\breve{\tau})^{T}\breve{%
\alpha })+\kappa _{n}|\breve{D}\breve{\alpha }|_{1}\leq \frac{1}{n}%
\sum_{i=1}^{n}\rho (Y_{i},X_{i}(\breve{\tau})^{T}\alpha _{0})+\kappa
_{n}|\breve{D} \alpha _{0}|_{1}.
\end{align}
When $\delta_0 = 0$,  
$\rho(Y, X(\breve{\tau})^T\alpha_0) =  \rho(Y,X(\tau_0)^T\alpha_0)$.
Using this fact and \eqref{delta0-eq1}, 
we obtain the following inequality 
\begin{align}\label{basic-ineq-delta0}
R(\breve{\alpha },\breve{\tau}) &\leq 
\left[ \nu _{n}(\alpha_{0},\breve{\tau})-\nu _{n}(\breve{\alpha},\breve{\tau })\right] +\kappa
_{n}|\breve{D}\alpha _{0}|_{1}-\kappa _{n}|\breve{D}\breve{\alpha }|_{1}.
\end{align}%

As in Step 3 in the proof of Theorem \ref{th2.3}, we apply 
 Lemma \ref{lem-emp} to 
$[\nu _{n} (\alpha _{0},\breve{\tau})-\nu _{n}(\breve{\alpha },\breve{\tau})]$
 with $a_{n}\ $and $b_{n}$ replaced by $a_{n}/2\ 
$and $b_{n}/2$. Then 
we can rewrite the basic inequality in \eqref{basic-ineq-delta0} by%
\begin{equation*}
\kappa _{n}\left\vert \breve{D} \alpha _{0}\right\vert _{1}\geq R\left( \breve{\alpha},%
\breve{\tau}\right) +\kappa _{n}\left\vert \breve{D}\breve{\alpha}\right\vert
_{1}-\frac{1}{2}\kappa _{n}\left\vert \breve{D}\left( \breve{\alpha}-\alpha
_{0}\right) \right\vert _{1} \; \text{ w.p.a.1}.
\end{equation*}
Now adding $\kappa _{n}\left\vert \breve{D}\left( \breve{\alpha}-\alpha
_{0}\right) \right\vert _{1}$ on both sides of the inequality above
and using the fact that $ \left\vert  \alpha_{0j} \right\vert _{1} - \left\vert \breve{\alpha}_{j} \right\vert
_{1} + \left\vert \left( \breve{\alpha}_{j} -\alpha
_{0j}\right) \right\vert _{1} = 0$ for $j \notin J$,
we have that  w.p.a.1,
\begin{align*}
2\kappa _{n}\left\vert \breve{D}\left( \breve{\alpha}%
-\alpha _{0}\right) _{J}\right\vert _{1}  
& 
\geq R\left( \breve{\alpha},\breve{\tau}\right) +\frac{1}{2}\kappa
_{n}\left\vert \breve{D}\left( \breve{\alpha}-\alpha _{0}\right) \right\vert
_{1}.  
\end{align*}%
Therefore, we have obtained the desired result.
\end{proof}

\noindent
\textbf{Step 3 }: Suppose that $\delta_0 = 0$. Then   
\begin{align*}
R\left( \breve{\alpha},\breve{\tau}\right) =O_P(\kappa_n^2s) \ \ \text{ and } \ \
\left\vert  \breve{\alpha}-\alpha _{0} \right\vert = O_P\left( \kappa _{n}s\right).
\end{align*}

\begin{proof}
By Step 2, 
\begin{equation}\label{aphat-delta0}
4\left\vert \breve{D}\left( \breve{\alpha}-\alpha _{0}\right) _{J}\right\vert
_{1}\geq \left\vert \breve{D}\left( \breve{\alpha}-\alpha _{0}\right)
\right\vert _{1}=\left\vert \breve{D}\left( \breve{\alpha}-\alpha _{0}\right)
_{J}\right\vert _{1}+\left\vert \breve{D}\left( \breve{\alpha}-\alpha _{0}\right)
_{J^{c}}\right\vert _{1},  
\end{equation}%
which enables us to apply the compatibility condition in Assumption \ref{ass2.7}. 

Recall that $\|Z\|_2=(EZ^2)^{1/2}$ for a random variable $Z.$ 
Note that for $s=|J(\alpha_0)|_0$,
\begin{align}\label{a.11add-delta0}
\begin{split}
& R\left( \breve{\alpha},\breve{\tau}\right) +\frac{1}{2}\kappa
_{n}\left\vert \breve{D}\left( \breve{\alpha}-\alpha _{0}\right) \right\vert
_{1} \\
&\leq_{(1)} 2\kappa _{n}\left\vert \breve{D}\left( \breve{\alpha}-\alpha _{0}\right)
_{J}\right\vert _{1} \\
& \leq_{(2)} 2\kappa _{n}\bar D\left\Vert X(\breve{\tau})^T(\breve\alpha-\alpha_0)\right\Vert _{2}\sqrt{s}/\phi  \\
&\leq_{(3)} \frac{4\kappa _{n}^{2}\bar D^{2}s}{2 \tilde{c} \phi ^{2}}+\frac{\tilde{c}}{2}%
\left\Vert X(\breve{\tau})^T(\breve\alpha-\alpha_0)\right\Vert _{2}^{2},
\end{split}
\end{align}
where (1) is from the basic inequality \eqref{BI1-delta0} in Step 2,
(2) is by the compatibility condition (Assumption \ref{ass2.7}), and
(3) is from the
inequality  that $uv\leq v^{2}/(2\tilde{c})+\tilde{c} u^2/2$ for any $\tilde{c} >0$.

Note that
\begin{align}\label{ec.55}
\lefteqn{\left\Vert X(\tau)^T\alpha-X(\tau)^T\alpha_0\right\Vert _{2}^{2} } \cr
&=_{(1)} \mathbb{E}\left[ (X^T(\theta-\beta_0))^21\{Q > \tau\} \right]
+ \mathbb{E}\left[ (X^T(\beta-\beta_0))^21\{Q \leq \tau\} \right] \cr
&\leq_{(2)} (\eta^\ast)^{-1} \mathbb{E}\left[ \left( \rho \left( Y,X^{T}\theta \right) -\rho \left( Y,X^{T}\beta_{0}\right) \right) 
1\left\{ Q>\tau \right\}  \right] \cr
&+  (\eta^\ast)^{-1} \mathbb{E}\left[ \left( \rho \left( Y,X^{T}\beta \right) -\rho \left( Y,X^{T}\beta_{0}\right) \right) 
1\left\{ Q\leq \tau \right\} \right] \cr
&\leq_{(3)} (\eta^\ast)^{-1} R(\alpha, \tau),
\end{align}
where (1) is simply an identity,  (2) from  Assumption \ref{ass2.2-delta0} \ref{ass2.2-delta0:itm2} , and (3) is  
due to \eqref{eqaa2.2-delta0}.
Hence,  (\ref{a.11add-delta0})  with $\tilde{c}=\eta^\ast$ implies that 
\begin{align}\label{oracle-ineq-delta0}
 R\left( \breve{\alpha},\breve{\tau}\right) +\kappa
_{n}\left\vert \breve{D}\left( \breve{\alpha}-\alpha _{0}\right) \right\vert
_{1}\leq \frac{4\kappa _{n}^{2}\bar D^{2}s}{\eta^\ast \phi ^{2}}.
\end{align}
Therefore,  $R\left( \breve{\alpha},\breve{\tau}\right) =O_P(\kappa_n^2s)$.
Also, $\left\vert  \breve{\alpha}-\alpha _{0} \right\vert = O_P\left( \kappa _{n}s\right)$
since $D(\breve\tau)\geq \underline{D}$ w.p.a.1
by Assumption \ref{a:setting} \ref{a:setting:itm3}.
\end{proof}

\subsection{Proof of Theorem  \ref{th4.2}}

  
We first prove part (i) when the minimum signal condition holds. 

When $\tau_0$ is not identifiable ($\delta_0=0$), $\widehat\tau$ obtained in the second-step estimation can be any value in $\mathcal{T}$.  Note that Lemmas \ref{la.1} and \ref{la.2} are stated and proved for this case as well.  Similar to the proof of Theorem \ref{th3.1},   by Lemma \ref{l4.1},  
in quantile  regression models, $M_n=Cs^{1/2}$ for some $C>0$.  Hence all the required conditions in Lemmas \ref{la.1} and \ref{la.2}
are satisfied by the conditions imposed in Theorem  \ref{th4.2}. Then by Lemmas \ref{la.1} and \ref{la.2},  w.p.a.1,  for any $\alpha =(\alpha _{J},\alpha
_{J^{c}})\in \mathcal{H}$, 
\begin{equation*}
\widetilde{S}_{n}(\bar \alpha_J, 0)=\bar Q_n(\bar\alpha_J)\leq \bar Q_n(\alpha_J)=\widetilde S_n(\alpha_J, 0)\leq \widetilde S_n(\alpha).
\end{equation*}%
Hence $(\bar\alpha_J,0)$ is a local minimizer of $\widetilde S_n$, which is also a global minimizer due to the convexity. This implies that w.p.a.1,  $\widetilde\alpha=(\widetilde\alpha_J,\widetilde\alpha_{J^c})$ satisfies: $\widetilde\alpha_{J^c}=0$, and  $\widetilde\alpha_J=\bar\alpha_J$, so 
\begin{equation*}
\left\vert \widetilde{\alpha }_{J}-\alpha _{0J}\right\vert _{2}=O_P \left( \sqrt{%
\frac{s\log s}{n}} \; \right),\quad \left\vert \widetilde{\alpha }_{J}-\alpha _{0J}\right\vert _{1}=O_P \left( s\sqrt{%
\frac{\log s}{n}} \; \right).
\end{equation*}
Also,  note that $R(\alpha_0,\widehat\tau)=0$ when $\delta_0=0$. Hence by (\ref{ec.47}),   
$$
R(\widetilde\alpha,\widehat\tau) \leq 2R(\alpha_0,\widehat\tau)+ 3\mu_n\bar D|\widetilde\alpha-\alpha_0|_1=O_P(\nu_ns\sqrt{\frac{\log s}{n}}).
$$

We now prove part (ii) without the minimum signal condition.  The proof is very similar to that of Theorem \ref{th3.1-without}. Hence we   provide the proof briefly.  In fact (\ref{ec.47}) still holds by the same argument.
But now $R(\alpha_0,\widehat\tau)=0$.
 Hence for $\Delta=\widetilde\alpha-\alpha_0$, $$
R(\widetilde\alpha,\widehat\tau) +\frac{1}{2}\mu_n|(\widehat D\Delta)_{J^c}|_1\leq \frac{3}{2}\mu_n|\widehat D\Delta_J|_1\leq  2\mu_n\bar D\sqrt{s}\|X(\widehat\tau)^T\Delta\|_2/\sqrt{\phi},
$$
where the last inequality follows from Assumption \ref{ass2.7}.
By (\ref{ec.55}), 
$
 \left\Vert X(\widehat\tau)^T\Delta\right\Vert _{2}^{2} \leq C R(\widetilde\alpha, \widehat\tau),
$
for some $C>0$.
  This implies, for some generic constant $C>0$, 
 $
 R(\widetilde\alpha,\widehat\tau)^2\leq \mu_n^2sC R(\widetilde\alpha,\widehat\tau).
 $
  It follows that 
 $$
 R(\widetilde\alpha,\widehat\tau)\leq  \mu_n^2sC,
 $$ and 
 $$
| \Delta|_1^2\leq Cs\|X(\widehat\tau)\Delta\|_2^2\leq   CsR(\widetilde\alpha,\widehat\tau)\leq Cs^2\mu_n^2.
$$
  
\newpage  

\section{Additional Simulation Results: Different $\tau_0$ and distributions of $Q$}\label{sec:add-sim-results-1}


Tables \ref{tb:diff-U}--\ref{tb:diff-chi2} summarize simulation results when the change point $\tau_0$ and the distribution of $Q_i$ vary. We set $\gamma=0.5$, i.e.\ median regression, and $n=200$ for all designs. We consider three different distributions of $Q_i$: Uniform$[0,1]$, $N(0,1)$, and $\chi^2(1)$. The change point parameter $\tau_0$ varies over $0.3,0.4,\ldots,0.7$ quantiles of each $Q_i$ distribution. We can confirm the following two results from these simulation studies. First, the performance of $\widehat{\tau}$ measured by the root-mean-squared error depends on the density of $Q_i$ distribution. For instance, it is quite uniform over different $\tau_0$ when $Q_i$ follows Uniform$[0,1]$. However, when $Q_i$ follows $N(0,1)$ or $\chi^2(1)$, it performs better when $\tau_0$ is located at a point with higher density of $Q_i$ distribution. Second, the mean squared error of $\widehat{\ap}$ and the oracle proportion get better when $\tau_0$ smaller. It might be caused by the simulation design, $X_i\cdot 1(Q_i>\tau_0)$, as it will generate less zeros when $\tau_0$ is smaller and help increase the signal from $X_i$'s.

\begin{table}[htbp]
\centering
\caption{Different $\tau_0$ and $Q_i$ dist.: $Q_i \sim Unif[0,1]$}\label{tb:diff-U}
\scriptsize
\begin{tabular}{llcclcccc}
  \hline
 & & Excess Risk & E[$J(\widehat{\ap})$] & MSE of $\widehat{\ap} ~(\widehat{\ap}_{J_0} / \widehat{\ap}_{J_0^c})$ & Pred. Er. & RMSE of $\widehat{\tau}$ & C. Prob. of $\widehat{\tau}$ & Oracle Prop. \\ 
  \hline
  \multirow{6}{*}{$\tau_0=0.3$} &Oracle 1 & 0.008 &      NA & 0.016 (NA / NA)& 0.282 & NA  & NA  & NA \\ 
  & Oracle 2 & 0.017 &   NA    & 0.016 (NA / NA)  & 0.451 & 0.010 & 0.951 & NA \\ 
  & Step 1 & 0.039   & 5.557 & 0.206 ( 0.182 / 0.024) & 0.697 & 0.011 & 0.950 & 0.026 \\ 
  &Step 2 & 0.040   &   NA    & NA (NA / NA)  & 0.700 & 0.011 & 0.949 & NA \\ 
  &Step 3a & 0.038  & 5.536 & 0.201 ( 0.177 / 0.024) & 0.687 & 0.011 & 0.947 & 0.026 \\ 
  &Step 3b & 0.041  & 2.042 & 0.145 ( 0.134 / 0.011) & 0.717 & 0.012 & 0.924 & 0.475 \\ 
   \\
  \multirow{6}{*}{$\tau_0=0.4$}& Oracle 1 & 0.008 & NA  & 0.014 (NA / NA)& 0.287 & NA  & NA  & NA \\ 
  & Oracle 2 & 0.017 & NA  & 0.014 (NA / NA)& 0.458 & 0.011 & 0.956 & NA \\ 
  &Step 1 & 0.039 & 5.590 & 0.228 ( 0.201 / 0.027) & 0.706 & 0.011 & 0.955 & 0.019 \\ 
  &Step 2 & 0.037 & NA  & NA (NA / NA)& 0.707 & 0.011 & 0.955 & NA \\ 
  &Step 3a & 0.034 & 5.578 & 0.226 ( 0.199 / 0.027) & 0.695 & 0.011 & 0.949 & 0.018 \\ 
  &Step 3b & 0.040 & 2.203 & 0.147 ( 0.131 / 0.017) & 0.704 & 0.011 & 0.933 & 0.492 \\ 
  \\
  \multirow{6}{*}{$\tau_0=0.5$}&Oracle 1 & 0.008 & NA  & 0.012 (NA / NA)& 0.287 & NA  & NA  & NA \\ 
  & Oracle 2 & 0.018 & NA  & 0.012 (NA / NA)& 0.470 & 0.010 & 0.951 & NA \\ 
  &Step 1 & 0.042 & 5.698 & 0.262 ( 0.230 / 0.032) & 0.706 & 0.011 & 0.944 & 0.020 \\ 
  &Step 2 & 0.042 & NA  & NA (NA / NA)& 0.711 & 0.011 & 0.939 & NA \\ 
  &Step 3a & 0.041 & 5.680 & 0.256 ( 0.224 / 0.032) & 0.696 & 0.011 & 0.941 & 0.020 \\ 
  &Step 3b & 0.041 & 2.343 & 0.167 ( 0.142 / 0.025) & 0.714 & 0.011 & 0.931 & 0.443 \\ 
  \\
  \multirow{6}{*}{$\tau_0=0.6$}&Oracle 1 & 0.008 & NA  & 0.013 (NA / NA)& 0.295 & NA  & NA  & NA \\ 
  & Oracle 2 & 0.017 & NA  & 0.013 (NA / NA) & 0.475 & 0.011 & 0.947 & NA \\ 
  &Step 1 & 0.042 & 5.869 & 0.344 ( 0.303 / 0.041) & 0.731 & 0.013 & 0.937 & 0.012 \\ 
  &Step 2 & 0.042 & NA  & NA (NA / NA)& 0.742 & 0.013 & 0.930 & NA \\ 
  &Step 3a & 0.039 & 5.855 & 0.336 ( 0.296 / 0.040) & 0.730 & 0.013 & 0.928 & 0.012 \\ 
  &Step 3b & 0.041 & 2.467 & 0.249 ( 0.204 / 0.046) & 0.734 & 0.012 & 0.923 & 0.382 \\ 
  \\
  \multirow{6}{*}{$\tau_0=0.7$} &Oracle 1 & 0.007 & NA  & 0.012 (NA / NA) & 0.280 & NA  & NA  & NA \\ 
  & Oracle 2 & 0.018 & NA  & 0.012  (NA / NA)& 0.470 & 0.010 & 0.949 & NA \\ 
  &Step 1 & 0.041 & 5.978 & 0.464 ( 0.407 / 0.057) & 0.729 & 0.012 & 0.954 & 0.016 \\ 
  &Step 2 & 0.042 & NA  &  NA (NA / NA)& 0.737 & 0.012 & 0.951 & NA \\ 
  &Step 3a & 0.041 & 5.981 & 0.456 ( 0.400 / 0.056) & 0.718 & 0.012 & 0.953 & 0.020 \\ 
  &Step 3b & 0.040 & 2.549 & 0.386 ( 0.303 / 0.083) & 0.706 & 0.012 & 0.944 & 0.319 \\ 
  \hline
   \multicolumn{9}{p{\textwidth}}{\footnotesize \emph{Note: } For all designs, $J(\alpha_{\gamma})=2$, $\gamma=0.5$, and $n=200$. See the note below Table \ref{tb:base-0.25} for other notation. 
}
\end{tabular}
\end{table}

\begin{table}[htbp]
\centering
\caption{Different $\tau_0$ and $Q_i$ dist.: $Q_i \sim N(0,1)$}\label{tb:diff-N}
\scriptsize
\begin{tabular}{llcclcccc}
  \hline
 & & Excess Risk & E[$J(\widehat{\ap})$] & MSE of $\widehat{\ap} ~(\widehat{\ap}_{J_0} / \widehat{\ap}_{J_0^c})$ & Pred. Er. & RMSE of $\widehat{\tau}$ & C. Prob. of $\widehat{\tau}$ & Oracle Prop. \\ 
  \hline
  \multirow{6}{*}{$\tau_0=-0.52$} & Oracle 1 & 0.008 & NA & 0.017 (NA / NA) &  0.294 & NA & NA & NA \\ 
  &Oracle 2 & 0.018 & NA & 0.017 (NA / NA) &  0.500 & 0.034 & 0.949 & NA \\ 
  &Step 1 & 0.037 & 5.389 & 0.191 ( 0.168 / 0.023 ) & 0.689 & 0.036 & 0.953 & 0.024 \\ 
  &Step 2 & 0.039 & NA &  NA (NA / NA)& 0.690 & 0.036 & 0.943 & NA \\ 
  &Step 3a & 0.039 & 5.382 & 0.187 ( 0.164 / 0.023) & 0.677 & 0.036 & 0.938 & 0.023 \\ 
  &Step 3b & 0.042 & 2.248 & 0.132 ( 0.125 / 0.008) & 0.695 & 0.042 & 0.918 & 0.523 \\ 
   \\
  \multirow{6}{*}{$\tau_0=-0.25$}& Oracle 1 & 0.008 & NA & 0.014 (NA / NA) &  0.292 & NA & NA & NA \\ 
  &Oracle 2 & 0.018 & NA & 0.014 (NA / NA) &  0.482 & 0.028 & 0.954 & NA \\ 
  &Step 1 & 0.041 & 5.722 & 0.231 ( 0.204 / 0.027) & 0.708 & 0.027 & 0.950 & 0.022 \\ 
  &Step 2 & 0.034 & NA &  NA (NA / NA)& 0.719 & 0.028 & 0.945 & NA \\ 
  &Step 3a & 0.039 & 5.724 & 0.226 ( 0.199 / 0.027) & 0.717 & 0.028 & 0.943 & 0.022 \\ 
  &Step 3b & 0.042 & 2.231 & 0.145 ( 0.129 / 0.016) & 0.702 & 0.029 & 0.938 & 0.474 \\  
  \\
  \multirow{6}{*}{$\tau_0=0$}&Oracle 1 & 0.008 & NA & 0.013(NA / NA) &  0.291 & NA & NA & NA \\ 
  &Oracle 2 & 0.016 & NA & 0.013 (NA / NA) & 0.464 & 0.025 & 0.968 & NA \\ 
  &Step 1 & 0.038 & 5.709 & 0.275 ( 0.242 / 0.033) & 0.709 & 0.028 & 0.953 & 0.024 \\ 
  &Step 2 & 0.040 & NA &  NA (NA / NA)& 0.706 & 0.028 & 0.957 & NA \\ 
  &Step 3a & 0.038 & 5.682 & 0.271 ( 0.238 / 0.033 )& 0.691 & 0.028 & 0.956 & 0.023 \\ 
  &Step 3b & 0.042 & 2.309 & 0.184 ( 0.156 / 0.029) & 0.711 & 0.027 & 0.948 & 0.458 \\  
  \\
  \multirow{6}{*}{$\tau_0=0.25$}&Oracle 1 & 0.008 & NA & 0.012 (NA / NA) &  0.292 & NA & NA & NA \\ 
  &Oracle 2 & 0.017 & NA & 0.012 (NA / NA) &  0.474 & 0.028 & 0.958 & NA \\ 
  &Step 1 & 0.041 & 5.829 & 0.359 ( 0.316 / 0.043) & 0.718 & 0.029 & 0.959 & 0.016 \\ 
  &Step 2 & 0.043 & NA &  NA (NA / NA)& 0.732 & 0.030 & 0.949 & NA \\ 
  &Step 3a & 0.039 & 5.841 & 0.351 ( 0.308 / 0.042) & 0.730 & 0.030 & 0.950 & 0.016 \\ 
  &Step 3b & 0.038 & 2.456 & 0.269 ( 0.219 / 0.050) & 0.711 & 0.030 & 0.941 & 0.378 \\  
  \\
  \multirow{6}{*}{$\tau_0=0.52$} &Oracle 1 & 0.008 & NA & 0.012 (NA / NA) &  0.286 & NA & NA & NA \\ 
  &Oracle 2 & 0.017 & NA & 0.012(NA / NA) &  0.466 & 0.031 & 0.964 & NA \\ 
  &Step 1 & 0.043 & 5.929 & 0.455 ( 0.400 / 0.055) & 0.759 & 0.034 & 0.953 & 0.012 \\ 
  &Step 2 & 0.041 & NA &  NA (NA / NA)& 0.748 & 0.034 & 0.947 & NA \\ 
  &Step 3a & 0.037 & 5.932 & 0.445 ( 0.390 / 0.055) & 0.736 & 0.033 & 0.945 & 0.010 \\ 
  &Step 3b & 0.042 & 2.529 & 0.395 ( 0.310 / 0.084) & 0.750 & 0.033 & 0.940 & 0.300 \\ 
  \hline
   \multicolumn{9}{p{\textwidth}}{\footnotesize \emph{Note: } For all designs, $J(\alpha_{\gamma})=2$, $\gamma=0.5$, and $n=200$. Note that $Quant_{0.3}(Q_i)\approx-0.52$, $Quant_{0.4}(Q_i)\approx -0.25$, $ Quant_{0.5}(Q_i) =0$.   See the note below Table \ref{tb:base-0.25} for other notation. 
}
\end{tabular}
\end{table}

\begin{table}[hbtp]
\centering
\caption{Different $\tau_0$ and $Q_i$ dist.: $Q_i \sim \chi^2(1)$}\label{tb:diff-chi2}
\scriptsize
\begin{tabular}{llcclcccc}
  \hline
 & & Excess Risk & E[$J(\widehat{\ap})$] & MSE of $\widehat{\ap} ~(\widehat{\ap}_{J_0} / \widehat{\ap}_{J_0^c})$ & Pred. Er. & RMSE of $\widehat{\tau}$ & C. Prob. of $\widehat{\tau}$ & Oracle Prop. \\ 
  \hline
  \multirow{6}{*}{$\tau_0=0.15$} & Oracle 1 & 0.008 &  NA & 0.017 (NA / NA) & 0.293 &  NA &  NA & NA \\ 
  &Oracle 2 & 0.017 &  NA & 0.017(NA / NA) & 0.461 & 0.012 & 0.978 & NA \\ 
  &Step 1 & 0.038 & 5.523 & 0.211 ( 0.187 / 0.025) & 0.701 & 0.012 & 0.979 & 0.032 \\ 
  &Step 2 & 0.034 &  NA &  NA (NA / NA)& 0.721 & 0.011 & 0.980 & NA \\ 
  &Step 3a & 0.036 & 5.524 & 0.207 ( 0.182 / 0.025) & 0.697 & 0.011 & 0.980 & 0.029 \\ 
  &Step 3b & 0.037 & 2.023 & 0.137 ( 0.126 / 0.010) & 0.692 & 0.012 & 0.966 & 0.523 \\ 
   \\
  \multirow{6}{*}{$\tau_0=0.27$}& Oracle 1 & 0.008 &  NA & 0.014 (NA / NA) & 0.286 &  NA &  NA & NA \\ 
  &Oracle 2 & 0.017 &  NA & 0.014 (NA / NA) & 0.448 & 0.015 & 0.957 & NA \\ 
  &Step 1 & 0.036 & 5.562 & 0.229 ( 0.202 / 0.027) & 0.720 & 0.016 & 0.951 & 0.026 \\ 
  &Step 2 & 0.036 &  NA &  NA (NA / NA)& 0.712 & 0.016 & 0.950 & NA \\ 
  &Step 3a & 0.038 & 5.558 & 0.225 ( 0.199 / 0.027) & 0.694 & 0.016 & 0.947 & 0.028 \\ 
  &Step 3b & 0.040 & 2.206 & 0.138 ( 0.124 / 0.014) & 0.693 & 0.015 & 0.945 & 0.507 \\   
  \\
  \multirow{6}{*}{$\tau_0=0.45$}&Oracle 1 & 0.008 &  NA & 0.011 (NA / NA) & 0.291 &  NA &  NA & NA \\ 
  &Oracle 2 & 0.016 &  NA & 0.011 (NA / NA) & 0.461 & 0.022 & 0.942 & NA \\ 
  &Step 1 & 0.036 & 5.810 & 0.291 ( 0.256 / 0.035) & 0.718 & 0.022 & 0.934 & 0.017 \\ 
  &Step 2 & 0.038 &  NA &  NA (NA / NA)& 0.722 & 0.021 & 0.930 & NA \\ 
  &Step 3a & 0.041 & 5.834 & 0.288 ( 0.253 / 0.035) & 0.706 & 0.021 & 0.930 & 0.019 \\ 
  &Step 3b & 0.041 & 2.353 & 0.207 ( 0.171 / 0.036) & 0.712 & 0.021 & 0.919 & 0.439 \\ 
  \\
  \multirow{6}{*}{$\tau_0=0.71$}&Oracle 1 & 0.009 &  NA & 0.012 (NA / NA) & 0.288 &  NA &  NA & NA \\ 
  &Oracle 2 & 0.018 &  NA & 0.012 (NA / NA) & 0.485 & 0.030 & 0.933 & NA \\ 
  &Step 1 & 0.035 & 5.883 & 0.348 ( 0.307 / 0.042) & 0.717 & 0.031 & 0.934 & 0.015 \\ 
  &Step 2 & 0.042 &  NA &  NA (NA / NA)& 0.741 & 0.032 & 0.923 & NA \\ 
  &Step 3a & 0.038 & 5.866 & 0.337 ( 0.296 / 0.041) & 0.726 & 0.032 & 0.922 & 0.014 \\ 
  &Step 3b & 0.038 & 2.386 & 0.240 ( 0.197 / 0.044) & 0.724 & 0.032 & 0.909 & 0.397 \\   
  \\
  \multirow{6}{*}{$\tau_0=1.07$} &Oracle 1 & 0.008 &  NA & 0.013(NA / NA) & 0.291 &  NA &  NA & NA \\ 
  &Oracle 2 & 0.017 &  NA & 0.013 (NA / NA) &0.473 & 0.044 & 0.936 & NA \\ 
  &Step 1 & 0.043 & 5.967 & 0.459 ( 0.404 / 0.055) & 0.740 & 0.049 & 0.923 & 0.008 \\ 
  &Step 2 & 0.041 &  NA &  NA (NA / NA)& 0.752 & 0.050 & 0.922 & NA \\ 
  &Step 3a & 0.036 & 5.932 & 0.445 ( 0.390 / 0.054) & 0.738 & 0.050 & 0.920 & 0.010 \\ 
  &Step 3b & 0.044 & 2.486 & 0.381 ( 0.303 / 0.078) & 0.740 & 0.048 & 0.918 & 0.317 \\ 
  \hline
   \multicolumn{9}{p{\textwidth}}{\footnotesize \emph{Note: } For all designs, $J(\alpha_{\gamma})=2$, $\gamma=0.5$, and $n=200$. Note that $\tau_0$ values are $0.3, 0.4, \ldots, 0.7$ quantiles of $\chi^2(1)$.   See the note below Table \ref{tb:base-0.25} for other notation. 
}
\end{tabular}
\end{table}

\newpage

\section{Additional Simulation Results: Sensitivity Analyses}\label{sec:add-sim-results}

Tables \ref{tb:sen-gm-s1}--\ref{tb:sen-a-s3b} summarize the simulation results of sensitivity analysis on tuning parameters. We set $\gamma=0.5$, i.e.\ median regression, and $n=200$ for all designs. We make variation on four constants of tuning parameters:  $\gm^*$ of $\overline{\Lambda}_{1-\gm^*}$, $c_1$ of $\kappa_n$ and $\omega_n$, $c_2$ of $\mu_n$, and $a$ of the signal adaptive weight $w_j$. Recall that they are set to $\gm^*=0.1$, $c_1=1.1$, $c_2=\log\log n$, and $a=3.7$ following the existing literature and some preliminary simulations. We make changes over the range between $-15\%$ and $+15\%$ of the suggested values. Since $\gm^*$ and $c_1$ are relevant for all estimation steps, we report the sensitivity analysis results for all steps: Tables \ref{tb:sen-gm-s1}--\ref{tb:sen-gm-s3b} and Tables \ref{tb:sen-c1-s1}--\ref{tb:sen-c1-s3b}. However, we only report the results of Step 3b for $c_2$ and $a$ as they affect only the last step: Table \ref{tb:sen-c2-s3b} and Table \ref{tb:sen-a-s3b}. These simulation studies confirm that the proposed estimators are robust to some variation in tuning parameters. Both $\gm^*$ and $c_1$ show some tendency that a smaller penalty size (larger $\gm^*$ and smaller $c_1$) improves the prediction error slightly. Table \ref{tb:sen-c2-s3b} shows quite stable oracle proportions unless $c_2$ is too small. The constant $a$ for the signal adaptive weight shows quite uniform performance over different values. 
Figures \ref{fig:sen-gm-s1}--\ref{fig:sen-a-s3b} present graphical representation of the sensitivity analyses
reported in Tables \ref{tb:sen-gm-s1}--\ref{tb:sen-a-s3b}.


\begin{table}[htbp]
\centering
\caption{Sensitivity Analysis of $\gamma^*$: Step 1}\label{tb:sen-gm-s1}
\scriptsize
\begin{tabular}{rccccccccc}
  \hline
 Changes & Excess Risk & E[$J_0(\widehat{\ap})$] & MSE of $\widehat{\ap}$ & MSE of $\widehat{\ap}_{J_0}$ & MSE of $\widehat{\ap}_{J_0^c}$ & Pred. Er. & RMSE of $\widehat{\tau}$ & C. Prob. of $\widehat{\tau}$ & Oracle Prop \\ 
  \hline

-15\% & 0.039 & 5.776 & 0.279 & 0.246 & 0.033 & 0.730 & 0.011 & 0.945 & 0.015 \\ 
  -12\% & 0.039 & 5.717 & 0.280 & 0.247 & 0.033 & 0.727 & 0.011 & 0.937 & 0.019 \\ 
  -9\% & 0.040 & 5.720 & 0.278 & 0.245 & 0.033 & 0.739 & 0.012 & 0.931 & 0.013 \\ 
  -6\% & 0.039 & 5.635 & 0.277 & 0.244 & 0.033 & 0.732 & 0.012 & 0.945 & 0.016 \\ 
  -3\% & 0.040 & 5.767 & 0.282 & 0.248 & 0.034 & 0.733 & 0.012 & 0.940 & 0.017 \\ 
  0\% & 0.040 & 5.790 & 0.279 & 0.245 & 0.034 & 0.723 & 0.011 & 0.950 & 0.015 \\ 
  +3\% & 0.043 & 5.782 & 0.275 & 0.242 & 0.033 & 0.741 & 0.011 & 0.944 & 0.017 \\ 
  +6\% & 0.038 & 5.711 & 0.272 & 0.239 & 0.033 & 0.715 & 0.011 & 0.944 & 0.018 \\ 
  +9\% & 0.041 & 5.745 & 0.278 & 0.245 & 0.034 & 0.697 & 0.010 & 0.961 & 0.018 \\ 
  +12\% & 0.040 & 5.730 & 0.272 & 0.240 & 0.032 & 0.735 & 0.010 & 0.957 & 0.012 \\ 
  +15\% & 0.042 & 5.809 & 0.271 & 0.240 & 0.032 & 0.713 & 0.011 & 0.949 & 0.010 \\  
   \hline
\end{tabular}
\end{table}

\newpage

\begin{table}[htbp]
\centering
\caption{Sensitivity Analysis of $\gamma^*$: Step 2}\label{tb:sen-gm-s2}
\scriptsize
\begin{tabular}{rcccc}
  \hline
 Changes & Excess Risk & Pred. Er. & RMSE of $\widehat{\tau}$ & C. Prob. of $\widehat{\tau}$ \\ 
  \hline

  -15\% & 0.035  & 0.726 & 0.011 & 0.937  \\ 
  -12\% & 0.035  & 0.721 & 0.011 & 0.937  \\ 
  -9\% & 0.036  & 0.731 & 0.012 & 0.933  \\ 
  -6\% & 0.036  & 0.727 & 0.012 & 0.934  \\ 
  -3\% & 0.042  & 0.719 & 0.012 & 0.931  \\ 
  0\% & 0.036  & 0.729 & 0.011 & 0.946 \\ 
  +3\% & 0.044  & 0.728 & 0.012 & 0.936  \\ 
  +6\% & 0.038  & 0.719 & 0.011 & 0.936  \\ 
  +9\% & 0.040  & 0.703 & 0.010 & 0.956  \\ 
  +12\% & 0.039  & 0.701 & 0.010 & 0.948  \\ 
  +15\% & 0.040  & 0.720 & 0.011 & 0.946  \\ 
    \hline
\end{tabular}
\end{table}


\begin{table}[htbp]
\centering
\caption{Sensitivity Analysis of $\gamma^*$: Step 3a}\label{tb:sen-gm-s3a}
\scriptsize
\begin{tabular}{rccccccccc}
  \hline
 Changes & Excess Risk & E[$J_0(\widehat{\ap})$] & MSE of $\widehat{\ap}$ & MSE of $\widehat{\ap}_{J_0}$ & MSE of $\widehat{\ap}_{J_0^c}$ & Pred. Er. & RMSE of $\widehat{\tau}$ & C. Prob. of $\widehat{\tau}$ & Oracle Prop \\ 
  \hline

  -15\% & 0.040 & 5.793 & 0.277 & 0.244 & 0.033 & 0.722 & 0.011 & 0.941 & 0.014 \\ 
  -12\% & 0.040 & 5.750 & 0.275 & 0.243 & 0.032 & 0.716 & 0.011 & 0.938 & 0.022 \\ 
  -9\% & 0.041 & 5.732 & 0.273 & 0.240 & 0.033 & 0.726 & 0.012 & 0.934 & 0.015 \\ 
  -6\% & 0.040 & 5.699 & 0.272 & 0.239 & 0.033 & 0.724 & 0.012 & 0.935 & 0.017 \\ 
  -3\% & 0.040 & 5.795 & 0.278 & 0.245 & 0.034 & 0.729 & 0.012 & 0.933 & 0.018 \\ 
  0\% & 0.039 & 5.910 & 0.272 & 0.239 & 0.033 & 0.717 & 0.011 & 0.947 & 0.017 \\ 
  +3\% & 0.036 & 5.782 & 0.269 & 0.237 & 0.033 & 0.714 & 0.012 & 0.938 & 0.018 \\ 
  +6\% & 0.039 & 5.726 & 0.267 & 0.235 & 0.033 & 0.699 & 0.011 & 0.938 & 0.018 \\ 
  +9\% & 0.040 & 5.747 & 0.272 & 0.239 & 0.033 & 0.688 & 0.010 & 0.959 & 0.017 \\ 
  +12\% & 0.038 & 5.740 & 0.267 & 0.236 & 0.032 & 0.707 & 0.010 & 0.949 & 0.012 \\ 
  +15\% & 0.034 & 5.836 & 0.267 & 0.235 & 0.032 & 0.703 & 0.011 & 0.944 & 0.009 \\    \hline
\end{tabular}
\end{table}


\begin{table}[htbp]
\centering
\caption{Sensitivity Analysis of $\gamma^*$: Step 3b}\label{tb:sen-gm-s3b}
\scriptsize
\begin{tabular}{rccccccccc}
  \hline
 Changes & Excess Risk & E[$J_0(\widehat{\ap})$] & MSE of $\widehat{\ap}$ & MSE of $\widehat{\ap}_{J_0}$ & MSE of $\widehat{\ap}_{J_0^c}$ & Pred. Er. & RMSE of $\widehat{\tau}$ & C. Prob. of $\widehat{\tau}$ & Oracle Prop \\ 
  \hline

-15\% & 0.041 & 2.331 & 0.185 & 0.157 & 0.028 & 0.726 & 0.011 & 0.932 & 0.429 \\ 
  -12\% & 0.040 & 2.309 & 0.179 & 0.153 & 0.025 & 0.720 & 0.011 & 0.928 & 0.447 \\ 
  -9\% & 0.042 & 2.338 & 0.195 & 0.165 & 0.030 & 0.734 & 0.012 & 0.922 & 0.417 \\ 
  -6\% & 0.042 & 2.296 & 0.189 & 0.161 & 0.028 & 0.732 & 0.012 & 0.924 & 0.449 \\ 
  -3\% & 0.043 & 2.311 & 0.186 & 0.158 & 0.028 & 0.721 & 0.012 & 0.924 & 0.435 \\ 
  0\% & 0.040 & 2.330 & 0.182 & 0.155 & 0.027 & 0.702 & 0.011 & 0.929 & 0.428 \\ 
  +3\% & 0.041 & 2.333 & 0.176 & 0.150 & 0.027 & 0.725 & 0.012 & 0.922 & 0.434 \\ 
  +6\% & 0.037 & 2.299 & 0.173 & 0.146 & 0.026 & 0.689 & 0.011 & 0.932 & 0.467 \\ 
  +9\% & 0.040 & 2.325 & 0.177 & 0.150 & 0.027 & 0.696 & 0.010 & 0.944 & 0.450 \\ 
  +12\% & 0.036 & 2.300 & 0.170 & 0.144 & 0.025 & 0.682 & 0.011 & 0.931 & 0.465 \\ 
  +15\% & 0.038 & 2.326 & 0.180 & 0.150 & 0.029 & 0.686 & 0.011 & 0.931 & 0.455 \\ 
  \hline
\end{tabular}
\end{table}


\newpage

\begin{table}[htbp]
\centering
\caption{Sensitivity Analysis of $c_1$: Step 1}\label{tb:sen-c1-s1}
\scriptsize
\begin{tabular}{rccccccccc}
  \hline
 Changes & Excess Risk & E[$J_0(\widehat{\ap})$] & MSE of $\widehat{\ap}$ & MSE of $\widehat{\ap}_{J_0}$ & MSE of $\widehat{\ap}_{J_0^c}$ & Pred. Er. & RMSE of $\widehat{\tau}$ & C. Prob. of $\widehat{\tau}$ & Oracle Prop \\ 
  \hline
-15\% & 0.039 & 5.811 & 0.266 & 0.233 & 0.033 & 0.695 & 0.011 & 0.948 & 0.027 \\ 
  -12\% & 0.034 & 5.656 & 0.273 & 0.239 & 0.034 & 0.707 & 0.012 & 0.946 & 0.010 \\ 
  -9\% & 0.039 & 5.870 & 0.273 & 0.241 & 0.032 & 0.699 & 0.010 & 0.964 & 0.014 \\ 
  -6\% & 0.036 & 5.787 & 0.274 & 0.241 & 0.033 & 0.707 & 0.009 & 0.965 & 0.016 \\ 
  -3\% & 0.043 & 5.807 & 0.274 & 0.242 & 0.032 & 0.716 & 0.011 & 0.944 & 0.008 \\ 
  0\% & 0.040 & 5.790 & 0.279 & 0.245 & 0.034 & 0.723 & 0.011 & 0.950 & 0.015 \\ 
  +3\% & 0.039 & 5.736 & 0.281 & 0.248 & 0.033 & 0.730 & 0.011 & 0.951 & 0.016 \\ 
  +6\% & 0.040 & 5.727 & 0.287 & 0.252 & 0.035 & 0.734 & 0.011 & 0.945 & 0.011 \\ 
  +9\% & 0.042 & 5.846 & 0.284 & 0.251 & 0.033 & 0.745 & 0.011 & 0.939 & 0.015 \\ 
  +12\% & 0.047 & 5.952 & 0.309 & 0.274 & 0.035 & 0.753 & 0.012 & 0.947 & 0.012 \\ 
  +15\% & 0.041 & 5.828 & 0.291 & 0.257 & 0.033 & 0.734 & 0.010 & 0.961 & 0.013 \\ 
   \hline
\end{tabular}
\end{table}


\begin{table}[htbp]
\centering
\caption{Sensitivity Analysis of $c_1$: Step 2}\label{tb:sen-c1-s2}
\scriptsize
\begin{tabular}{rcccc}
  \hline
 Changes & Excess Risk & Pred. Er. & RMSE of $\widehat{\tau}$ & C. Prob. of $\widehat{\tau}$  \\ 
  \hline
  -15\% & 0.038 &  0.697 & 0.011 & 0.943 \\ 
  -12\% & 0.037 & 0.713 & 0.012 & 0.943 \\ 
  -9\% & 0.036 &  0.698 & 0.010 & 0.953 \\ 
  -6\% & 0.034 &  0.711 & 0.009 & 0.954 \\ 
  -3\% & 0.040 &  0.723 & 0.011 & 0.941 \\ 
   0\% & 0.036 &  0.729 & 0.011 & 0.946 \\ 
  +3\% & 0.035 &  0.724 & 0.011 & 0.943 \\ 
  +6\% & 0.040 &  0.742 & 0.011 & 0.944 \\ 
  +9\% & 0.038 &  0.753 & 0.011 & 0.931 \\ 
  +12\% & 0.044 & 0.753 & 0.011 & 0.940 \\ 
  +15\% & 0.046 & 0.745 & 0.009 & 0.960 \\ 
   \hline
\end{tabular}
\end{table}


\begin{table}[htbp]
\centering
\caption{Sensitivity Analysis of $c_1$: Step 3a}\label{tb:sen-c1-s3a}
\scriptsize
\begin{tabular}{rccccccccc}
  \hline
 Changes & Excess Risk & E[$J_0(\widehat{\ap})$] & MSE of $\widehat{\ap}$ & MSE of $\widehat{\ap}_{J_0}$ & MSE of $\widehat{\ap}_{J_0^c}$ & Pred. Er. & RMSE of $\widehat{\tau}$ & C. Prob. of $\widehat{\tau}$ & Oracle Prop \\ 
  \hline
-15\% & 0.035 & 5.783 & 0.262 & 0.230 & 0.033 & 0.676 & 0.011 & 0.947 & 0.026 \\ 
  -12\% & 0.037 & 5.691 & 0.270 & 0.237 & 0.034 & 0.696 & 0.012 & 0.946 & 0.011 \\ 
  -9\% & 0.040 & 5.854 & 0.270 & 0.238 & 0.032 & 0.689 & 0.010 & 0.956 & 0.017 \\ 
  -6\% & 0.039 & 5.810 & 0.271 & 0.238 & 0.033 & 0.705 & 0.009 & 0.960 & 0.015 \\ 
  -3\% & 0.034 & 5.832 & 0.268 & 0.237 & 0.031 & 0.706 & 0.011 & 0.943 & 0.010 \\ 
  0\% & 0.039 & 5.910 & 0.272 & 0.239 & 0.033 & 0.717 & 0.011 & 0.947 & 0.017 \\ 
  +3\% & 0.040 & 5.747 & 0.278 & 0.244 & 0.033 & 0.720 & 0.011 & 0.944 & 0.016 \\ 
  +6\% & 0.041 & 5.728 & 0.281 & 0.246 & 0.035 & 0.728 & 0.011 & 0.945 & 0.014 \\ 
  +9\% & 0.042 & 5.789 & 0.280 & 0.247 & 0.033 & 0.715 & 0.011 & 0.935 & 0.015 \\ 
  +12\% & 0.038 & 5.952 & 0.304 & 0.269 & 0.035 & 0.729 & 0.011 & 0.945 & 0.012 \\ 
  +15\% & 0.038 & 5.818 & 0.285 & 0.252 & 0.033 & 0.735 & 0.010 & 0.962 & 0.016 \\ 
   \hline
\end{tabular}
\end{table}


\begin{table}[htbp]
\centering
\caption{Sensitivity Analysis of $c_1$: Step 3b}\label{tb:sen-c1-s3b}
\scriptsize
\begin{tabular}{rccccccccc}
  \hline
 Changes & Excess Risk & E[$J_0(\widehat{\ap})$] & MSE of $\widehat{\ap}$ & MSE of $\widehat{\ap}_{J_0}$ & MSE of $\widehat{\ap}_{J_0^c}$ & Pred. Er. & RMSE of $\widehat{\tau}$ & C. Prob. of $\widehat{\tau}$ & Oracle Prop \\ 
  \hline
-15\% & 0.035 & 2.341 & 0.158 & 0.129 & 0.028 & 0.649 & 0.011 & 0.937 & 0.468 \\ 
  -12\% & 0.038 & 2.328 & 0.190 & 0.156 & 0.034 & 0.678 & 0.012 & 0.936 & 0.415 \\ 
  -9\% & 0.038 & 2.346 & 0.185 & 0.156 & 0.029 & 0.676 & 0.010 & 0.946 & 0.442 \\ 
  -6\% & 0.037 & 2.332 & 0.164 & 0.138 & 0.026 & 0.674 & 0.010 & 0.943 & 0.453 \\ 
  -3\% & 0.039 & 2.343 & 0.187 & 0.157 & 0.030 & 0.694 & 0.011 & 0.931 & 0.447 \\ 
  0\% & 0.040 & 2.330 & 0.182 & 0.155 & 0.027 & 0.702 & 0.011 & 0.929 & 0.428 \\ 
  +3\% & 0.041 & 2.331 & 0.184 & 0.157 & 0.027 & 0.723 & 0.011 & 0.935 & 0.436 \\ 
  +6\% & 0.047 & 2.309 & 0.193 & 0.164 & 0.028 & 0.745 & 0.012 & 0.932 & 0.436 \\ 
  +9\% & 0.045 & 2.341 & 0.170 & 0.148 & 0.022 & 0.732 & 0.011 & 0.931 & 0.434 \\ 
  +12\% & 0.049 & 2.385 & 0.225 & 0.189 & 0.035 & 0.757 & 0.012 & 0.925 & 0.427 \\ 
  +15\% & 0.045 & 2.354 & 0.204 & 0.175 & 0.029 & 0.755 & 0.010 & 0.944 & 0.424 \\ 
   \hline
\end{tabular}
\end{table}



\begin{table}[htbp]
\centering
\caption{Sensitivity Analysis of $c_2$: Step 3b}\label{tb:sen-c2-s3b}
\scriptsize
\begin{tabular}{rccccccccc}
  \hline
 Changes & Excess Risk & E[$J_0(\widehat{\ap})$] & MSE of $\widehat{\ap}$ & MSE of $\widehat{\ap}_{J_0}$ & MSE of $\widehat{\ap}_{J_0^c}$ & Pred. Er. & RMSE of $\widehat{\tau}$ & C. Prob. of $\widehat{\tau}$ & Oracle Prop \\ 
  \hline
  -15\% & 0.040 & 2.455 & 0.173 & 0.146 & 0.027 & 0.699 & 0.012 & 0.928 & 0.409 \\ 
  -12\% & 0.040 & 2.416 & 0.175 & 0.148 & 0.028 & 0.698 & 0.011 & 0.929 & 0.414 \\ 
  -9\% & 0.041 & 2.390 & 0.175 & 0.148 & 0.027 & 0.702 & 0.011 & 0.922 & 0.431 \\ 
  -6\% & 0.039 & 2.373 & 0.167 & 0.141 & 0.025 & 0.695 & 0.011 & 0.931 & 0.432 \\ 
  -3\% & 0.041 & 2.349 & 0.166 & 0.142 & 0.025 & 0.708 & 0.011 & 0.932 & 0.443 \\ 
  0\% & 0.040 & 2.330 & 0.182 & 0.155 & 0.027 & 0.702 & 0.011 & 0.929 & 0.428 \\ 
  +3\% & 0.042 & 2.324 & 0.192 & 0.162 & 0.031 & 0.716 & 0.010 & 0.932 & 0.464 \\ 
  +6\% & 0.045 & 2.308 & 0.186 & 0.158 & 0.028 & 0.738 & 0.012 & 0.917 & 0.460 \\ 
  +9\% & 0.045 & 2.273 & 0.194 & 0.164 & 0.030 & 0.739 & 0.011 & 0.932 & 0.457 \\ 
  +12\% & 0.047 & 2.269 & 0.192 & 0.163 & 0.029 & 0.757 & 0.011 & 0.926 & 0.438 \\ 
  +15\% & 0.047 & 2.266 & 0.189 & 0.161 & 0.028 & 0.762 & 0.012 & 0.927 & 0.446 \\  
   \hline
\end{tabular}
\end{table}



\begin{table}[htbp]
\centering
\caption{Sensitivity Analysis of $a$: Step 3b}\label{tb:sen-a-s3b}
\scriptsize
\begin{tabular}{rccccccccc}
  \hline
 Changes & Excess Risk & E[$J_0(\widehat{\ap})$] & MSE of $\widehat{\ap}$ & MSE of $\widehat{\ap}_{J_0}$ & MSE of $\widehat{\ap}_{J_0^c}$ & Pred. Er. & RMSE of $\widehat{\tau}$ & C. Prob. of $\widehat{\tau}$ & Oracle Prop \\ 
  \hline
  -15\% & 0.040 & 2.599 & 0.183 & 0.152 & 0.031 & 0.696 & 0.011 & 0.932 & 0.435 \\ 
  -12\% & 0.040 & 2.343 & 0.176 & 0.148 & 0.028 & 0.699 & 0.011 & 0.935 & 0.436 \\ 
  -9\% & 0.041 & 2.356 & 0.176 & 0.149 & 0.028 & 0.706 & 0.011 & 0.929 & 0.433 \\ 
  -6\% & 0.041 & 2.371 & 0.175 & 0.148 & 0.027 & 0.713 & 0.011 & 0.924 & 0.421 \\ 
  -3\% & 0.041 & 2.342 & 0.166 & 0.142 & 0.025 & 0.710 & 0.011 & 0.928 & 0.450 \\ 
  0\% & 0.040 & 2.330 & 0.182 & 0.155 & 0.027 & 0.702 & 0.011 & 0.929 & 0.428 \\ 
  +3\% & 0.042 & 2.349 & 0.174 & 0.147 & 0.027 & 0.718 & 0.011 & 0.927 & 0.443 \\ 
  +6\% & 0.042 & 2.350 & 0.173 & 0.148 & 0.026 & 0.720 & 0.011 & 0.929 & 0.445 \\ 
  +9\% & 0.046 & 2.340 & 0.182 & 0.155 & 0.027 & 0.735 & 0.011 & 0.923 & 0.449 \\ 
  +12\% & 0.046 & 2.327 & 0.184 & 0.156 & 0.028 & 0.733 & 0.011 & 0.922 & 0.457 \\ 
  +15\% & 0.044 & 2.339 & 0.185 & 0.157 & 0.028 & 0.724 & 0.011 & 0.935 & 0.453 \\ 
   \hline
\end{tabular}
\end{table}

\begin{figure}[tbph]
\centering
\caption{Sensitivity Analysis of $\gm^*$: Step 1}\label{fig:sen-gm-s1}
\begin{minipage}{.33\textwidth}
\centering
\includegraphics[width=\linewidth]{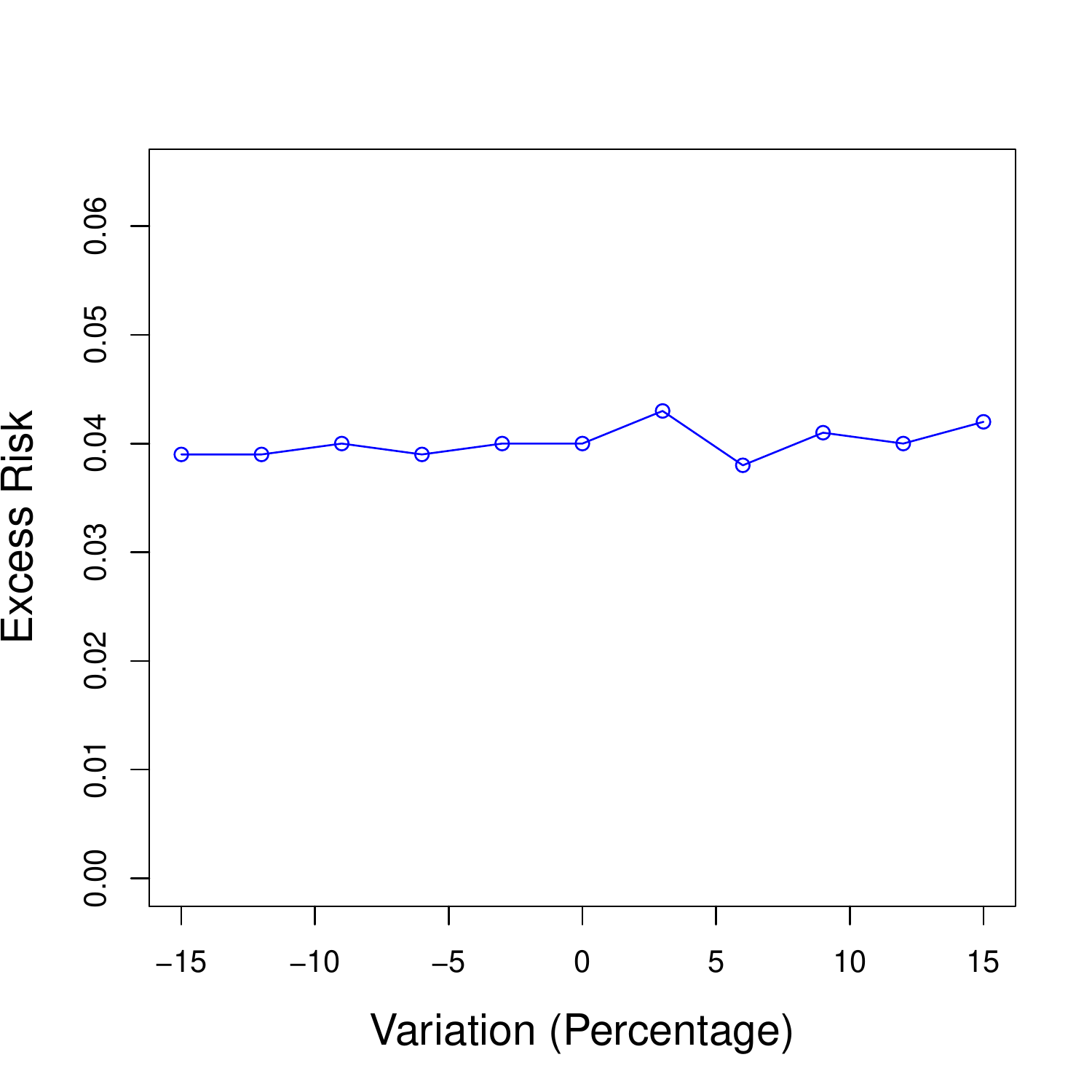}
\end{minipage}\hfill
\begin{minipage}{.33\textwidth}
\centering
\includegraphics[width=\linewidth]{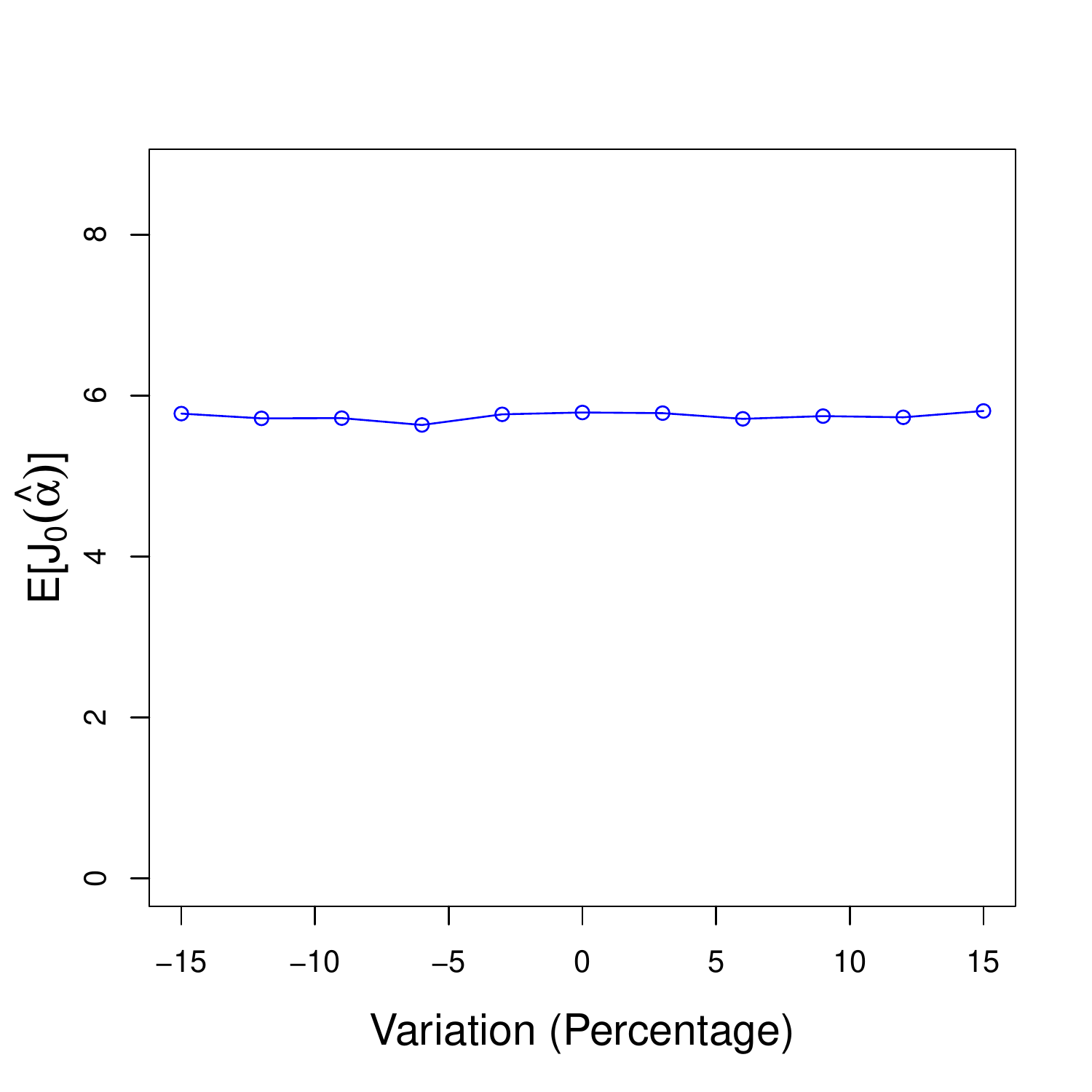}
\end{minipage}\hfill
\begin{minipage}{.33\textwidth}
\centering
\includegraphics[width=\linewidth]{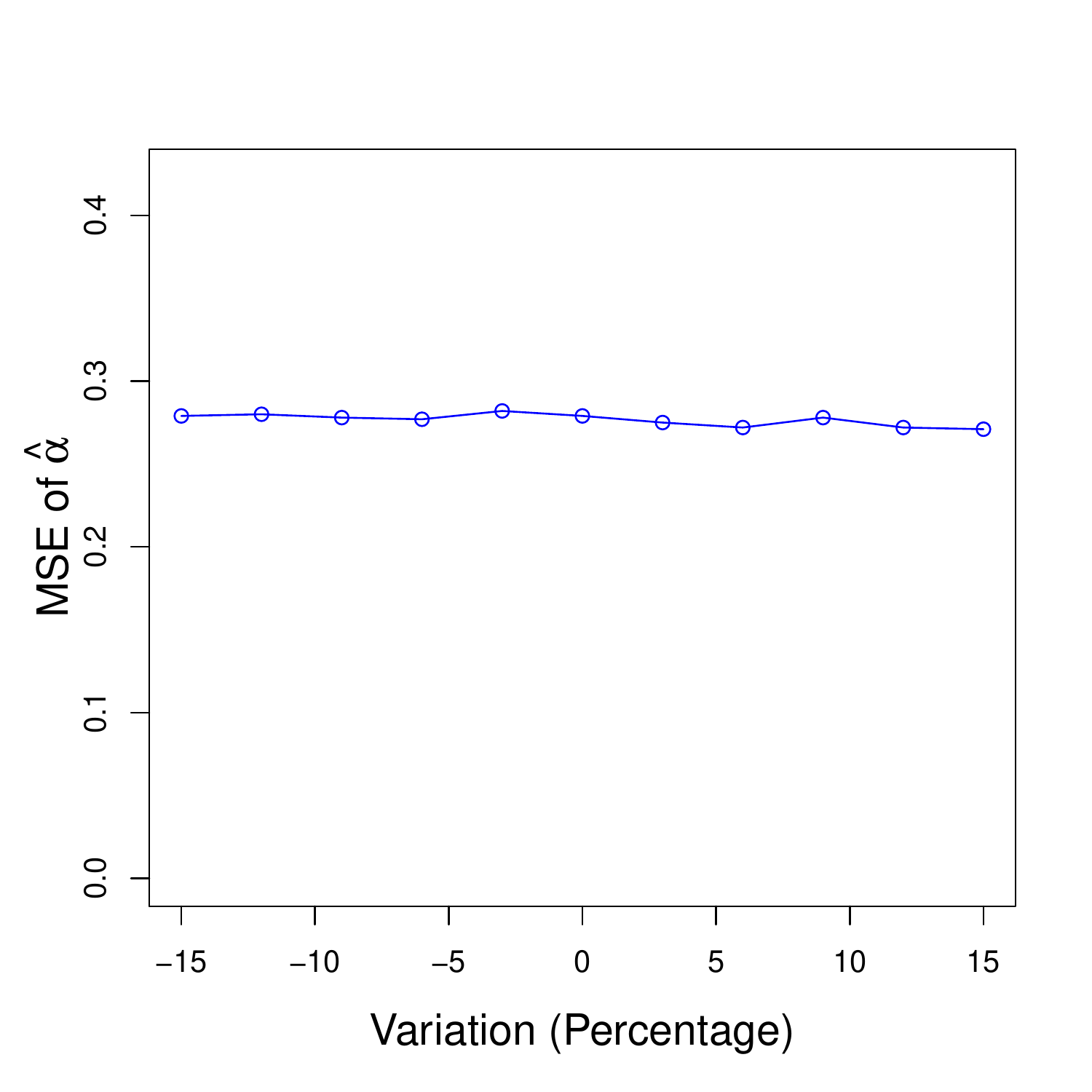}
\end{minipage}
\begin{minipage}{.33\textwidth}
\centering
\includegraphics[width=\linewidth]{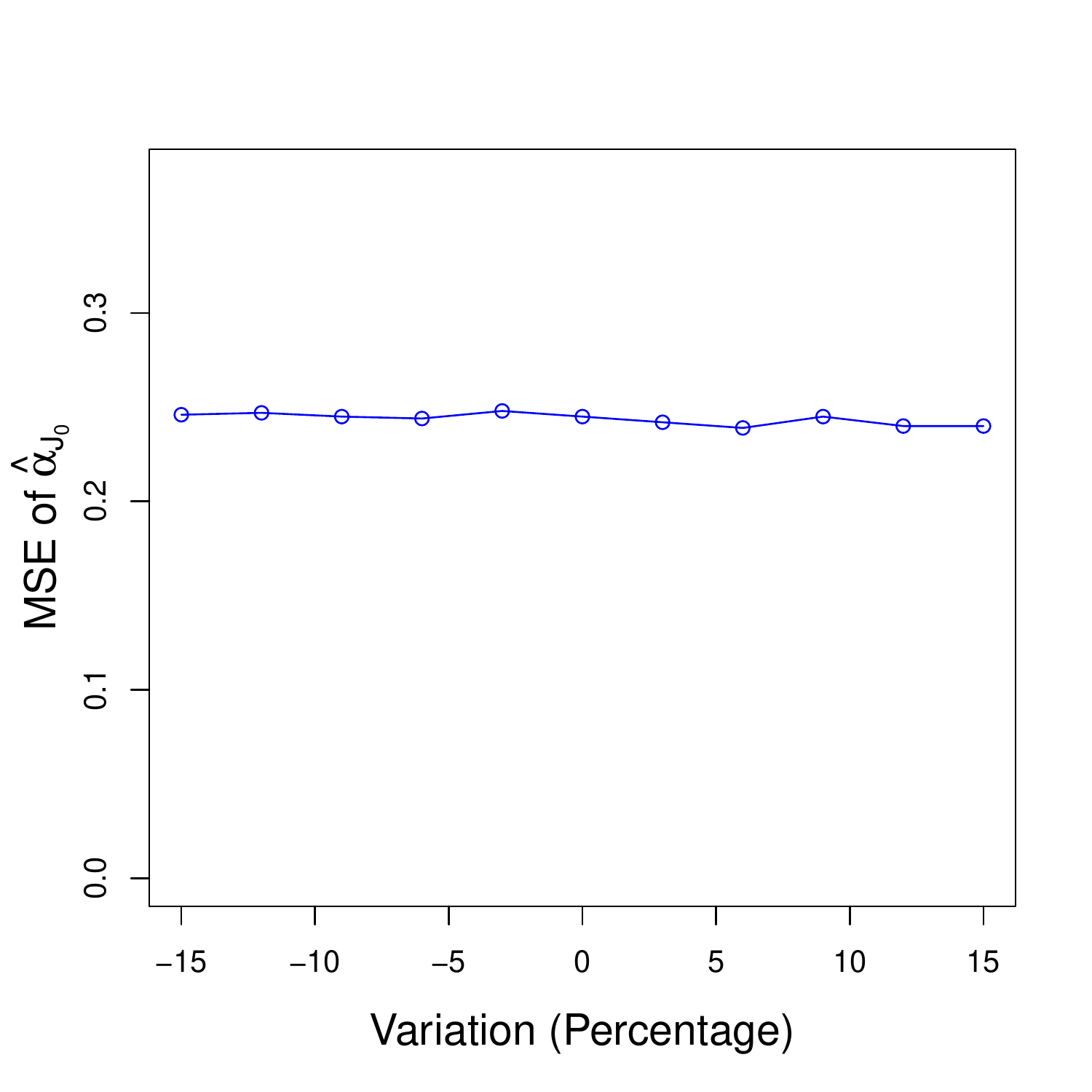}
\end{minipage}\hfill
\begin{minipage}{.33\textwidth}
\centering
\includegraphics[width=\linewidth]{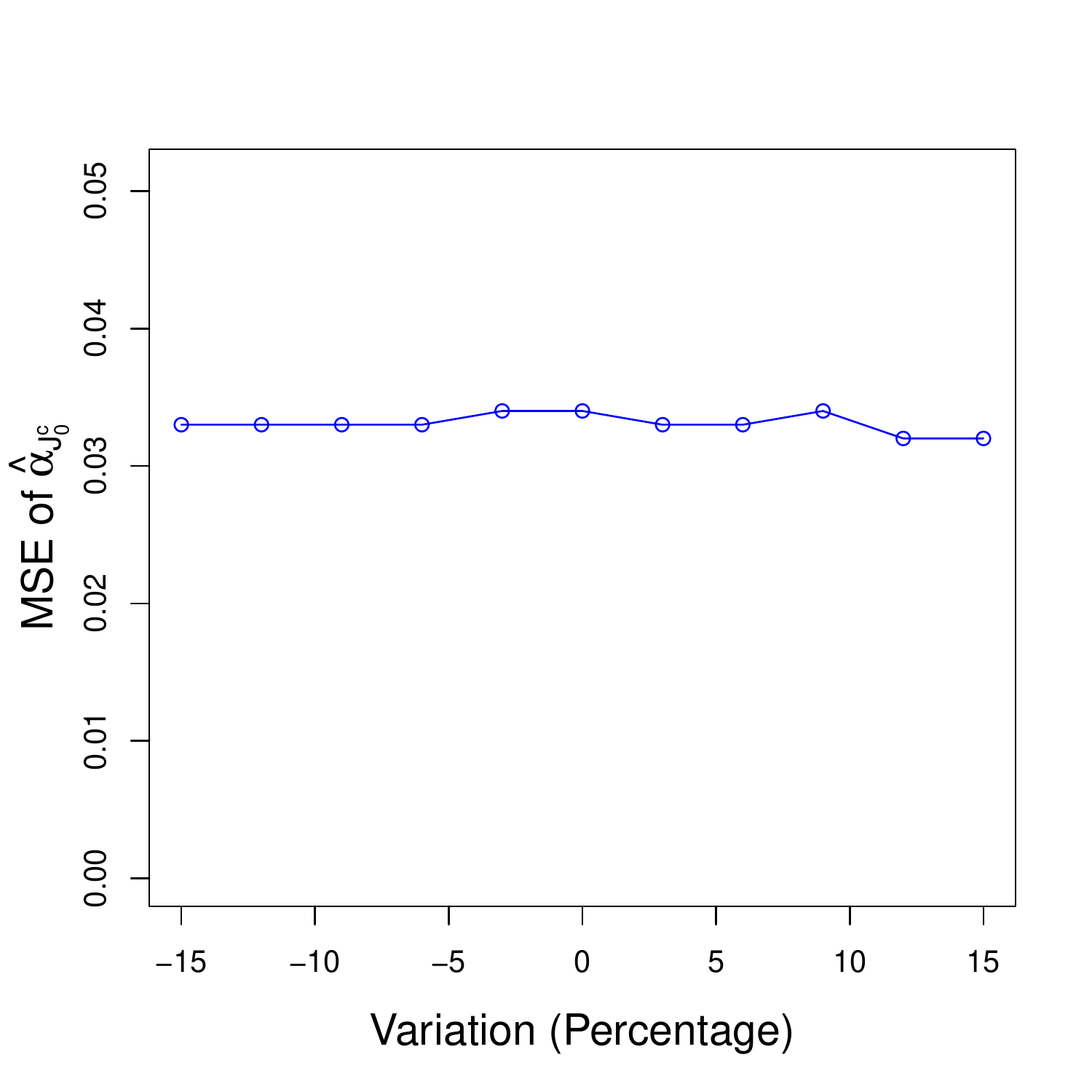}
\end{minipage}\hfill
\begin{minipage}{.33\textwidth}
\centering
\includegraphics[width=\linewidth]{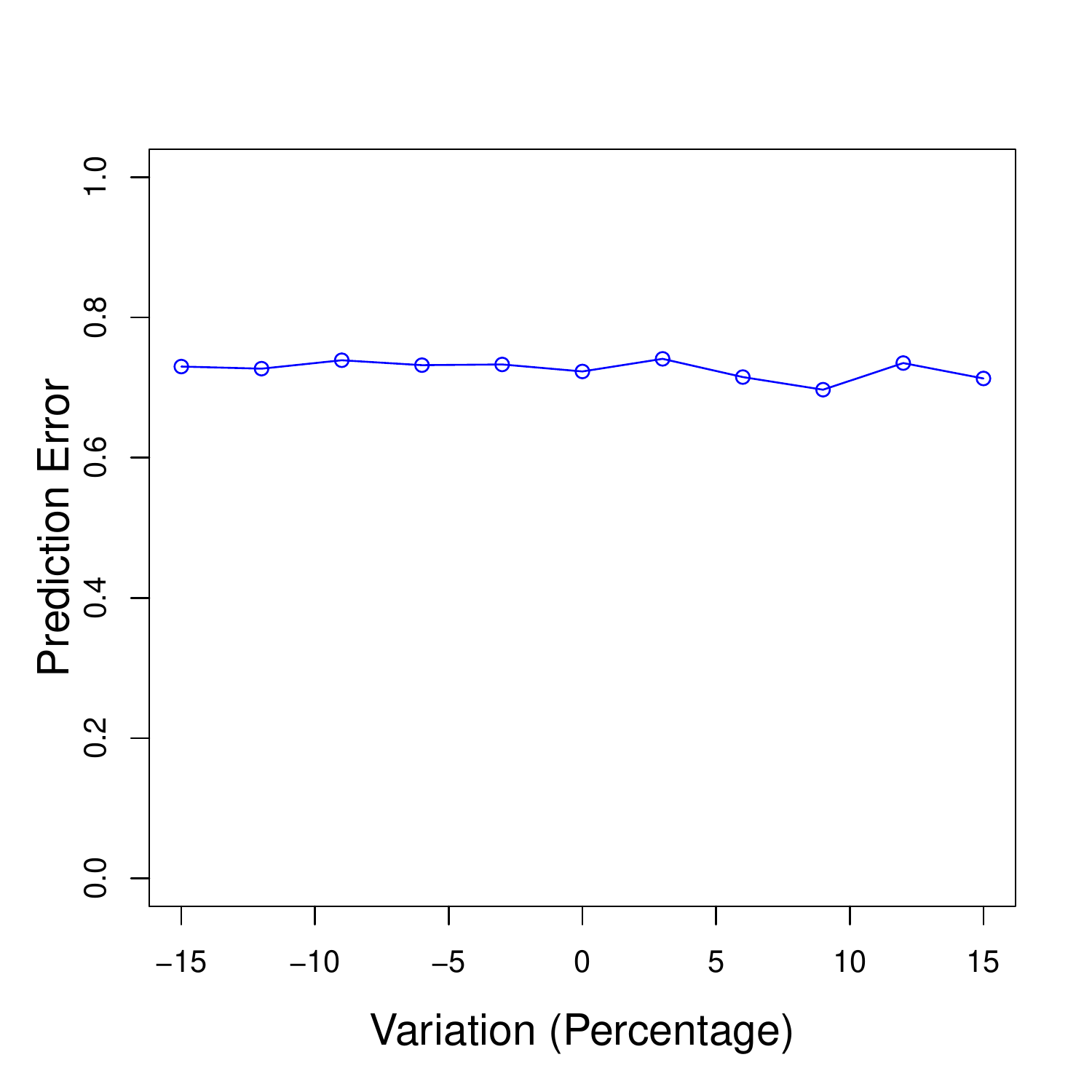}
\end{minipage}
\begin{minipage}{.33\textwidth}
\centering
\includegraphics[width=\linewidth]{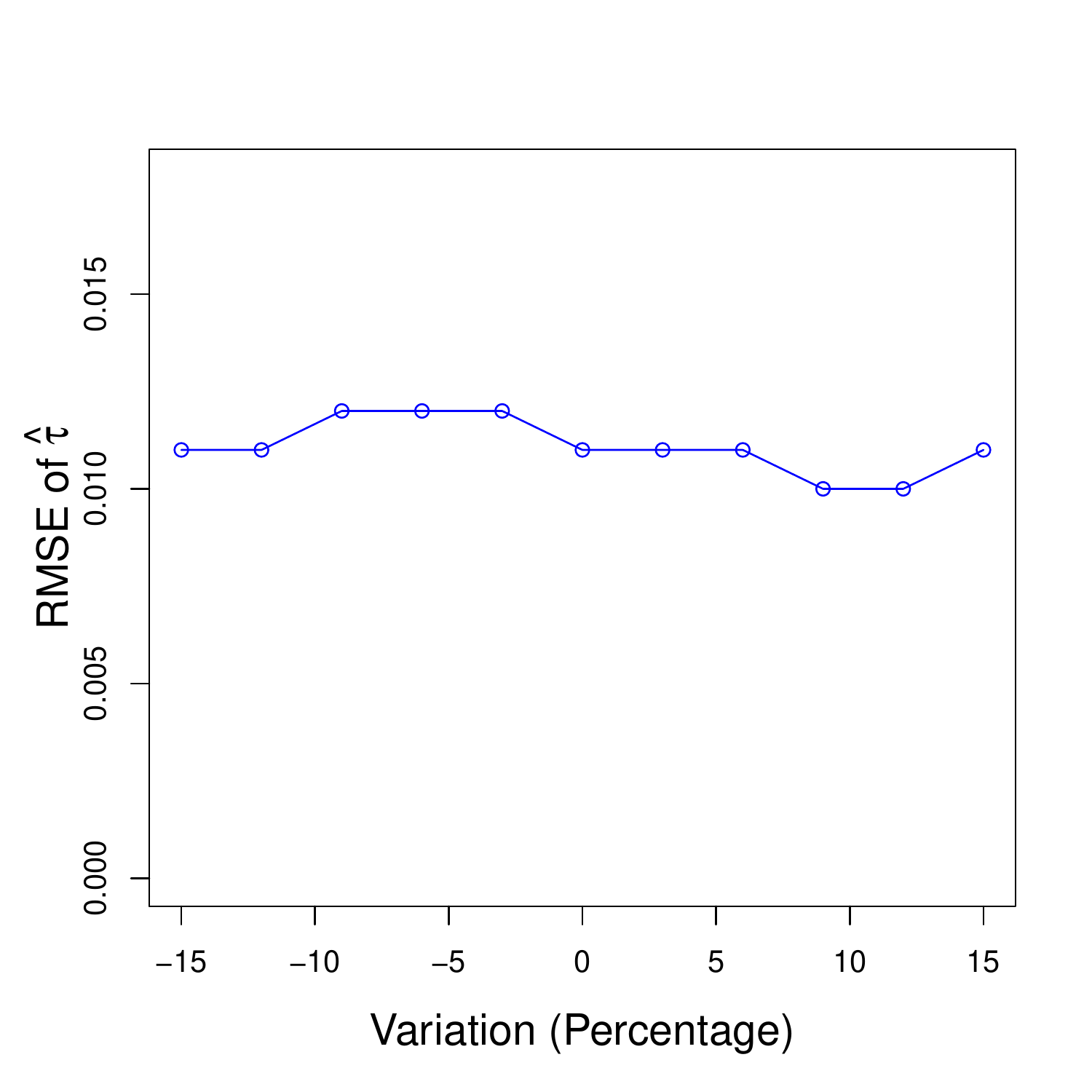}
\end{minipage}\hfill
\begin{minipage}{.33\textwidth}
\centering
\includegraphics[width=\linewidth]{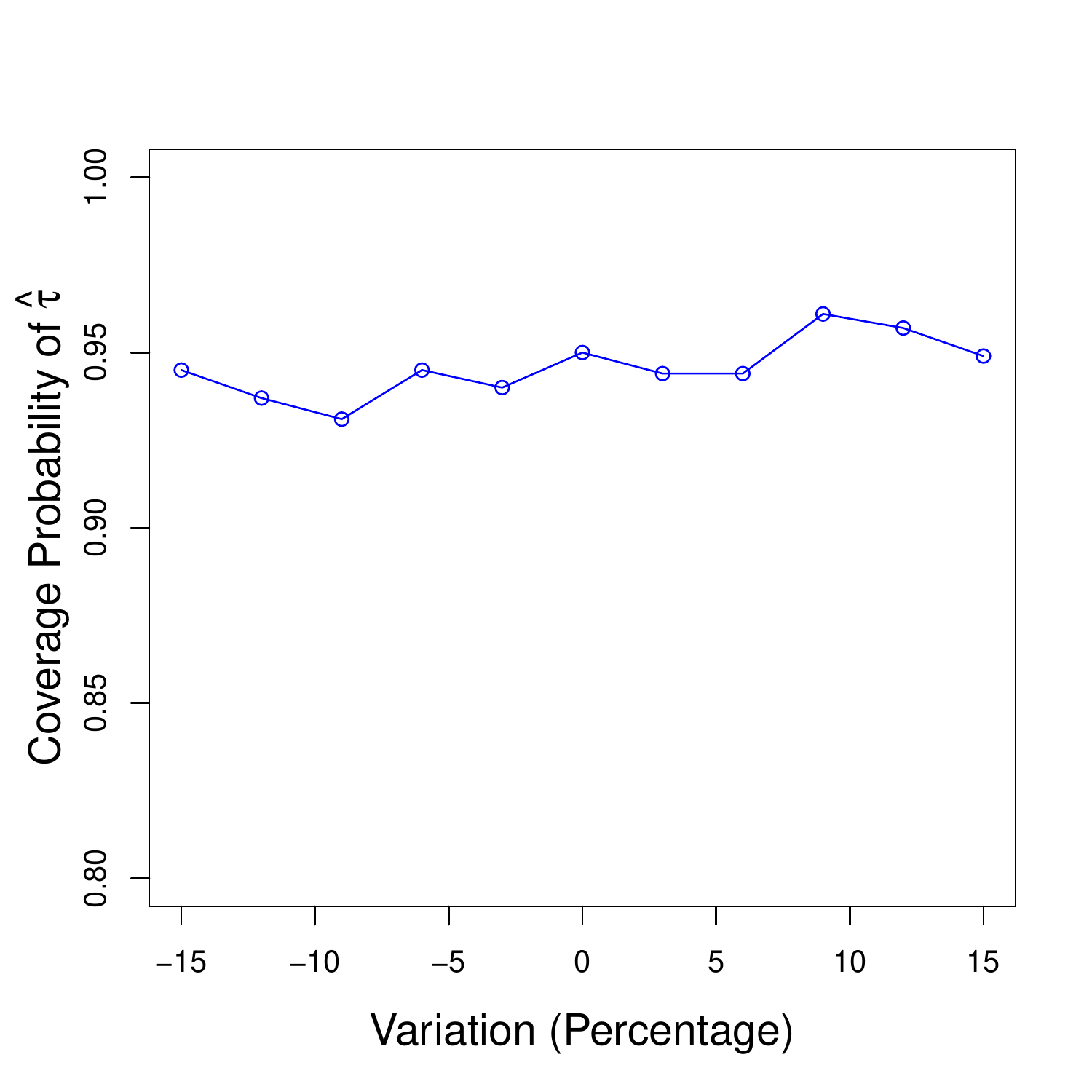}
\end{minipage}\hfill
\begin{minipage}{.33\textwidth}
\centering
\includegraphics[width=\linewidth]{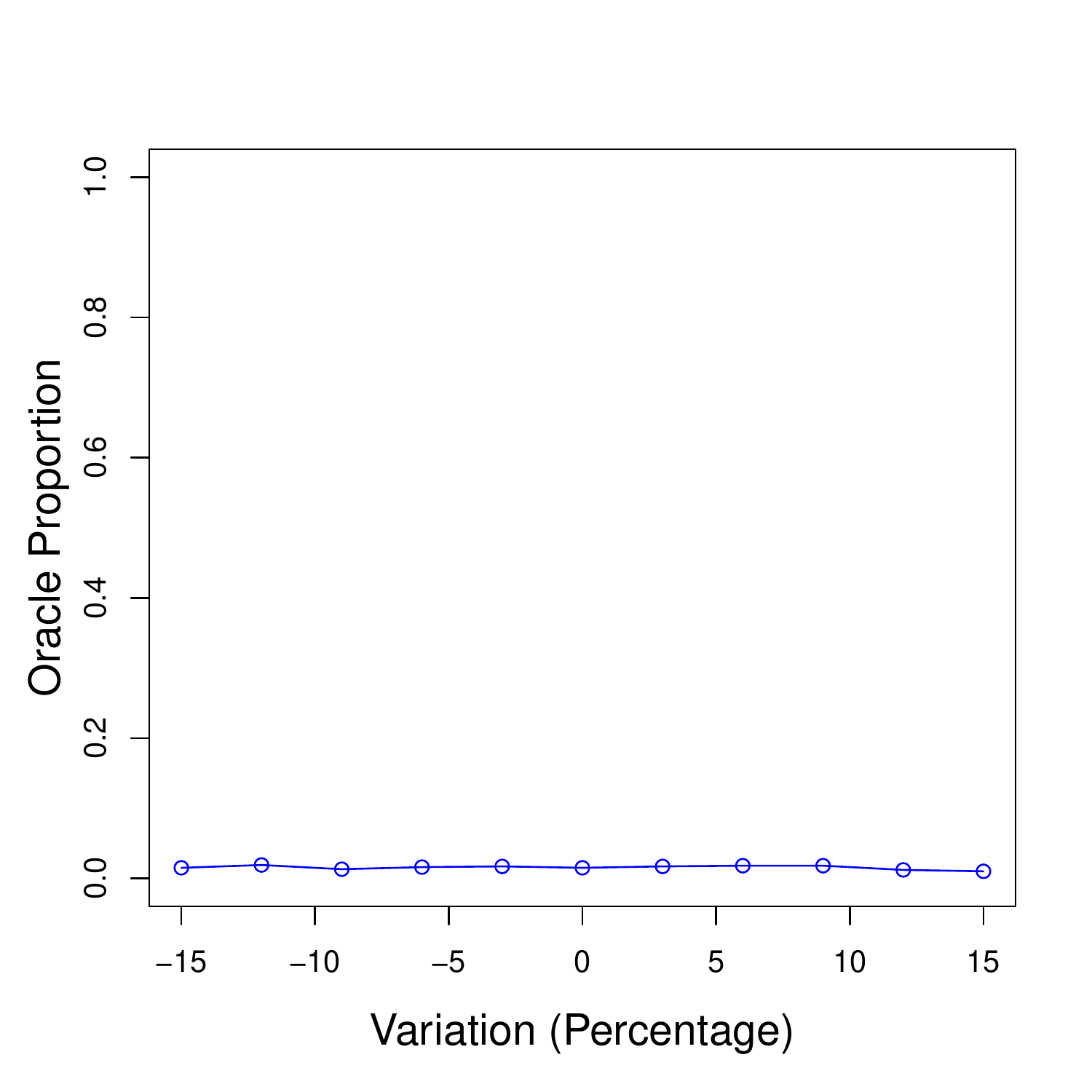}
\end{minipage}
\end{figure}

\newpage

\begin{figure}[tbph]
\centering
\caption{Sensitivity Analysis of $\gm^*$: Step 2}\label{fig:sen-gm-s2}
\begin{minipage}{.49\textwidth}
\centering
\includegraphics[width=\linewidth]{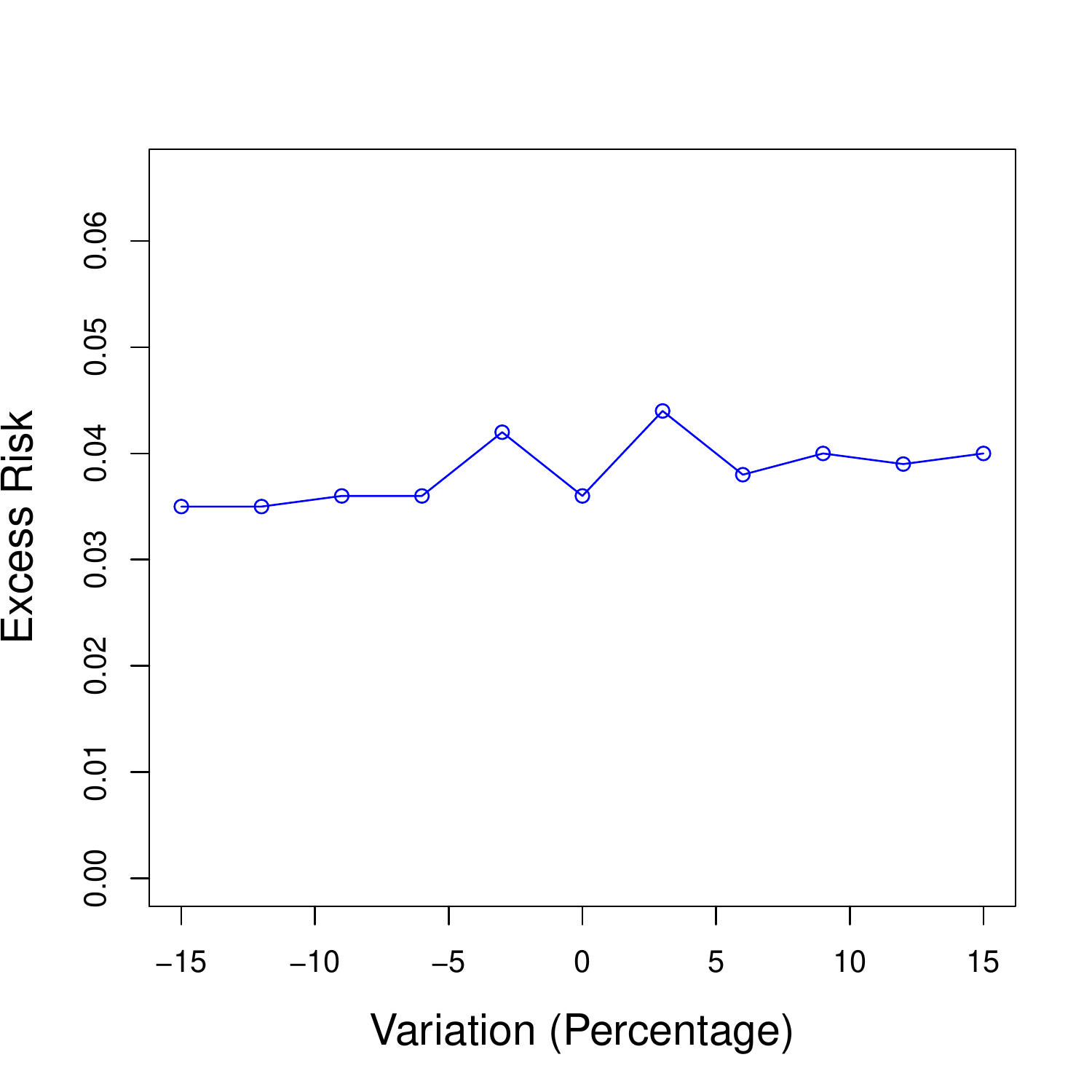}
\end{minipage}\hfill
\begin{minipage}{.49\textwidth}
\centering
\includegraphics[width=\linewidth]{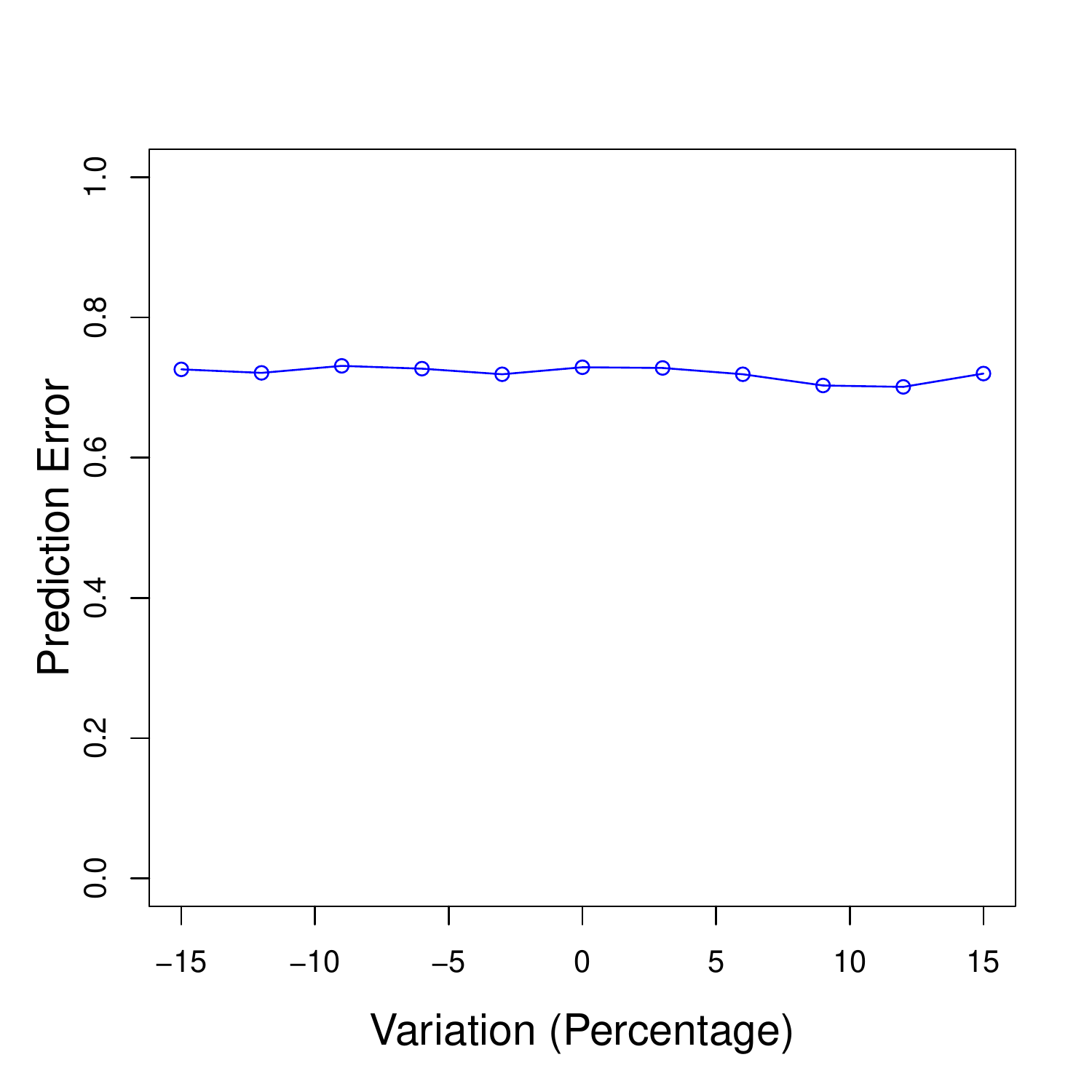}
\end{minipage}\hfill
\begin{minipage}{.49\textwidth}
\centering
\includegraphics[width=\linewidth]{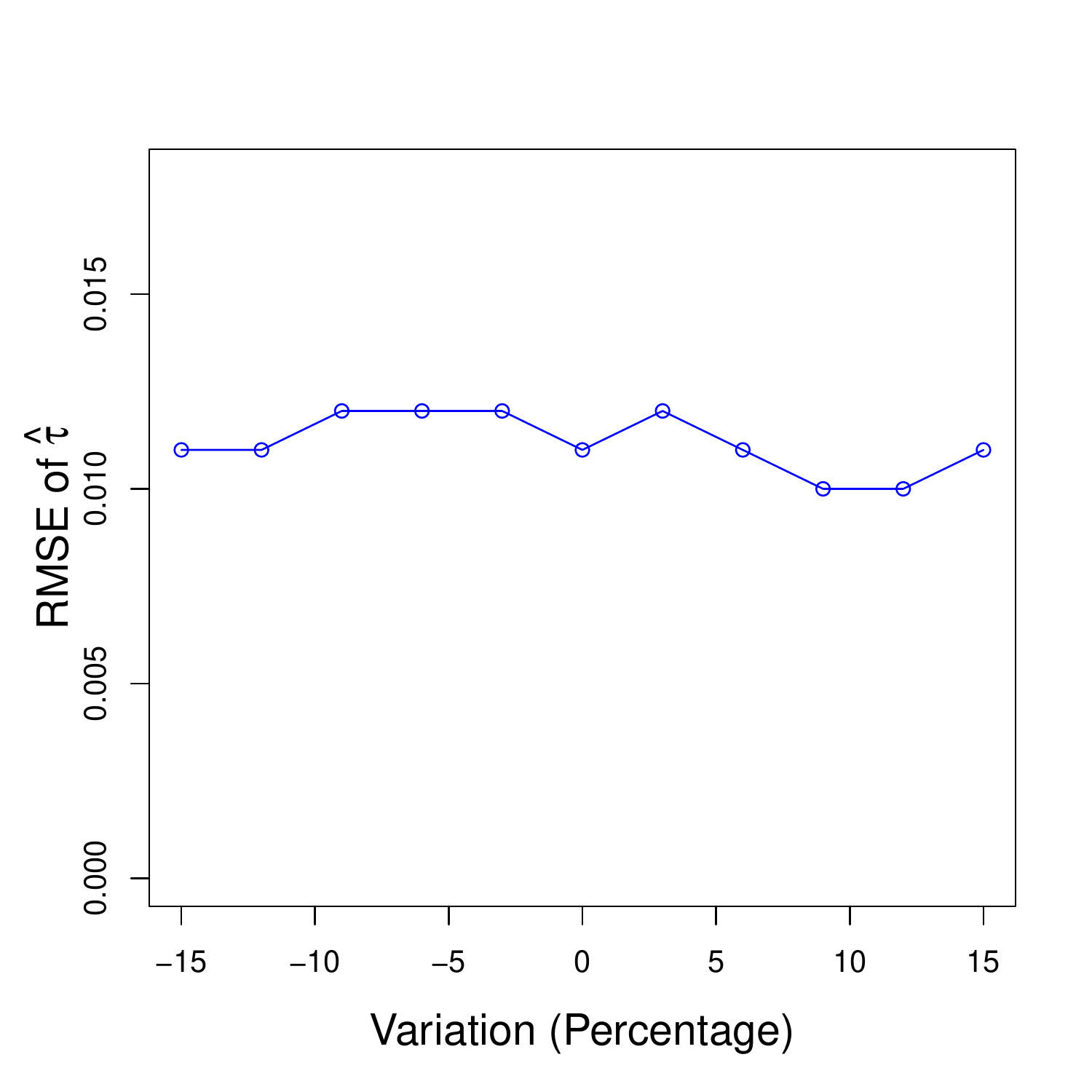}
\end{minipage}
\begin{minipage}{.49\textwidth}
\centering
\includegraphics[width=\linewidth]{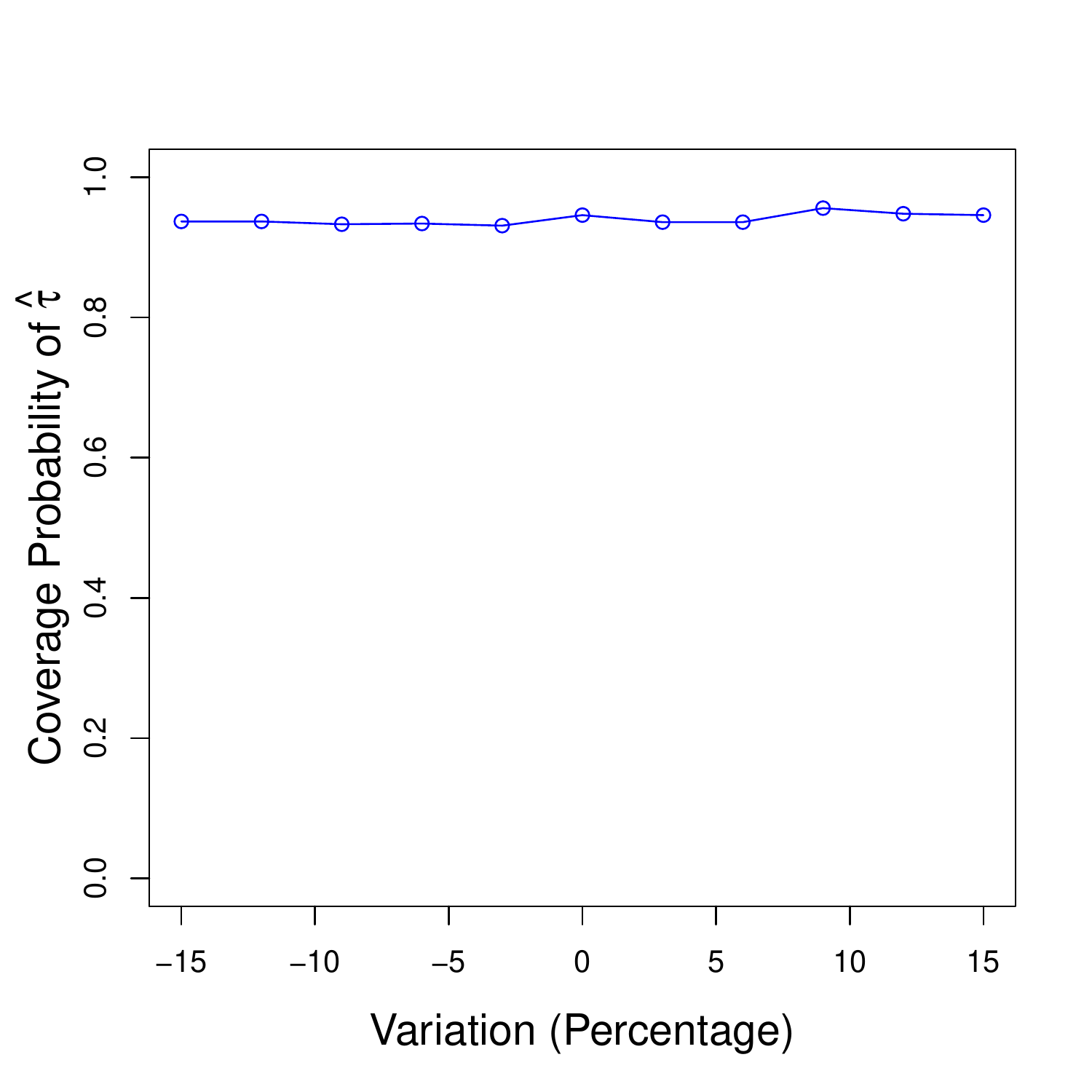}
\end{minipage}\hfill
\end{figure}

\newpage

\begin{figure}[tbph]
\centering
\caption{Sensitivity Analysis of $\gm^*$: Step 3a}\label{fig:sen-gm-s3a}
\begin{minipage}{.33\textwidth}
\centering
\includegraphics[width=\linewidth]{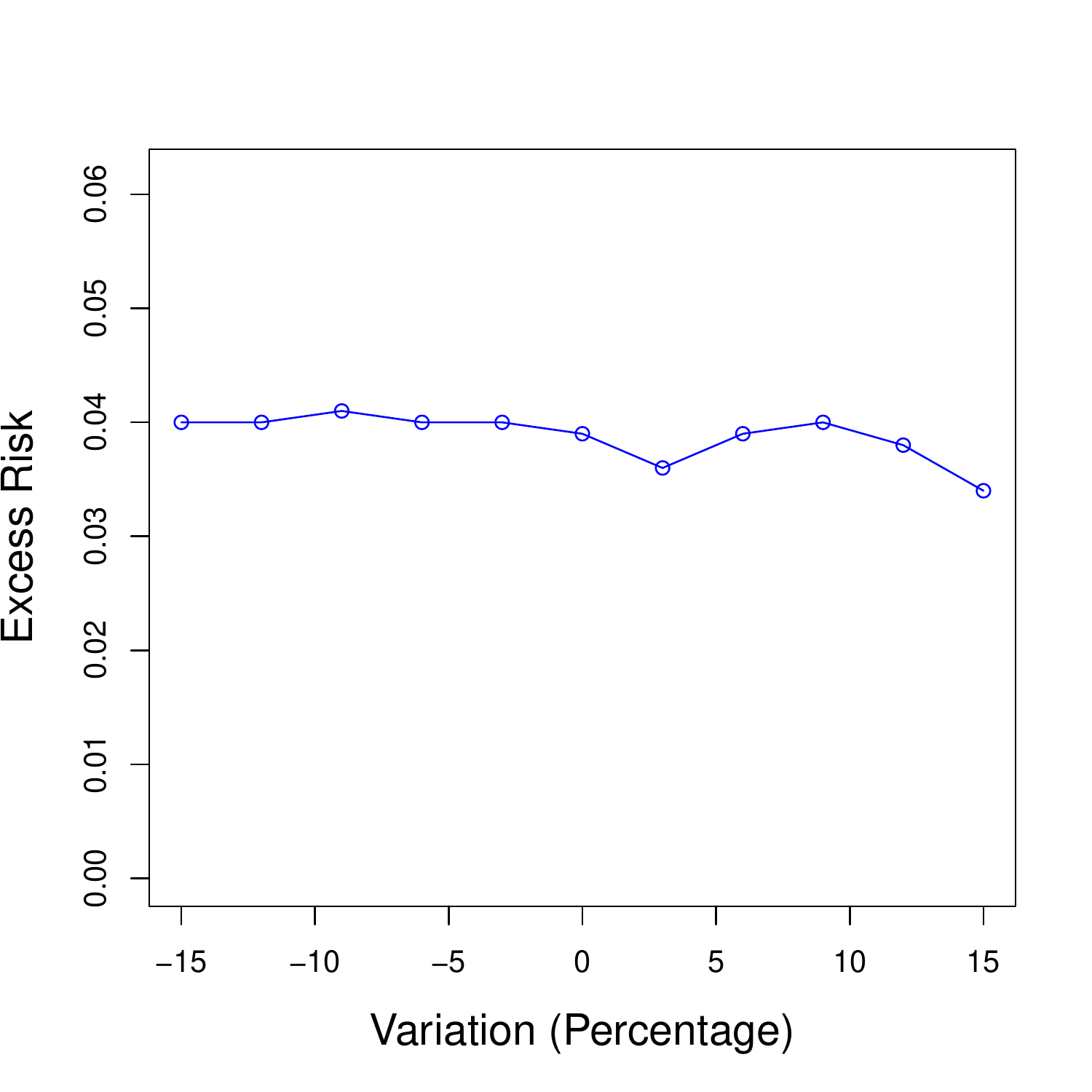}
\end{minipage}\hfill
\begin{minipage}{.33\textwidth}
\centering
\includegraphics[width=\linewidth]{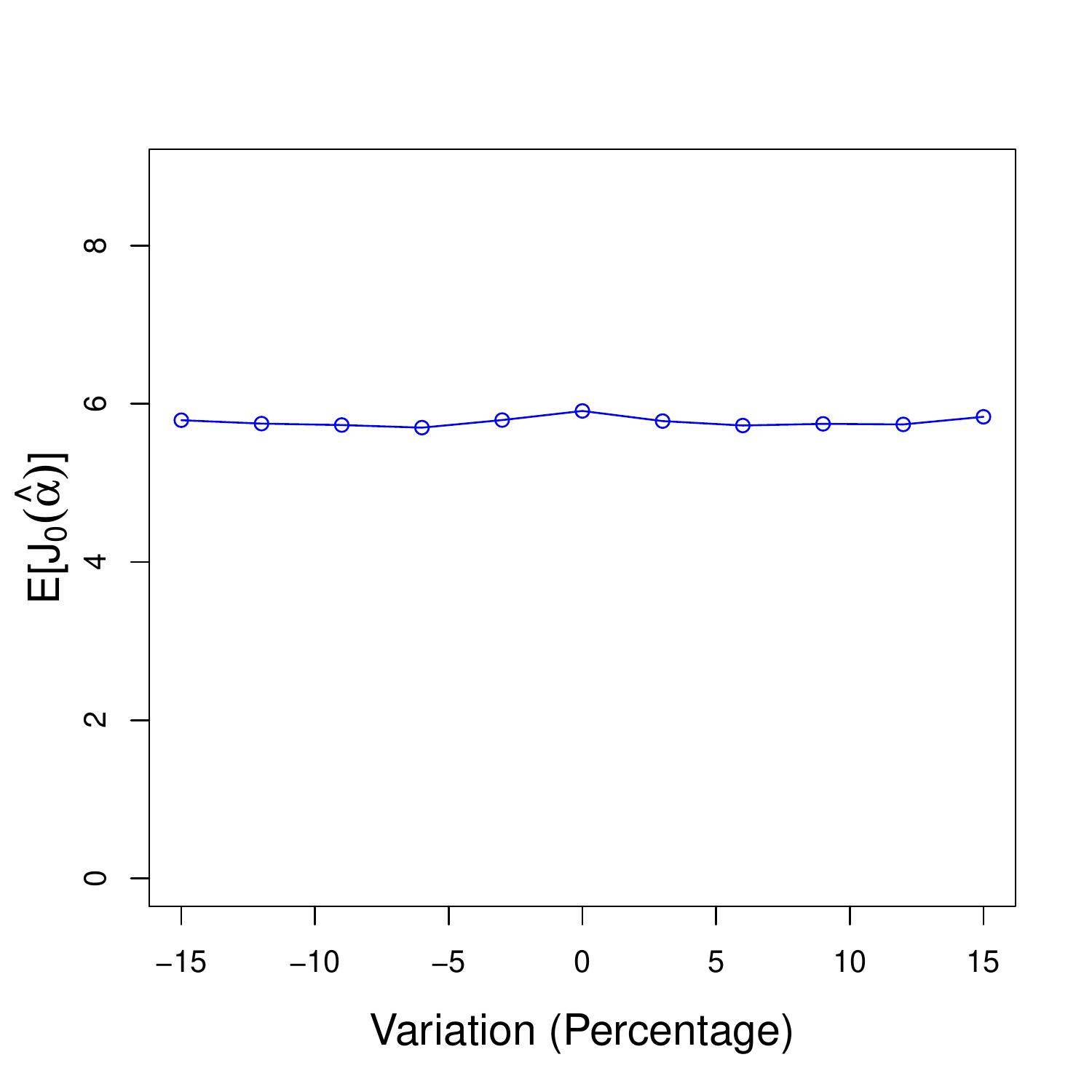}
\end{minipage}\hfill
\begin{minipage}{.33\textwidth}
\centering
\includegraphics[width=\linewidth]{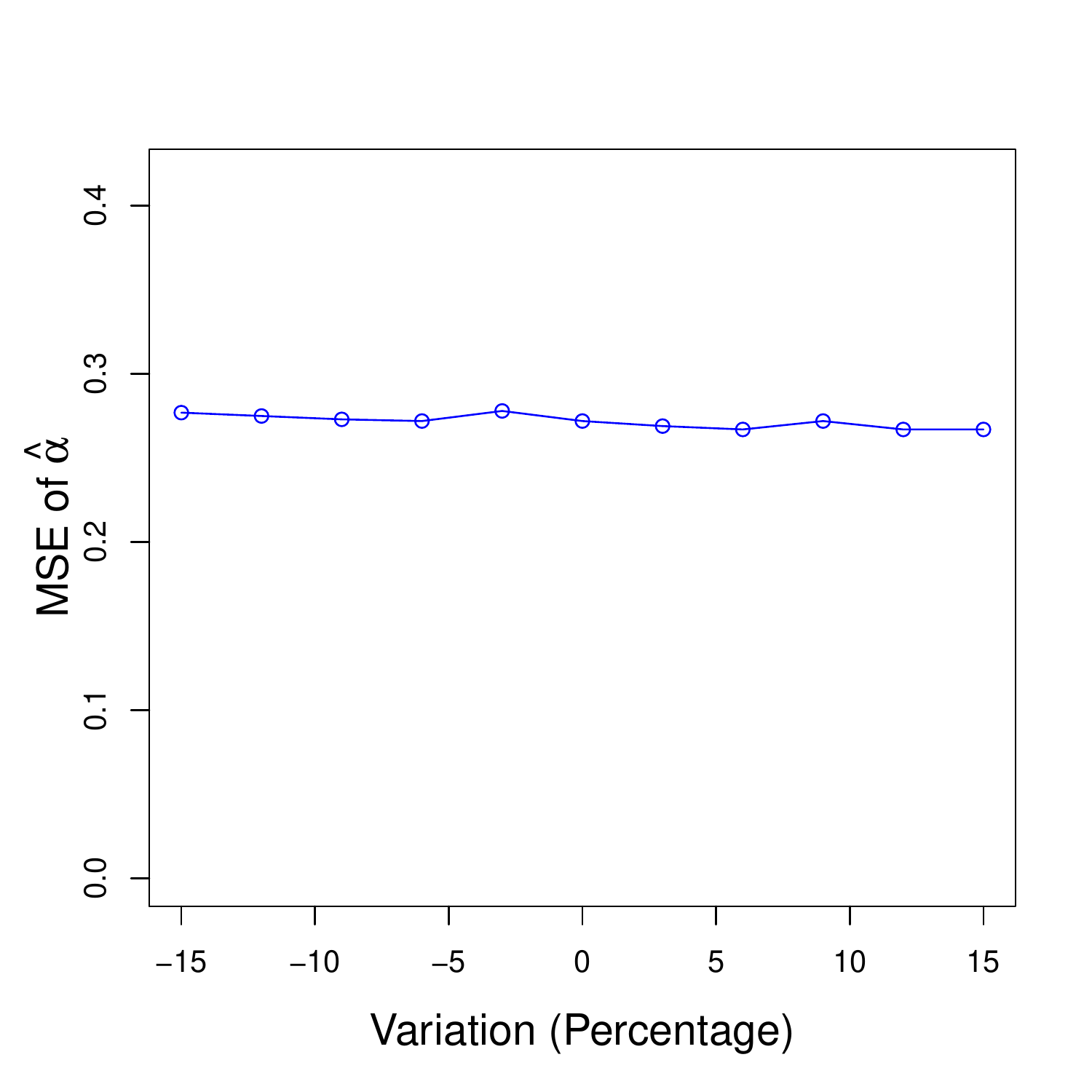}
\end{minipage}
\begin{minipage}{.33\textwidth}
\centering
\includegraphics[width=\linewidth]{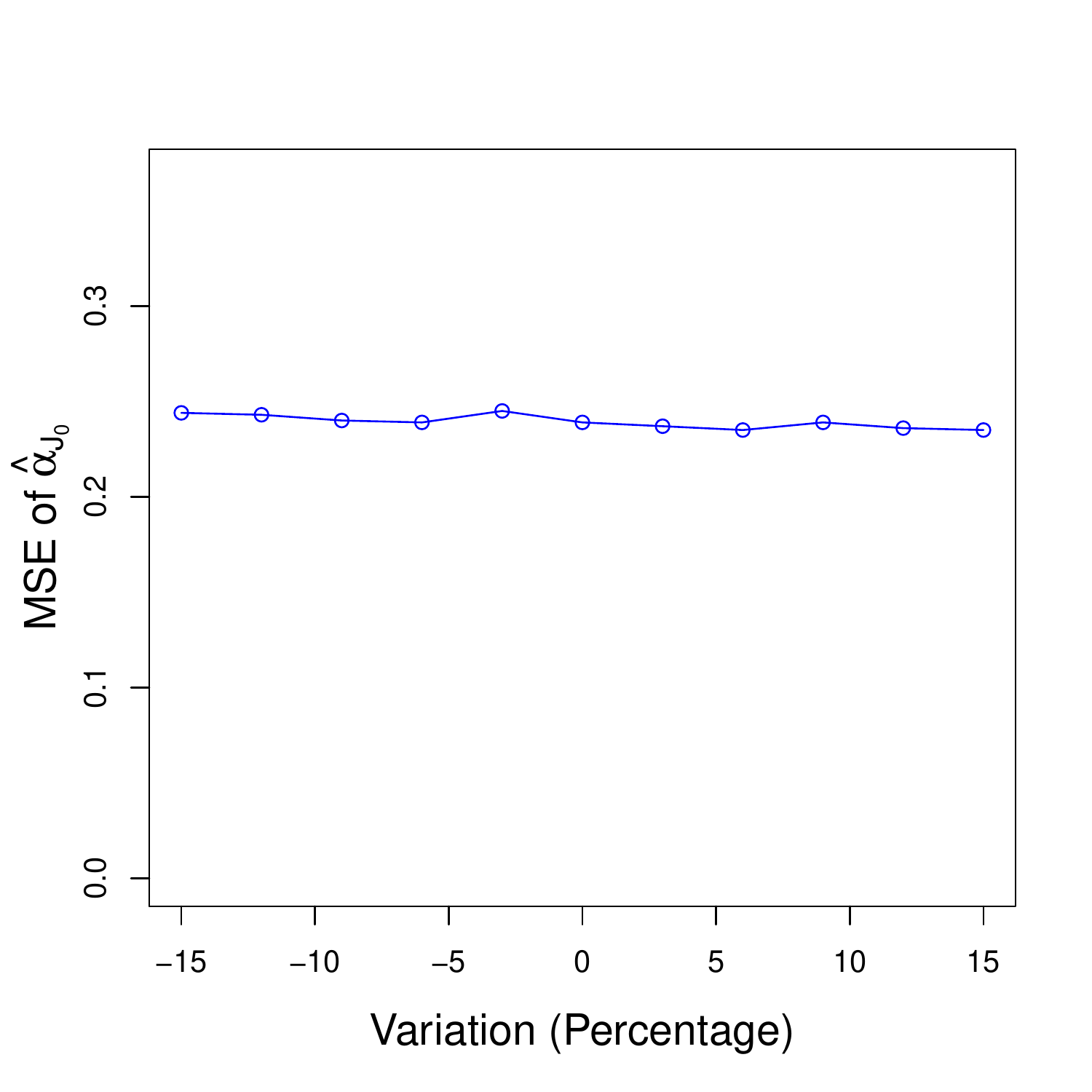}
\end{minipage}\hfill
\begin{minipage}{.33\textwidth}
\centering
\includegraphics[width=\linewidth]{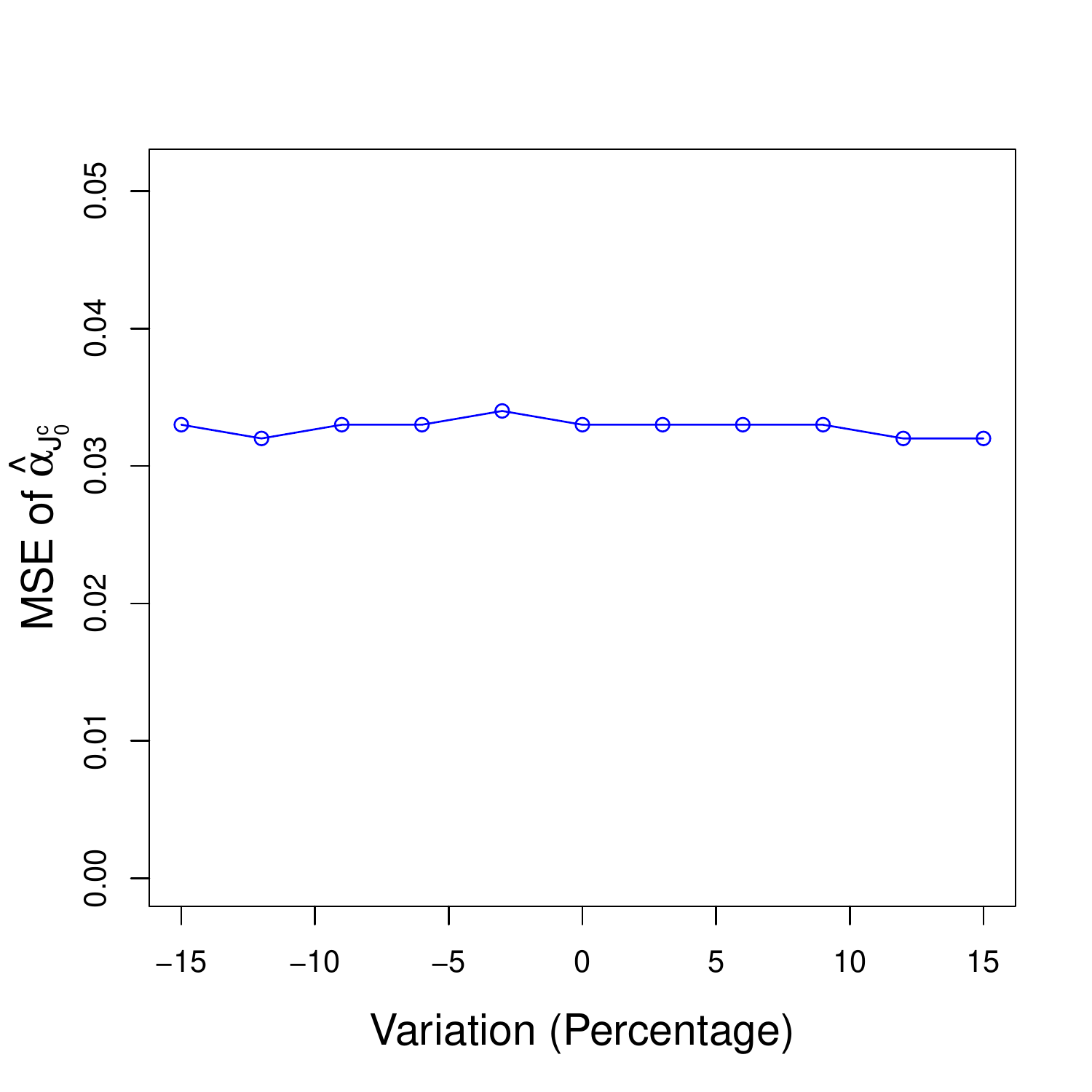}
\end{minipage}\hfill
\begin{minipage}{.33\textwidth}
\centering
\includegraphics[width=\linewidth]{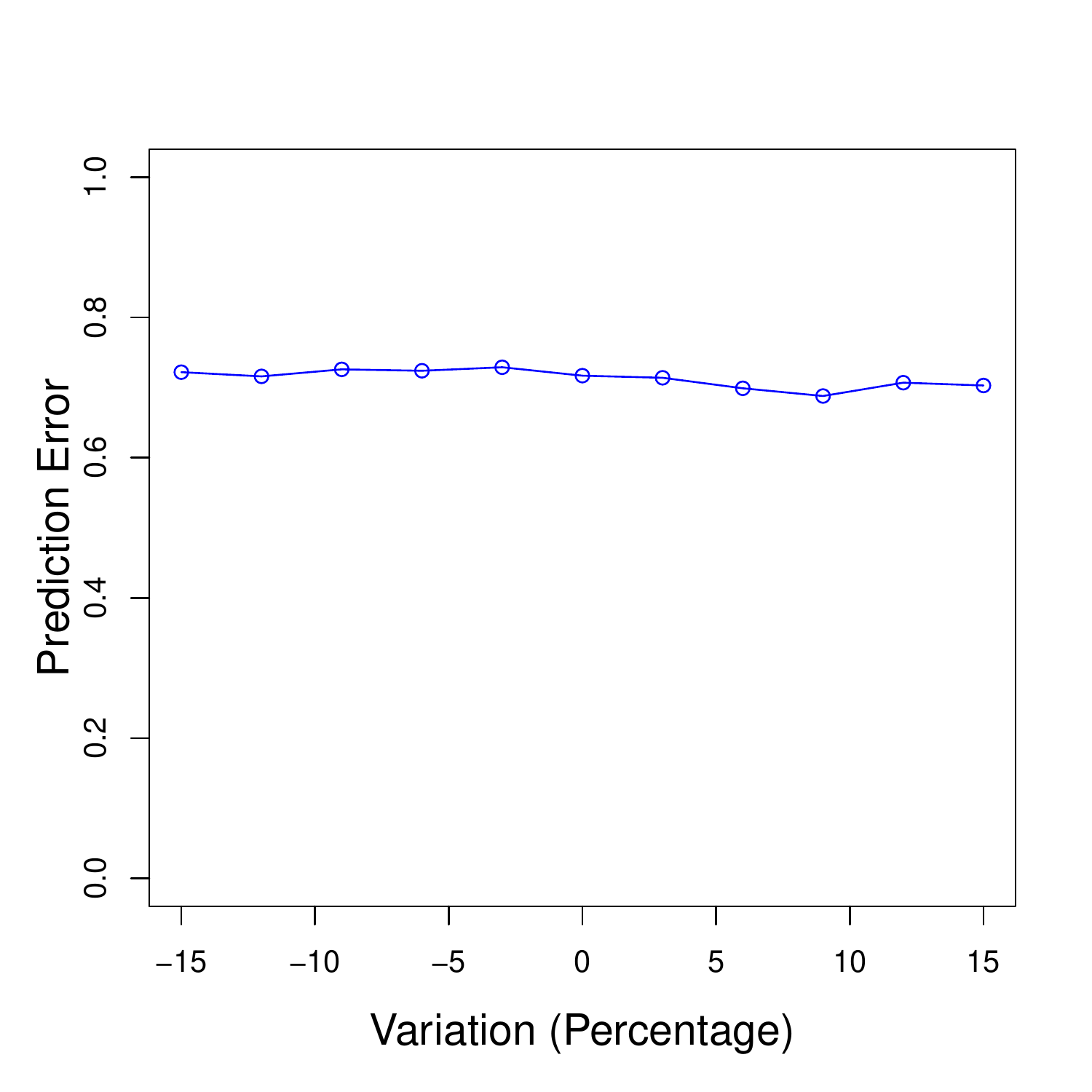}
\end{minipage}
\begin{minipage}{.33\textwidth}
\centering
\includegraphics[width=\linewidth]{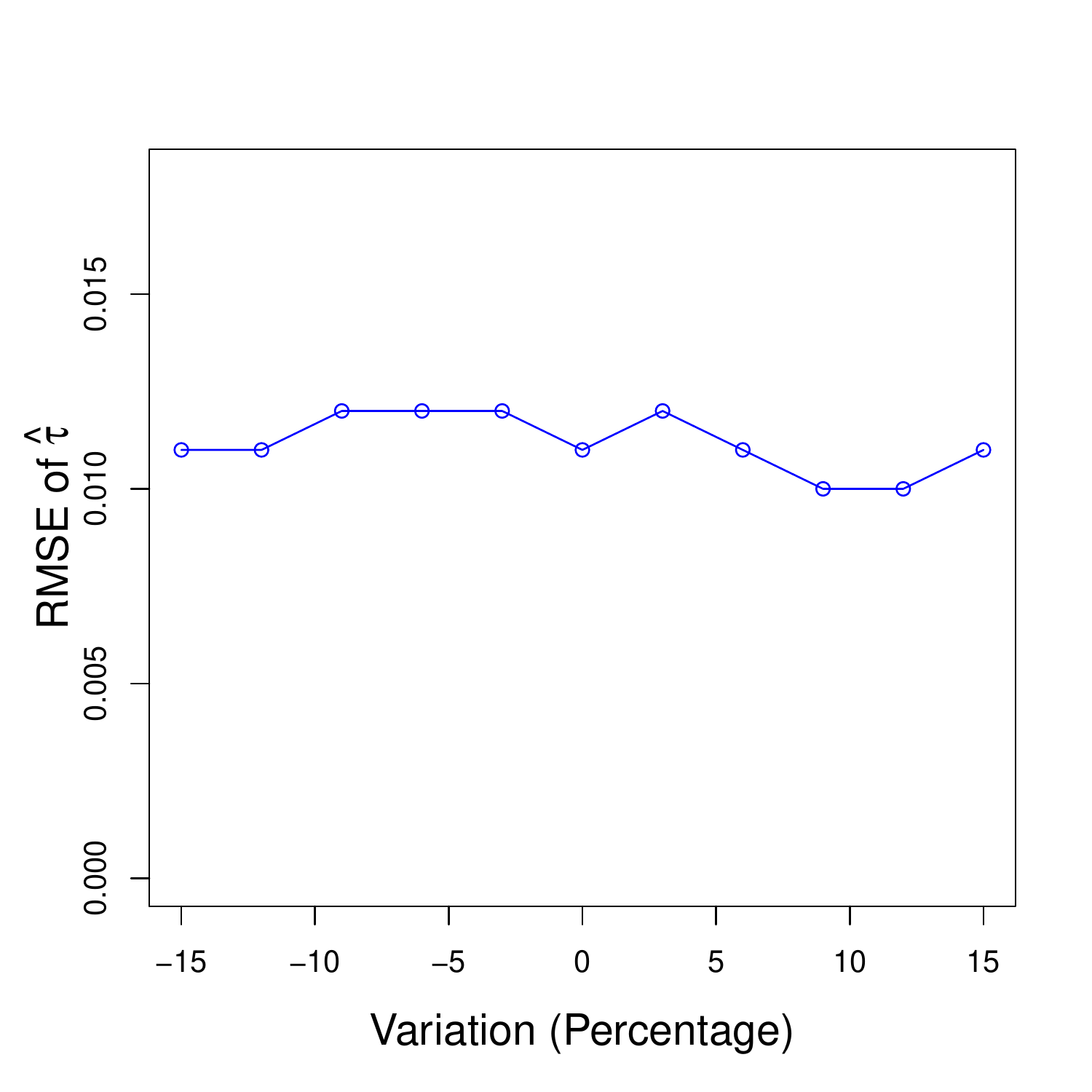}
\end{minipage}\hfill
\begin{minipage}{.33\textwidth}
\centering
\includegraphics[width=\linewidth]{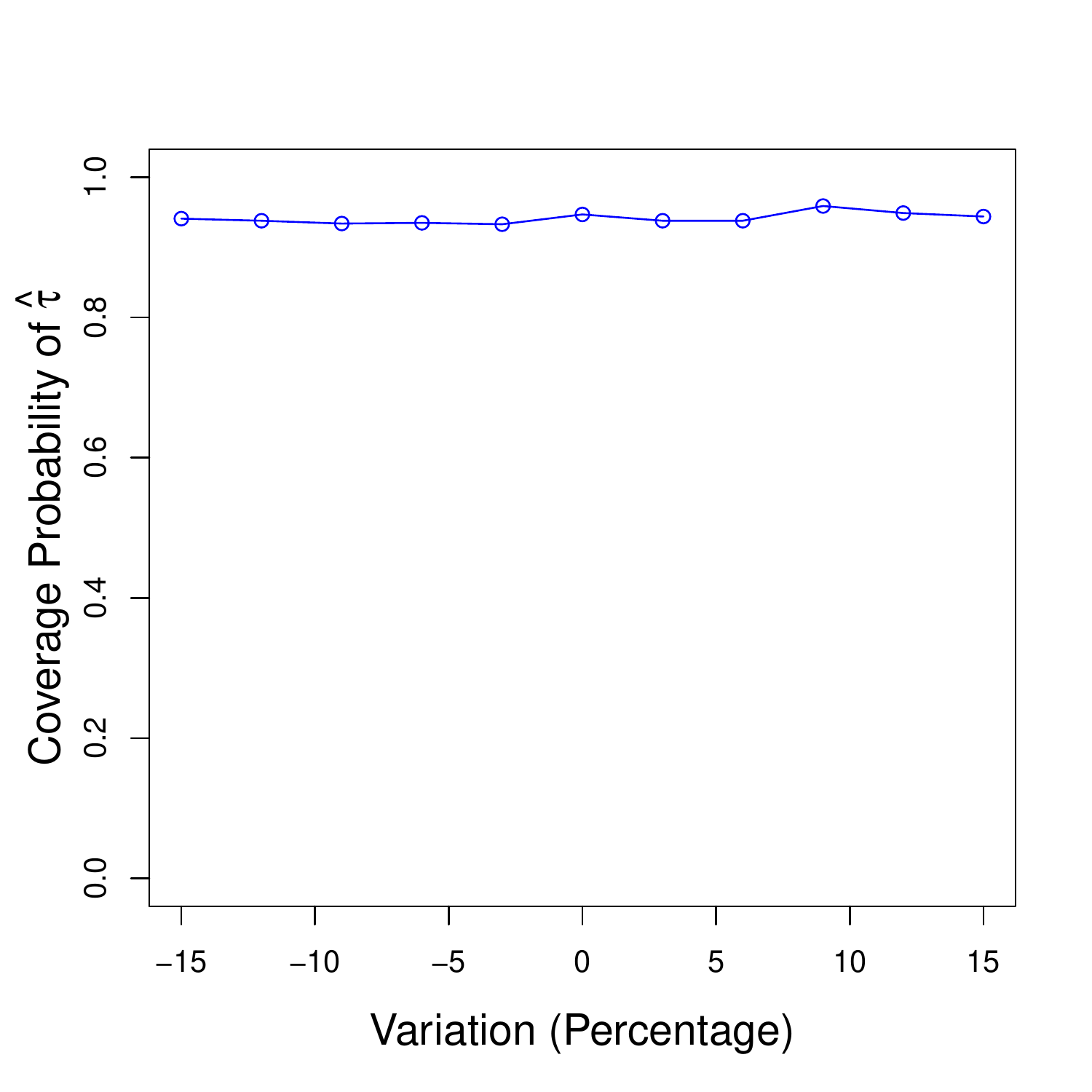}
\end{minipage}\hfill
\begin{minipage}{.33\textwidth}
\centering
\includegraphics[width=\linewidth]{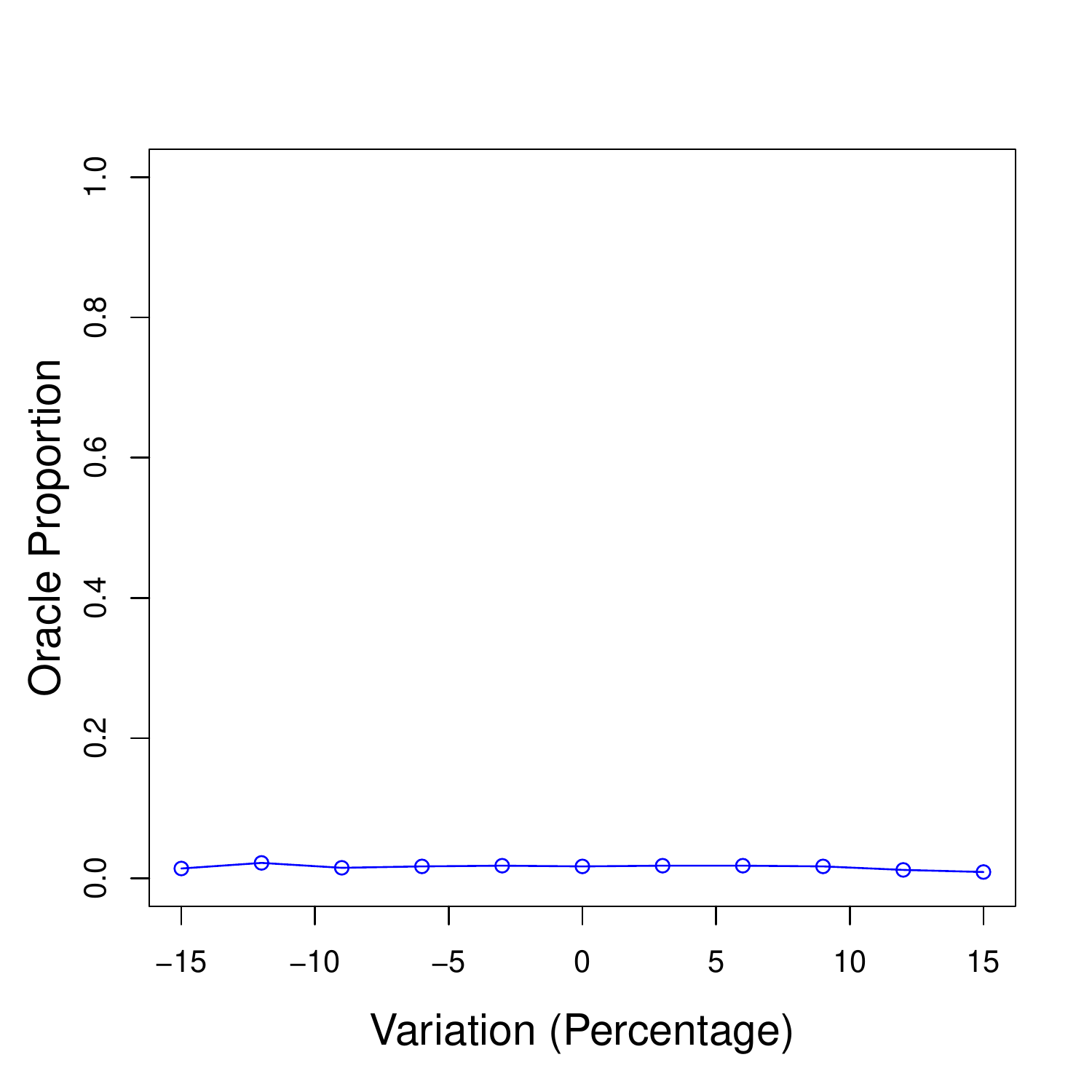}
\end{minipage}
\end{figure}

\newpage

\begin{figure}[tbph]
\centering
\caption{Sensitivity Analysis of $\gm^*$: Step 3b}\label{fig:sen-gm-s3b}
\begin{minipage}{.33\textwidth}
\centering
\includegraphics[width=\linewidth]{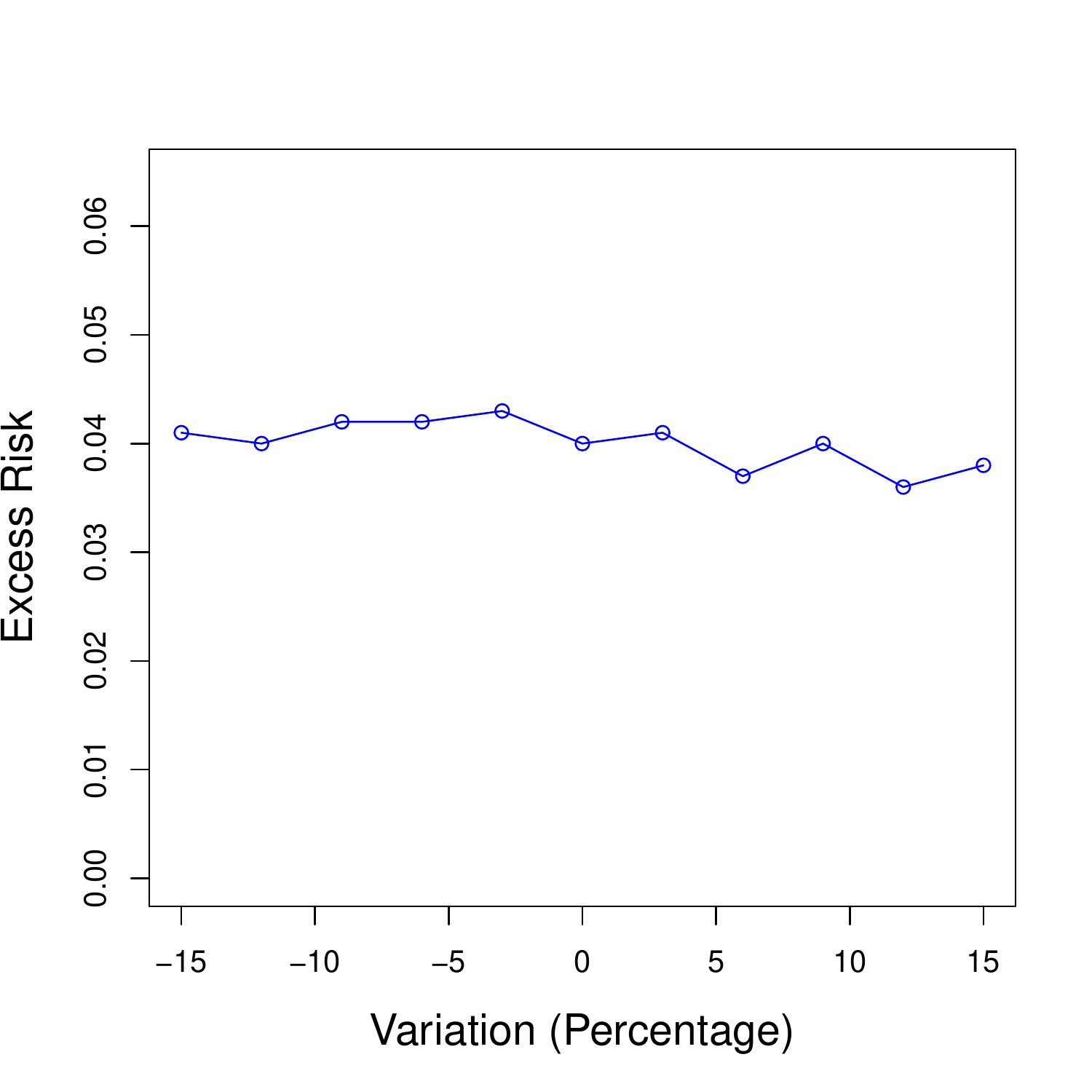}
\end{minipage}\hfill
\begin{minipage}{.33\textwidth}
\centering
\includegraphics[width=\linewidth]{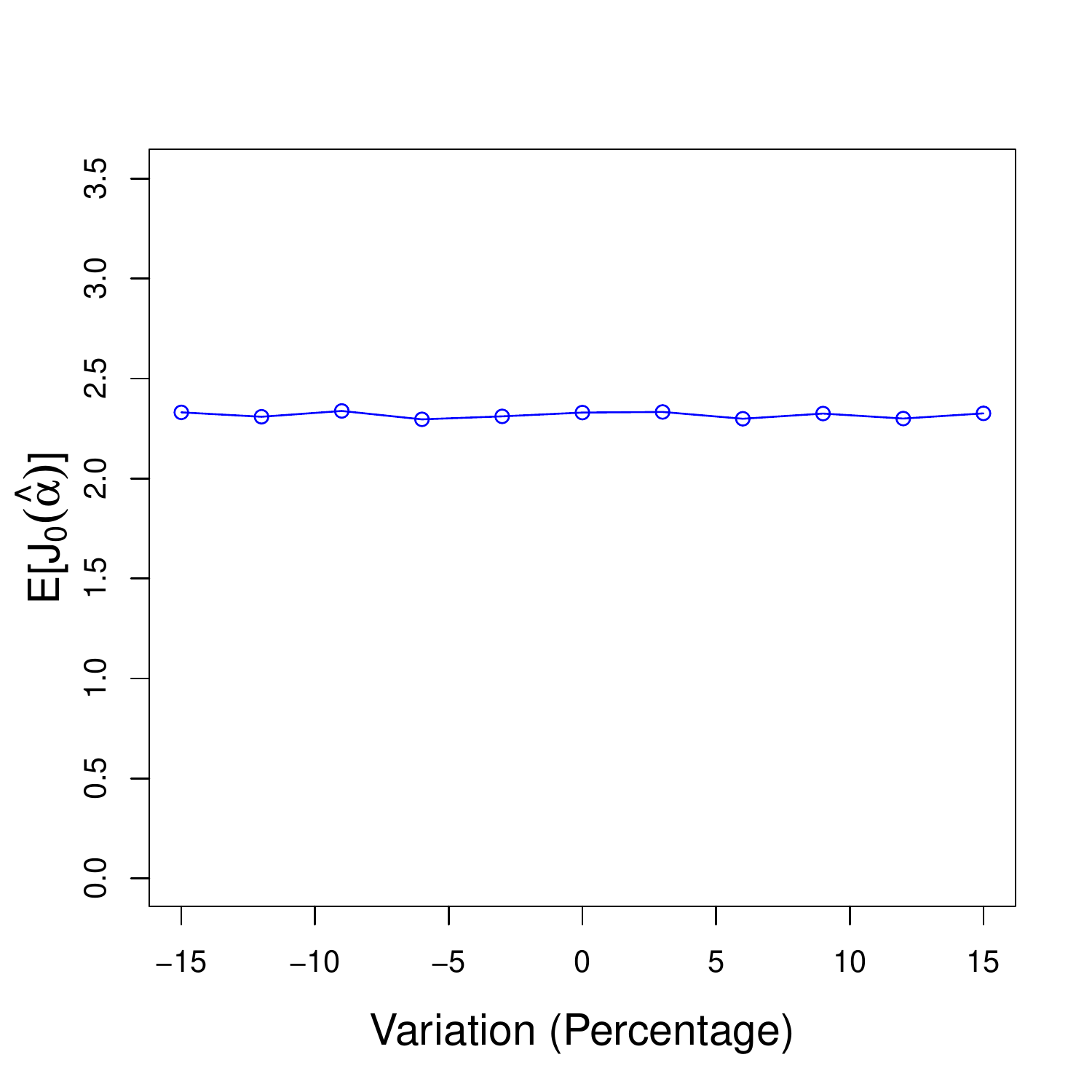}
\end{minipage}\hfill
\begin{minipage}{.33\textwidth}
\centering
\includegraphics[width=\linewidth]{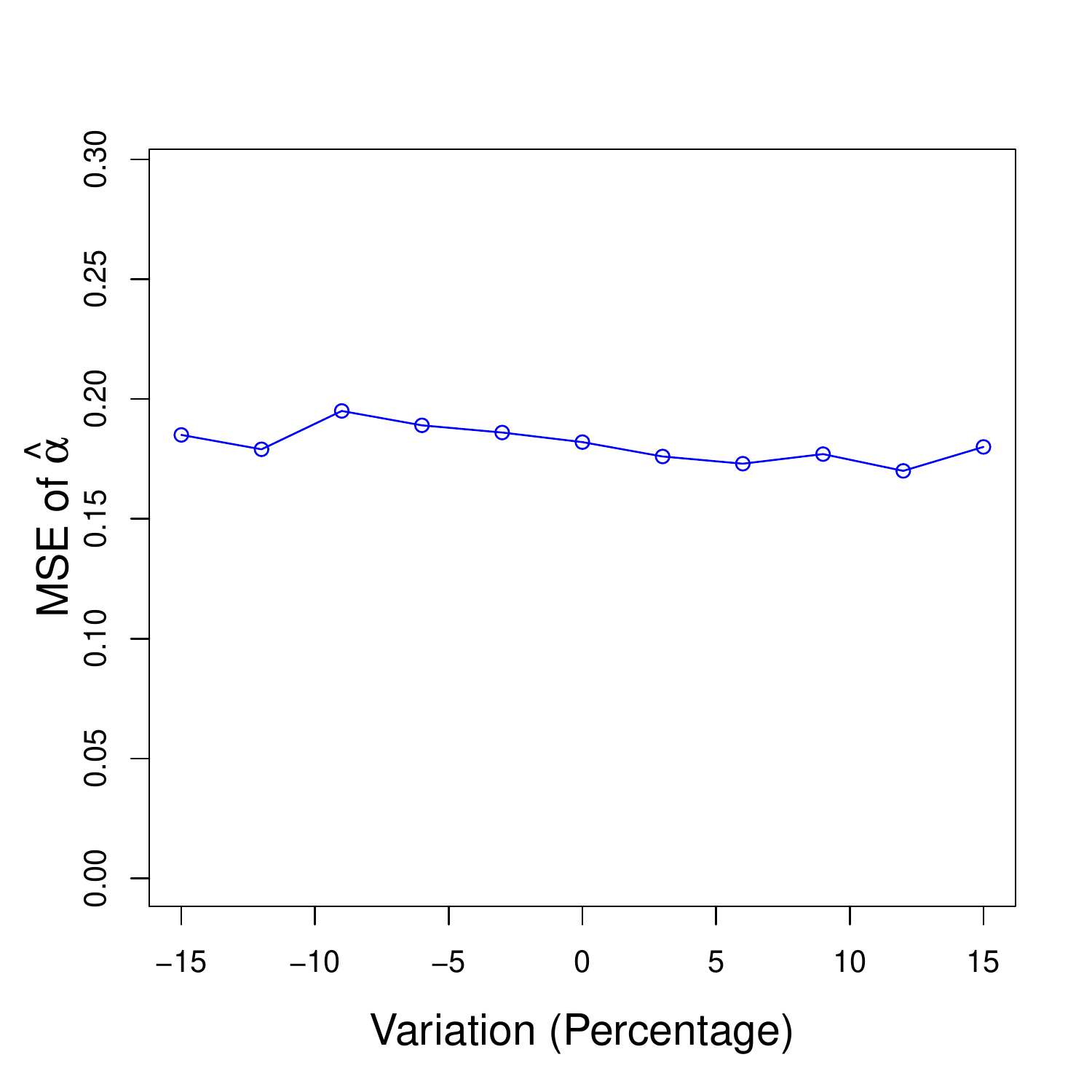}
\end{minipage}
\begin{minipage}{.33\textwidth}
\centering
\includegraphics[width=\linewidth]{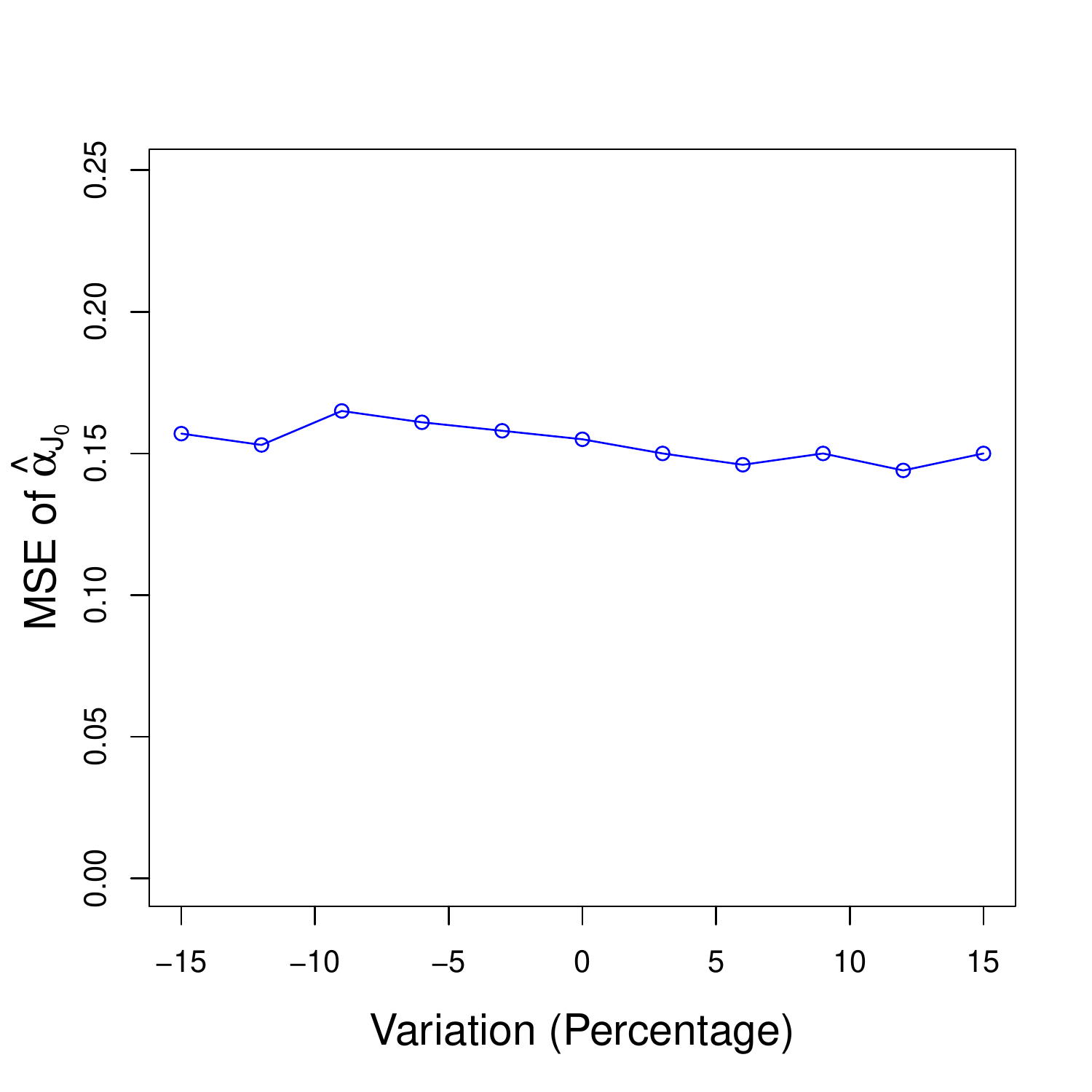}
\end{minipage}\hfill
\begin{minipage}{.33\textwidth}
\centering
\includegraphics[width=\linewidth]{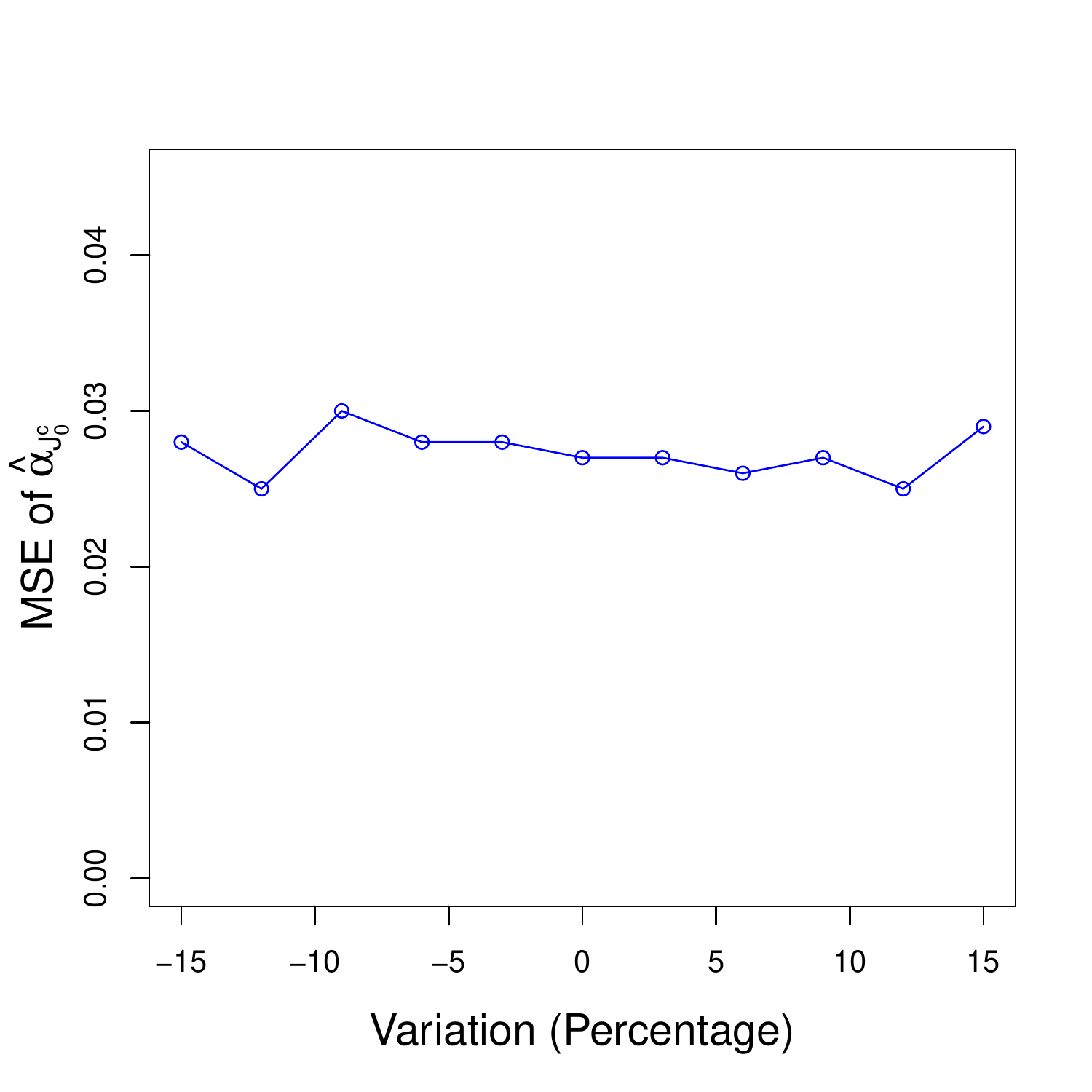}
\end{minipage}\hfill
\begin{minipage}{.33\textwidth}
\centering
\includegraphics[width=\linewidth]{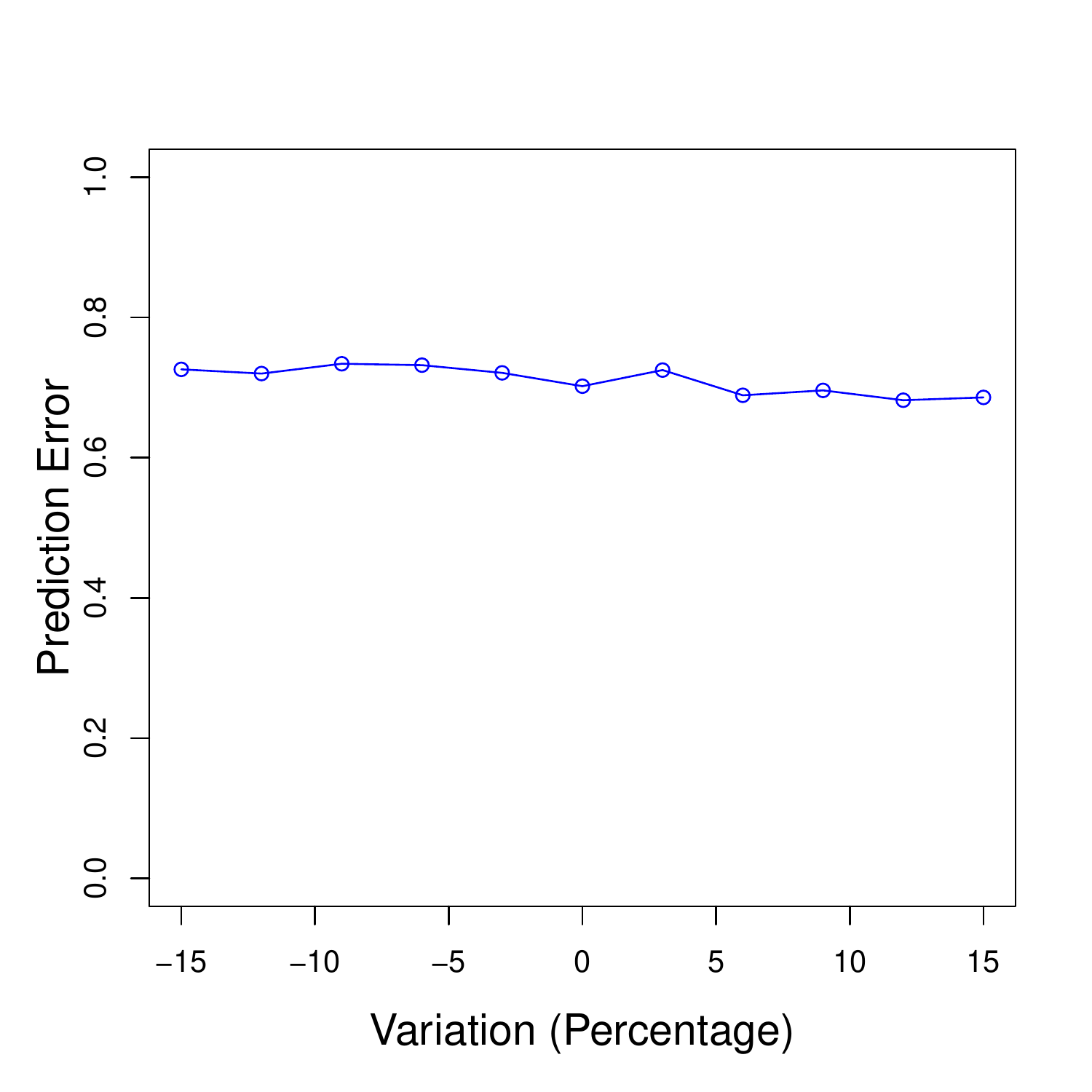}
\end{minipage}
\begin{minipage}{.33\textwidth}
\centering
\includegraphics[width=\linewidth]{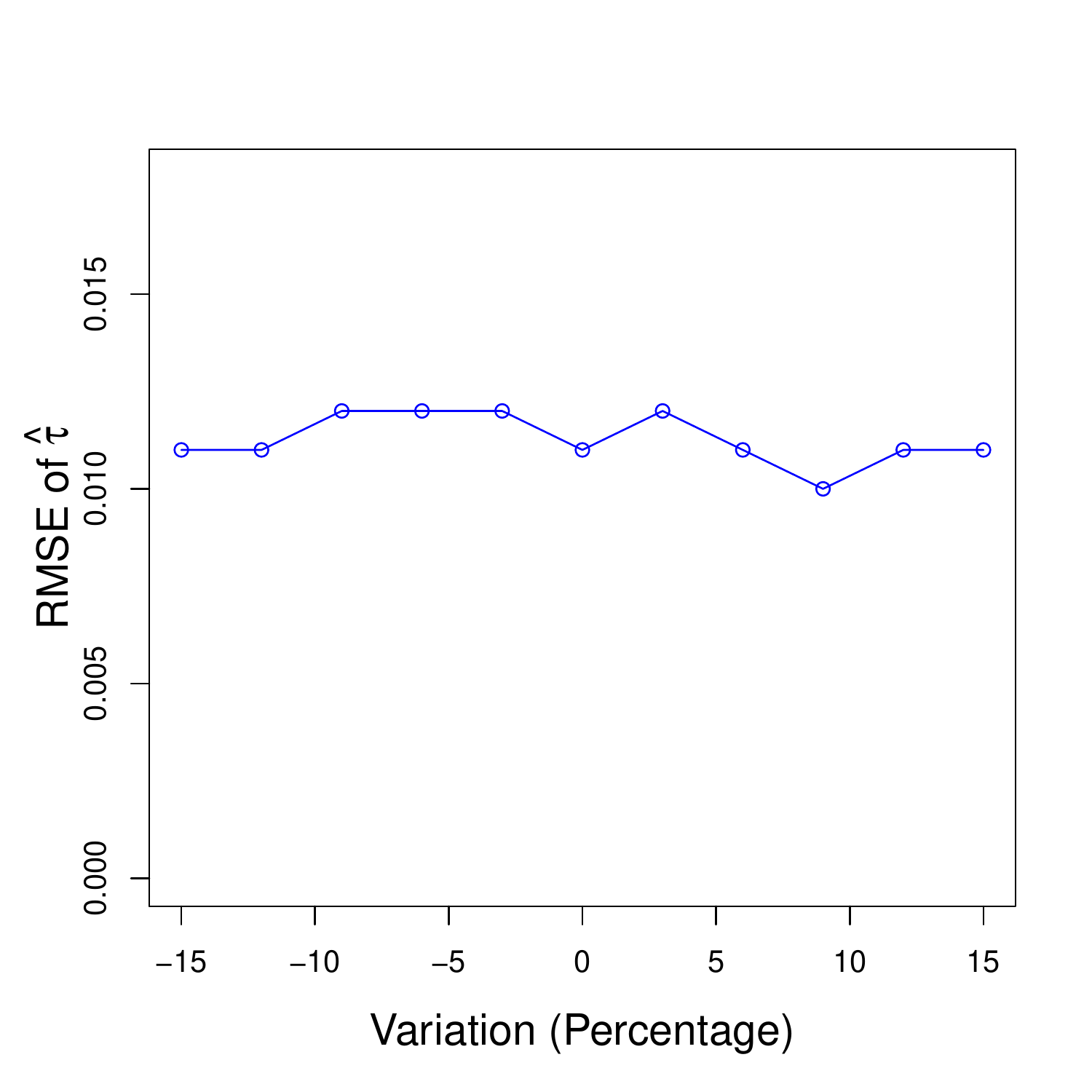}
\end{minipage}\hfill
\begin{minipage}{.33\textwidth}
\centering
\includegraphics[width=\linewidth]{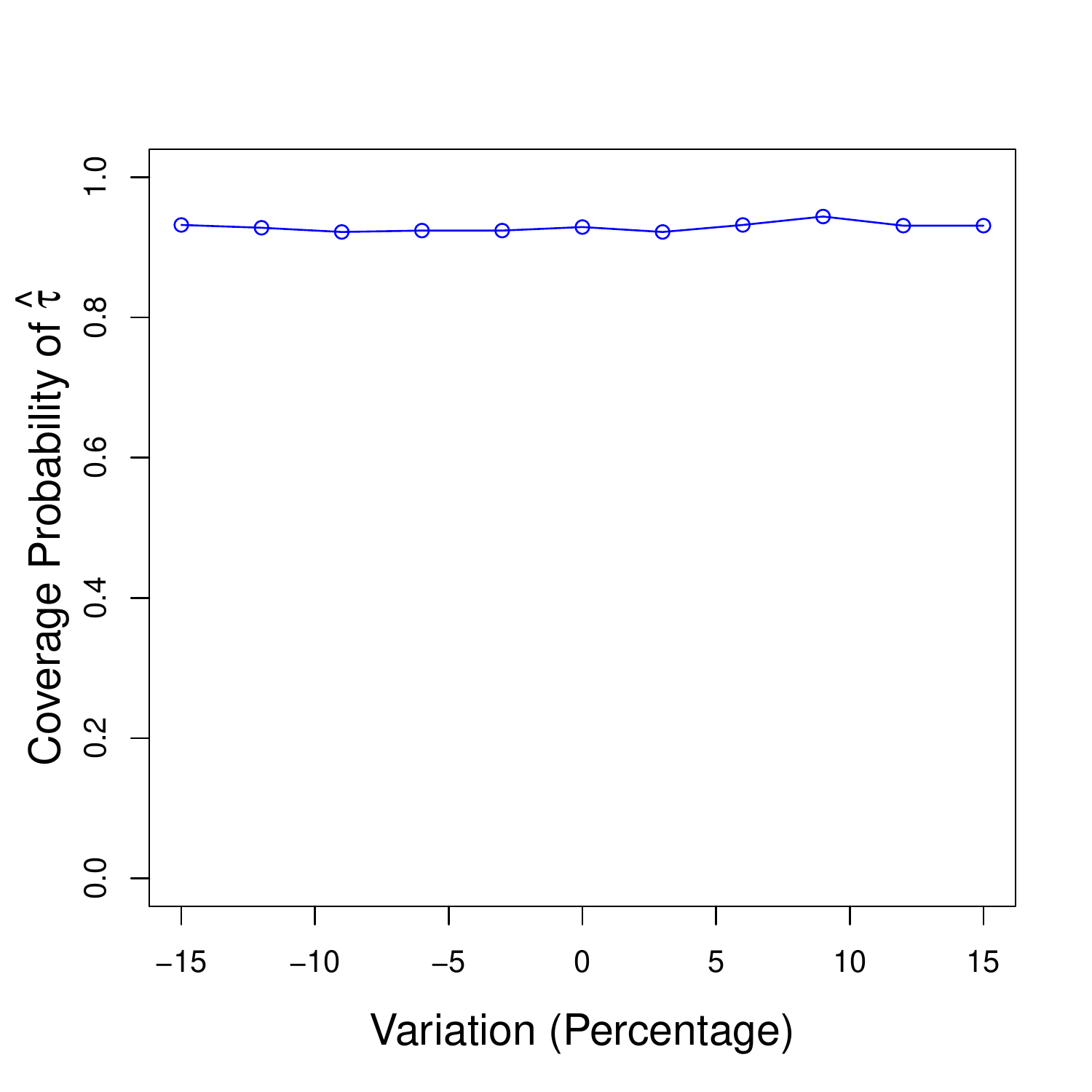}
\end{minipage}\hfill
\begin{minipage}{.33\textwidth}
\centering
\includegraphics[width=\linewidth]{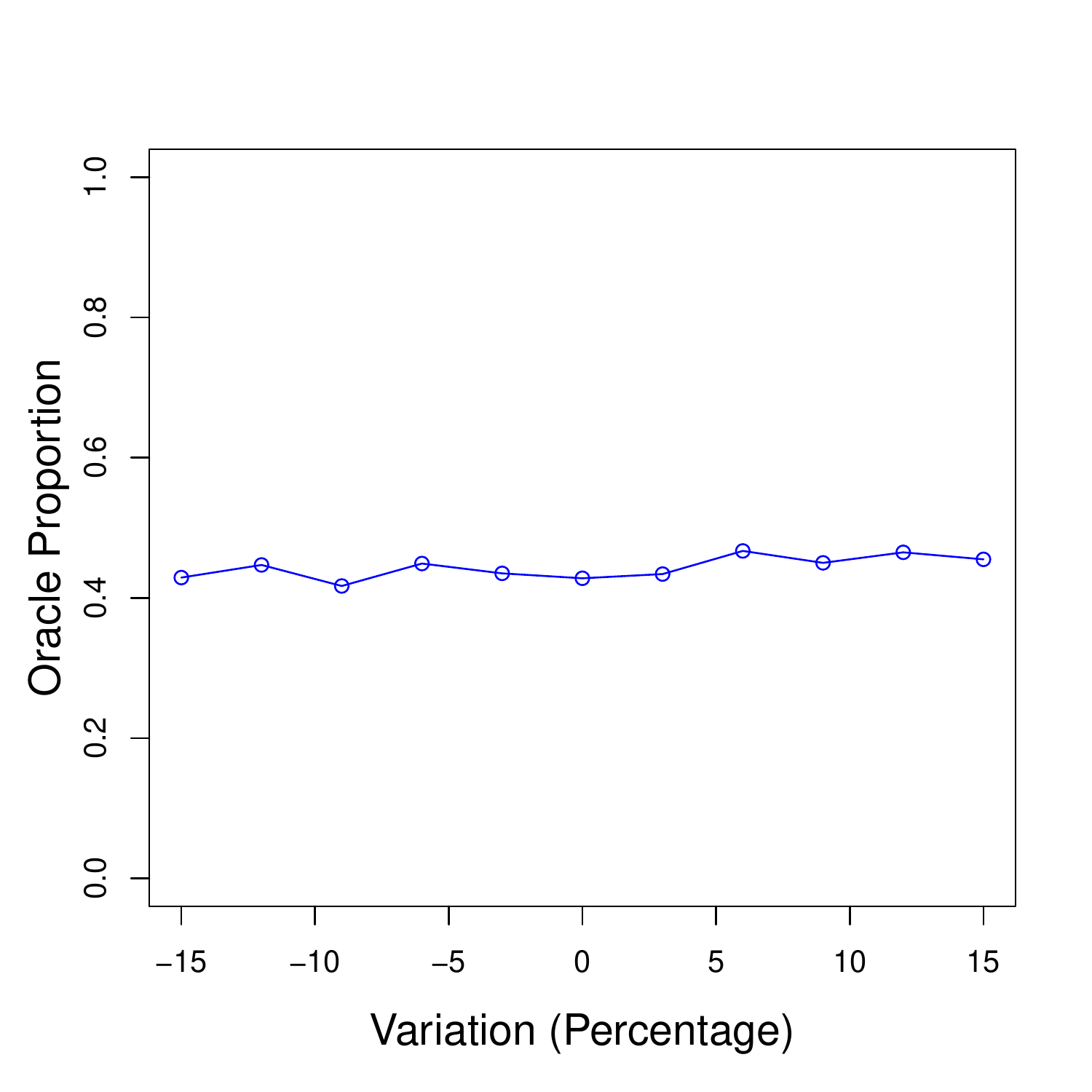}
\end{minipage}
\end{figure}

\newpage

%
%
%

\begin{figure}[tbph]
\centering
\caption{Sensitivity Analysis of $c_1$: Step 1}\label{fig:sen-c1-s1}
\begin{minipage}{.33\textwidth}
\centering
\includegraphics[width=\linewidth]{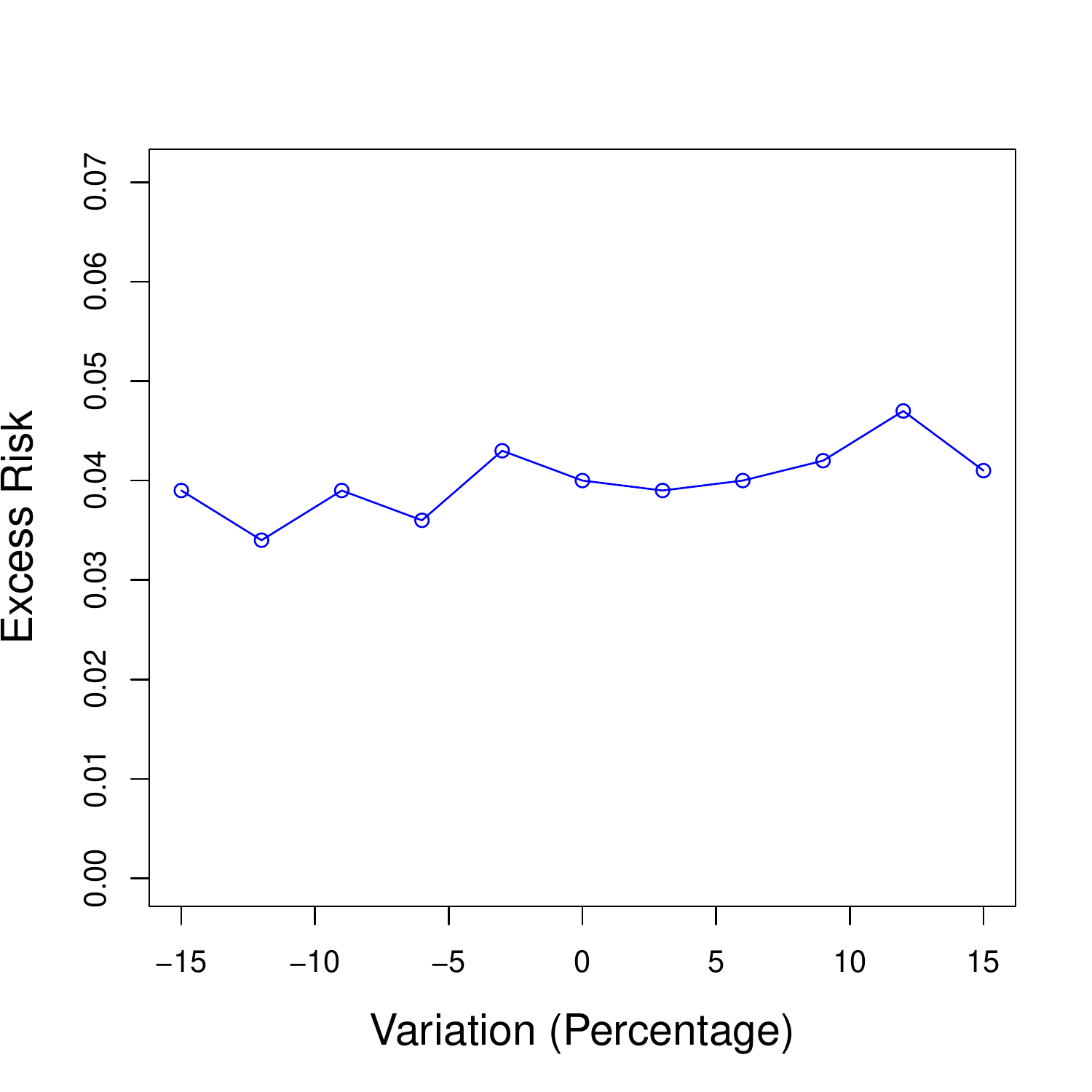}
\end{minipage}\hfill
\begin{minipage}{.33\textwidth}
\centering
\includegraphics[width=\linewidth]{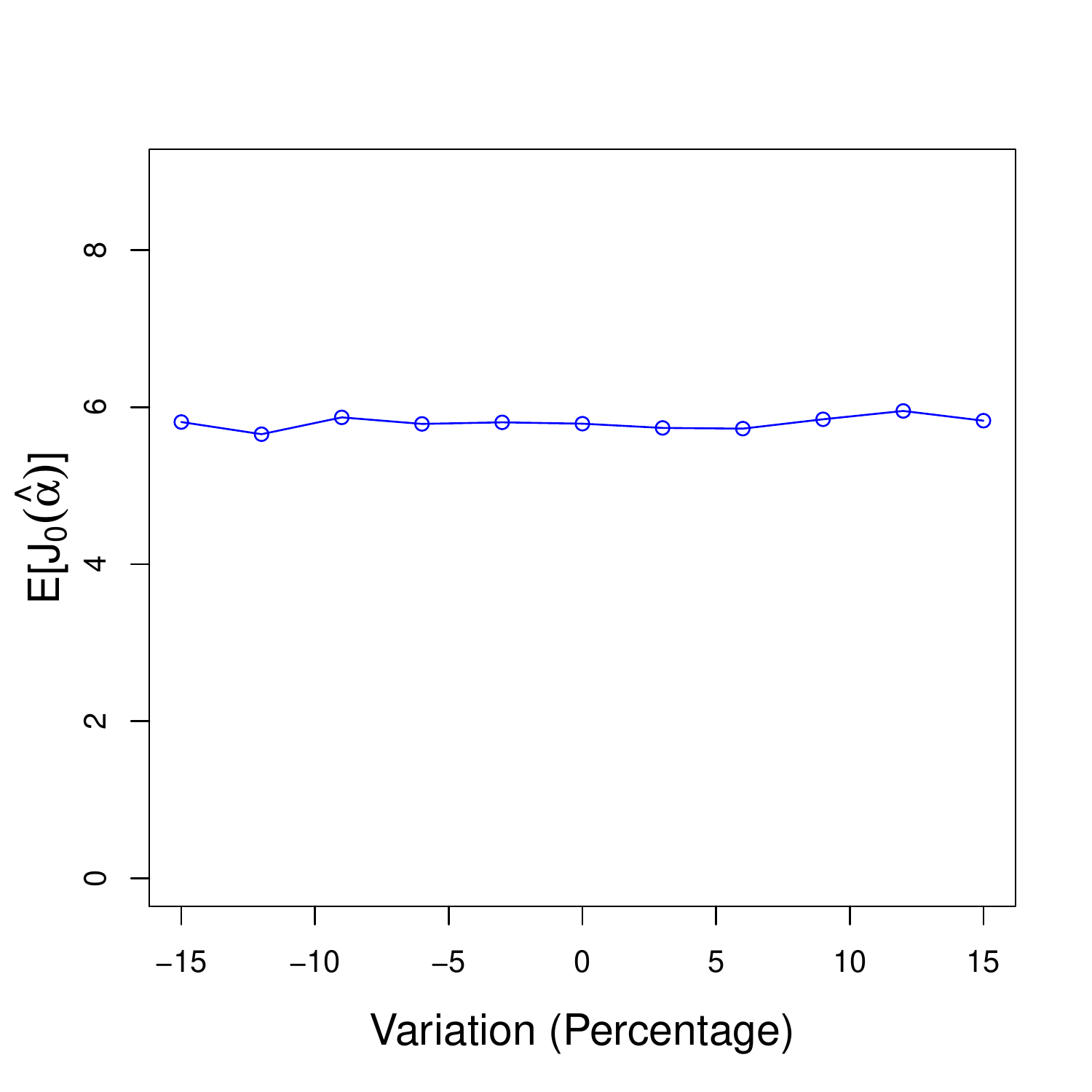}
\end{minipage}\hfill
\begin{minipage}{.33\textwidth}
\centering
\includegraphics[width=\linewidth]{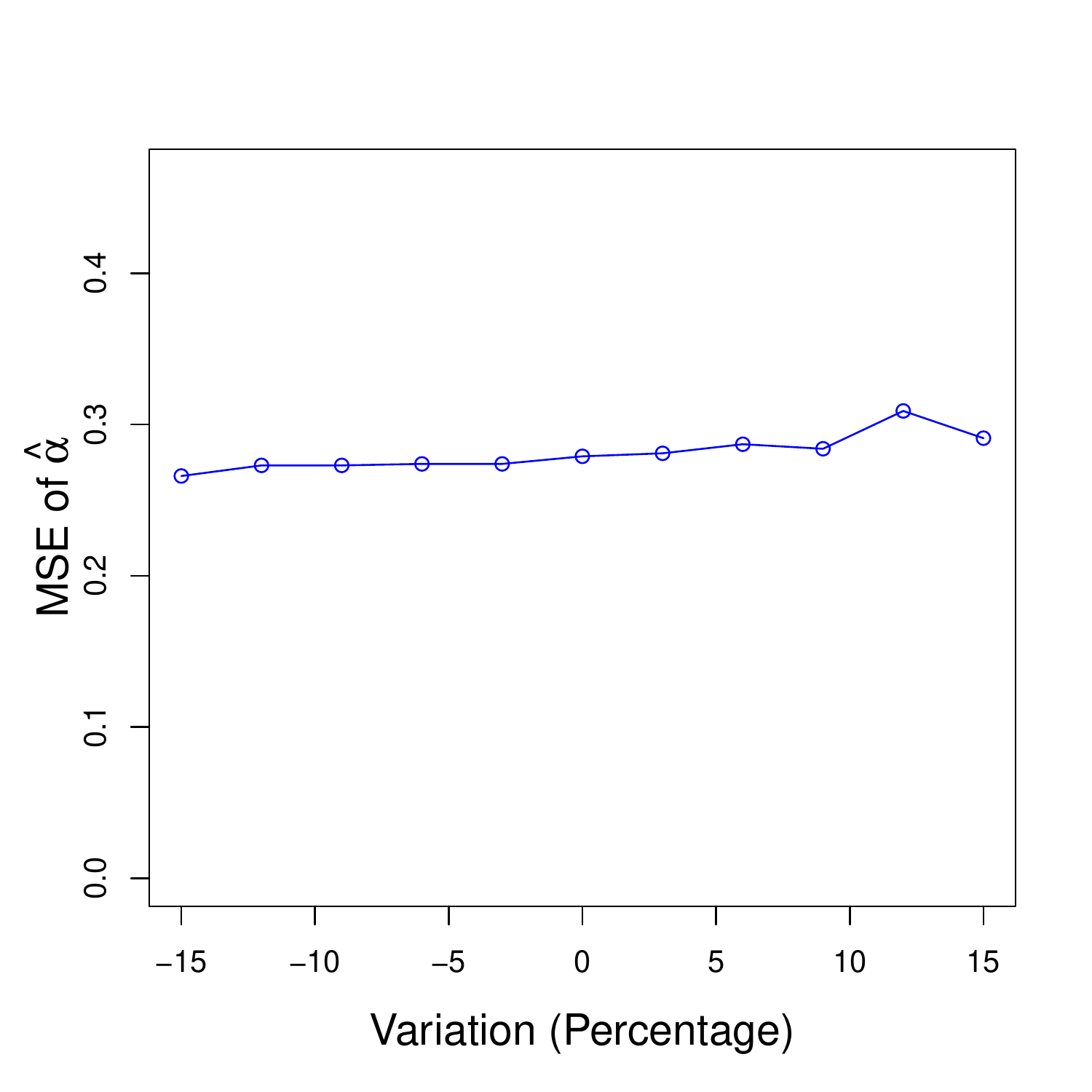}
\end{minipage}
\begin{minipage}{.33\textwidth}
\centering
\includegraphics[width=\linewidth]{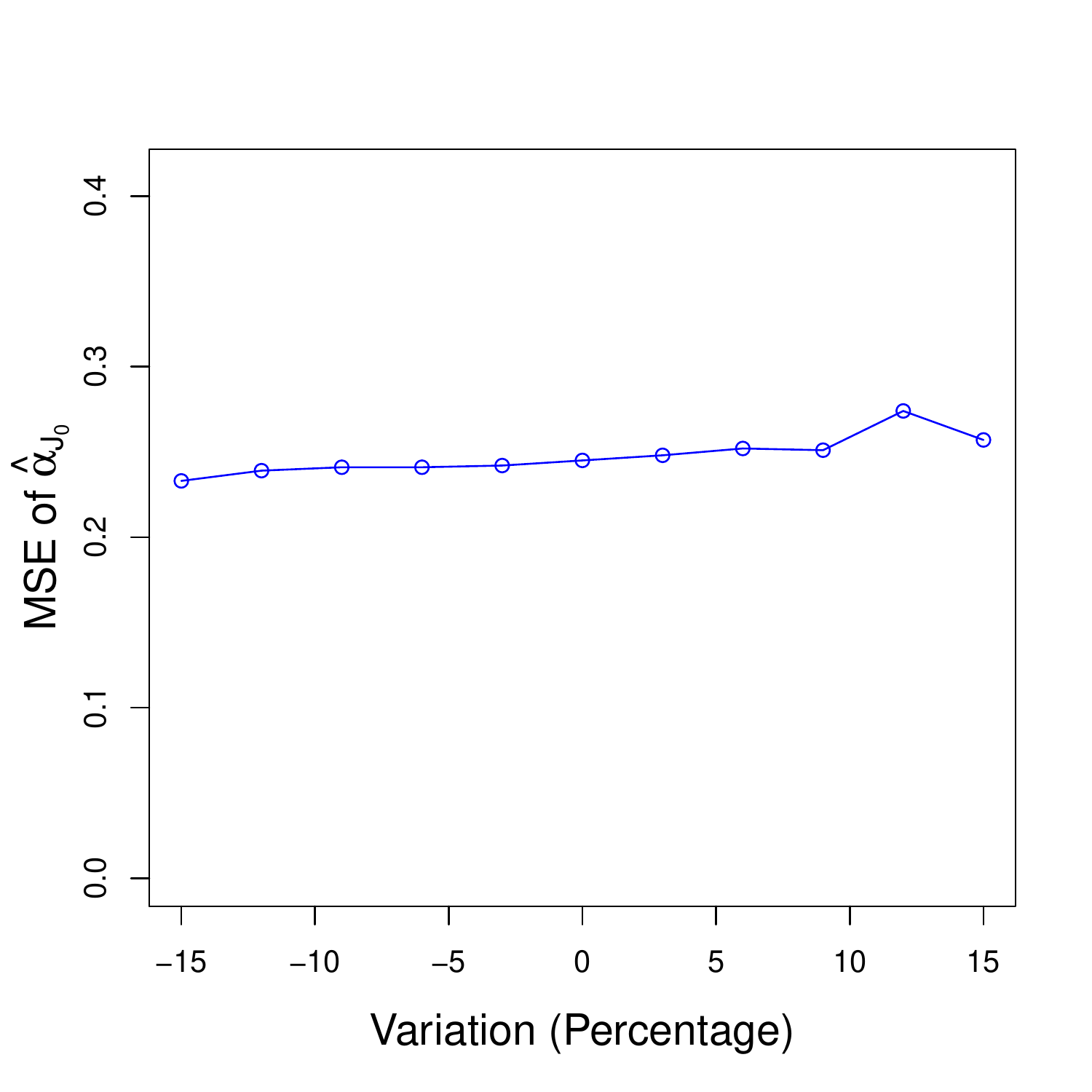}
\end{minipage}\hfill
\begin{minipage}{.33\textwidth}
\centering
\includegraphics[width=\linewidth]{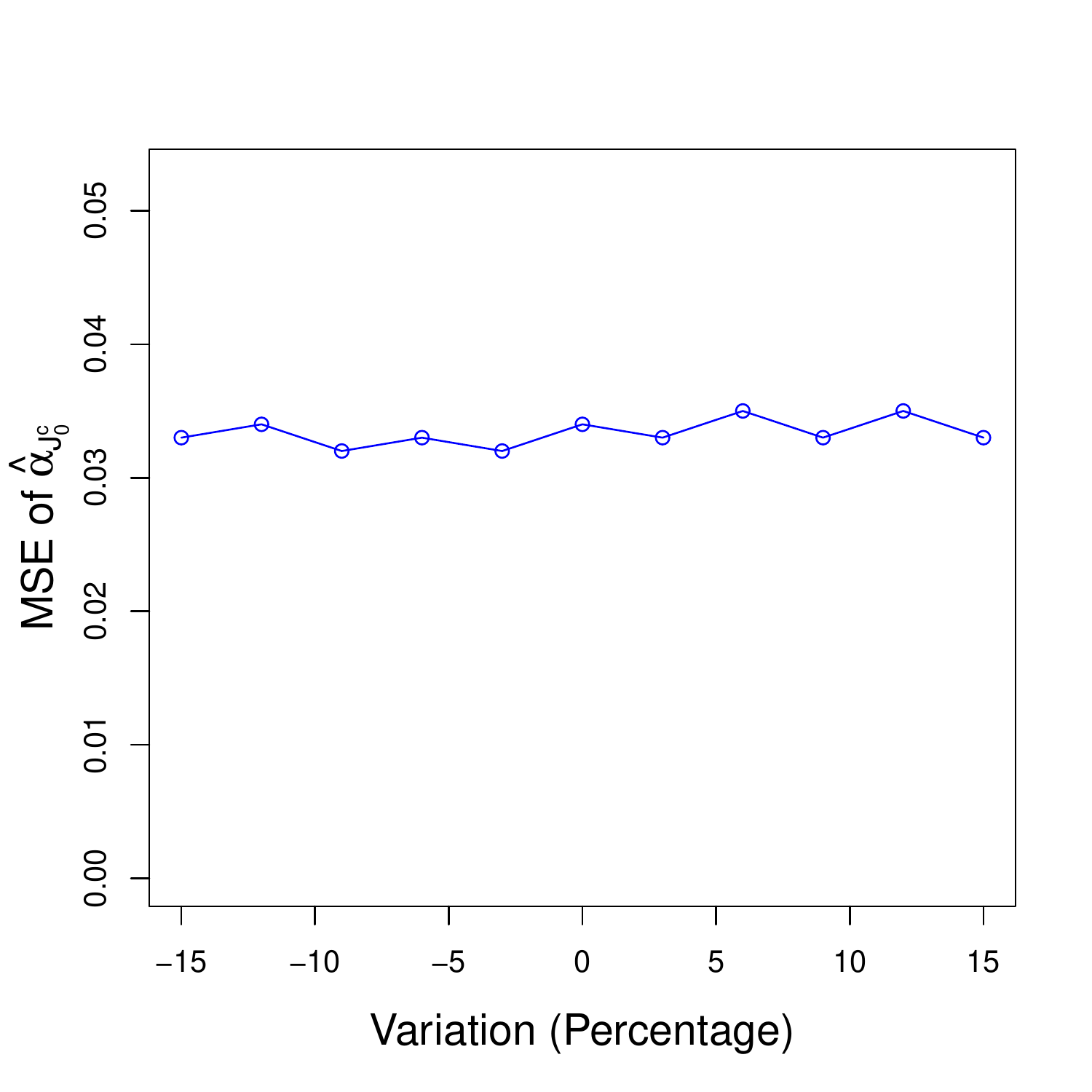}
\end{minipage}\hfill
\begin{minipage}{.33\textwidth}
\centering
\includegraphics[width=\linewidth]{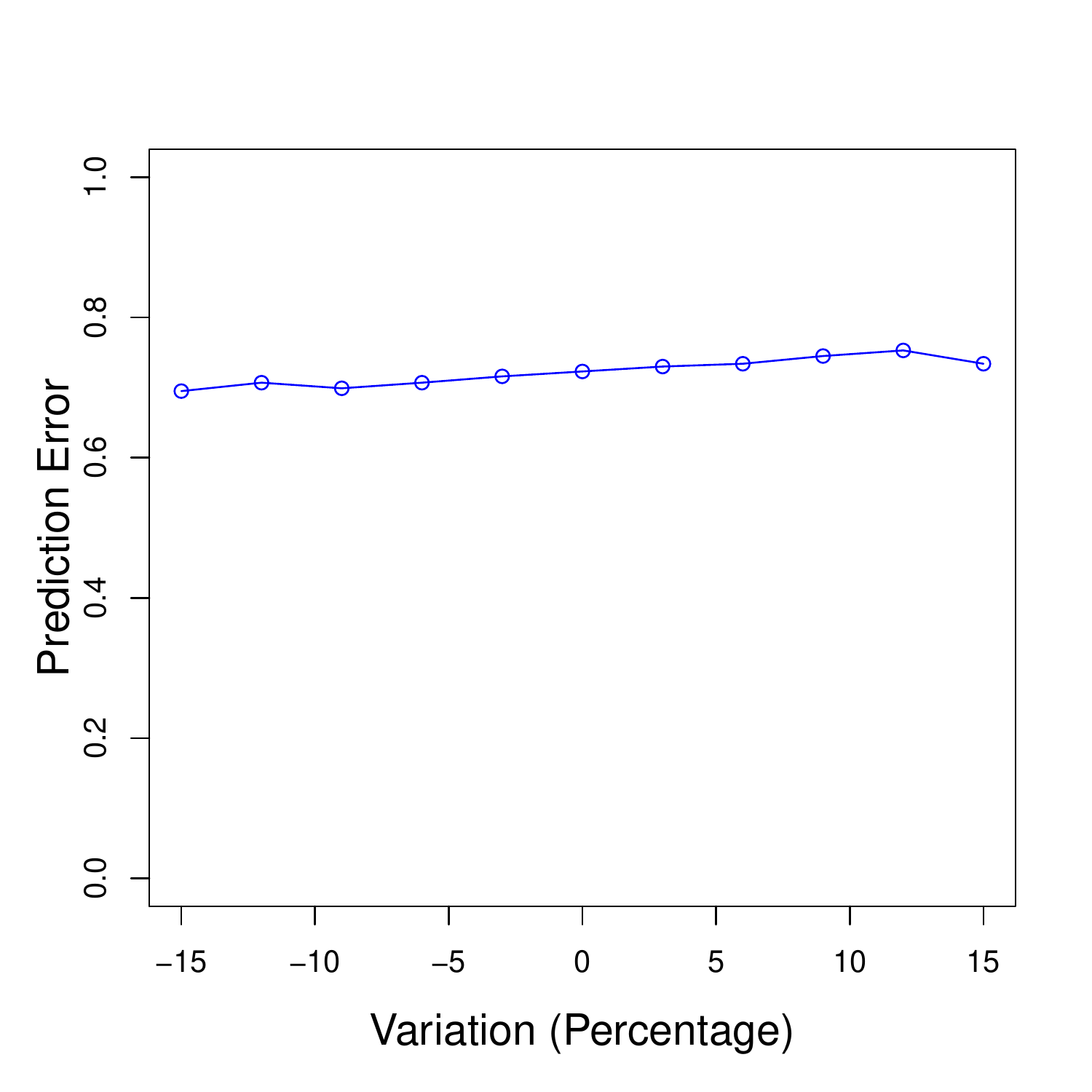}
\end{minipage}
\begin{minipage}{.33\textwidth}
\centering
\includegraphics[width=\linewidth]{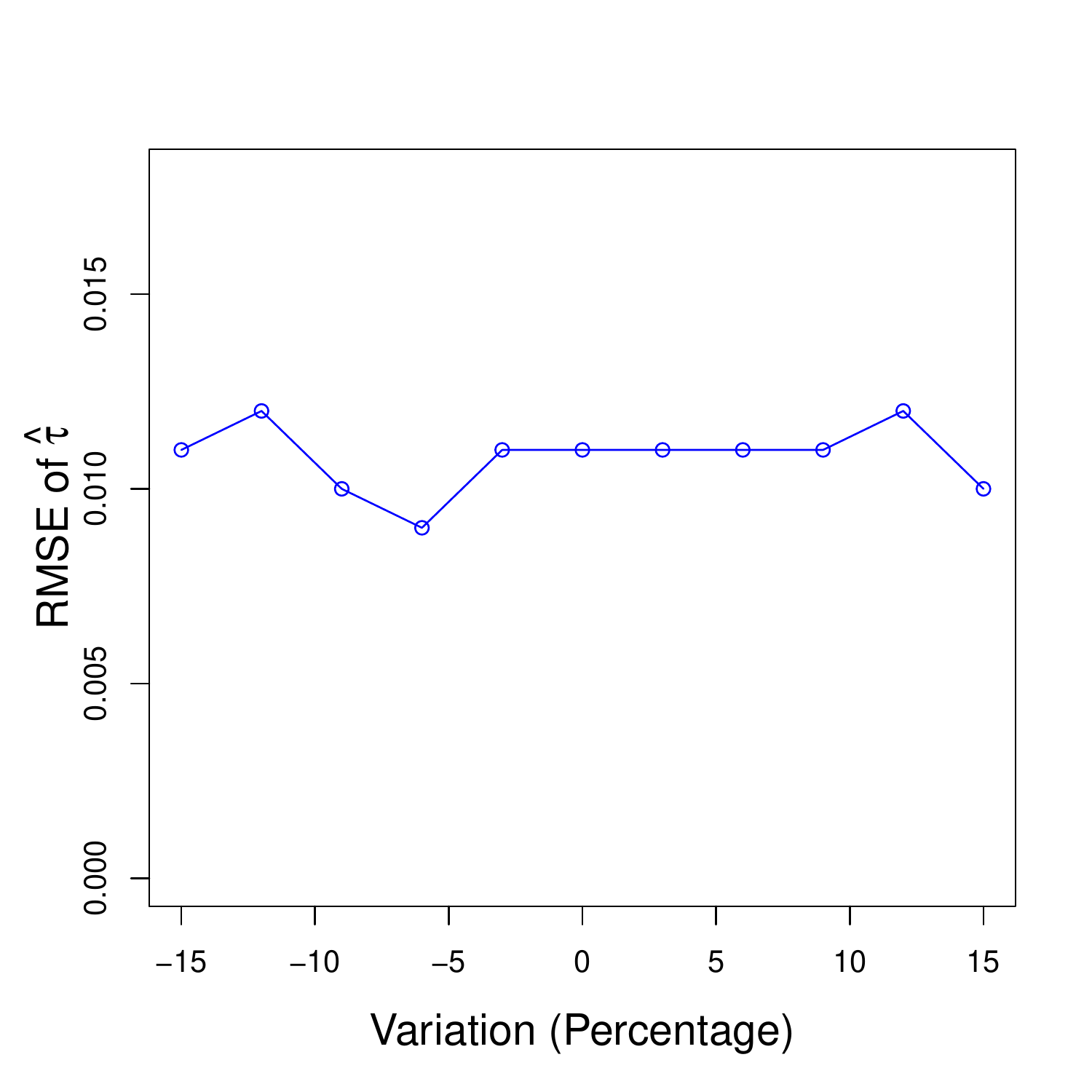}
\end{minipage}\hfill
\begin{minipage}{.33\textwidth}
\centering
\includegraphics[width=\linewidth]{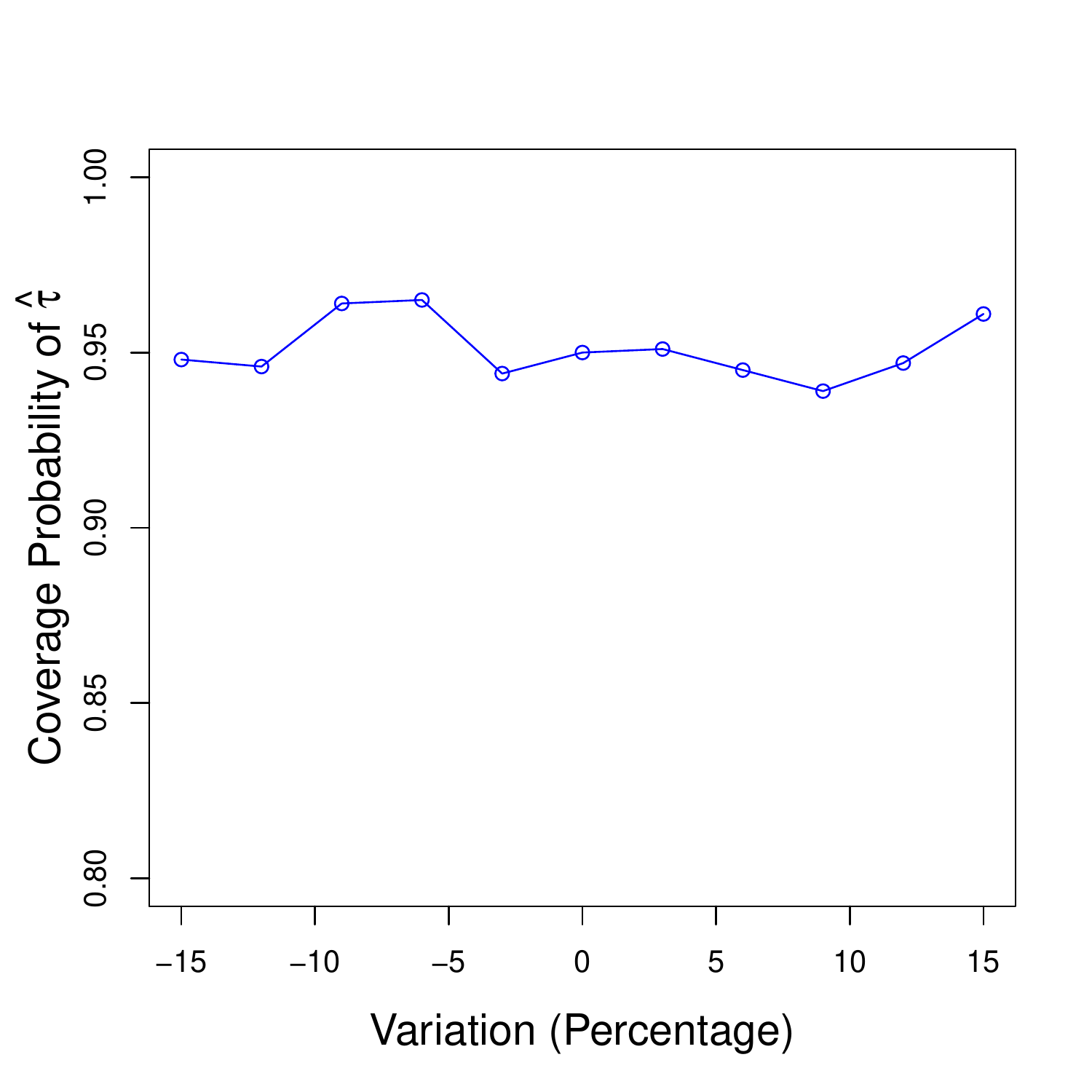}
\end{minipage}\hfill
\begin{minipage}{.33\textwidth}
\centering
\includegraphics[width=\linewidth]{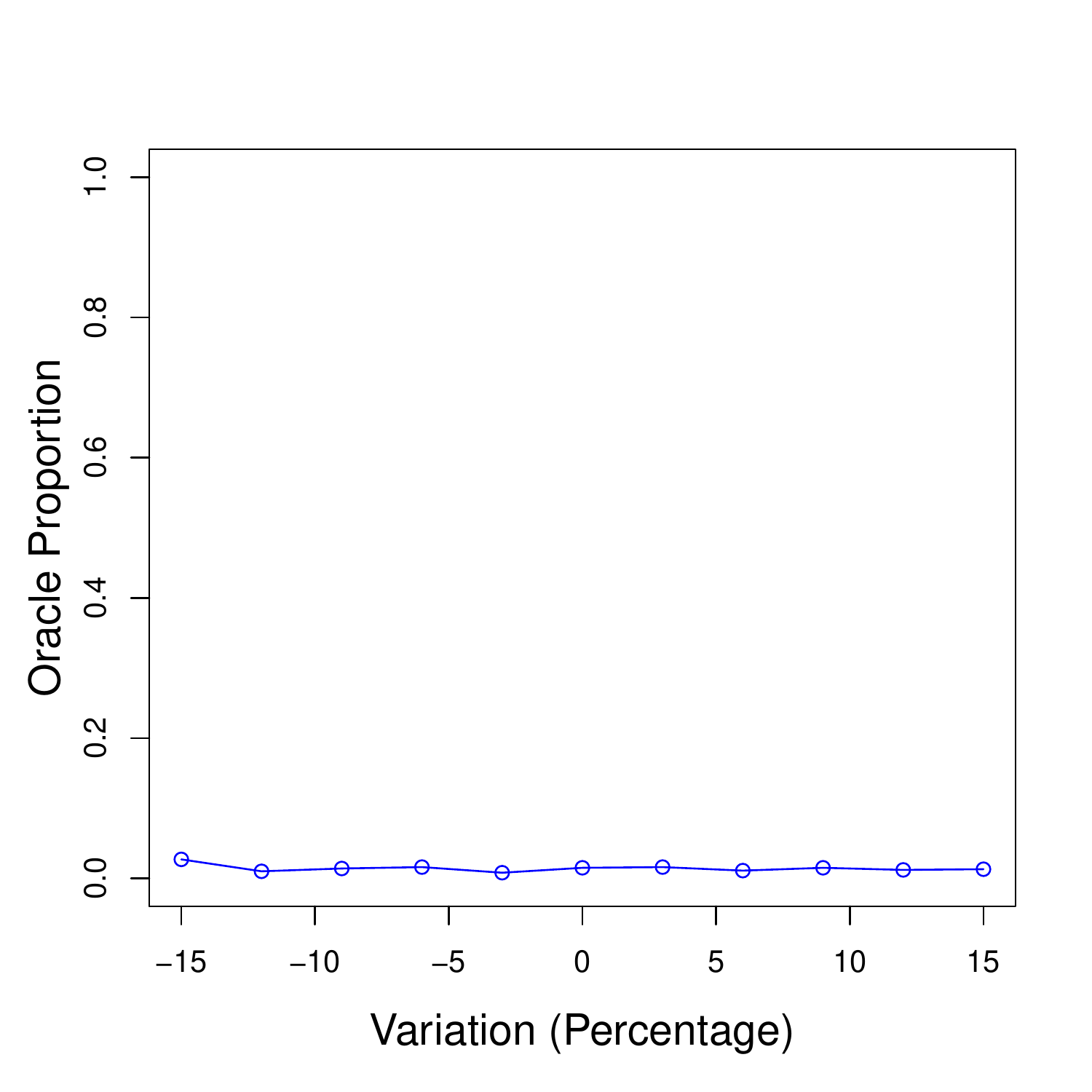}
\end{minipage}
\end{figure}

\newpage

\begin{figure}[tbph]
\centering
\caption{Sensitivity Analysis of $c_1$: Step 2}\label{fig:sen-c1-s2}
\begin{minipage}{.49\textwidth}
\centering
\includegraphics[width=\linewidth]{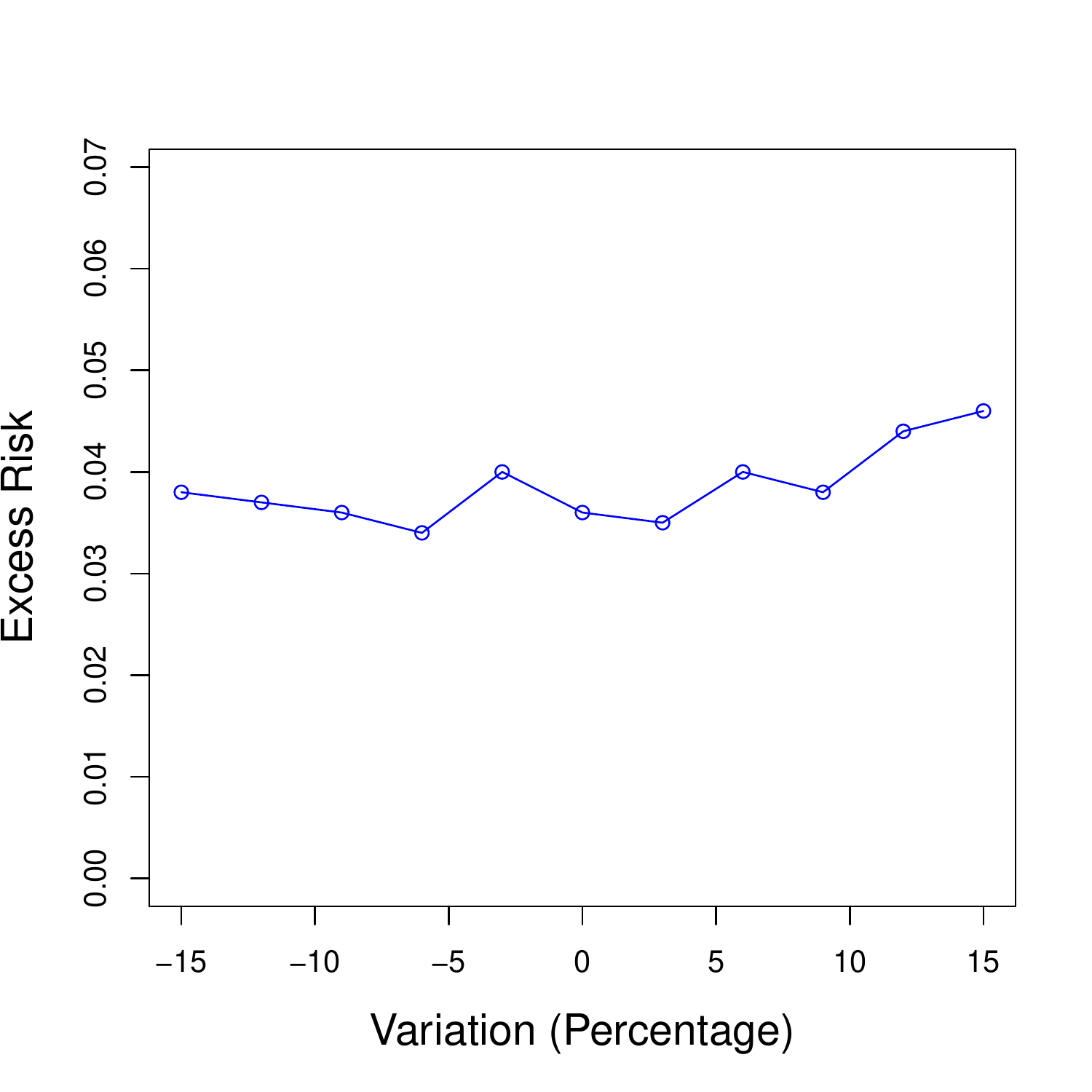}
\end{minipage}\hfill
\begin{minipage}{.49\textwidth}
\centering
\includegraphics[width=\linewidth]{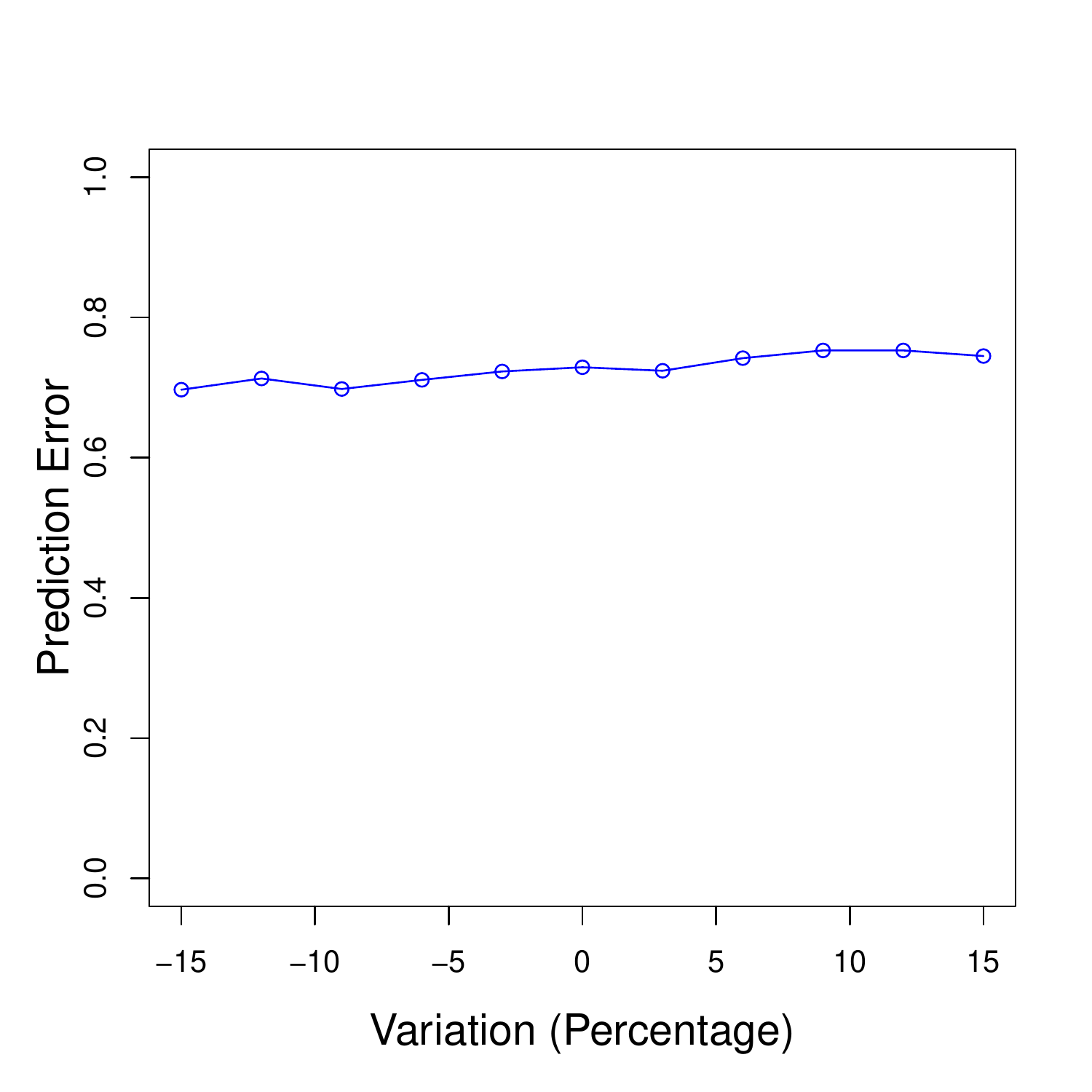}
\end{minipage}\hfill
\begin{minipage}{.49\textwidth}
\centering
\includegraphics[width=\linewidth]{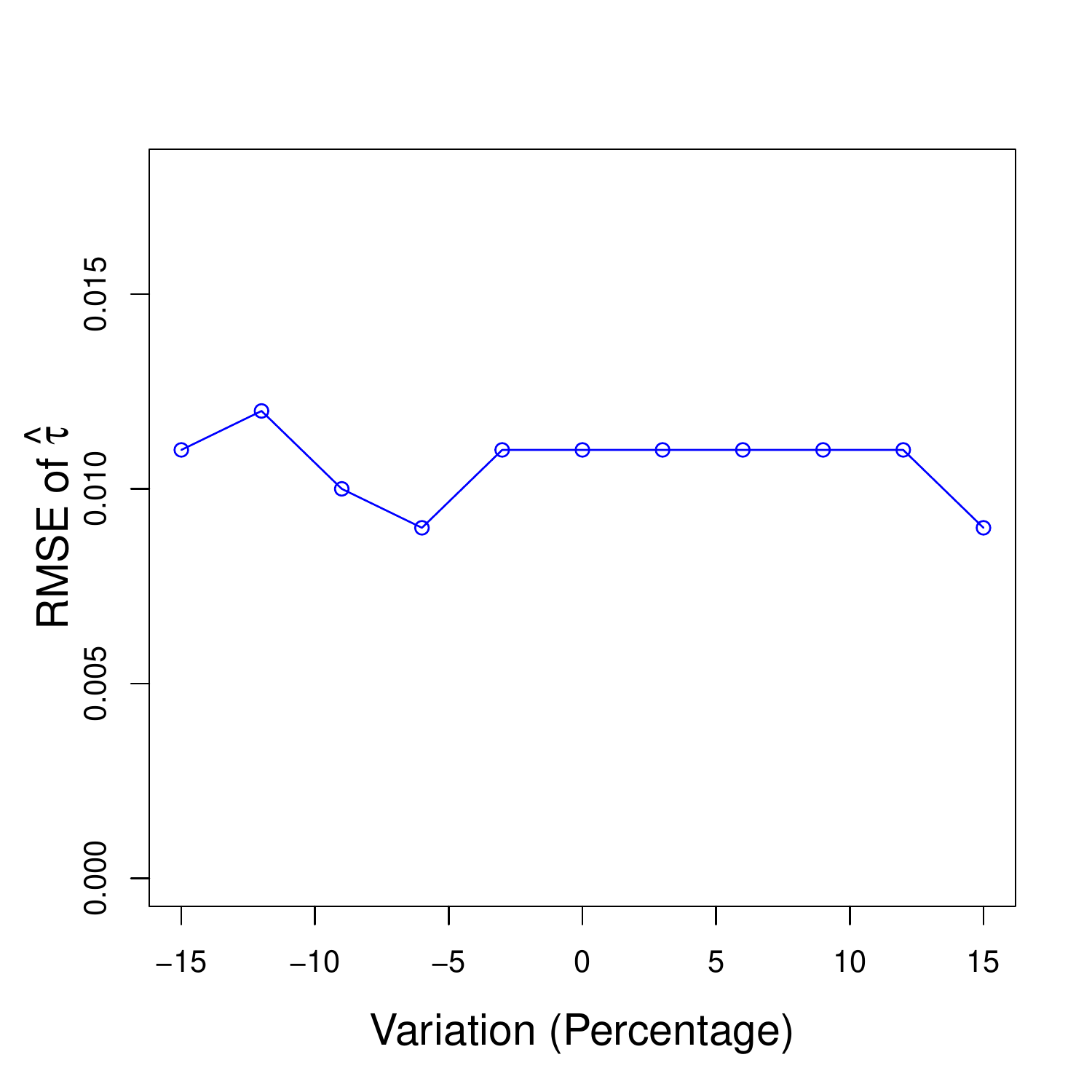}
\end{minipage}
\begin{minipage}{.49\textwidth}
\centering
\includegraphics[width=\linewidth]{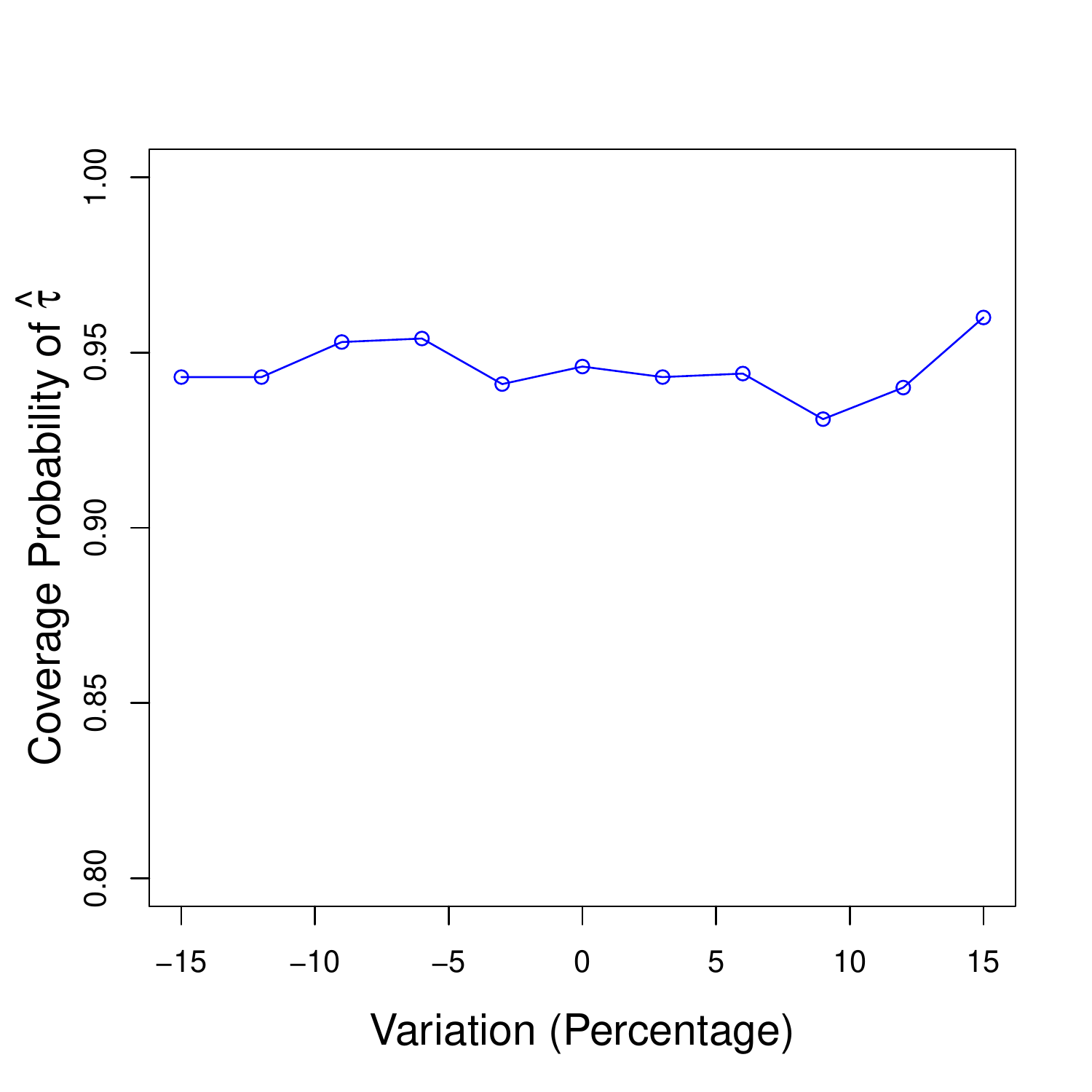}
\end{minipage}\hfill
\end{figure}

\newpage

\begin{figure}[tbph]
\centering
\caption{Sensitivity Analysis of $c_1$: Step 3a}\label{fig:sen-c1-s3a}
\begin{minipage}{.33\textwidth}
\centering
\includegraphics[width=\linewidth]{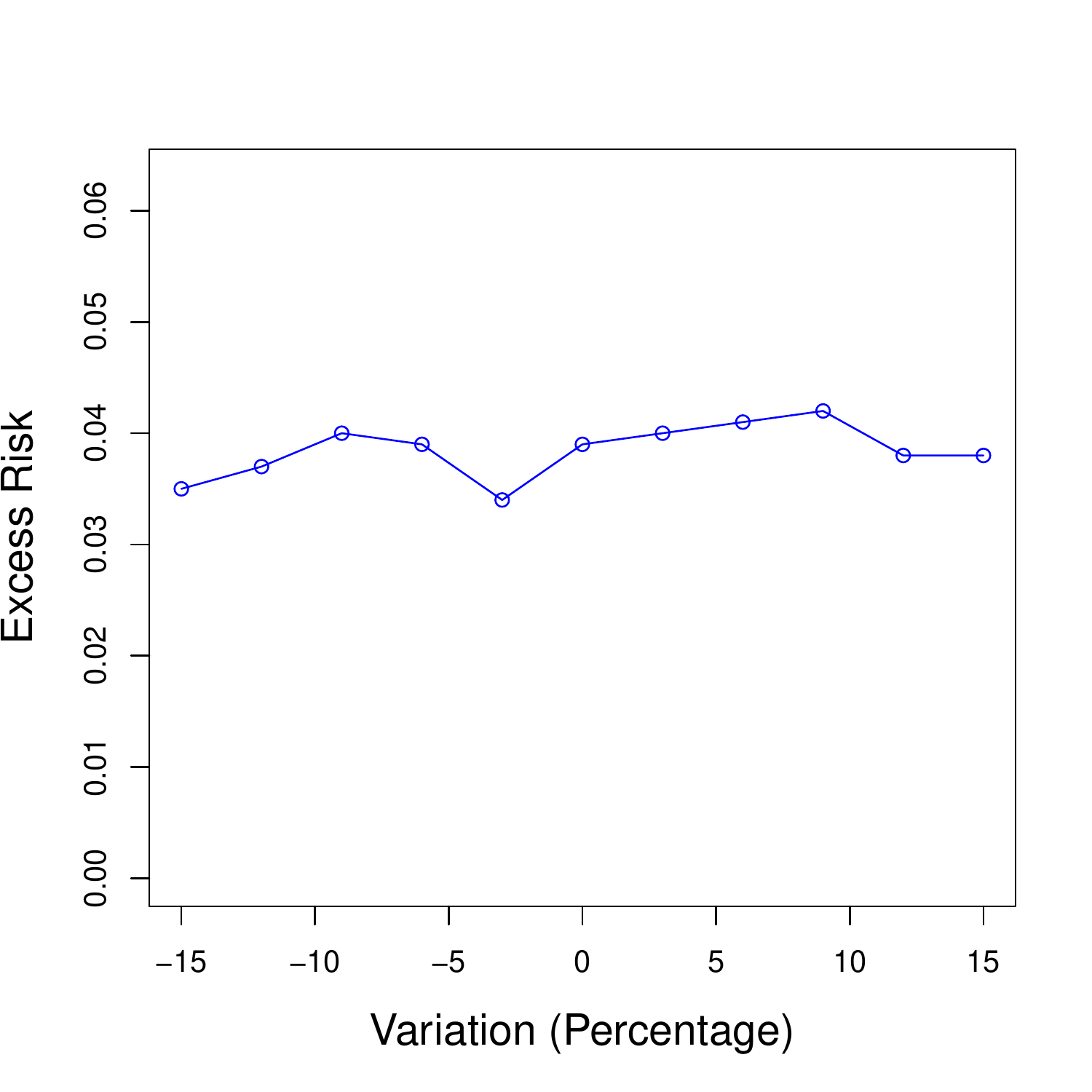}
\end{minipage}\hfill
\begin{minipage}{.33\textwidth}
\centering
\includegraphics[width=\linewidth]{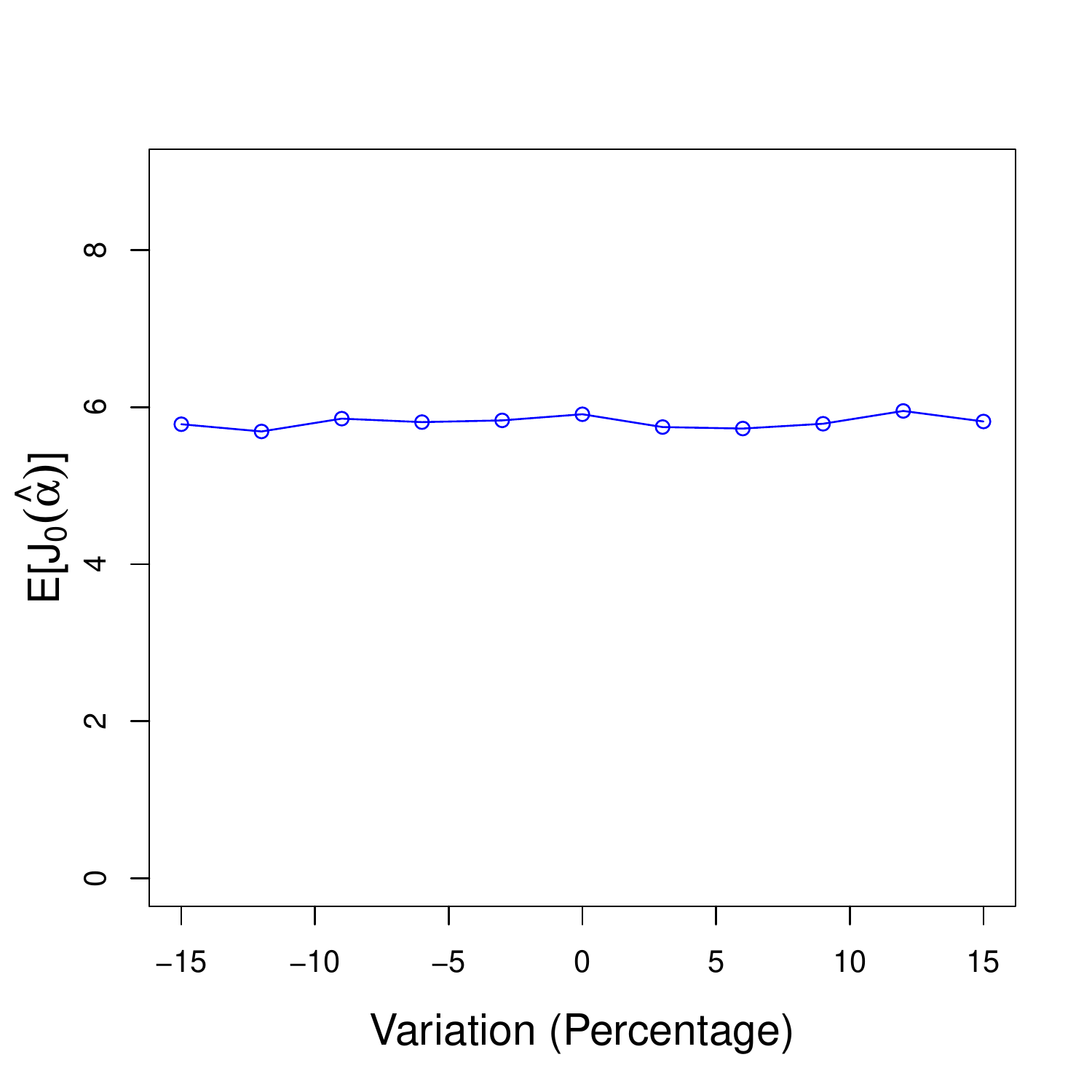}
\end{minipage}\hfill
\begin{minipage}{.33\textwidth}
\centering
\includegraphics[width=\linewidth]{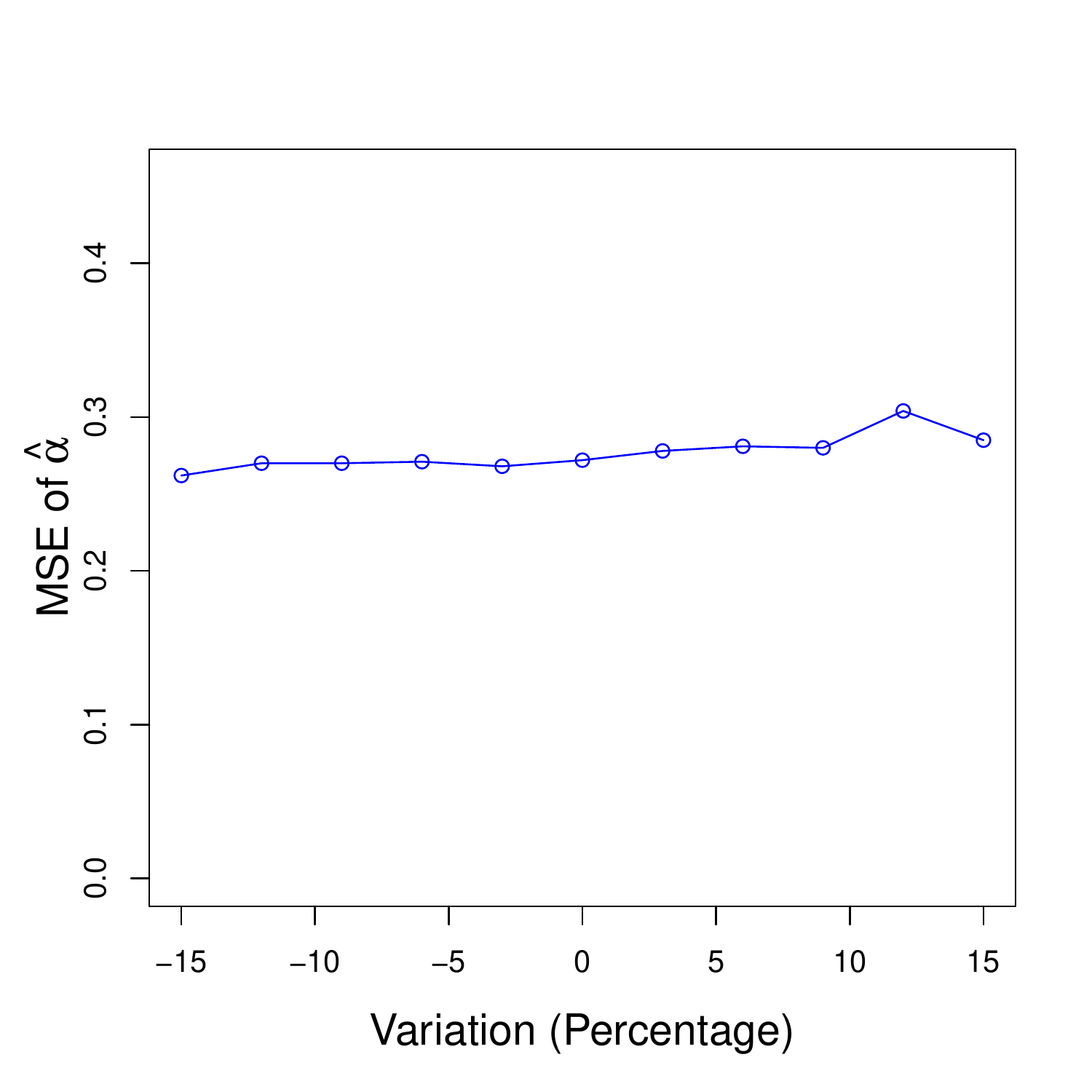}
\end{minipage}
\begin{minipage}{.33\textwidth}
\centering
\includegraphics[width=\linewidth]{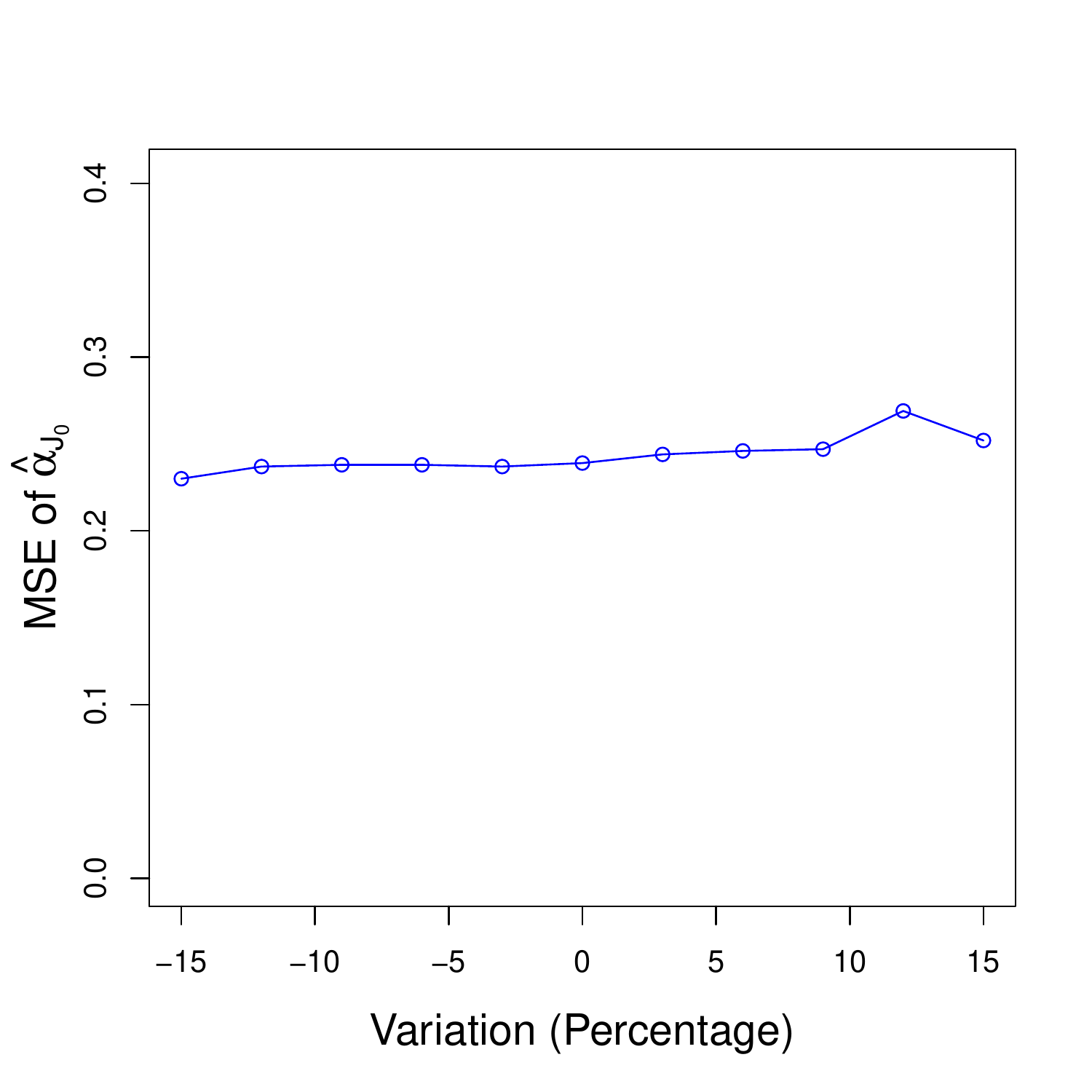}
\end{minipage}\hfill
\begin{minipage}{.33\textwidth}
\centering
\includegraphics[width=\linewidth]{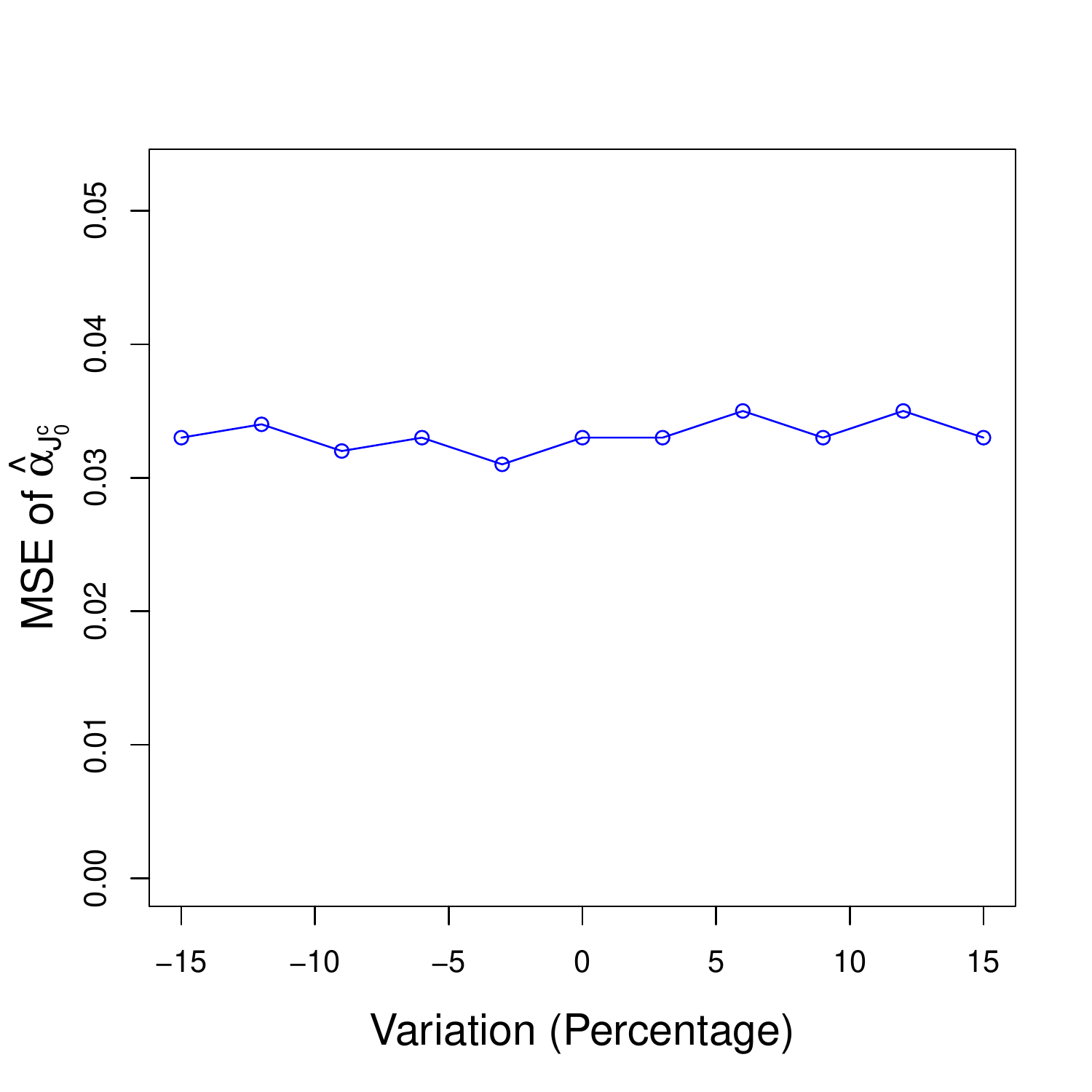}
\end{minipage}\hfill
\begin{minipage}{.33\textwidth}
\centering
\includegraphics[width=\linewidth]{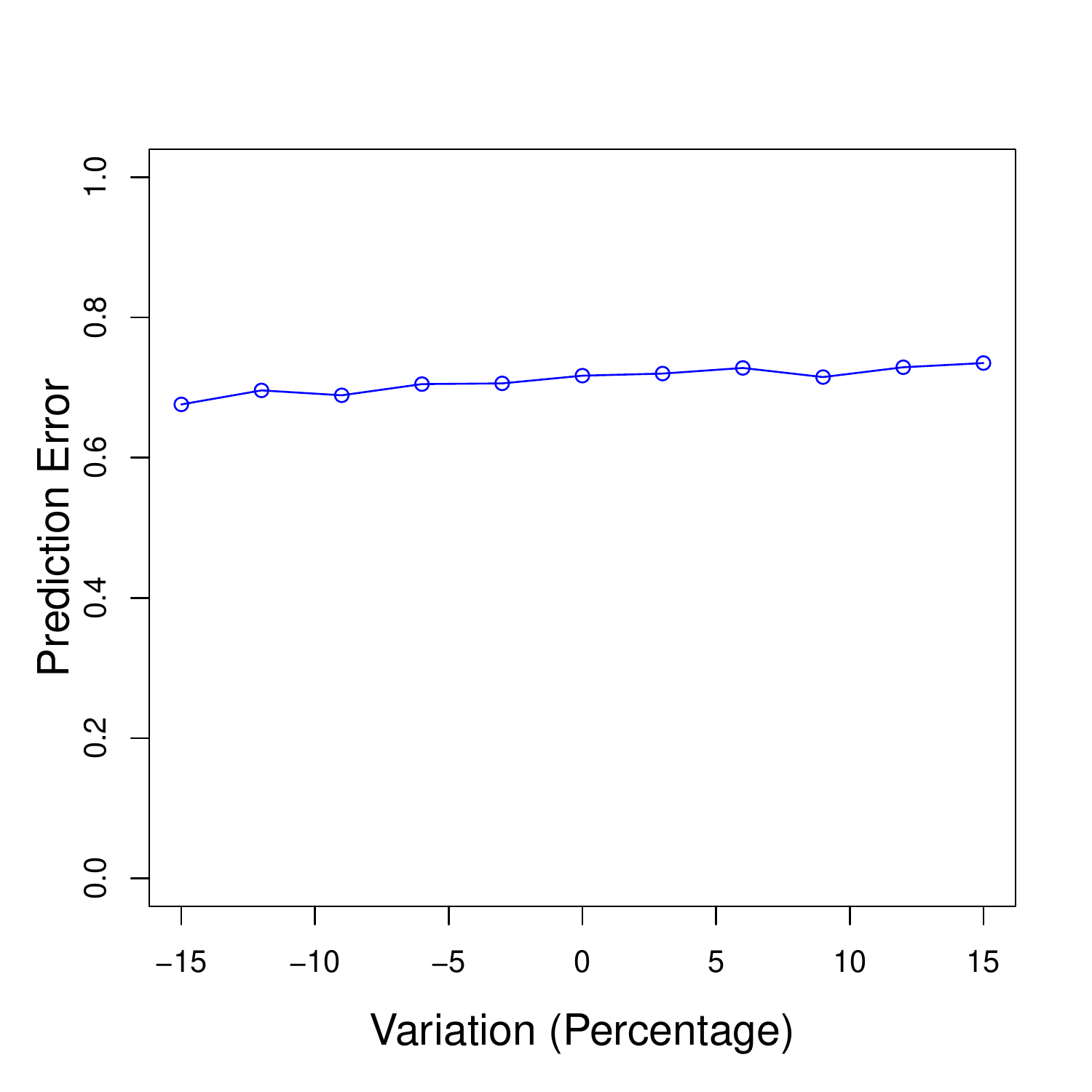}
\end{minipage}
\begin{minipage}{.33\textwidth}
\centering
\includegraphics[width=\linewidth]{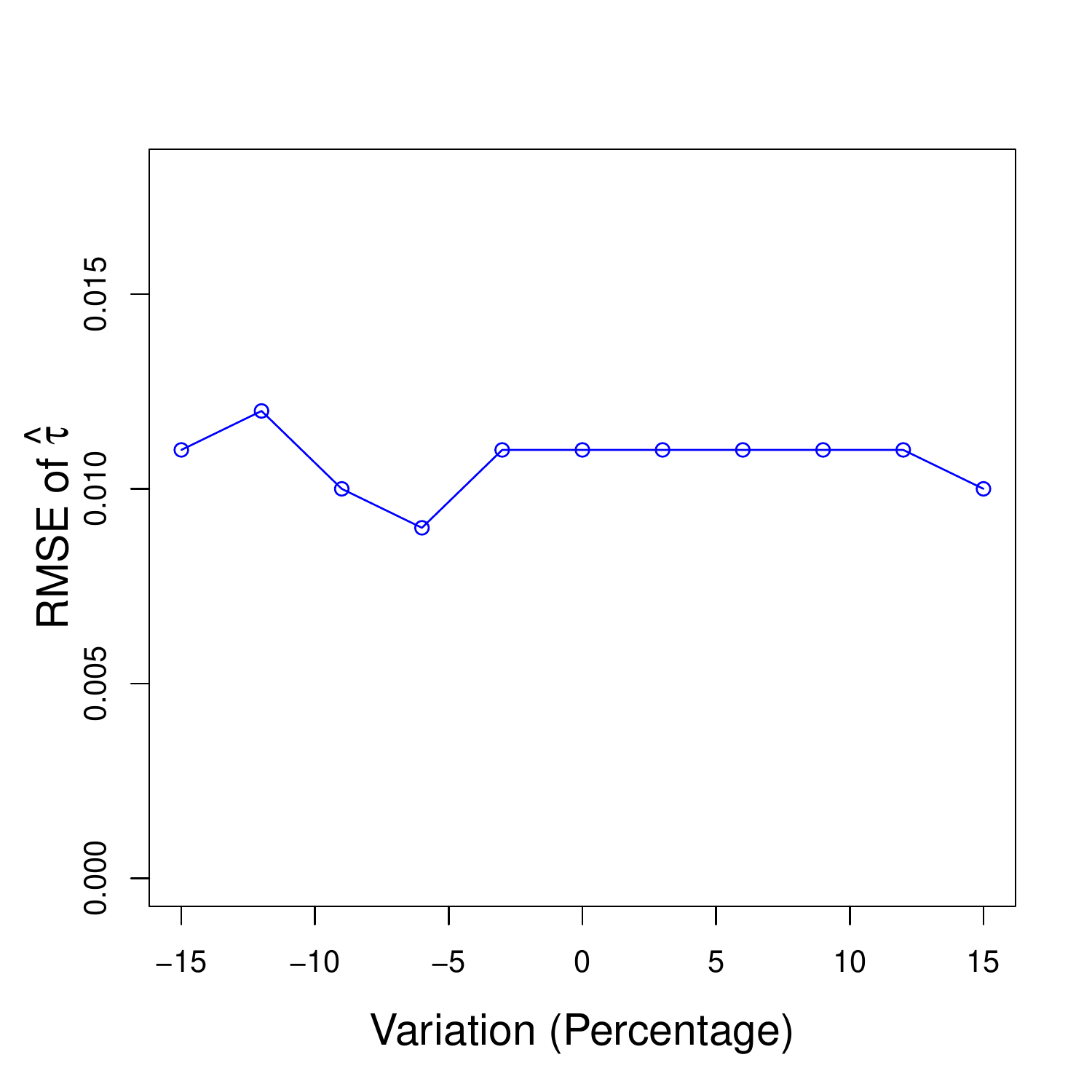}
\end{minipage}\hfill
\begin{minipage}{.33\textwidth}
\centering
\includegraphics[width=\linewidth]{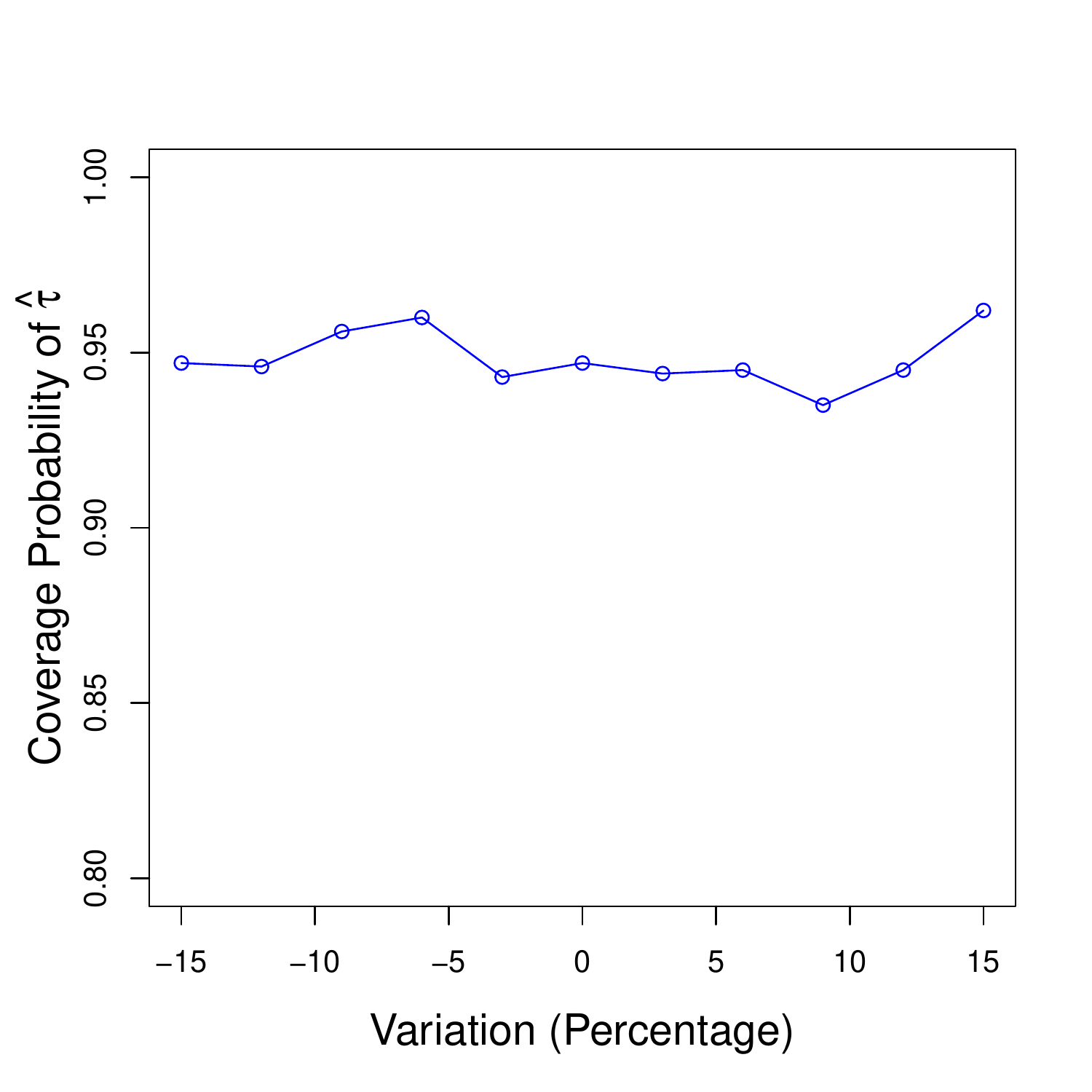}
\end{minipage}\hfill
\begin{minipage}{.33\textwidth}
\centering
\includegraphics[width=\linewidth]{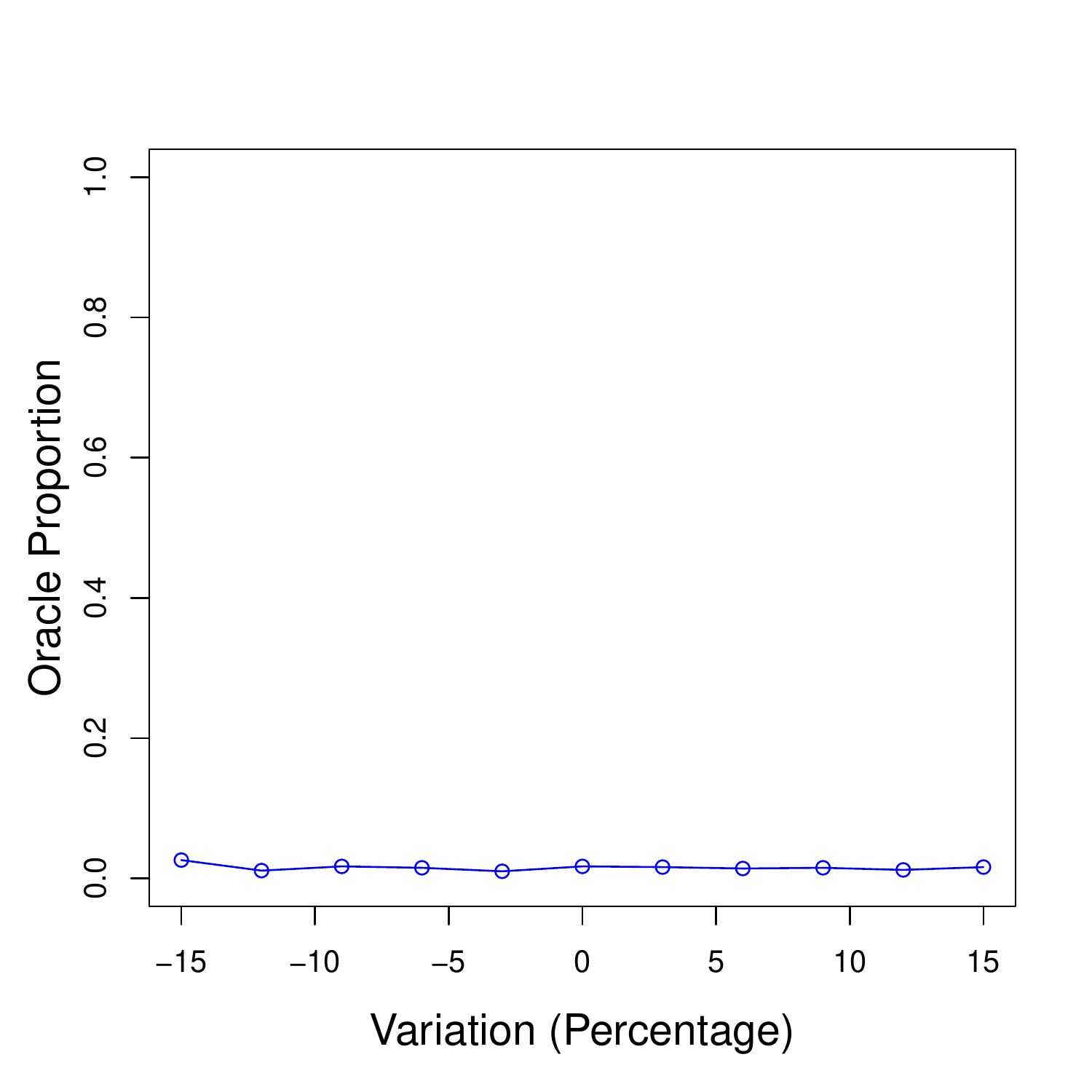}
\end{minipage}
\end{figure}

\newpage

\begin{figure}[tbph]
\centering
\caption{Sensitivity Analysis of $c_1$: Step 3b}\label{fig:sen-c1-s3b}
\begin{minipage}{.33\textwidth}
\centering
\includegraphics[width=\linewidth]{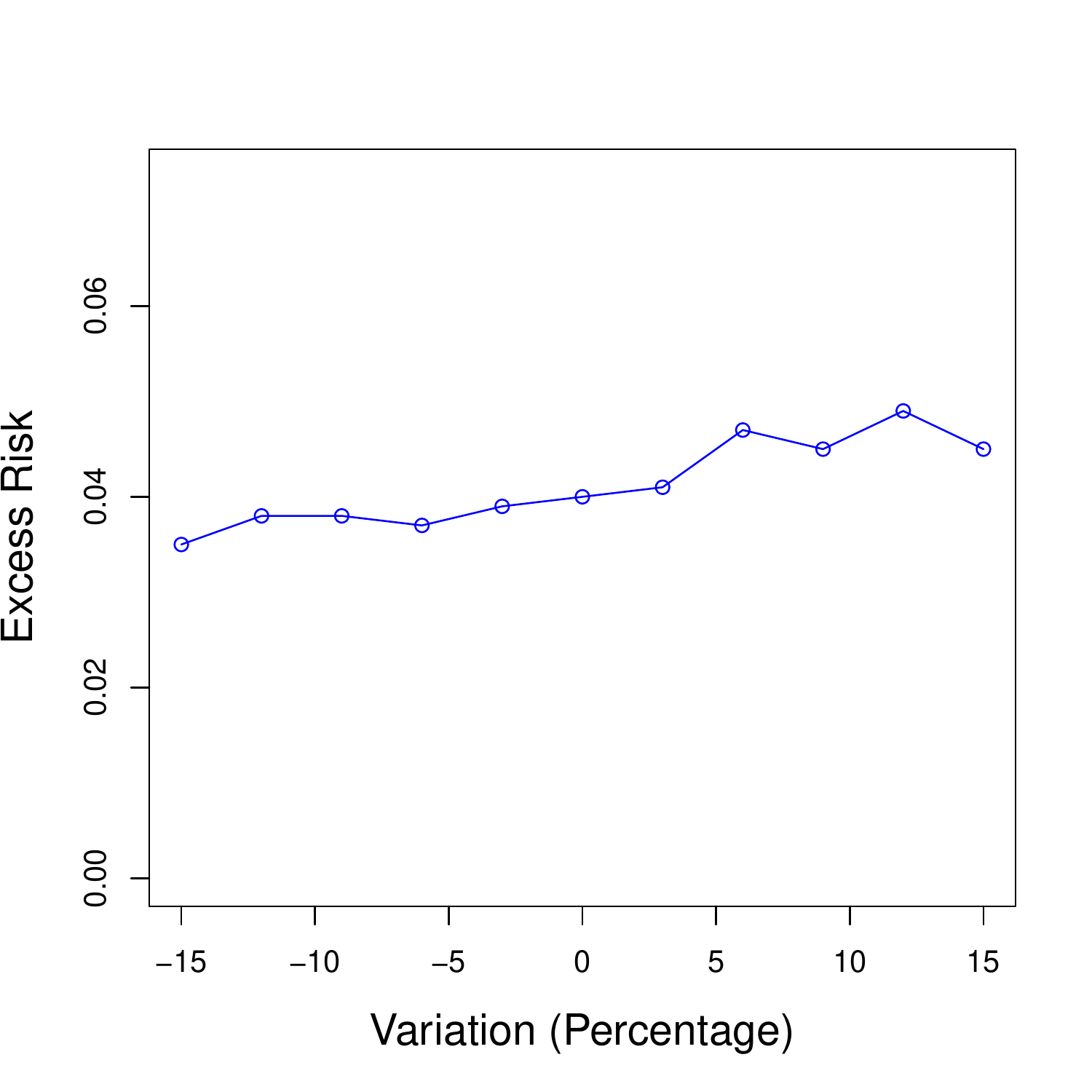}
\end{minipage}\hfill
\begin{minipage}{.33\textwidth}
\centering
\includegraphics[width=\linewidth]{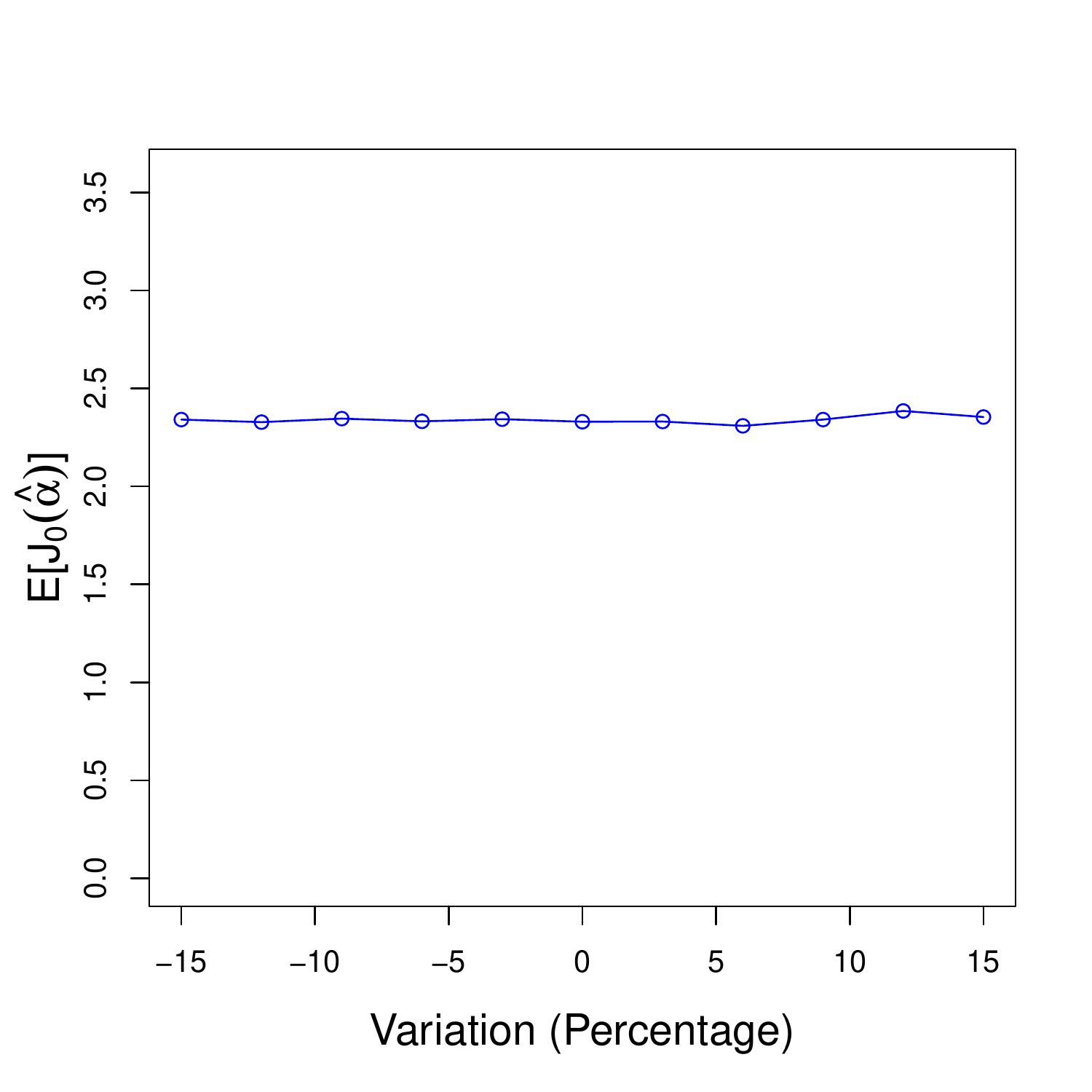}
\end{minipage}\hfill
\begin{minipage}{.33\textwidth}
\centering
\includegraphics[width=\linewidth]{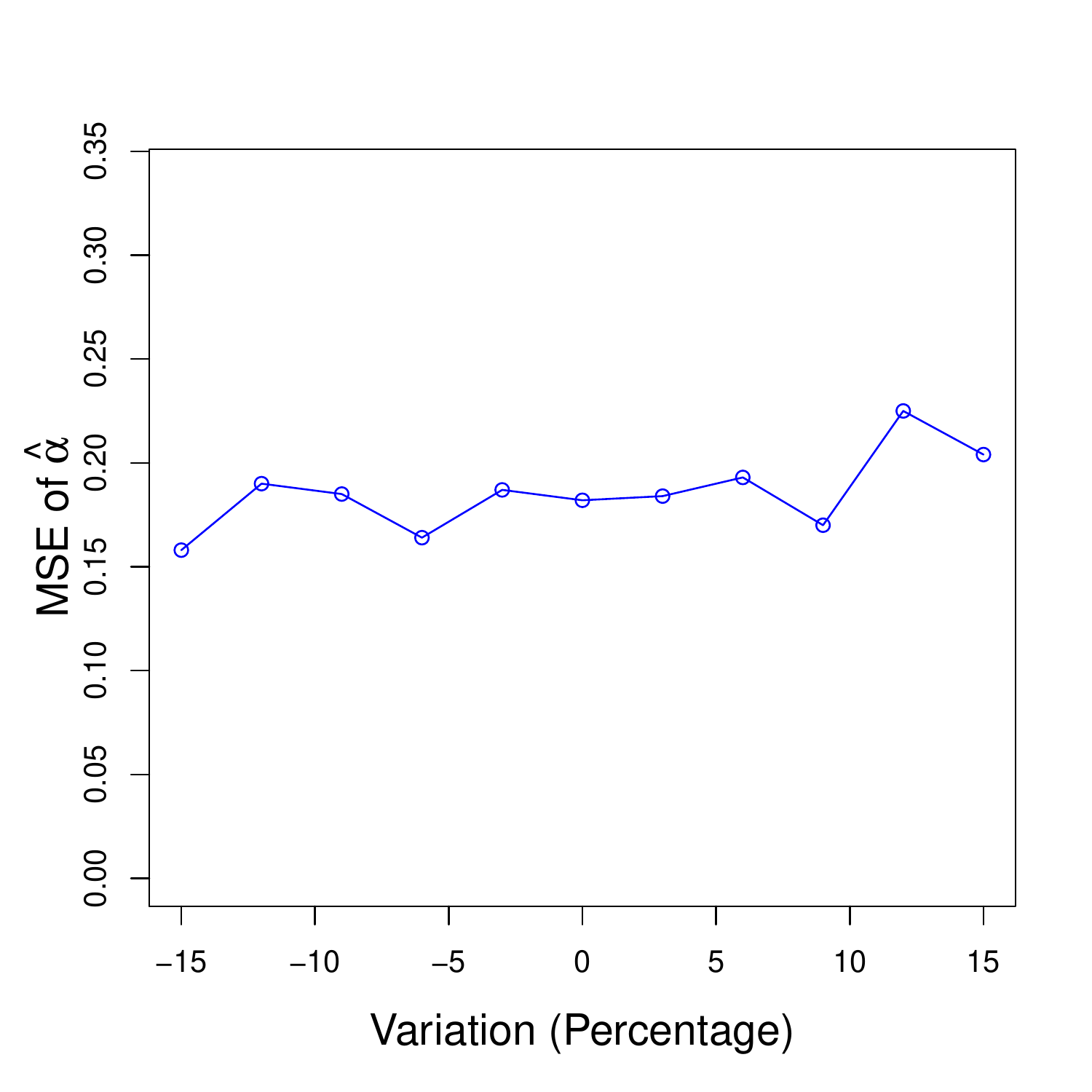}
\end{minipage}
\begin{minipage}{.33\textwidth}
\centering
\includegraphics[width=\linewidth]{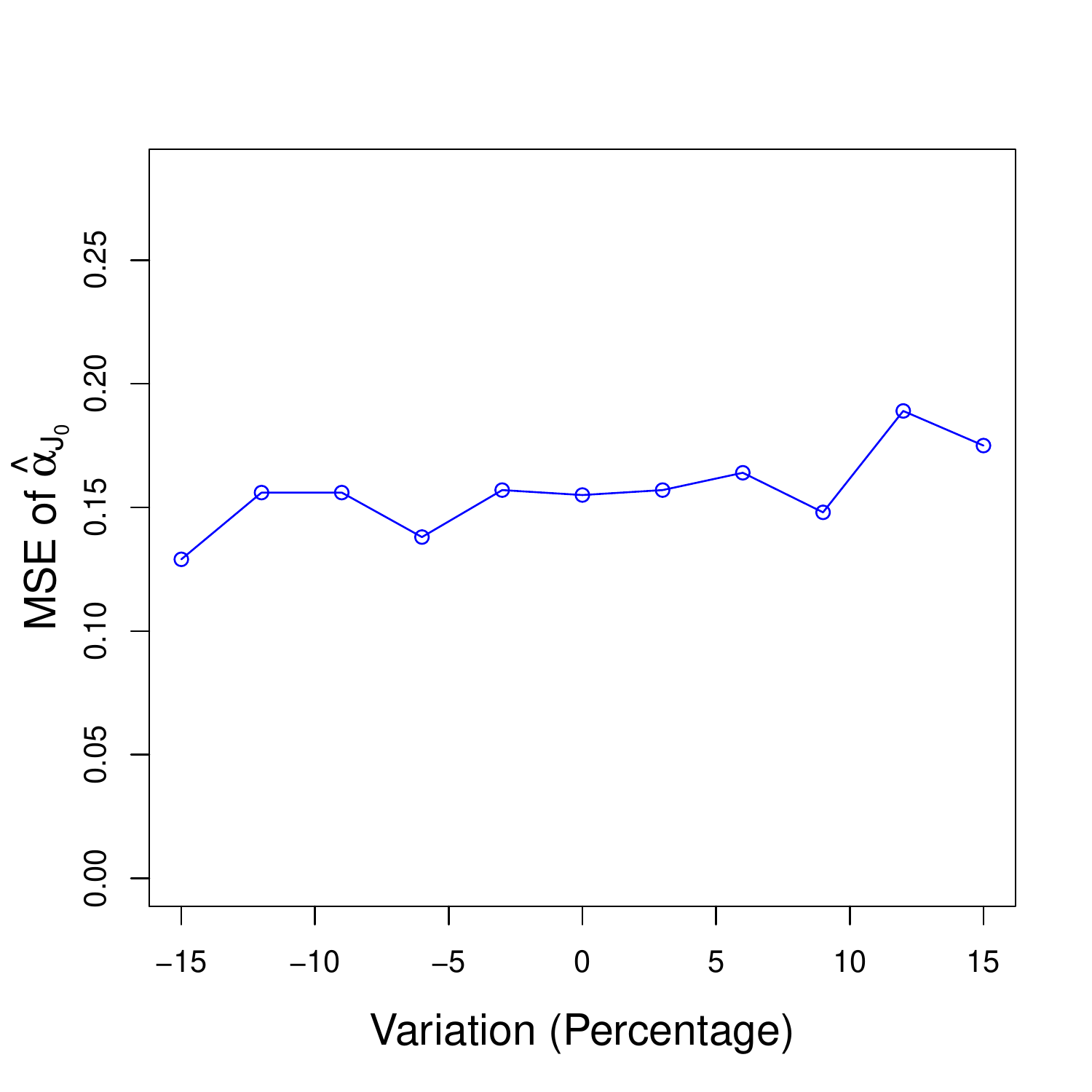}
\end{minipage}\hfill
\begin{minipage}{.33\textwidth}
\centering
\includegraphics[width=\linewidth]{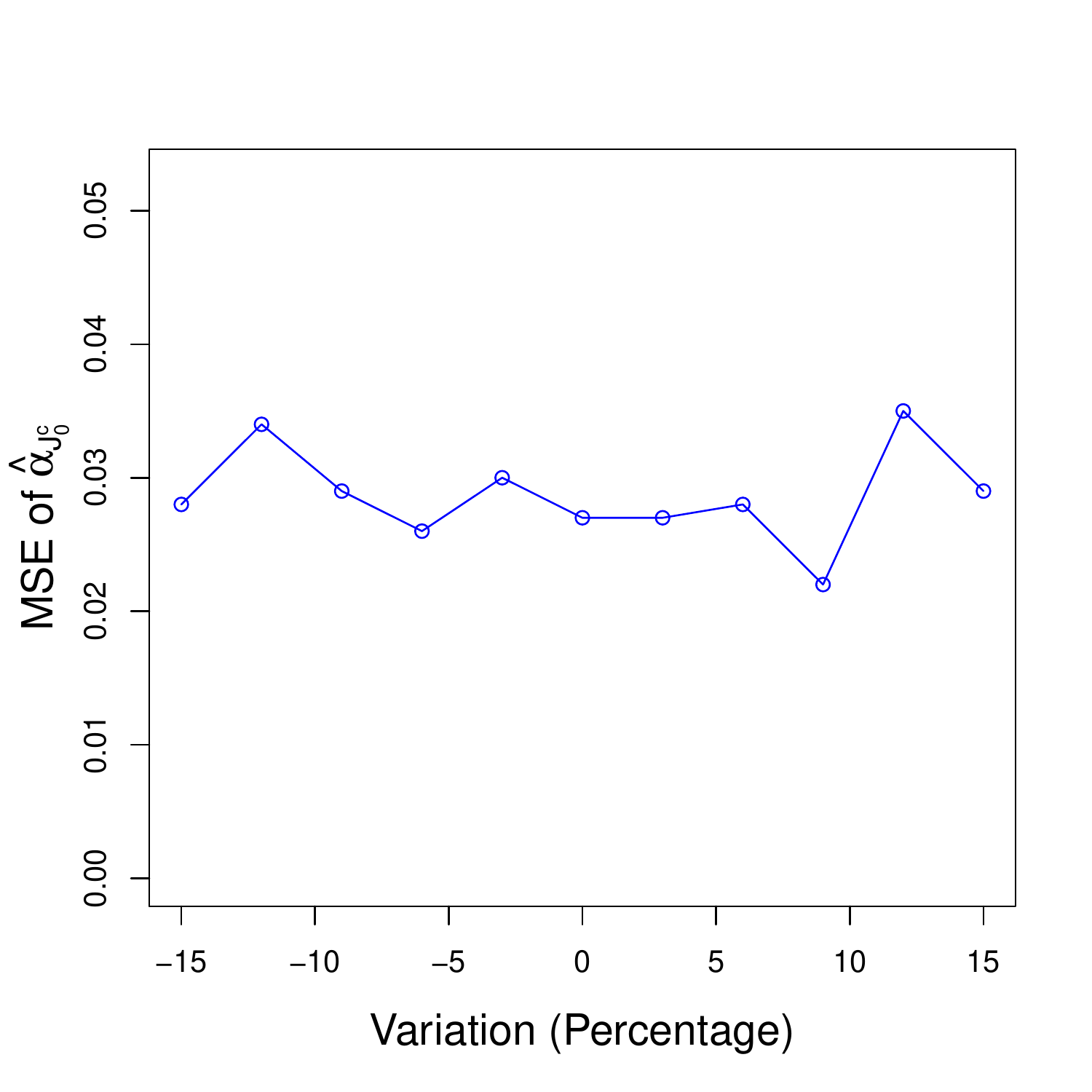}
\end{minipage}\hfill
\begin{minipage}{.33\textwidth}
\centering
\includegraphics[width=\linewidth]{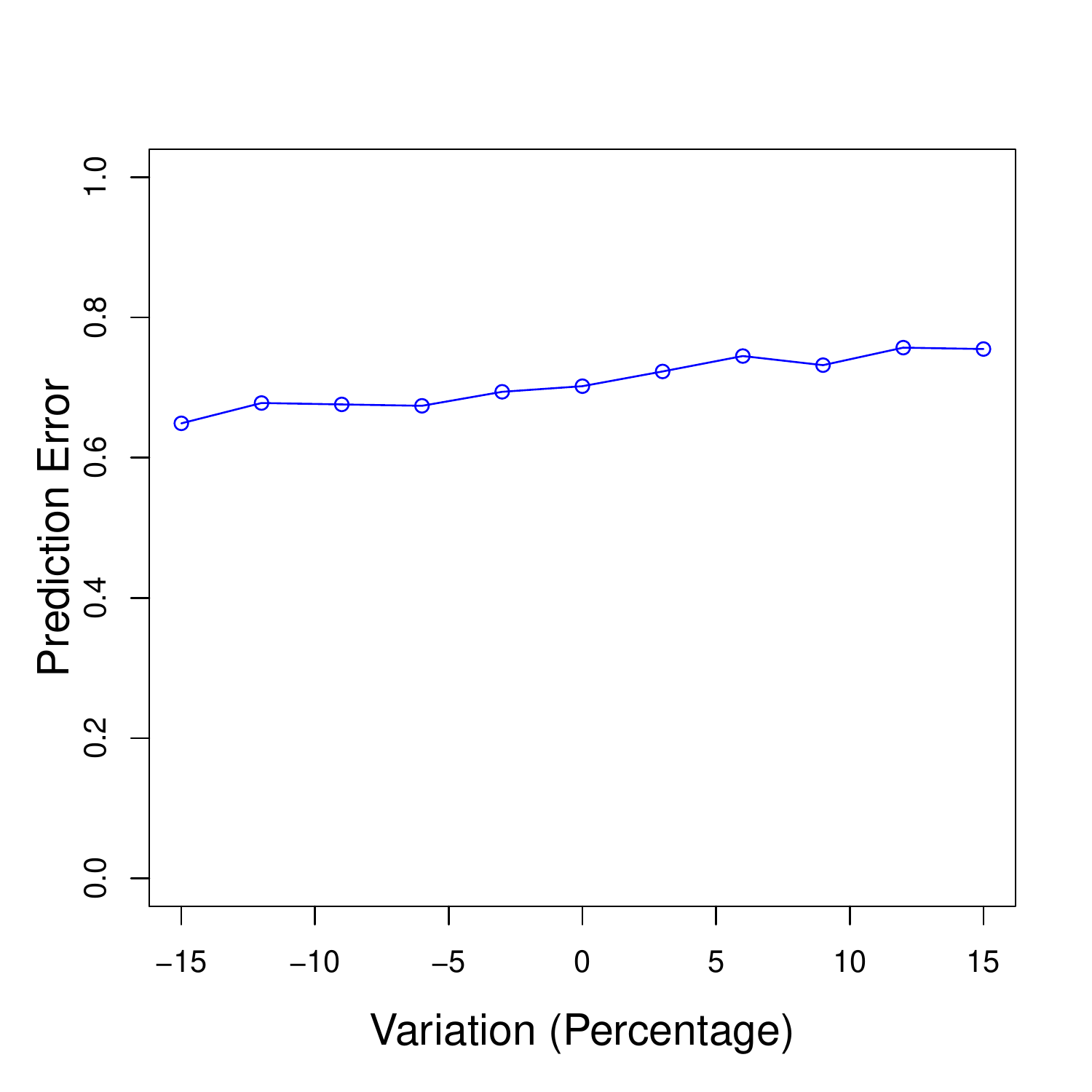}
\end{minipage}
\begin{minipage}{.33\textwidth}
\centering
\includegraphics[width=\linewidth]{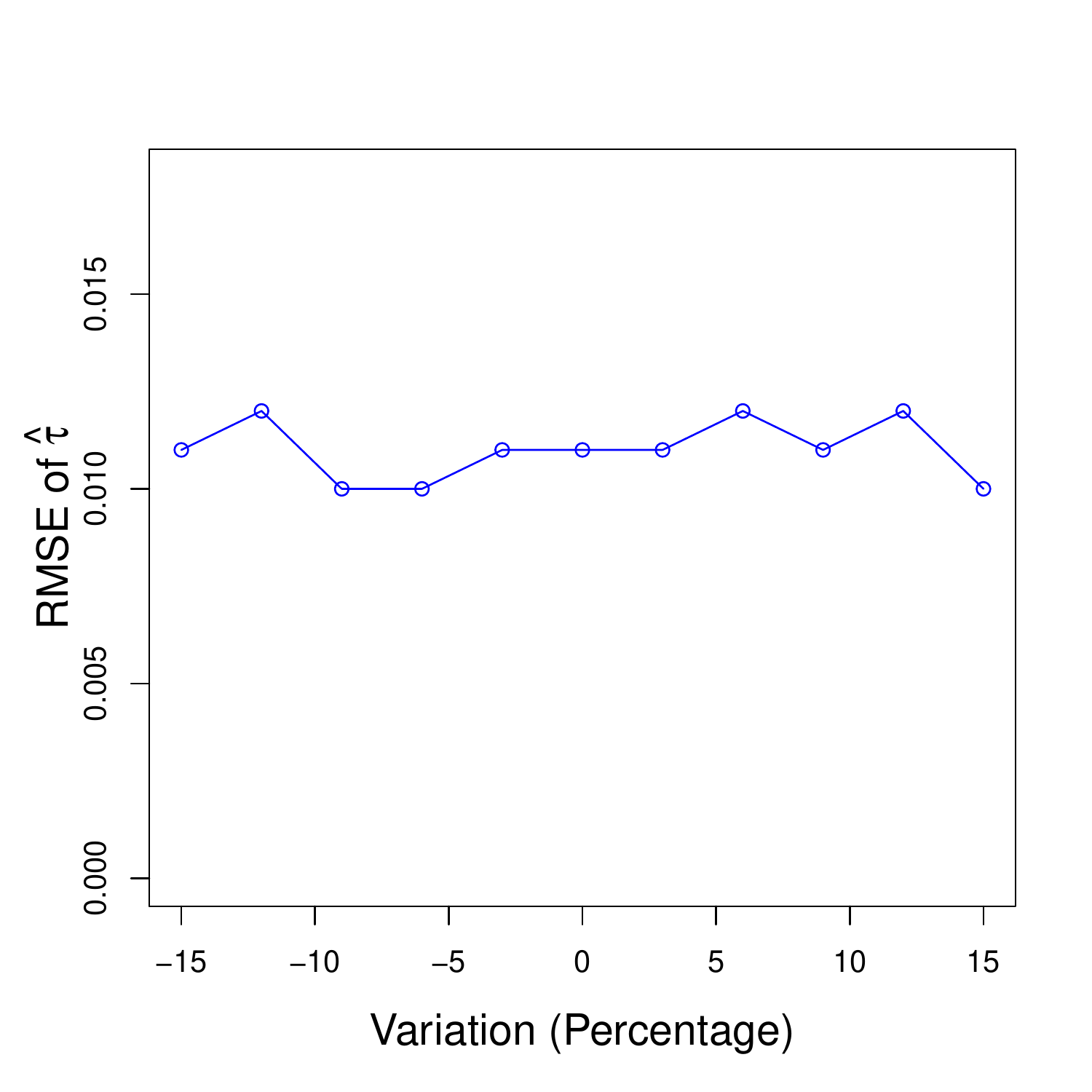}
\end{minipage}\hfill
\begin{minipage}{.33\textwidth}
\centering
\includegraphics[width=\linewidth]{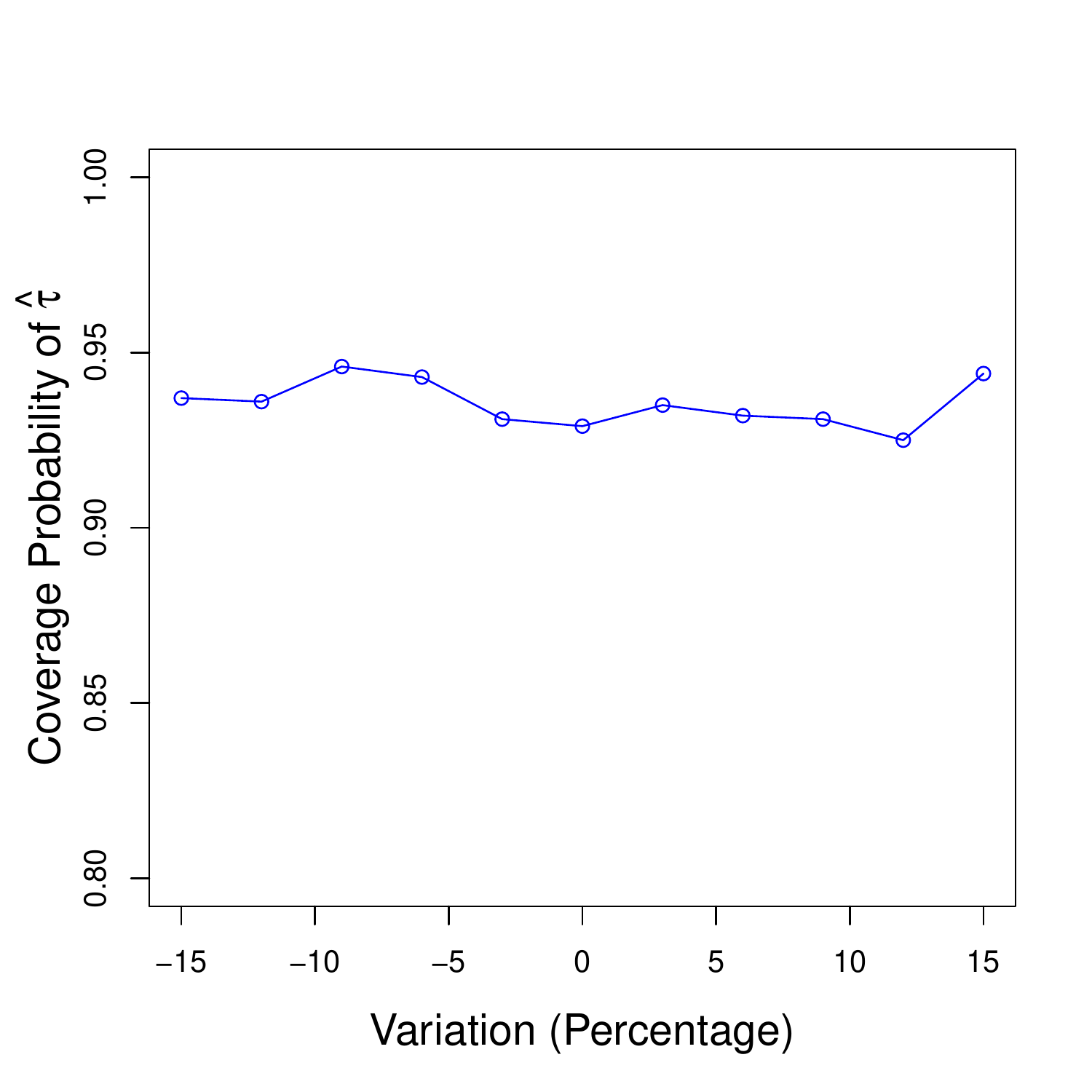}
\end{minipage}\hfill
\begin{minipage}{.33\textwidth}
\centering
\includegraphics[width=\linewidth]{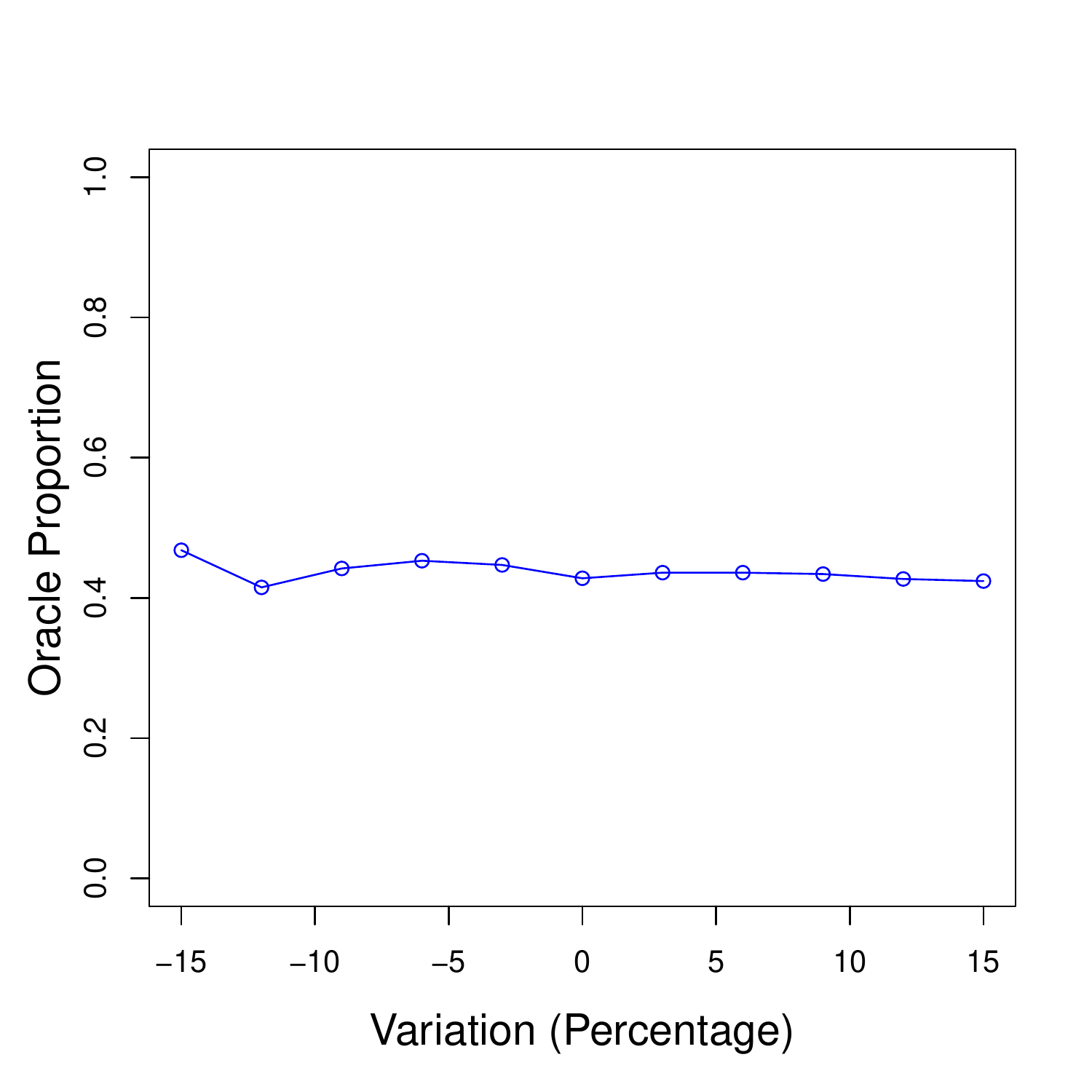}
\end{minipage}
\end{figure}

%
%

\newpage

\begin{figure}[tbph]
\centering
\caption{Sensitivity Analysis of $c_2$: Step 3b}\label{fig:sen-c2-s3b}
\begin{minipage}{.33\textwidth}
\centering
\includegraphics[width=\linewidth]{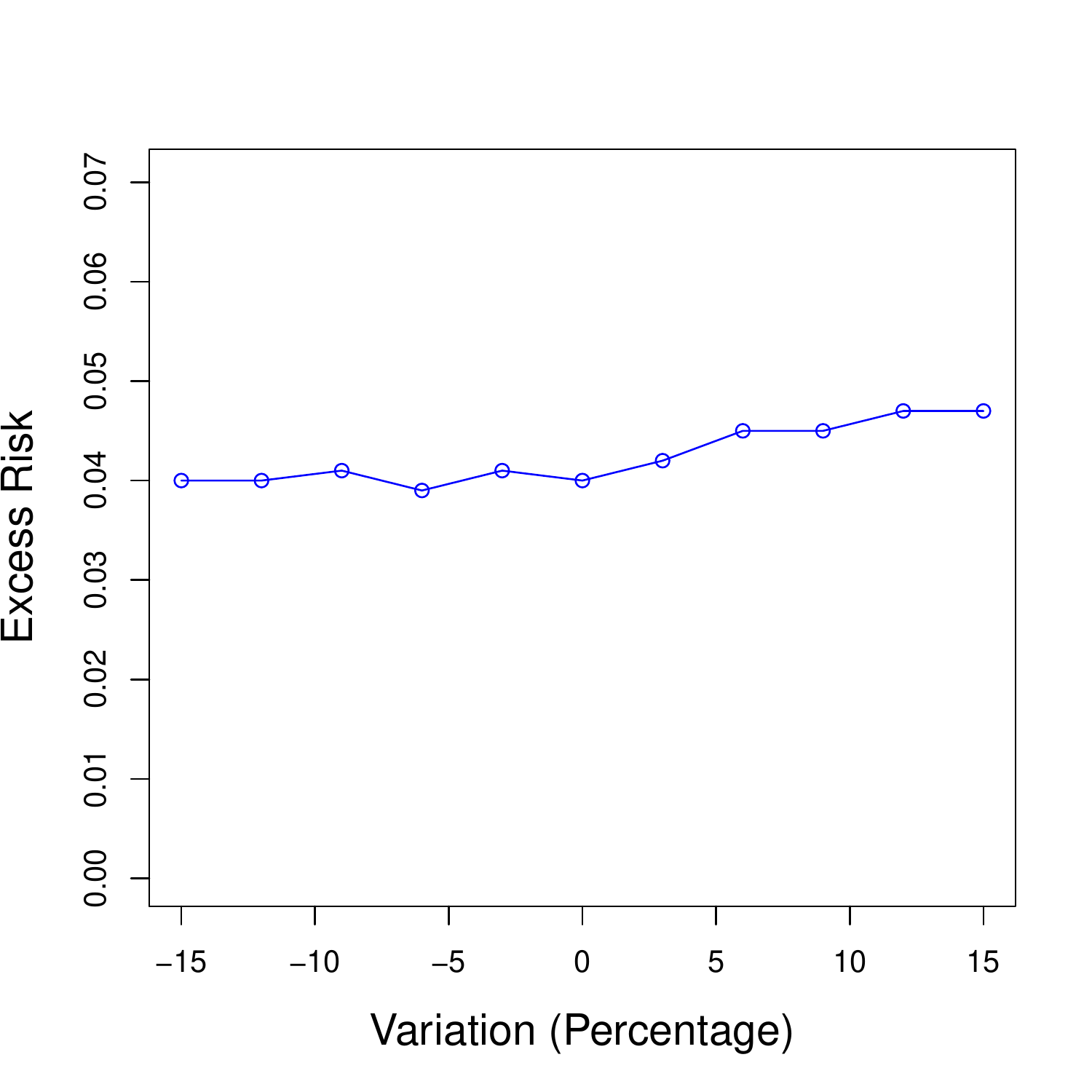}
\end{minipage}\hfill
\begin{minipage}{.33\textwidth}
\centering
\includegraphics[width=\linewidth]{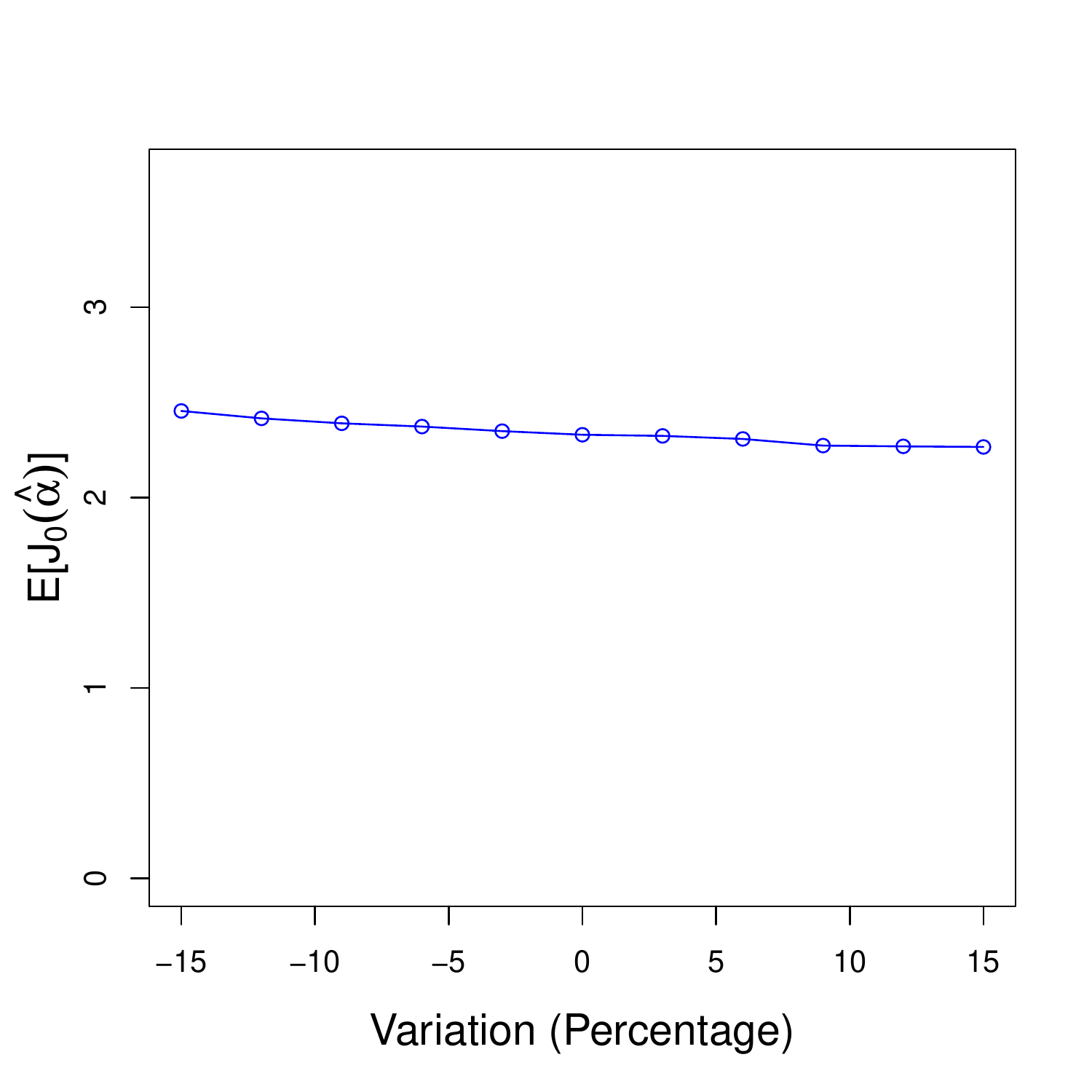}
\end{minipage}\hfill
\begin{minipage}{.33\textwidth}
\centering
\includegraphics[width=\linewidth]{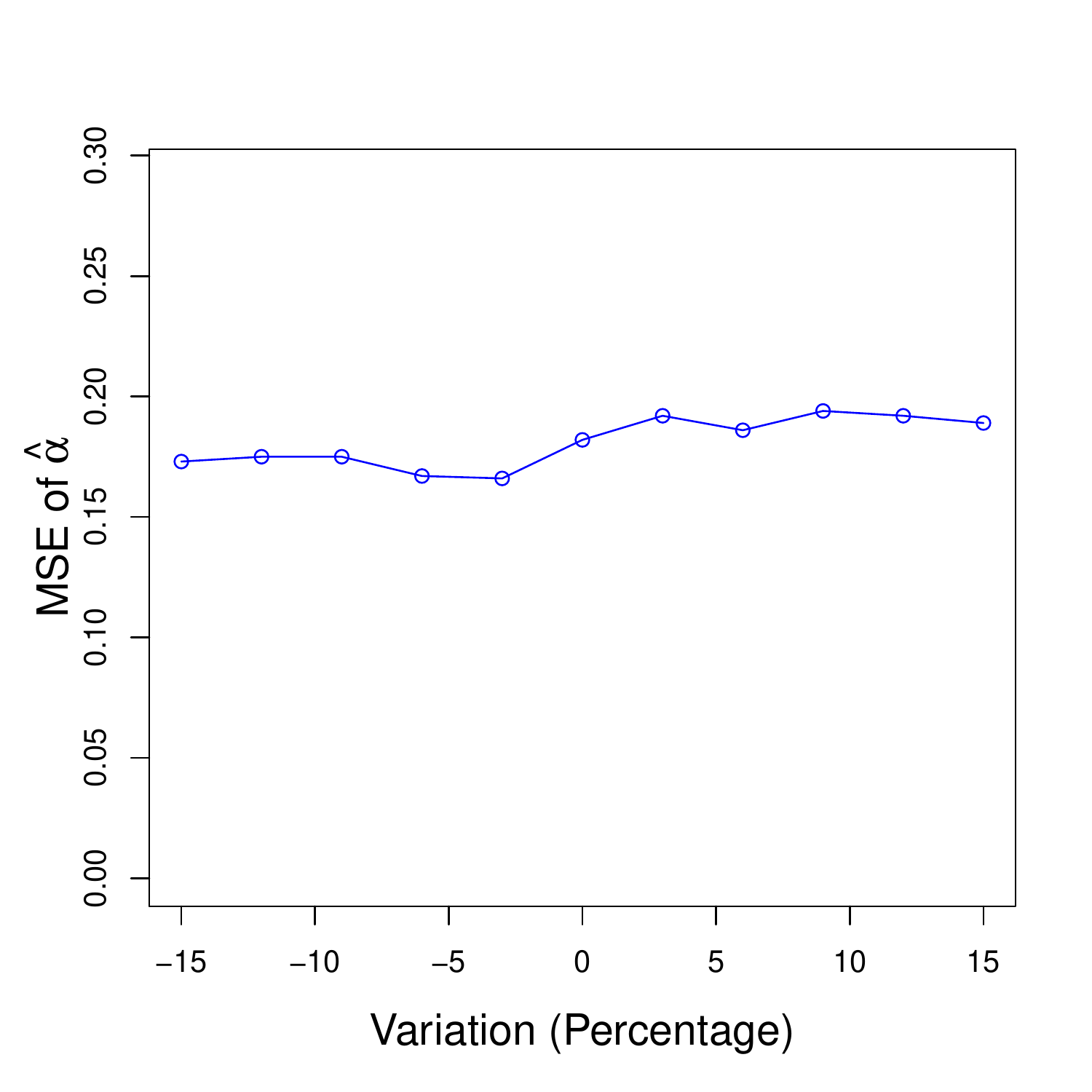}
\end{minipage}
\begin{minipage}{.33\textwidth}
\centering
\includegraphics[width=\linewidth]{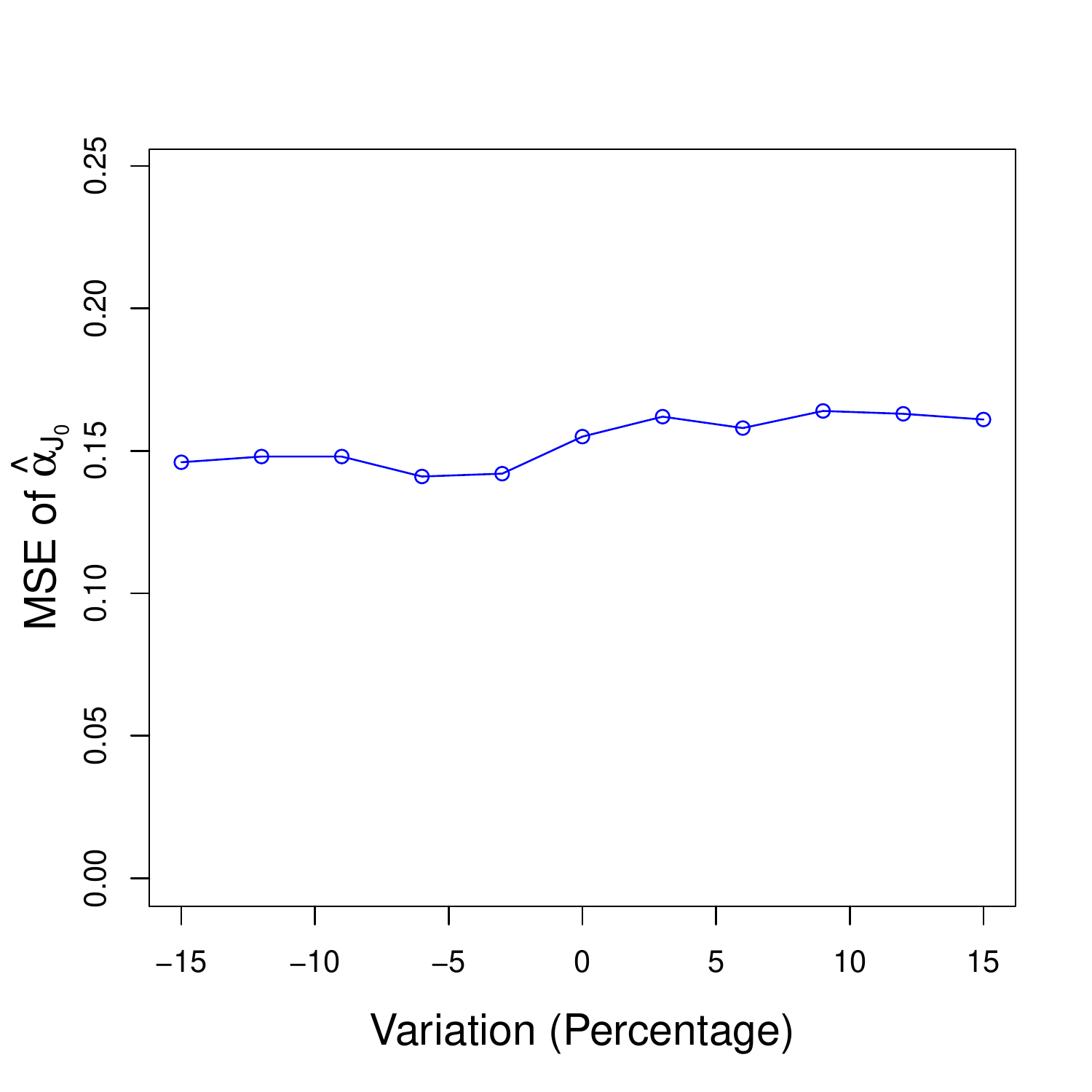}
\end{minipage}\hfill
\begin{minipage}{.33\textwidth}
\centering
\includegraphics[width=\linewidth]{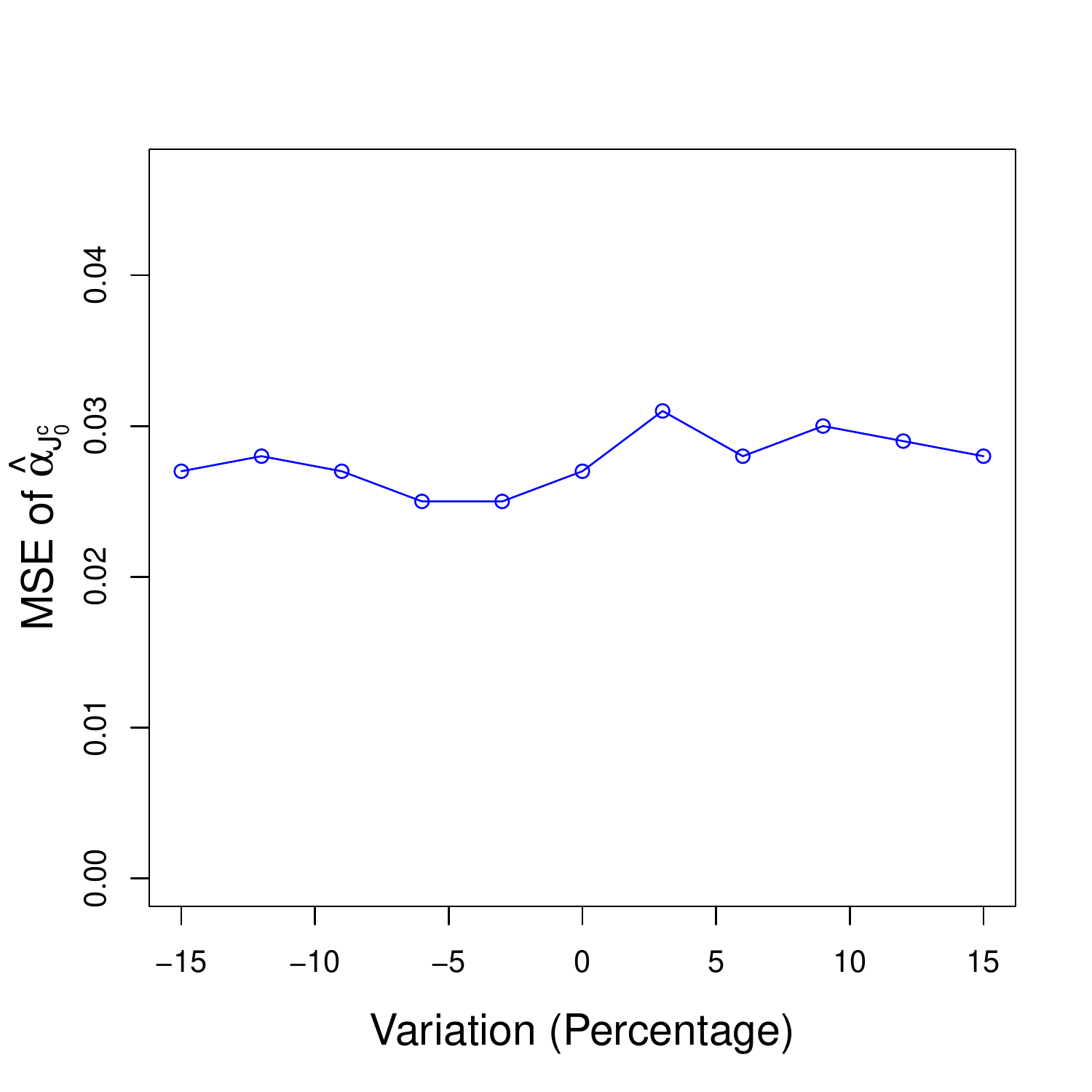}
\end{minipage}\hfill
\begin{minipage}{.33\textwidth}
\centering
\includegraphics[width=\linewidth]{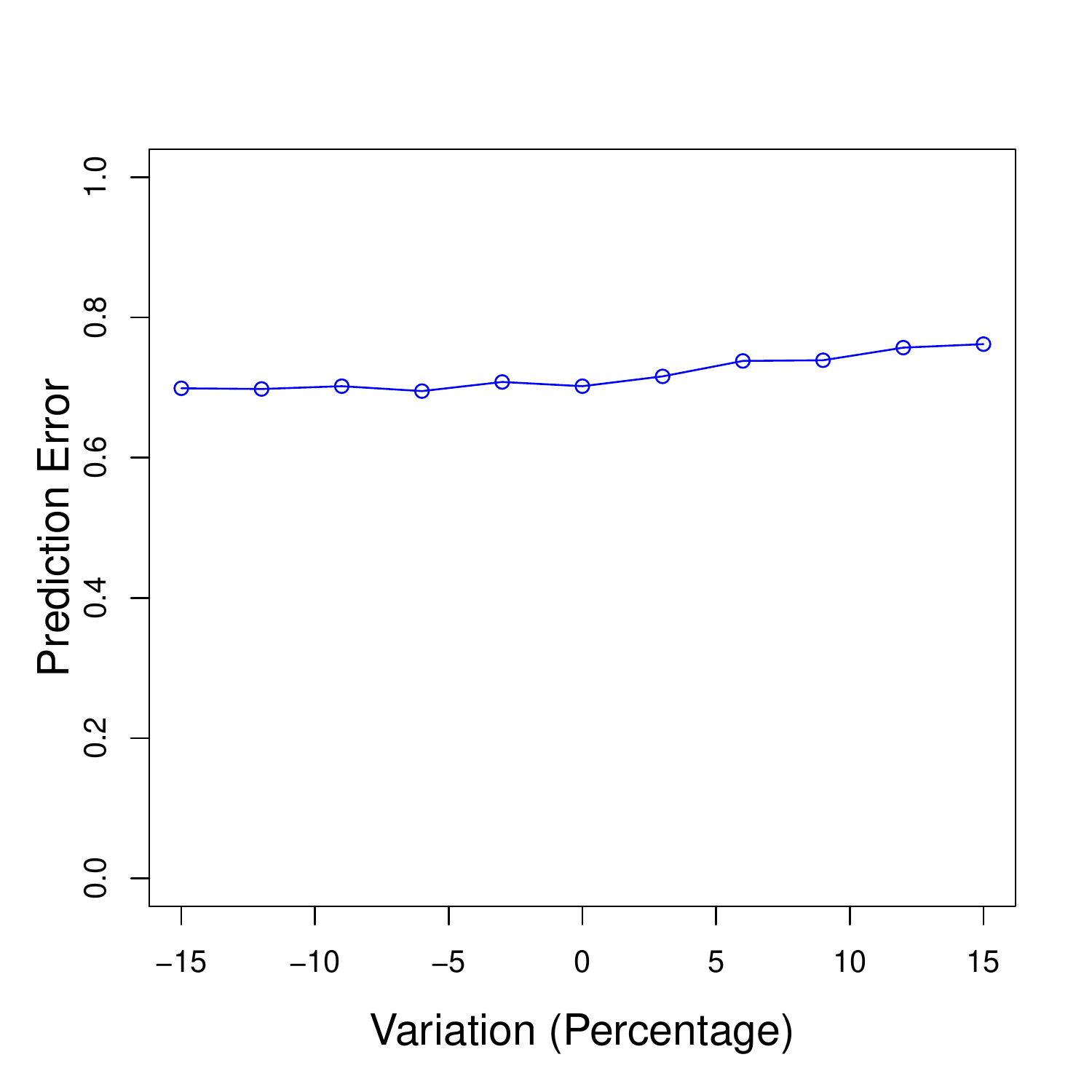}
\end{minipage}
\begin{minipage}{.33\textwidth}
\centering
\includegraphics[width=\linewidth]{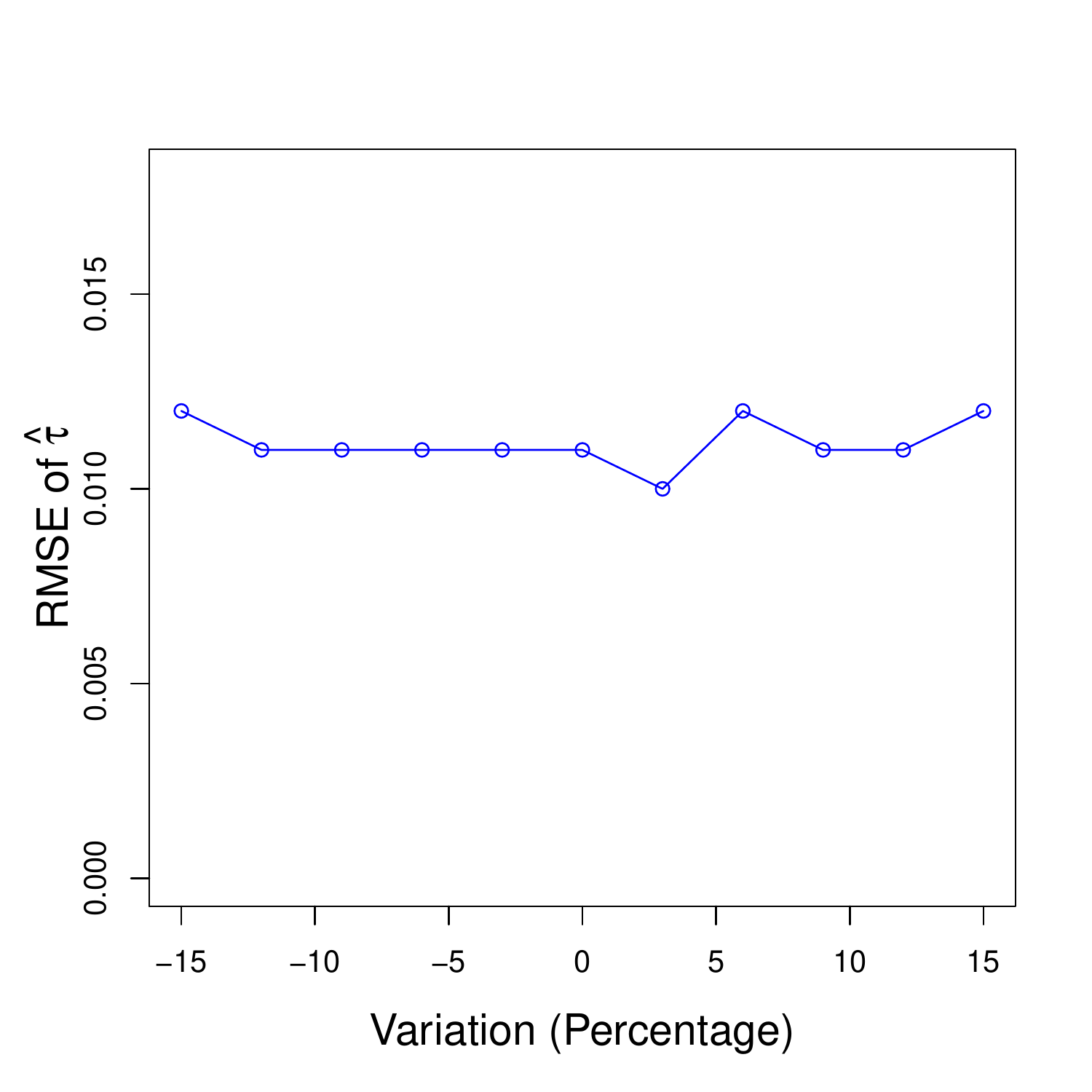}
\end{minipage}\hfill
\begin{minipage}{.33\textwidth}
\centering
\includegraphics[width=\linewidth]{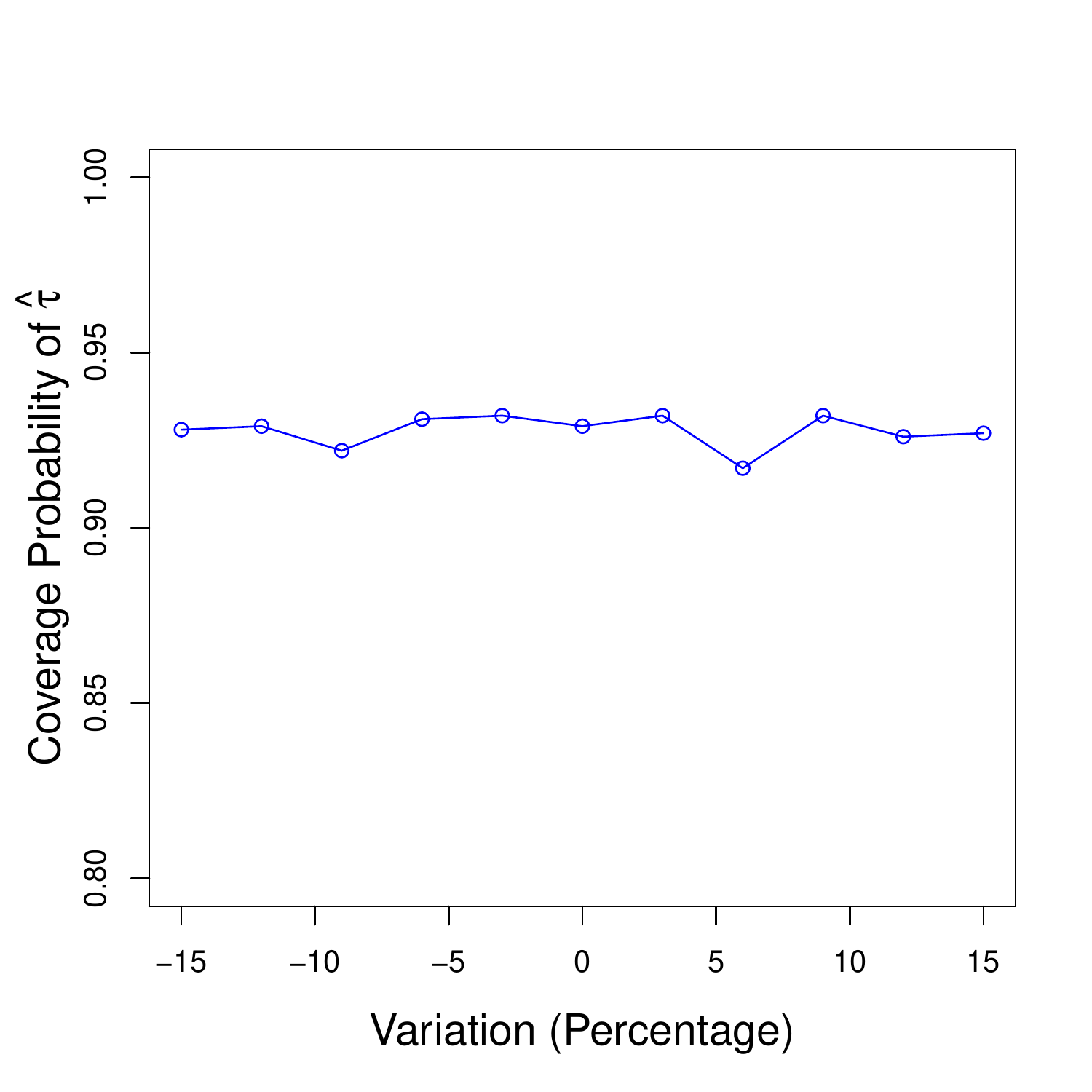}
\end{minipage}\hfill
\begin{minipage}{.33\textwidth}
\centering
\includegraphics[width=\linewidth]{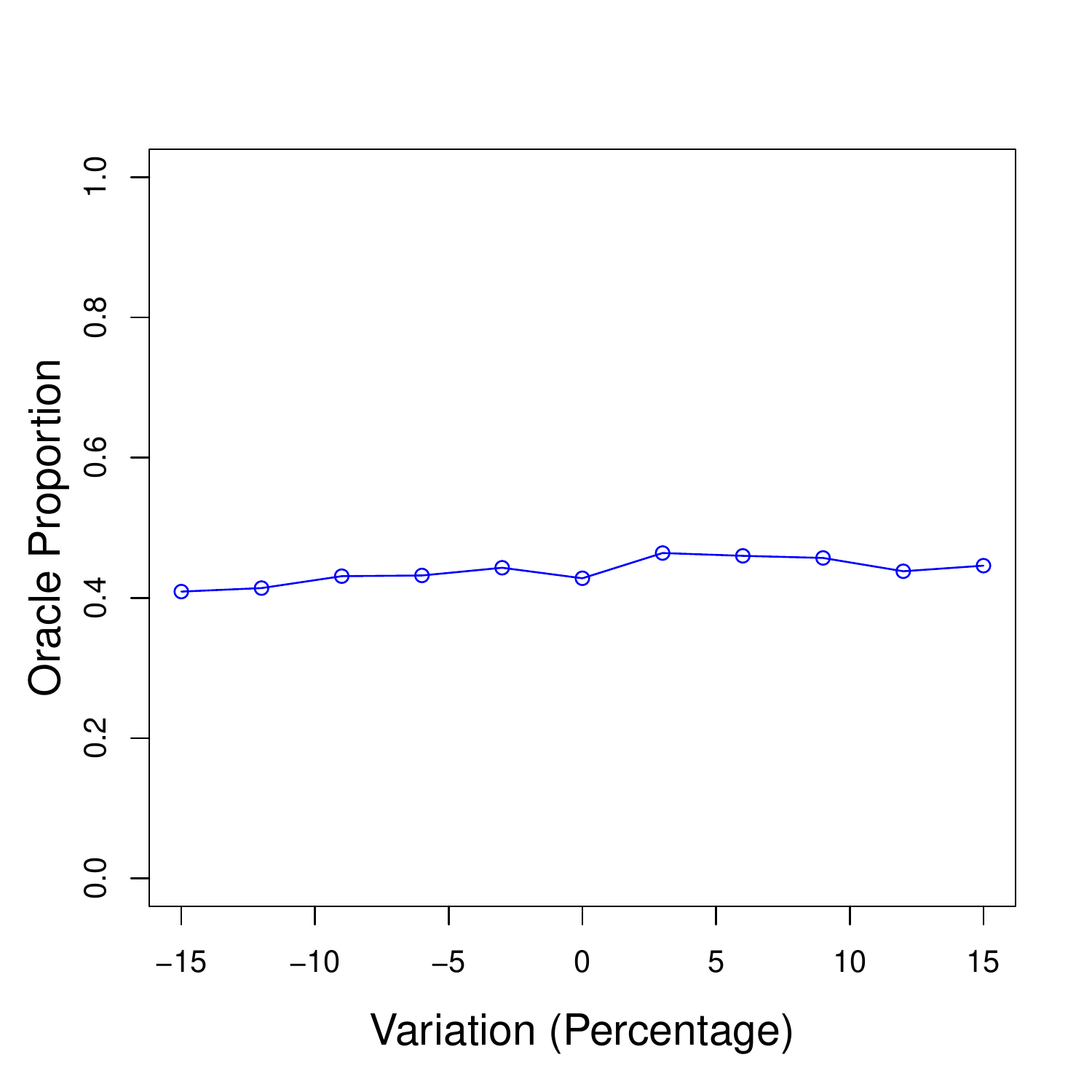}
\end{minipage}
\end{figure}

\newpage

%
%

\begin{figure}[tbph]
\centering
\caption{Sensitivity Analysis of $a$: Step 3b}\label{fig:sen-a-s3b}
\begin{minipage}{.33\textwidth}
\centering
\includegraphics[width=\linewidth]{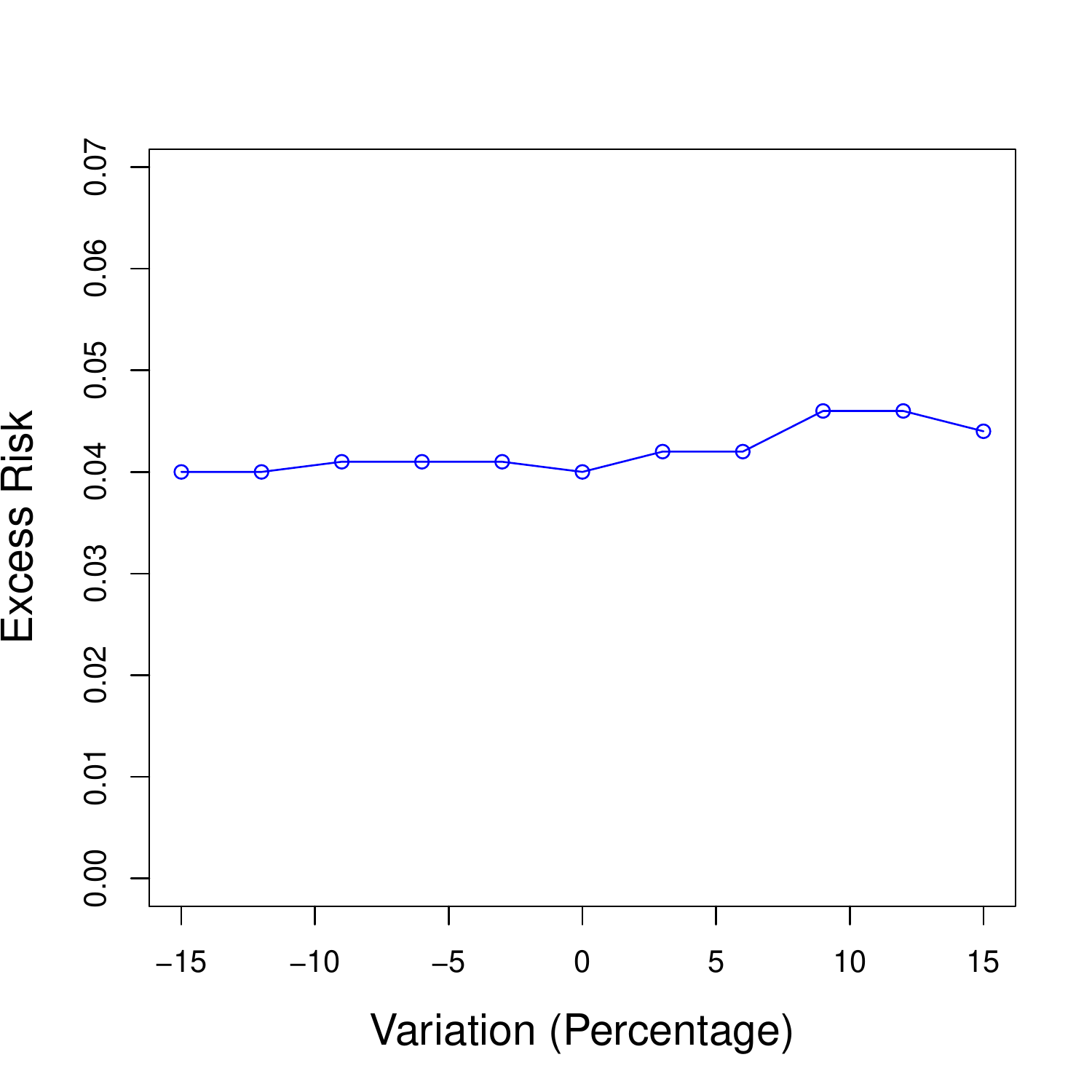}
\end{minipage}\hfill
\begin{minipage}{.33\textwidth}
\centering
\includegraphics[width=\linewidth]{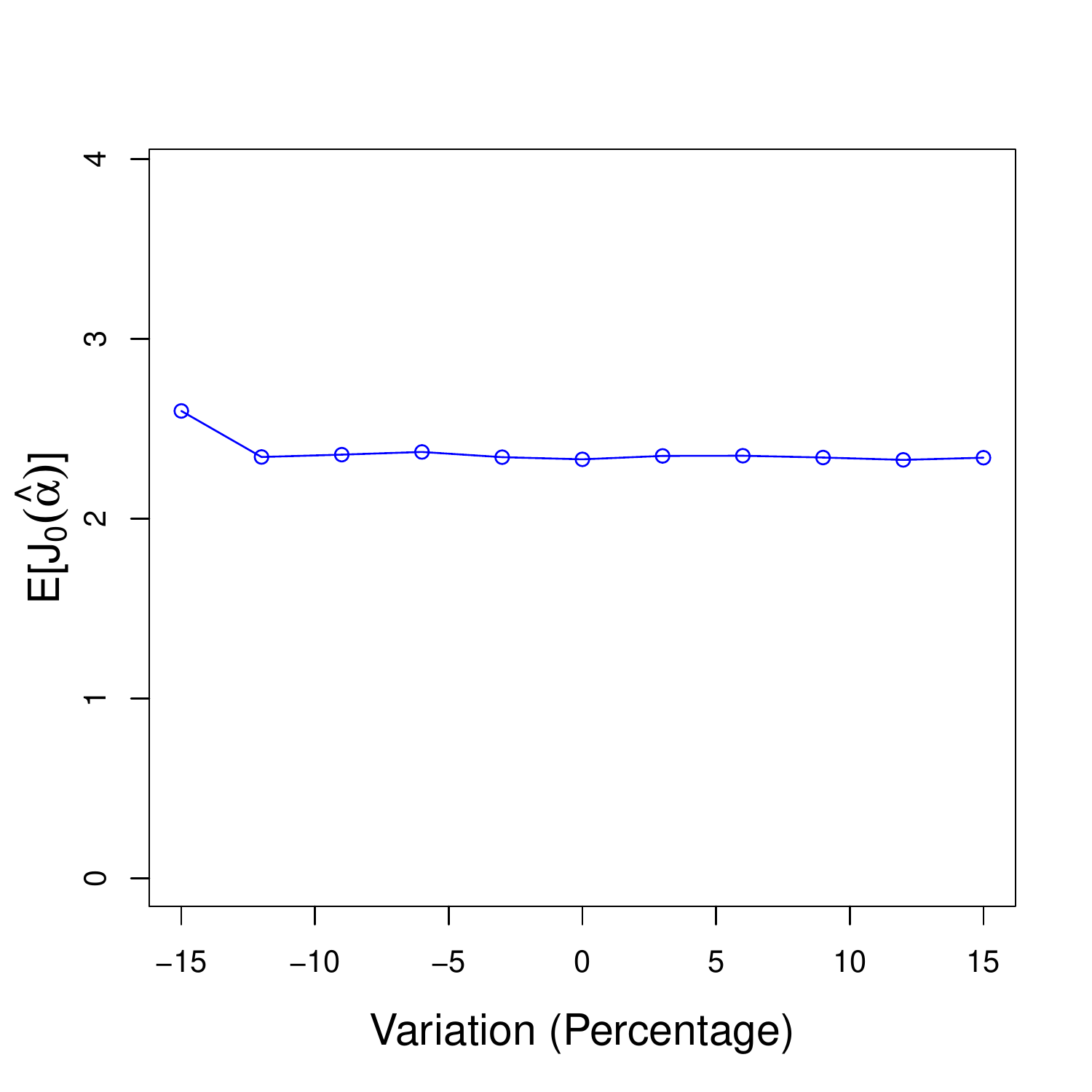}
\end{minipage}\hfill
\begin{minipage}{.33\textwidth}
\centering
\includegraphics[width=\linewidth]{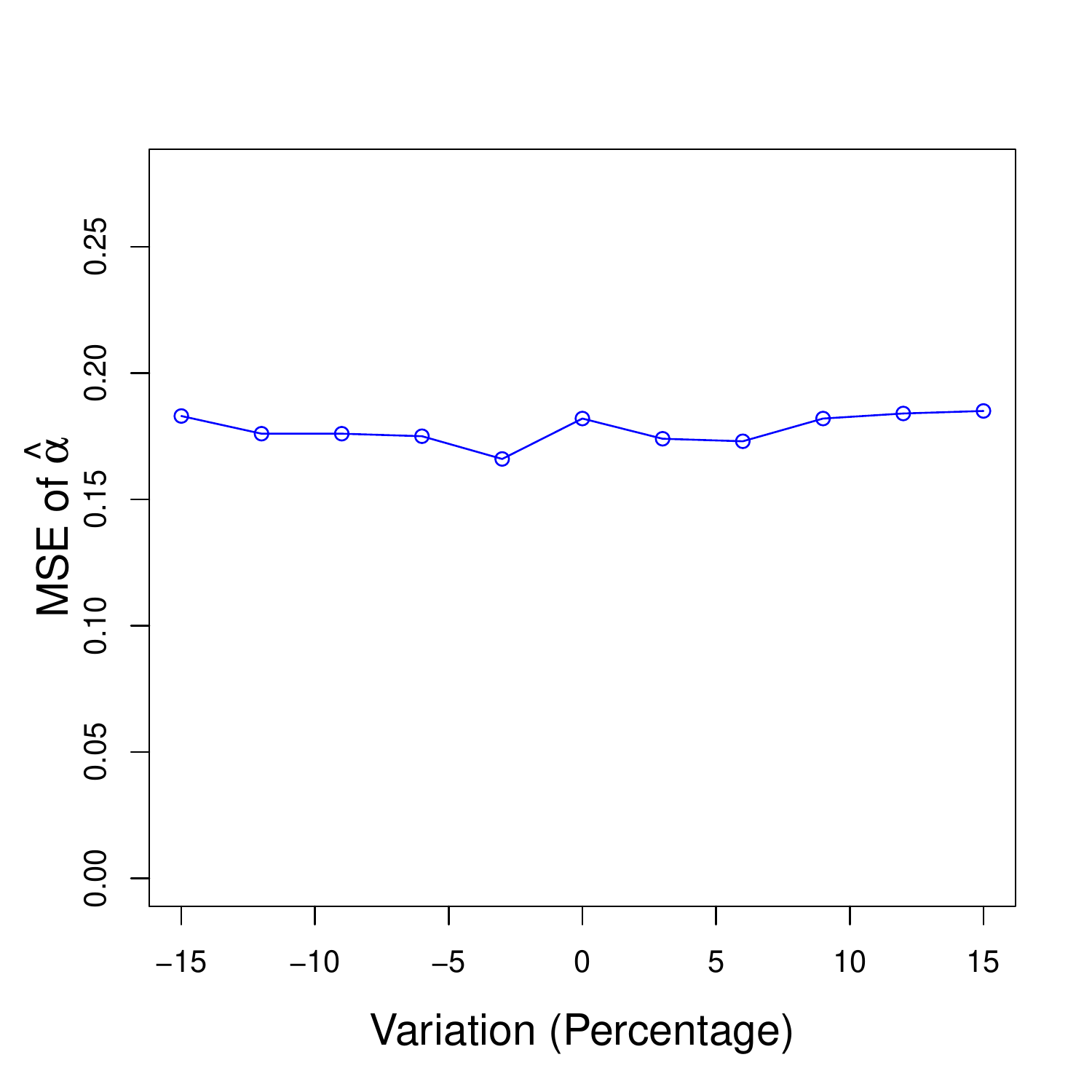}
\end{minipage}
\begin{minipage}{.33\textwidth}
\centering
\includegraphics[width=\linewidth]{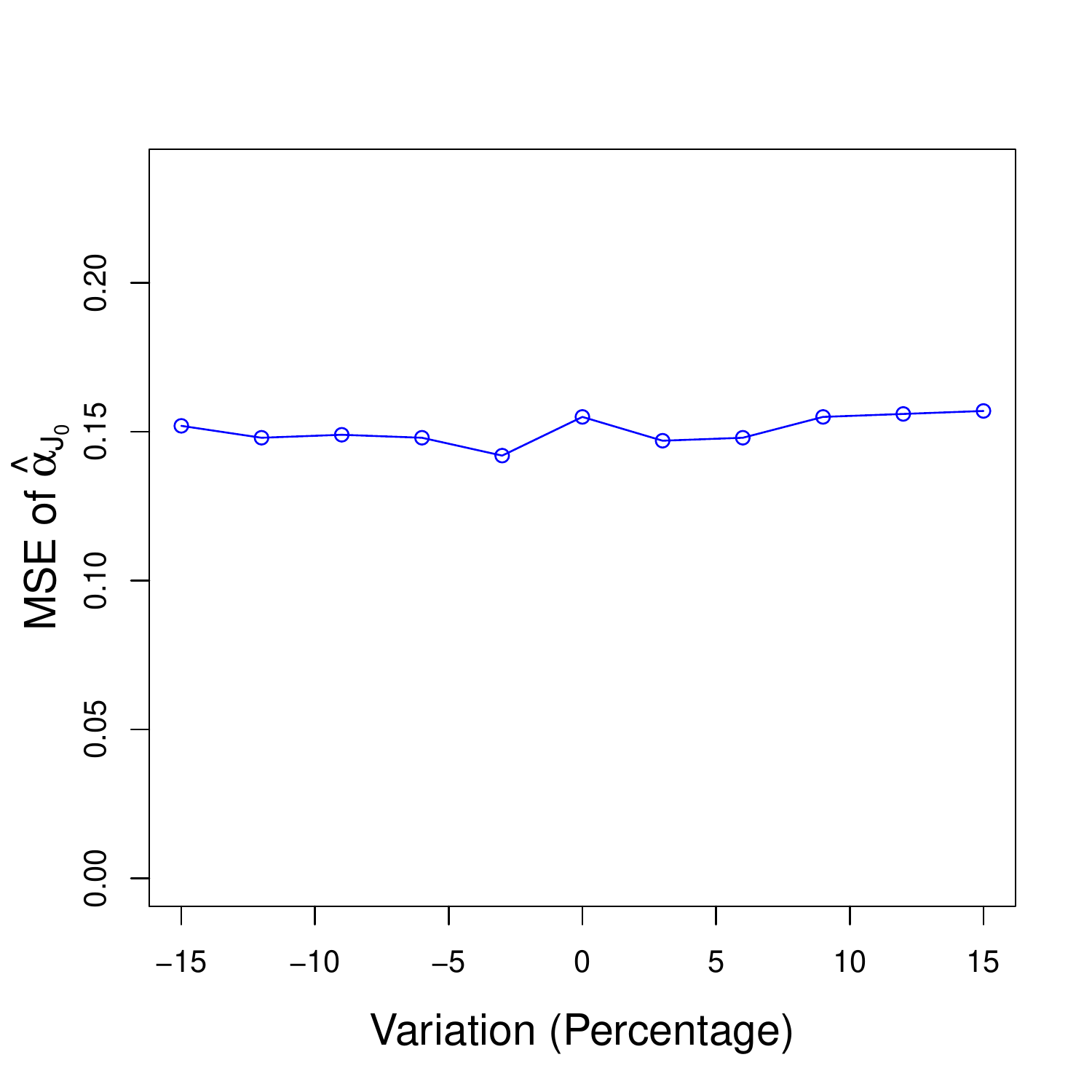}
\end{minipage}\hfill
\begin{minipage}{.33\textwidth}
\centering
\includegraphics[width=\linewidth]{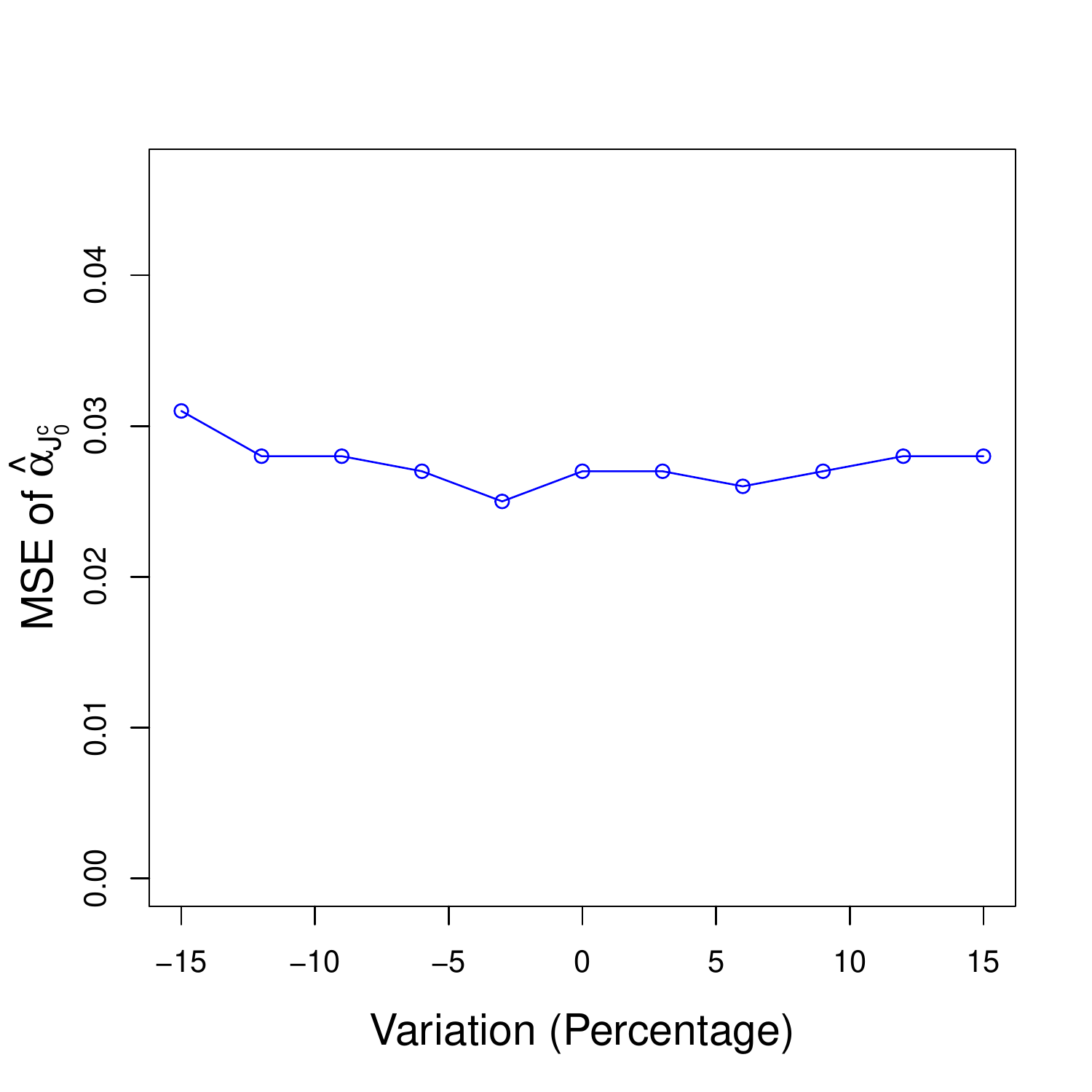}
\end{minipage}\hfill
\begin{minipage}{.33\textwidth}
\centering
\includegraphics[width=\linewidth]{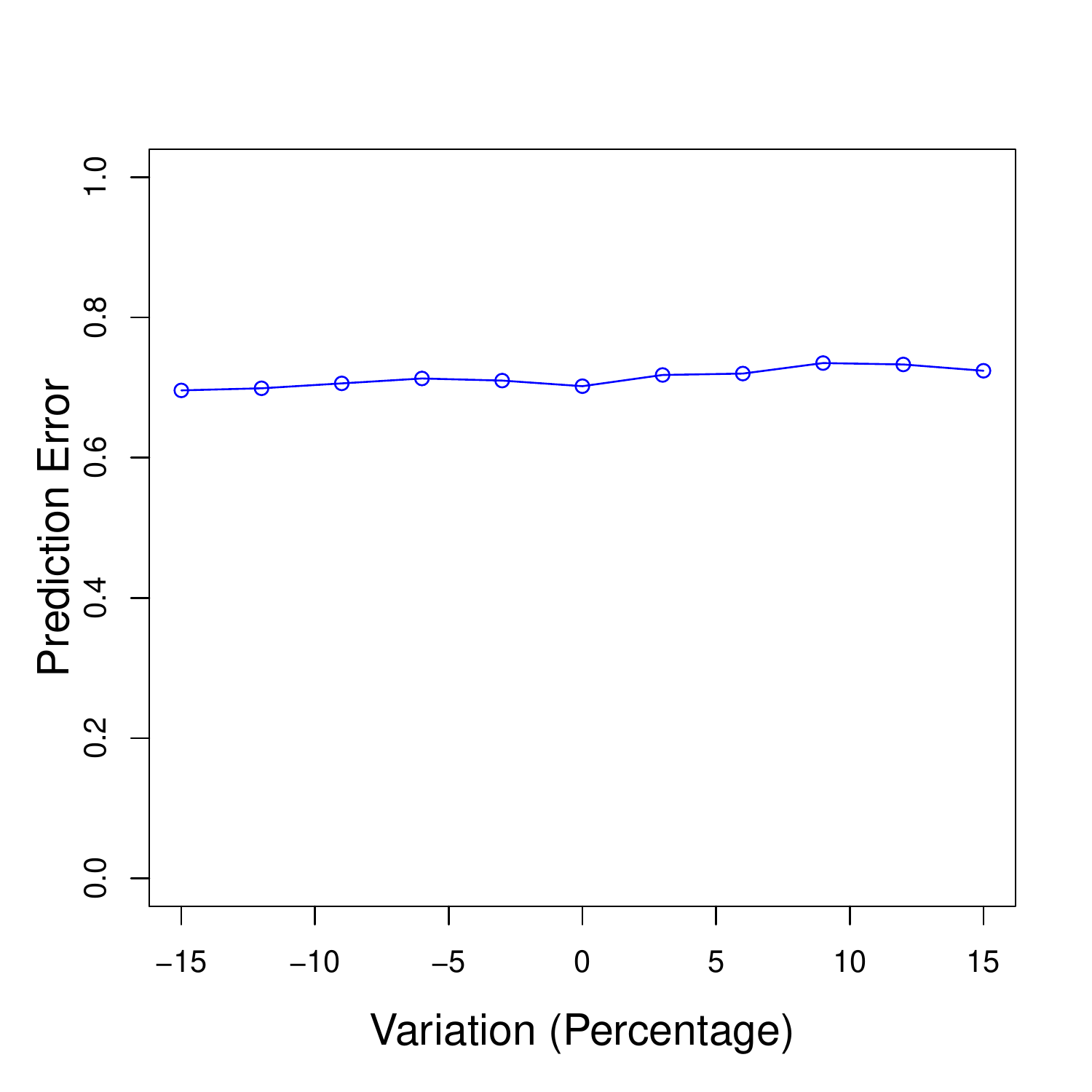}
\end{minipage}
\begin{minipage}{.33\textwidth}
\centering
\includegraphics[width=\linewidth]{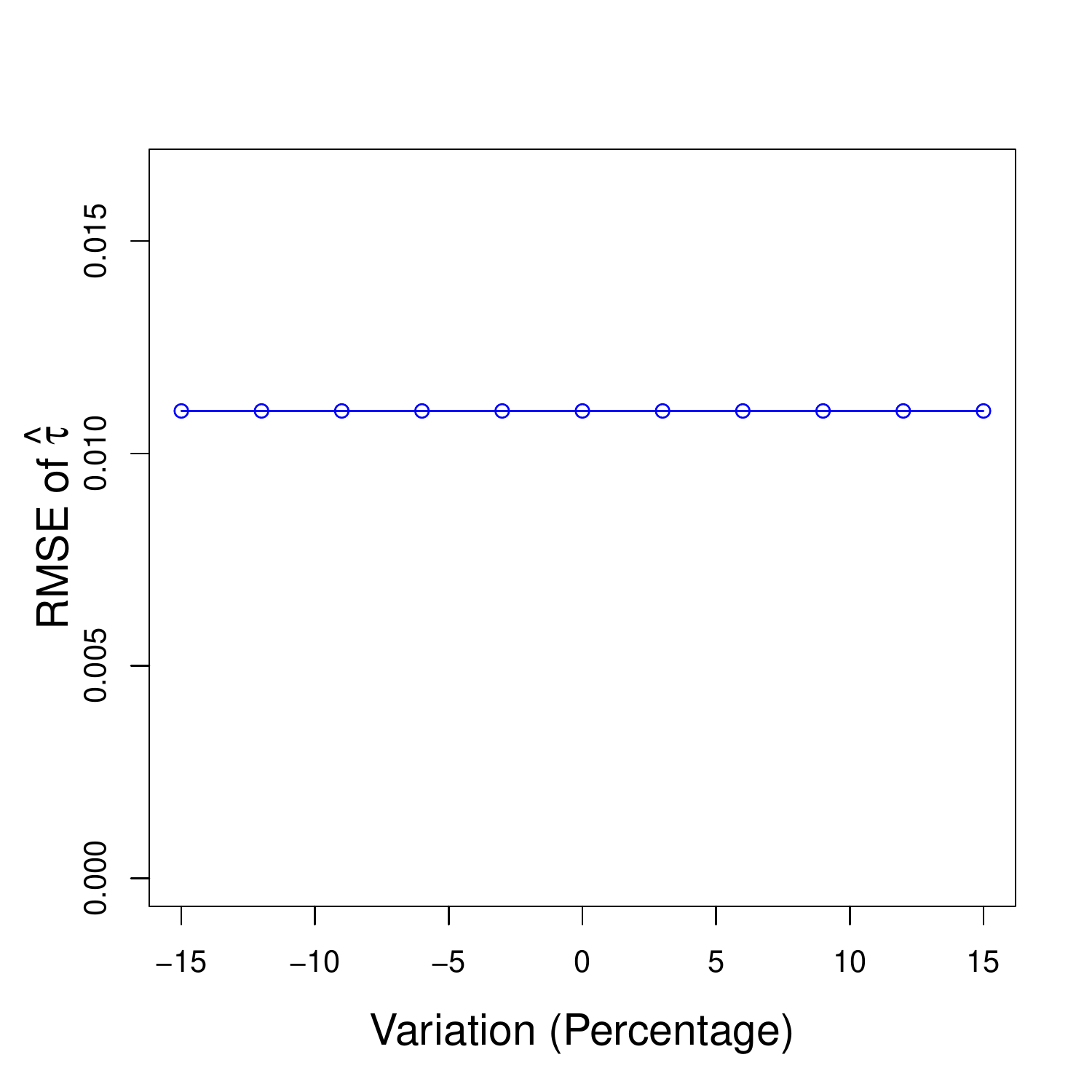}
\end{minipage}\hfill
\begin{minipage}{.33\textwidth}
\centering
\includegraphics[width=\linewidth]{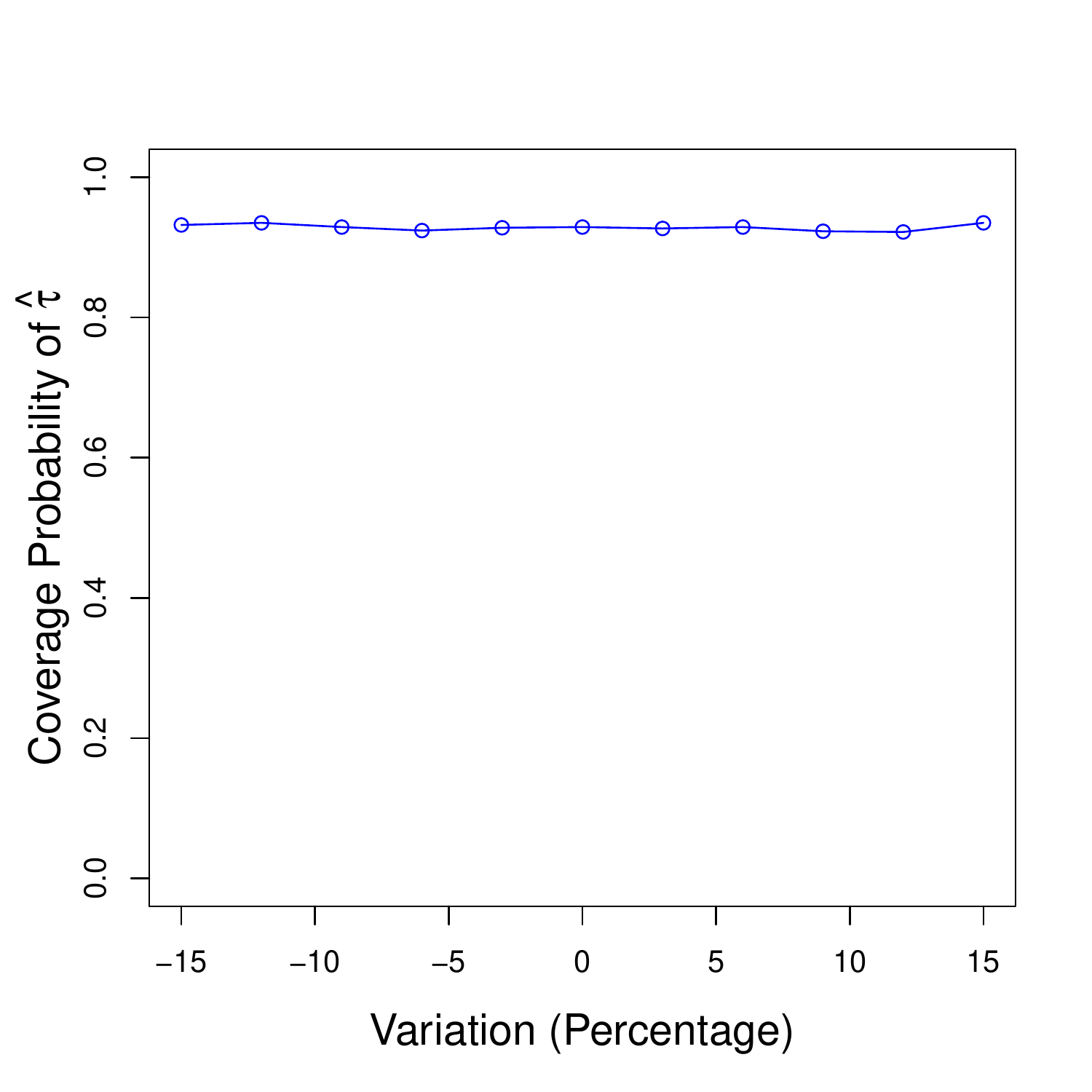}
\end{minipage}\hfill
\begin{minipage}{.33\textwidth}
\centering
\includegraphics[width=\linewidth]{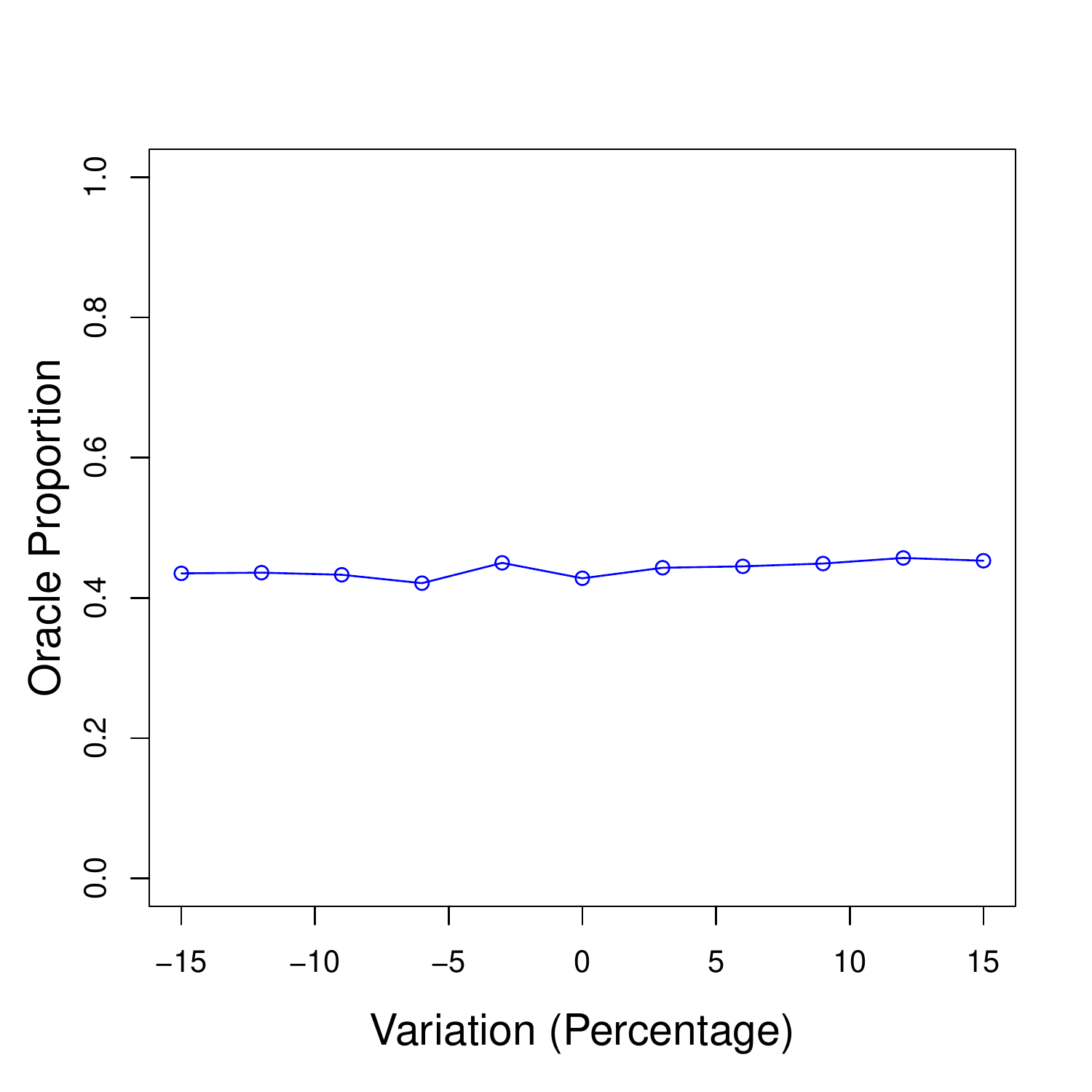}
\end{minipage}
\end{figure}

\newpage

\section{Estimating a Change Point in Racial Segregation: Additional Tables}\label{app:emp-tables-figures}

\begin{table}[htbp]
\footnotesize
\centering
\caption{Selected Covariates $\gm=0.25$}\label{tb:sel-0.25}
\begin{tabular}{lp{10cm}l}
  \\
  \hline
 & Selected Regressors &  \\ 
  \hline
    \underline{6 control variables}\\

  No Interaction         & 3 6 & \\ 
  Two-way Interaction    & 1:3 3:5 4:5 5:6 & \\ 
  Three-way Interaction  & 4:5 5:6 1:2:3 1:3:5 1:4:6 2:3:5 3:4:5 3:5:6 & \\ 
  Four-way Interaction   & 4:5 5:6 2:3:5 1:3:4:5 1:4:5:6 3:4:5:6 & \\ 
  Five-way Interaction   & 4:5 5:6 2:3:5 1:4:5:6 1:2:3:4:5 2:3:4:5:6 & \\ 
  Six-way Interaction    & 4:5 5:6 2:3:5 1:4:5:6 1:2:3:4:5 2:3:4:5:6 & \\ 
\\
    \underline{12 control variables}\\
 No Interaction & 3 6 & \\ 
  Two-way Interaction & 1:5 3:5-sq 5:6 6:5-sq 3-sq:5-sq & \\ 
  Three-way Interaction & 5:6 3:5:5-sq 4:1-sq:4-sq 5:1-sq:6-sq 5:3-sq:5-sq 5:4-sq:6-sq 1-sq:3-sq:5-sq 1-sq:5-sq:6-sq & \\ 
  Four-way Interaction & 5:6 2:5:6 1:3:1-sq:4-sq 1:5:6:6-sq 2:1-sq:5-sq:6-sq 3:5:2-sq:4-sq 3:5:2-sq:5-sq 3:1-sq:4-sq:6-sq 5:1-sq:2-sq:6-sq 5:1-sq:3-sq:5-sq 5:2-sq:3-sq:5-sq 6:3-sq:4-sq:6-sq & \\ 
  Five-way Interaction & 5:6 2:5:6 1:3:4:1-sq:4-sq 1:3:1-sq:4-sq:6-sq 1:5:6:2-sq:6-sq 1:5:3-sq:4-sq:5-sq 1:1-sq:2-sq:5-sq:6-sq 2:5:1-sq:2-sq:6-sq 2:5:2-sq:3-sq:5-sq 3:5:1-sq:3-sq:5-sq 4:6:3-sq:4-sq:6-sq 4:1-sq:2-sq:5-sq:6-sq & \\ 
  Six-way Interaction & 2:5:6 1:3:4:1-sq:4-sq 2:5:1-sq:2-sq:6-sq 2:5:2-sq:3-sq:5-sq 1:2:5:6:2-sq:6-sq 1:2:1-sq:2-sq:5-sq:6-sq 1:3:4:5:3-sq:5-sq 1:3:4:6:1-sq:4-sq 1:3:2-sq:4-sq:5-sq:6-sq 1:4:1-sq:3-sq:4-sq:6-sq 2:4:1-sq:2-sq:5-sq:6-sq 3:4:6:3-sq:4-sq:6-sq 3:5:1-sq:2-sq:3-sq:5-sq 4:5:1-sq:2-sq:3-sq:5-sq  & \\ 
     \hline
     \multicolumn{3}{p{\textwidth}}{\footnotesize{\emph{Note}: Numbers 1 to 6 refer to 6 tract-level control variables: the unemployment rate(1), the log of mean family income(2), the fractions of vacant(3), renter-occupied housing units(4), and single-unit(5), and the fraction of workers who use public transport to travel to work(6). Notation `-sq' stands for the squared variable. The colon (:) denotes interaction between covariates. For example, 1:2 stands for interaction between the unemployment rate and the log of the mean family income. 
  }}
\end{tabular}
\end{table}

\begin{table}[htbp]
\footnotesize
\centering
\caption{Selected Covariates $\gm=0.50$}\label{tb:sel-0.50}
\begin{tabular}{lp{10cm}l}
  \\
  \hline
 & Selected Regressors &  \\ 
  \hline
    \underline{6 control variables}\\

No Interaction         & 1 3 6                                    \\
Two-way Interaction    & 1:3 3:5 4:5 5:6                          \\
Three-way Interaction  & 4:5 1:3:5 2:3:5 2:4:6 2:5:6 3:4:5        \\
Four-way Interaction   & 4:5 5:6 2:3:5 1:3:4:5 1:4:5:6 2:3:4:6    \\
Five-way Interaction   & 4:5 5:6 2:3:5 1:4:5:6 2:3:4:6 1:2:3:4:5  \\
Six-way Interaction    & 4:5 5:6 2:3:5 1:4:5:6 2:3:4:6 1:2:3:4:5  \\
\\
    \underline{12 control variables}\\
No Interaction         &  1 3 6 5-sq \\
Two-way Interaction    &  1 1:5 3:5-sq 4:5 5:6 6:4-sq 3-sq:5-sq \\
Three-way Interaction  & 1 5 1:5 3:5:5-sq 5:6:2-sq 5:1-sq:6-sq 5:4-sq:6-sq 1-sq:3-sq:5-sq 1-sq:5-sq:6-sq 2-sq:3-sq:5-sq 4-sq:5-sq:6-sq \\
Four-way Interaction   & 1 5 5:6:2-sq 1:4:5-sq:6-sq 1:5:6:6-sq 1:3-sq:4-sq:5-sq 2:2-sq:3-sq:5-sq 3:5:2-sq:5-sq 3:1-sq:4-sq:6-sq 4:5:6:6-sq 5:1-sq:2-sq:6-sq 5:1-sq:3-sq:5-sq 1-sq:2-sq:5-sq:6-sq \\
Five-way Interaction   & 1 5 2:2-sq:3-sq:5-sq 5:1-sq:3-sq:5-sq 1:4:2-sq:5-sq:6-sq 1:5:3-sq:4-sq:5-sq 2:3:5:2-sq:5-sq 2:5:1-sq:2-sq:6-sq 2:5:2-sq:3-sq:5-sq 2:1-sq:2-sq:5-sq:6-sq 3:4:1-sq:4-sq:6-sq 3:5:2-sq:5-sq:6-sq 3:6:1-sq:4-sq:6-sq 4:5:6:2-sq:6-sq 4:1-sq:2-sq:5-sq:6-sq \\
Six-way Interaction    & 1 5 2:2-sq:3-sq:5-sq 5:1-sq:3-sq:5-sq 1:5:3-sq:4-sq:5-sq 2:3:5:2-sq:5-sq 2:5:1-sq:2-sq:6-sq 2:5:2-sq:3-sq:5-sq 2:1-sq:2-sq:5-sq:6-sq 1:3:1-sq:3-sq:5-sq:6-sq 1:3:2-sq:4-sq:5-sq:6-sq 2:3:5:2-sq:5-sq:6-sq 2:4:1-sq:2-sq:5-sq:6-sq 3:4:6:1-sq:4-sq:6-sq \\
     \hline
     \multicolumn{3}{p{\textwidth}}{\footnotesize{\emph{Note}: Numbers 1 to 6 refer to 6 tract-level control variables: the unemployment rate(1), the log of mean family income(2), the fractions of vacant(3), renter-occupied housing units(4), and single-unit(5), and the fraction of workers who use public transport to travel to work(6). Notation `-sq' stands for the squared variable. The colon (:) denotes interaction between covariates. For example, 1:2 stands for interaction between the unemployment rate and the log of the mean family income. 
  }}
\end{tabular}
\end{table}

\begin{table}[htbp]
\footnotesize
\centering
\caption{Selected Covariates $\gm=0.75$}\label{tb:sel-0.75}
\begin{tabular}{lp{10cm}l}
  \\
  \hline
 & Selected Covariates &\\
 \hline
    \underline{6 control variables}\\

No Interaction         & 3 5 \\
Two-way Interaction    & 1:3 3:5 4:5 4:6 5:6 \\
Three-way Interaction  & 1:2 4:5 1:3:5 1:5:6 2:3:5 2:4:5 2:5:6 3:4:5 3:4:6 3:5:6 \\
Four-way Interaction   & 1:2 4:5 2:3:5 2:5:6 3:5:6 1:3:4:5 1:4:5:6 2:3:4:6 \\
Five-way Interaction   & 1:2 4:5 2:3:5 2:5:6 3:5:6 1:4:5:6 2:3:4:6 1:2:3:4:5 \\
Six-way Interaction    & 1:2 4:5 2:3:5 2:5:6 3:5:6 1:4:5:6 2:3:4:6 1:2:3:4:5   \\
\\
    \underline{12 control variables}\\
No Interaction         & 1 3 5-sq \\
Two-way Interaction    &  1:2 1:3 1:1-sq 3:5-sq 4:5 5:6 6:6-sq 3-sq:4-sq 3-sq:5-sq 4-sq:6-sq \\
Three-way Interaction  & 1 3:5:5-sq 3:1-sq:5-sq 3:2-sq:5-sq 3:5-sq:6-sq 4:3-sq:4-sq 5:6:2-sq 5:1-sq:6-sq 5:4-sq:6-sq 1-sq:3-sq:5-sq 1-sq:5-sq:6-sq 2-sq:3-sq:5-sq \\
Four-way Interaction   & 1 1-sq:3-sq:5-sq 1:3:4:5 1:3:4-sq:5-sq 1:4:5-sq:6-sq 2:3:2-sq:5-sq 2:5:6:2-sq 2:2-sq:3-sq:5-sq 3:4:2-sq:4-sq 3:5:2-sq:5-sq 3:5:5-sq:6-sq 3:1-sq:4-sq:6-sq 4:1-sq:5-sq:6-sq 5:1-sq:3-sq:5-sq 1-sq:2-sq:5-sq:6-sq \\
Five-way Interaction   & 1 1-sq:3-sq:5-sq 1:3:4:5 1:3:4-sq:5-sq 1:4:5-sq:6-sq 2:3:2-sq:5-sq 2:5:6:2-sq 2:2-sq:3-sq:5-sq 3:4:2-sq:4-sq 3:5:2-sq:5-sq 3:5:5-sq:6-sq 3:1-sq:4-sq:6-sq 4:1-sq:5-sq:6-sq 5:1-sq:3-sq:5-sq 1-sq:2-sq:5-sq:6-sq \\
Six-way Interaction    & 1 5:6 1-sq:3-sq:5-sq 1:3:4:5 2:3:2-sq:5-sq 2:2-sq:3-sq:5-sq 5:1-sq:3-sq:5-sq 2:3:4:2-sq:4-sq 2:3:5:2-sq:5-sq 2:1-sq:2-sq:5-sq:6-sq 2:3:5:2-sq:5-sq:6-sq 2:4:5:1-sq:3-sq:5-sq 2:4:1-sq:2-sq:5-sq:6-sq 3:4:5:1-sq:2-sq:5-sq 4:6:1-sq:3-sq:4-sq:6-sq \\
     \hline
     \multicolumn{3}{p{\textwidth}}{\footnotesize{\emph{Note}: Numbers 1 to 6 refer to 6 tract-level control variables: the unemployment rate(1), the log of mean family income(2), the fractions of vacant(3), renter-occupied housing units(4), and single-unit(5), and the fraction of workers who use public transport to travel to work(6). Notation `-sq' stands for the squared variable. The colon (:) denotes interaction between covariates. For example, 1:2 stands for interaction between the unemployment rate and the log of the mean family income. 
  }}
\end{tabular}
\end{table}


\begin{table}[htbp]
\footnotesize
\centering
\caption{Full Estimation Results  from  Mean Regression (Untrimmed Data)}\label{tb:Chicago-new}
\begin{tabular}{lccccc}
  \\
  \hline
 & \multirow{2}{*}{No.\  of Reg.\ } & No.\  of    & \multirow{2}{*}{$\widehat{\tau}$} &  \multirow{2}{*}{CI  for $\tau_0$} &\multirow{2}{*}{$\widehat{\delta}$} \\ 
  & & Selected Reg.\ &    &  \\ 
  \hline
    \underline{6 control variables}\\

No Interaction          &    26        &        25  &  3.25  &  NA  &  -21.53 \\
Two-way Interaction     &    41        &        34  &  3.25  &  NA  &  -17.39 \\
Three-way Interaction   &    61        &        40  &  3.25  &  NA  &  -16.60 \\
Four-way Interaction    &    76        &        50  &  3.25  &  NA  &  -15.80 \\
Five-way Interaction    &    82        &        49  &  3.25  &  NA  &  -16.12 \\
Six-way Interaction     &    83        &        50  &  3.25  &  NA  &  -16.14 \\

\\
    \underline{12 control variables}\\
No Interaction           &    32       &         29  &  3.25  &  NA  &  -19.94 \\
Two-way Interaction      &    98       &         54  &  3.25  &  NA  &  -15.27 \\
Three-way Interaction    &   318       &         80  &  3.25  &  NA  &  -15.33 \\
Four-way Interaction     &   813       &        103  &  3.25  &  NA  &  -15.16 \\
Five-way Interaction     &  1605       &        129  &  3.25  &  NA  &  -15.27 \\
Six-way Interaction      &  2529       &        142  &  3.25  &  NA  &  -15.55 \\

     \hline
     \multicolumn{6}{p{\textwidth}}{\footnotesize{\emph{Note}: The sample size of untrimmed data is $n=1,813$. The parameter ${\tau}_0$ is estimated by the grid search on the 591 equi-spaced points over $[1,60]$. As in the simulation studies, the tuning parameters are set from Step 1  in median regression.
  }}
\end{tabular}
\end{table}

\begin{table}[htbp]
\footnotesize
\centering
\caption{Full Estimation Results  from  Mean Regression (Trimmed Data)}\label{tb:Chicago-new-trimmed}
\begin{tabular}{lccccc}
  \\
  \hline
 & \multirow{2}{*}{No.\  of Reg.\ } & No.\  of Selected Reg.\   & \multirow{2}{*}{$\widehat{\tau}$} &  \multirow{2}{*}{CI  for $\tau_0$} &\multirow{2}{*}{$\widehat{\delta}$} \\ 
  & & in Step 3b &    &  \\ 
  \hline
    \underline{6 control variables}\\

No Interaction           &    26        &        24  &  3.35 &   NA  &   -6.32 \\
Two-way Interaction      &    41        &        32  &  3.25 &   NA  &   -6.00 \\
Three-way Interaction    &    61        &        37  &  3.25 &   NA  &   -7.01 \\
Four-way Interaction     &    76        &        35  &  3.25 &   NA  &   -6.68 \\
Five-way Interaction     &    82        &        41  &  3.25 &   NA  &   -6.59 \\
Six-way Interaction      &    83        &        41  &  3.25 &   NA  &   -6.53 \\

\\
    \underline{12 control variables}\\
No Interaction           &     32    &             30  &   3.35  &   NA  &    -4.27 \\
Two-way Interaction      &     98    &             48  &   3.35  &   NA  &    -4.28 \\
Three-way Interaction    &    318    &             63  &   3.25  &   NA  &    -5.02 \\
Four-way Interaction     &    813    &             90  &   3.25  &   NA  &    -5.18 \\
Five-way Interaction     &   1605    &             88  &   3.25  &   NA  &    -5.27 \\
Six-way Interaction      &   2529    &            107  &   3.25  &   NA  &    -5.19 \\

     \hline
     \multicolumn{6}{p{\textwidth}}{\footnotesize{\emph{Note}: The trimmed data drop top and bottom $5\%$ observations based on $\{Y_i\}$ and the sample sizes decreases to $n=1,626$.  The parameter ${\tau}_0$ is estimated by the grid search on the 591 equi-spaced points over $[1,60]$. As in the simulation studies, the tuning parameters are set from Step 1  in median regression.
  }}
\end{tabular}
\end{table}

\newpage

\singlespacing

\bibliographystyle{ims}
\bibliography{liaoBib-SL}

\begin{thebibliography}{43}
\expandafter\ifx\csname natexlab\endcsname\relax\def\natexlab#1{#1}\fi
\expandafter\ifx\csname url\endcsname\relax
  \def\url#1{\texttt{#1}}\fi
\expandafter\ifx\csname urlprefix\endcsname\relax\def\urlprefix{URL }\fi

\bibitem[{Belloni and Chernozhukov(2011)}]{BC11}
\textsc{Belloni, A.} and \textsc{Chernozhukov, V.} (2011).
\newblock {$\ell_1$}-penalized quantile regression in high dimensional sparse
  models.
\newblock \textit{Annals of Statistics} \textbf{39} 82--130.

\bibitem[{Bickel et~al.(2009)Bickel, Ritov and Tsybakov}]{Bickeletal}
\textsc{Bickel, P.}, \textsc{Ritov, Y.} and \textsc{Tsybakov, A.} (2009).
\newblock Simultaneous analysis of {Lasso} and {Dantzig} selector.
\newblock \textit{Annals of Statistics} \textbf{37} 1705--1732.

\bibitem[{Bradic et~al.(2011)Bradic, Fan and Wang}]{Bradic2011}
\textsc{Bradic, J.}, \textsc{Fan, J.} and \textsc{Wang, W.} (2011).
\newblock Penalized composite quasi-likelihood for ultrahigh dimensional
  variable selection.
\newblock \textit{Journal of the Royal Statistical Society: Series B
  (Statistical Methodology)} \textbf{73} 325--349.

\bibitem[{B\"{u}hlmann and van~de Geer(2011)}]{bulmann}
\textsc{B\"{u}hlmann, P.} and \textsc{van~de Geer, S.} (2011).
\newblock \textit{Statistics for high-dimensional data, methods, theory and
  applications}.
\newblock Springer, New York.

\bibitem[{Callot et~al.(2016)Callot, Caner, Kock and
  Riquelme}]{Callot:et-al:2016}
\textsc{Callot, L.}, \textsc{Caner, M.}, \textsc{Kock, A.~B.} and
  \textsc{Riquelme, J.~A.} (2016).
\newblock Sharp threshold detection based on sup-norm error rates in
  high-dimensional models.
\newblock \textit{Journal of Business \& Economic Statistics} , forthcoming.

\bibitem[{Card et~al.(2008)Card, Mas and Rothstein}]{card2008tipping}
\textsc{Card, D.}, \textsc{Mas, A.} and \textsc{Rothstein, J.} (2008).
\newblock Tipping and the dynamics of segregation.
\newblock \textit{Quarterly Journal of Economics}  177--218.

\bibitem[{Chan(1993)}]{chan1993consistency}
\textsc{Chan, K.-S.} (1993).
\newblock Consistency and limiting distribution of the least squares estimator
  of a threshold autoregressive model.
\newblock \textit{Annals of Statistics} \textbf{21} 520--533.

\bibitem[{Chan et~al.(2016)Chan, Ing, Li and Yau}]{Chan:et-al:2016}
\textsc{Chan, N.~H.}, \textsc{Ing, C.-K.}, \textsc{Li, Y.} and \textsc{Yau,
  C.~Y.} (2016).
\newblock Threshold estimation via group orthogonal greedy algorithm.
\newblock \textit{Journal of Business \& Economic Statistics} , forthcoming.

\bibitem[{Chan et~al.(2014)Chan, Yau and Zhang}]{chan2013group}
\textsc{Chan, N.~H.}, \textsc{Yau, C.~Y.} and \textsc{Zhang, R.-M.} (2014).
\newblock Group {LASSO} for structural break time series.
\newblock \textit{Journal of the American Statistical Association} \textbf{109}
  590--599.

\bibitem[{Cho and Fryzlewicz(2015)}]{cho2012multiple}
\textsc{Cho, H.} and \textsc{Fryzlewicz, P.} (2015).
\newblock Multiple-change-point detection for high dimensional time series via
  sparsified binary segmentation.
\newblock \textit{Journal of the Royal Statistical Society: Series B
  (Statistical Methodology)} \textbf{77} 475--507.

\bibitem[{Ciuperca(2013)}]{Ciuperca:13}
\textsc{Ciuperca, G.} (2013).
\newblock Quantile regression in high-dimension with breaking.
\newblock \textit{Journal of Statistical Theory and Applications} \textbf{12}
  288--305.

\bibitem[{Enikeeva and Harchaoui(2013)}]{enikeeva2013high}
\textsc{Enikeeva, F.} and \textsc{Harchaoui, Z.} (2013).
\newblock High-dimensional change-point detection with sparse alternatives.
\newblock \textit{arXiv preprint} \url{http://arxiv.org/abs/1312.1900}.

\bibitem[{Fan et~al.(2014)Fan, Fan and Barut}]{FFB}
\textsc{Fan, J.}, \textsc{Fan, Y.} and \textsc{Barut, E.} (2014).
\newblock Adaptive robust variable selection.
\newblock \textit{Annals of Statistics} \textbf{42} 324--351.

\bibitem[{Fan and Li(2001)}]{Fan01}
\textsc{Fan, J.} and \textsc{Li, R.} (2001).
\newblock Variable selection via nonconcave penalized likelihood and its oracle
  properties.
\newblock \textit{Journal of the American Statistical Association} \textbf{96}
  1348--1360.

\bibitem[{Fan and Lv(2011)}]{FL11}
\textsc{Fan, J.} and \textsc{Lv, J.} (2011).
\newblock Non-concave penalized likelihood with np-dimensionality.
\newblock \textit{IEEE Transactions on Information Theory} \textbf{57}
  5467--5484.

\bibitem[{Frick et~al.(2014)Frick, Munk and Sieling}]{Frick-et-al:14}
\textsc{Frick, K.}, \textsc{Munk, A.} and \textsc{Sieling, H.} (2014).
\newblock Multiscale change point inference.
\newblock \textit{Journal of the Royal Statistical Society: Series B
  (Statistical Methodology)} \textbf{76} 495--580.

\bibitem[{Hansen(1996)}]{Hansen:1996}
\textsc{Hansen, B.~E.} (1996).
\newblock Inference when a nuisance parameter is not identified under the null
  hypothesis.
\newblock \textit{Econometrica} \textbf{64} 413--430.

\bibitem[{Hansen(2000)}]{hansen2000sample}
\textsc{Hansen, B.~E.} (2000).
\newblock Sample splitting and threshold estimation.
\newblock \textit{Econometrica} \textbf{68} 575--603.

\bibitem[{He and Shao(2000)}]{He:Shao:00}
\textsc{He, X.} and \textsc{Shao, Q.-M.} (2000).
\newblock On parameters of increasing dimensions.
\newblock \textit{Journal of Multivariate Analysis} \textbf{73} 120--135.

\bibitem[{Koenker(2016)}]{quantreg}
\textsc{Koenker, R.} (2016).
\newblock \textit{quantreg: Quantile Regression}.
\newblock R package version 5.29.
\newline\urlprefix\url{https://CRAN.R-project.org/package=quantreg}

\bibitem[{Koenker and Bassett(1978)}]{Koenker:Bassett:1978}
\textsc{Koenker, R.} and \textsc{Bassett, G.} (1978).
\newblock Regression quantiles.
\newblock \textit{Econometrica}  33--50.

\bibitem[{Koenker and Mizera(2014)}]{Koenker:Mizera:14}
\textsc{Koenker, R.} and \textsc{Mizera, I.} (2014).
\newblock Convex optimization in {R}.
\newblock \textit{Journal of Statistical Software} \textbf{60} 1--23.

\bibitem[{Kosorok(2008)}]{Kosorok}
\textsc{Kosorok, M.~R.} (2008).
\newblock \textit{Introduction to Empirical Processes and Semiparametric
  Inference}.
\newblock Springer, New York.

\bibitem[{Kosorok and Song(2007)}]{kosorok2007}
\textsc{Kosorok, M.~R.} and \textsc{Song, R.} (2007).
\newblock Inference under right censoring for transformation models with a
  change-point based on a covariate threshold.
\newblock \textit{Annals of Statistics} \textbf{35} 957--989.

\bibitem[{Lee and Seo(2008)}]{Lee:Seo:08}
\textsc{Lee, S.} and \textsc{Seo, M.~H.} (2008).
\newblock Semiparametric estimation of a binary response model with a
  change-point due to a covariate threshold.
\newblock \textit{Journal of Econometrics} \textbf{144} 492--499.

\bibitem[{Lee et~al.(2011)Lee, Seo and Shin}]{lee2012testing}
\textsc{Lee, S.}, \textsc{Seo, M.~H.} and \textsc{Shin, Y.} (2011).
\newblock Testing for threshold effects in regression models.
\newblock \textit{Journal of the American Statistical Association} \textbf{106}
  220--231.

\bibitem[{Lee et~al.(2016)Lee, Seo and Shin}]{lee2012lasso}
\textsc{Lee, S.}, \textsc{Seo, M.~H.} and \textsc{Shin, Y.} (2016).
\newblock The lasso for high dimensional regression with a possible change
  point.
\newblock \textit{Journal of the Royal Statistical Society: Series B
  (Statistical Methodology)} \textbf{78} 193--210.

\bibitem[{Leonardi and B{\"u}hlmann(2016)}]{leonardi2016}
\textsc{Leonardi, F.} and \textsc{B{\"u}hlmann, P.} (2016).
\newblock Computationally efficient change point detection for high-dimensional
  regression.
\newblock \textit{arXiv preprint arXiv:1601.03704}
  \url{http://arxiv.org/abs/1601.03704}.

\bibitem[{Li and Ling(2012)}]{Li:Ling:12}
\textsc{Li, D.} and \textsc{Ling, S.} (2012).
\newblock On the least squares estimation of multiple-regime threshold
  autoregressive models.
\newblock \textit{Journal of Econometrics} \textbf{167} 240--253.

\bibitem[{Loh and Wainwright(2013)}]{loh2013regularized}
\textsc{Loh, P.-L.} and \textsc{Wainwright, M.~J.} (2013).
\newblock Regularized {$M$}-estimators with nonconvexity: Statistical and
  algorithmic theory for local optima.
\newblock In \textit{Advances in Neural Information Processing Systems 26}
  (C.~Burges, L.~Bottou, M.~Welling, Z.~Ghahramani and K.~Weinberger, eds.).
  Curran Associates, Inc., 476--484.

\bibitem[{Lov\'{a}sz and Vempala(2007)}]{logcave:07}
\textsc{Lov\'{a}sz, L.} and \textsc{Vempala, S.} (2007).
\newblock The geometry of logconcave functions and sampling algorithms.
\newblock \textit{Random Structures \& Algorithms} \textbf{30} 307--358.

\bibitem[{Negahban et~al.(2012)Negahban, Ravikumar, Wainwright and
  Yu}]{negahban2012}
\textsc{Negahban, S.~N.}, \textsc{Ravikumar, P.}, \textsc{Wainwright, M.~J.}
  and \textsc{Yu, B.} (2012).
\newblock A unified framework for high-dimensional analysis of {$M$}-estimators
  with decomposable regularizers.
\newblock \textit{Statistical Science} \textbf{27} 538--557.

\bibitem[{Pons(2003)}]{Pons2003}
\textsc{Pons, O.} (2003).
\newblock Estimation in a {Cox} regression model with a change-point according
  to a threshold in a covariate.
\newblock \textit{Annals of Statistics} \textbf{31} 442--463.

\bibitem[{Raskutti et~al.(2011)Raskutti, Wainwright and Yu}]{Rasku09}
\textsc{Raskutti, G.}, \textsc{Wainwright, M.} and \textsc{Yu, B.} (2011).
\newblock Minimax rates of estimation for high-dimensional linear regression
  over {$\ell_q$}-balls.
\newblock \textit{IEEE Transactions on Information Theory} \textbf{57}
  6976--6994.

\bibitem[{Seijo and Sen(2011{\natexlab{a}})}]{Seijo:Sen:11a}
\textsc{Seijo, E.} and \textsc{Sen, B.} (2011{\natexlab{a}}).
\newblock Change-point in stochastic design regression and the bootstrap.
\newblock \textit{Annals of Statistics} \textbf{39} 1580--1607.

\bibitem[{Seijo and Sen(2011{\natexlab{b}})}]{Seijo:Sen:11b}
\textsc{Seijo, E.} and \textsc{Sen, B.} (2011{\natexlab{b}}).
\newblock A continuous mapping theorem for the smallest argmax functional.
\newblock \textit{Electronic Journal of Statistics} \textbf{5} 421--439.

\bibitem[{Tong(1990)}]{tong1990non}
\textsc{Tong, H.} (1990).
\newblock \textit{Non-linear time series: a dynamical system approach}.
\newblock Oxford University Press.

\bibitem[{van~de Geer(2008)}]{geer}
\textsc{van~de Geer, S.~A.} (2008).
\newblock High-dimensional generalized linear models and the lasso.
\newblock \textit{Annals of Statistics} \textbf{36} 614--645.

\bibitem[{van~der Vaart and Wellner(1996)}]{VW}
\textsc{van~der Vaart, A.} and \textsc{Wellner, J.} (1996).
\newblock \textit{Weak convergence and empirical processes}.
\newblock Springer, New York.

\bibitem[{Wang(2013)}]{Wang:2013}
\textsc{Wang, L.} (2013).
\newblock The {$L_1$} penalized {LAD} estimator for high dimensional linear
  regression.
\newblock \textit{Journal of Multivariate Analysis} \textbf{120} 135--151.

\bibitem[{Wang et~al.(2012)Wang, Wu and Li}]{Wang:Wu:Li:2012}
\textsc{Wang, L.}, \textsc{Wu, Y.} and \textsc{Li, R.} (2012).
\newblock Quantile regression for analyzing heterogeneity in ultra-high
  dimension.
\newblock \textit{Journal of the American Statistical Association} \textbf{107}
  214--222.

\bibitem[{Zhao and Yu(2006)}]{zhao2006model}
\textsc{Zhao, P.} and \textsc{Yu, B.} (2006).
\newblock On model selection consistency of lasso.
\newblock \textit{Journal of Machine Learning Research} \textbf{7} 2541--2563.

\bibitem[{Zou and Li(2008)}]{zouli}
\textsc{Zou, H.} and \textsc{Li, R.} (2008).
\newblock One-step sparse estimations in non concave penalized likelihood
  models.
\newblock \textit{Annals of Statistics} \textbf{36} 1509--1533.

\end{thebibliography}

\end{document}